\documentclass[a4paper,10pt]{fullarticle}
\usepackage[british]{babel}
\usepackage{csquotes}
\usepackage{sciencestuff}
\usepackage[most]{tcolorbox}
\usepackage{tikz}
\usetikzlibrary{decorations.pathmorphing,
                decorations.markings,
                arrows,
                arrows.meta,
                intersections}
\usepackage{pgfplots}

\pgfplotsset{compat=1.18}
\usepgfplotslibrary{fillbetween}
\usepackage[algoruled,algosection,vlined,shortend,linesnumbered]{algorithm2e}

\AddToHook{cmd/appendix/before}{%
  \crefalias{section}{appendix}%
  \crefalias{subsection}{appendix}%
}

\input{arxiv.def}
\title{%
    Data Field Theory \\ \relax
    {\Large Theory and Applications of the Functional Renormalization Group for Signal Detection}}

\author[1]{Riccardo Finotello\emailfoot{riccardo.finotello@cea.fr}}
\author[2]{Vincent Lahoche\emailfoot{vincent.lahoche@cea.fr}}
\author[3]{Dine Ousmane Samary\emailfoot{dine.ousmanesamary@uac.bj}}
\author[2]{Parham Radpay\emailfoot{parham.radpay@cea.fr}}

\affil[1]{%
  Université Paris Saclay, \textsc{Cea}, Service de Génie Logiciel et de Simulation (\textsc{Sgls}),
  \protect \\
  Gif-sur-Yvette, F-91191, France
}
\affil[2]{%
  Université Paris-Saclay, \textsc{Cea}, 
  \protect \\
  Palaiseau, F-91120, France
}
\affil[3]{%
    Faculté des Sciences et Techniques (ICMPA-UNESCO Chair)
    \protect \\
    Université d'Abomey-Calavi, 072 BP 50, Benin
}

\date{}

\hypersetup{%
  pdfauthor={Riccardo Finotello, Vincent Lahoche, Dine Ousmane Samary, Parham Radpay},
  pdftitle={Data Field Theory: Theory and Applications of the Functional Renormalization Group for Signal Detection},
  pdfkeywords={%
    functional renormalization group,
    theoretical physics,
    data science,
    signal analysis,
    signal detection,
    random matrix theory
    },
  pdfsubject={renormalization group for signal detection}
}

\colorlet{thmframe}{teal!75!black}
\colorlet{thmbody}{teal!5!white}
\colorlet{thmtitlebg}{teal!85!black}
\newtcbtheorem[number within=section]{theorem}{Theorem}{%
    enhanced,
    breakable,
    colback=thmbody,
    colframe=thmframe,
    colbacktitle=thmtitlebg,
    coltitle=white,
    fonttitle=\bfseries,
    boxed title style={%
        boxrule=0pt,
        colframe=thmtitlebg,
        sharp corners
    },
    grow to left by=-5mm,
    grow to right by=-5mm,
    arc=2mm,
    boxrule=0.7pt,
    left=5mm,
    right=5mm,
    top=4.5mm,
    bottom=4.5mm,
    drop fuzzy shadow=black!45,
}{thm}
\makeatletter
\newcommand{\tcb@cnt@theoremautorefname}{Theorem}
\makeatother
\crefname{tcb@cnt@theorem}{theorem}{theorems}
\Crefname{tcb@cnt@theorem}{Theorem}{Theorems}

\colorlet{defframe}{orange!75!black}
\colorlet{defbody}{orange!5!white}
\colorlet{deftitlebg}{orange!85!black}
\newtcbtheorem[number within=section]{definition}{Definition}{%
    enhanced,
    breakable,
    colback=defbody,
    colframe=defframe,
    colbacktitle=deftitlebg,
    coltitle=white,
    fonttitle=\bfseries,
    boxed title style={%
        boxrule=0pt,
        colframe=deftitlebg,
        sharp corners
    },
    grow to left by=-5mm,
    grow to right by=-5mm,
    arc=2mm,
    boxrule=0.7pt,
    left=5mm,
    right=5mm,
    top=4.5mm,
    bottom=4.5mm,
    drop fuzzy shadow=black!45,
}{def}
\makeatletter
\newcommand{\tcb@cnt@definitionautorefname}{Definition}
\makeatother
\crefname{tcb@cnt@definition}{definition}{definitions}
\Crefname{tcb@cnt@definition}{Definition}{Definitions}

\colorlet{propframe}{blue!75!black}
\colorlet{propbody}{blue!5!white}
\colorlet{proptitlebg}{blue!85!black}
\newtcbtheorem[number within=section]{proposition}{Proposition}{%
    enhanced,
    breakable,
    colback=propbody,
    colframe=propframe,
    colbacktitle=proptitlebg,
    coltitle=white,
    fonttitle=\bfseries,
    boxed title style={%
        boxrule=0pt,
        colframe=proptitlebg,
        sharp corners
    },
    grow to left by=-5mm,
    grow to right by=-5mm,
    arc=2mm,
    boxrule=0.7pt,
    left=5mm,
    right=5mm,
    top=4.5mm,
    bottom=4.5mm,
    drop fuzzy shadow=black!45,
}{prop}
\makeatletter
\newcommand{\tcb@cnt@propositionautorefname}{Proposition}
\makeatother
\crefname{tcb@cnt@proposition}{proposition}{propositions}
\Crefname{tcb@cnt@proposition}{Proposition}{Propositions}

\colorlet{remarkframe}{red!75!black}
\colorlet{remarkbody}{red!5!white}
\colorlet{remarktitlebg}{red!85!black}
\newtcbtheorem[number within=section]{remark}{Remark}{%
    enhanced,
    breakable,
    colback=remarkbody,
    colframe=remarkframe,
    colbacktitle=remarktitlebg,
    coltitle=white,
    fonttitle=\bfseries,
    boxed title style={%
        boxrule=0pt,
        colframe=remarktitlebg,
        sharp corners
    },
    grow to left by=-5mm,
    grow to right by=-5mm,
    arc=2mm,
    boxrule=0.7pt,
    left=5mm,
    right=5mm,
    top=4.5mm,
    bottom=4.5mm,
    drop fuzzy shadow=black!45,
}{rem}
\makeatletter
\newcommand{\tcb@cnt@remarkautorefname}{Remark}
\makeatother
\crefname{tcb@cnt@remark}{remark}{remarks}
\Crefname{tcb@cnt@remark}{Remark}{Remarks}

\tcolorboxenvironment{proof}{%
    enhanced,
    breakable,
    arc=2mm,
    grow to left by=-5mm,
    grow to right by=-5mm,
    left=5mm,
    right=5mm,
    top=2mm,
    bottom=2mm,
    colback=gray!5,
    colframe=gray!60,
    fontupper=\footnotesize,
    drop fuzzy shadow=black!45,
}

\crefname{algocf}{algorithm}{algorithms}
\Crefname{algocf}{Algorithm}{Algorithms}

\newcommand{\rg}{\textsc{rg}\xspace}  
\newcommand{\hsi}{\textsc{hsi}\xspace}  
\newcommand{\ecm}{\textsc{ecm}\xspace}  
\newcommand{\rmt}{\textsc{rmt}\xspace}  
\newcommand{\dft}{\textsc{dft}\xspace}  
\newcommand{\mpdistr}{\textsc{MP}\xspace}  
\newcommand{\onepi}{\textsc{1PI}\xspace}  
\newcommand{\lod}{\textsc{lod}\xspace}  
\newcommand{\eft}{\textsc{eft}\xspace}  
\newcommand{\ir}{\textsc{ir}\xspace}  
\newcommand{\uv}{\textsc{uv}\xspace}  
\newcommand{\lpa}{\textsc{lpa}\xspace}  
\newcommand{\eaa}{\textsc{eaa}\xspace}  
\newcommand{\kl}{\textsc{KL}\xspace}  
\newcommand{\kde}{\textsc{kde}\xspace}  
\newcommand{\iid}{i.i.d.\xspace}  
\newcommand{\frg}{\textsc{frg}\xspace}  
\newcommand{\bbp}{\textsc{BBP}\xspace}  
\newcommand{\rgb}{\textsc{rgb}\xspace}  
\newcommand{\mnf}{\textsc{mnf}\xspace}  
\newcommand{\gsa}{\textsc{gsa}\xspace}  
\newcommand{\wfr}{\textsc{wfr}\xspace}  
\newcommand{\ipr}{\textsc{ipr}\xspace}  
\newcommand{\wf}{\textsc{WF}\xspace}  
\newcommand{\clt}{\textsc{clt}\xspace}  
\newcommand{\ojk}{\textsc{OJK}\xspace}  
\newcommand{\goe}{\textsc{Goe}\xspace}  

\begin{document}

\maketitle

\begin{abstract}
    We review the renormalization group framework for signal detection in high-dimensional data, tailored to the regime where the signal may be of extensive rank and does not separate from the noise bulk as isolated spikes.
The framework provides a conceptually simple criterion for distinguishing signal from noise within a quasi-continuous spectral region near a random-matrix universality class.
This scenario lies beyond the reach of standard methods such as the Baik--Ben Arous--P\'{e}ch\'{e} threshold, which requires eigenvalues to be cleanly separated from the bulk.
The renormalization group approach, by contrast, directly tracks spectral deformations and consistently yields a lower limit of detection without relying on spike separation.
We review results that identify the presence of a signal by testing the stability of the Gaussian fixed point of an effective field theory for the collective behaviour of the degrees of freedom in the spectral tail, where the signal resides.
We also discuss how the scale dependence of the canonical dimension, induced by the signal, manifests as a dimensional phase transition.

\end{abstract}

\highlights{%
    We pedagogically review theory and concrete use cases of the renormalization group for signal detection in data field theory.
}


\clearpage

\tableofcontents

\clearpage

\section{Introduction}\label{sec1}
As Artificial Intelligence (\ai) continues to advance, being able to analyse large quantities of data has become a necessary requirement to provide high-quality insights and build robust models for decision-making and scientific applications.
The detection of meaningful signals inside large datasets is a critical aspect of this process.
Today, large-scale data where information may be hidden in complex structures represents a major challenge, both theoretically and practically.
Understanding high-dimensional spaces and their specific characteristics, as well as implementing reliable and efficient algorithms, has indeed become an essential skill.
In this work, we review an original approach to this general problem, based on the Renormalization Group (\rg) and an analogous field theory model describing the effective behaviour of data.
This approach translates into the field of data analysis what physics, and in particular quantum field theory, was able to teach us.
The review aims at introducing the general aspects, as well as details, for both physicists and data analysts interested in techniques at the frontier of data analysis and theoretical physics.
\Cref{part1} introduces the necessary \rg concepts from first principles for readers unfamiliar with field-theoretic methods, while the application sections in \Cref{part3} are designed to be accessible independently of the technical developments in \Cref{part2}.

In what follows, data are represented as rectangular matrix $X$ of size $N \times P$, where $P$ denotes the dimensionality of the data and $N$ the number of samples.
We then have the following definition:
\begin{definition}{Empirical Correlation Matrix}{definitionECM}
    The \ecm $C$ is the $P \times P$ matrix with entries:
    \begin{equation}
        {C}_{ij} \defeq \frac{(C_0)_{ij}}{\sqrt{(C_0)_{ii} (C_0)_{jj}}},
        \label{corr}
    \end{equation}
    where $C_0$ is the covariance matrix:
    \begin{equation}
        C_0 \defeq \frac{X^{\mathsf{T}} X}{N}.
        \label{cov}
    \end{equation}
\end{definition}

Signal and anomaly detection deal with two main and difficult issues when trying to extract meaningful insights from the data:
\begin{enumerate}
    \item The study of the relevant prior governing the correlation structure in the dataset. Generally, the prior is the \textit{empirical correlation matrix} (\ecm).
    \item The identification of the distinction between \enquote{information} and \enquote{noise}, i.e., the challenging task of finding the relevant eigenvectors and eigenvalues of the \ecm.
\end{enumerate}
The definition \eqref{corr} discards high-variance contributions that could dominate the spectrum of the covariance matrix.
Several techniques in data analysis are able to further clean the signal from noise.
The well-known Principal Components Analysis (\pca) is an example~\cite{PCA1} and works by isolating the largest eigenvalues and the associated principal directions of the signal.
The interested reader can refer, for instance, to \Cref{AppD} for a review on one spike matrix model and the Baik-Ben Arous-P\'{e}ch\'{e} (\bbp) phase transition.
\Cref{fig1} illustrates the general structure of a typical experimental spectrum of eigenvalues, which is generally composed of two distinct parts:
\begin{enumerate}
    \item A quasi-continuous bulk, where the typical spacing $\delta$ between the eigenvalues satisfies $\min(N, P)^{-1} < \delta \ll \sqrt{\min(N, P)^{-1}}$ when the size $N$ of the correlation matrix is sufficiently large.
    \item Isolated spikes, which usually carry the most visible part of the information.
\end{enumerate}

\begin{figure}[t]
    \centering
    \begin{tikzpicture}
    \draw[->, thick] (0,0) -- (0,6.0) node[above] {$\mu(\lambda)$};
    \draw[->, thick] (0,0) -- (8.0,0) node[right] {$\lambda$};

    \draw[draw=black!80, fill=black!15, fill opacity=0.3]
    (0.49, 0.0)
    .. controls (0.92, 5.49) and (1.78, 5.02) .. (2.26, 4.41)
    .. controls (2.73, 3.80) and (2.81, 3.06) .. (3.05, 2.14)
    .. controls (3.29, 1.22) and (3.68, 0.13) .. (4.58, 0.0)
    -- (0.49, 0.0) -- cycle;
    \node[black, anchor=center] at (1.9, 2.00) {\textbf{Bulk}};

    \draw[red, thick, dashed]
    (0.49, 0.0)
    .. controls (0.94, 6.09) and (2.11, 6.71) .. (3.44, 0.0) node[text=red, below] {$\Lambda_{\text{cutoff}}$};
    \draw[->, red, thick] (4.00, 3.63) node[above, text=red, text width=3em] {\textbf{Noise model}} -- (4.00, 2.90) -- (2.80, 2.91);

    \draw[thin, densely dashed] (5.94, 0.24) ellipse (1.37 and 1.5);
    \node[black, anchor=center] at (5.94, 2.20) {\textbf{Spikes}};
    \draw[ultra thick, black!80] (5.23, 0.0) -- (5.23, 1.18);
    \draw[ultra thick, black!80] (5.41, 0.0) -- (5.41, 0.87);
    \draw[ultra thick, black!80] (5.86, 0.0) -- (5.86, 0.56);
    \draw[ultra thick, black!80] (6.29, 0.0) -- (6.29, 1.01);
\end{tikzpicture}
    \caption{%
        A typical experimental spectrum $\{\mu(\lambda)\}$, featuring a quasi-continuous bulk and a set of isolated spikes.
        The definition of \enquote{meaningful information}, and, specifically, the position of the cutoff $\Lambda_{\text{cutoff}}$, depends on the choice of the noise model.
    }
    \label{fig1}
\end{figure}
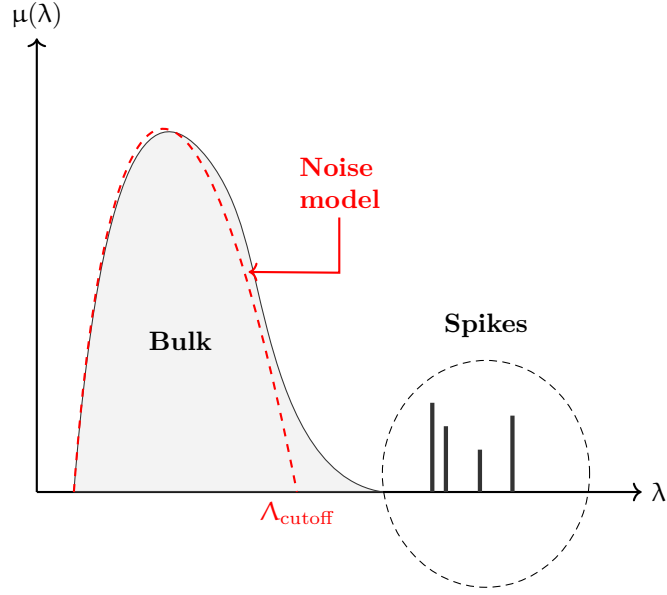

The general problem of isolating signal boils down to partitioning the spectrum into two parts: one that we call \enquote{noise}, which we seek to eliminate, and the other, \enquote{information}, which we need to identify and preserve.
However, there is no absolute definition of the latter.
There are, however, general principles: variables corresponding to information are generally correlated, whereas noise is essentially uncorrelated.
The definition of rank, that is establishing the number of redundancies in the data, is also a common criterion.
In simple terms, one could decide, for instance, to retain only the spikes, in which case the bulk of the spectrum would be treated entirely as noise.
Though effective, this approach has its practical limitations.
Most standard methods derived from \pca perform well when focusing on spikes but fail as soon as they encounter continuous spectra~\cite{kmean}.

A well-known example of this problem is provided by risk analysis in financial markets.
Famously, in the article~\cite{Bouchaud2}, the authors propose a method based on Random Matrix Theory (\rmt) to construct a decision boundary (i.e., choose a value of $\Lambda_{\text{cutoff}}$ in \Cref{fig1}) between relevant degrees-of-freedom (\dof), that is information, and irrelevant ones, or noise.
Their method is based on an \emph{analytical noise model} using the Marchenko-Pastur (\mpdistr) theorem:\footnote{%
    Historically, it was Wigner who first had the intuition that a large deterministic matrix could be replaced by a random matrix, subject only to the constraint of preserving the fundamental symmetries of the system~\cite{Wigner}.
    This departure from strict determinism revealed a profound phenomenon of universality: in highly complex systems, macroscopic statistical properties ultimately lose memory of the details of the original interaction.}
\begin{theorem}{Marchenko--Pastur Distribution}{thMP}
    Let $X$ be an $N \times P$ random matrix with \iid entries of variance $\sigma^2$.
    As $N, P \to \infty$ while keeping the ratio $q \defeq P/N$ fixed, the empirical eigenvalue distribution of the corresponding $P \times P$ random Wishart matrix $Z \defeq X^{\mathsf{T}} X / N$ converges weakly towards the \mpdistr distribution:
    \begin{equation}
        \mu_{\mpdistr}(\lambda)\defeq \frac{\sqrt{(\lambda_+-\lambda)(\lambda-\lambda_-)}}{2\pi \sigma^2 q \lambda},
        \label{MP}
    \end{equation}
    where $\lambda_\pm=\sigma^2 (1\pm \sqrt{q})^2$.
    Notice that we explicitly ignore the Dirac delta function at the origin in the case where $q > 1$.
    This contribution is irrelevant for our purposes since we focus our attention on the tail of the spectrum.
\end{theorem}
\Cref{numplot} illustrates the theorem for a Wishart matrix of size $10^4$.
The noise model considered thus rests essentially on an \emph{a priori} assumption on the weak correlation of noisy \dof which automatically brings a large-scale covariance matrix into the validity conditions of the theorem.

\begin{figure}[t]
    \centering
    \includegraphics[width=0.7\textwidth]{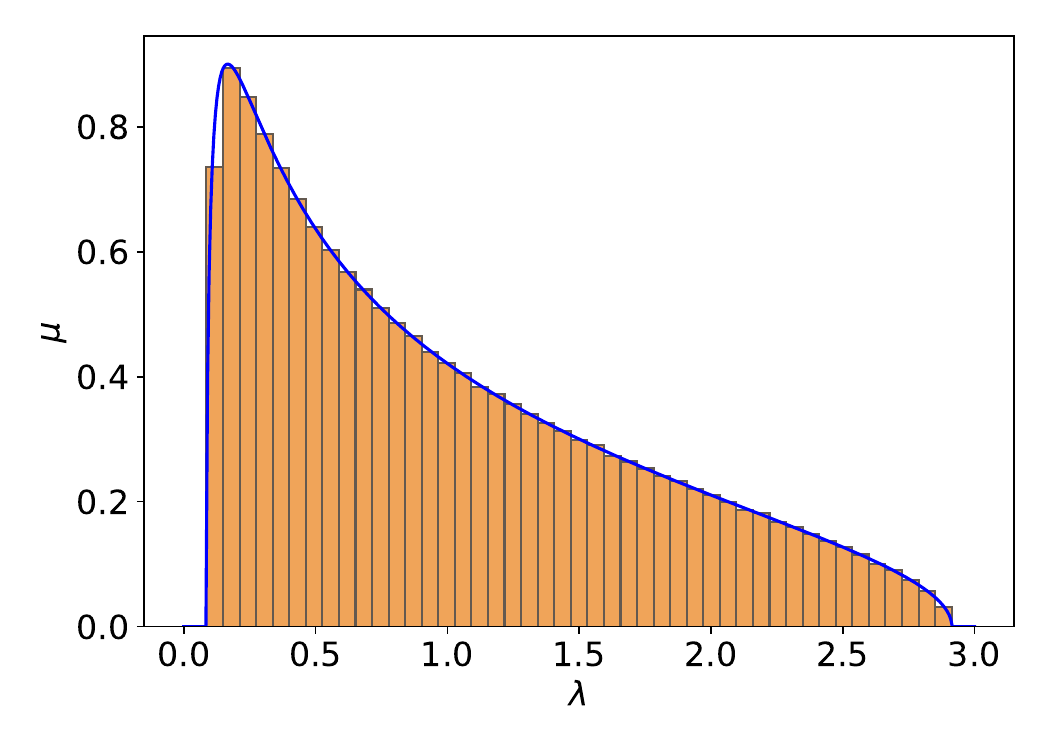}
    \caption{%
        Illustration of the convergence toward MP universal law.
        The histogram corresponds to the empirical spectrum of a white Wishart matrix with a sample size $4 \times 10^4$ and a ratio $q = 0.5$.
        The blue lines materialise the limit \mpdistr ($\mu_{\mpdistr}$) law.
    }
    \label{numplot}
\end{figure}

It should be noted, however, that while the original proof of \Cref{thm:thMP}~\cite{marvcenko1967distribution} assumes that the entries of the matrix $X$ are strictly \iid, other proofs have since relaxed this strict constraint and considered moderate, rapidly decaying correlations~\cite{bryson2021marchenko}.
In a sense, the universality class associated with the \mpdistr distribution covers a broader reality than that of the \iid hypothesis, which nonetheless seems so intuitive for defining noise.
This definition proves convenient, notably due to its quite broad universality, which one could say is quite helpful when dealing with data from different sources and of different nature.

The standard \bbp-based approach, while optimal for detecting isolated spikes that cleanly separate from the bulk, faces fundamental limitations when the signal rank is not negligible compared to the matrix dimensions.
As shown in the context of the extensive spike model\cite{Landau2023}, even moderate-rank signals (as low as $10$ for matrices of size $\sim 10^3$) deviate significantly from finite-rank predictions.
In this extensive regime, the signal does not manifest as a few isolated eigenvalues but as a collective deformation of the spectral tail, making the very notion of a \enquote{spike} ill-defined.
The \rg approach reviewed here circumvents this limitation.
Rather than searching for individual eigenvalues, it constructs an effective field theory whose propagator reproduces the full empirical spectral density.
In this framework, a pure-noise spectrum (the \mpdistr universality class) corresponds to a Gaussian fixed point with well-defined, scale-independent coupling dimensions.
The presence of a signal deforms the spectral tail, modifying the effective propagator and rendering the canonical dimensions scale-dependent.
We refer to this phenomenon as a dimensional phase transition, in analogy with the appearance of anomalous dimensions at critical points in statistical mechanics.
Tracking this scale dependence via the \rg flow equations reveals a sharp change in the relevance of the couplings at a characteristic scale, that is the point beyond which the spectrum can no longer be described by the pure-noise universality class.
This scale provides an objective cutoff $\Lambda_{\text{cutoff}}$ between signal and noise (see \Cref{fig1}), and the \rg approach consistently yields a lower limit of detection than the \bbp threshold, precisely because it does not require eigenvalues to be cleanly separated from the bulk.

This review follows a long line of research work on the use of \rg for signal and anomaly detection~\cite{RG1,RG2,RG3,RG4,RG5,RG6,RG7}.
The approach discussed here constitutes an interdisciplinary approach to data analysis using physics-inspired techniques, though it is certainly not an isolated case.
It is actually part of a lively bibliographic constellation of works that bring together various \rg-inspired methods for \ai and machine learning (\ml).
These approaches are well aligned with recent advancements in computational techniques based on physics-inspired arguments, such as the correspondence between Neural Networks and Quantum Field Theory (\textsc{nn--qft})~\cite{NNQFT1,NNQFT2,NNQFT3,NNQFT4,halverson2024tasi,demirtas2024neural,howard2025wilsonian}.
Various other approaches tackling different aspects of \ml (such as deep learning and diffusion models) using \rg-based techniques can also be mentioned, notably \cite{hou2023machine,beny2013deep,berman2023bayesian,berman2024inverse,masuki2025generative,cotler2023renormalizing}, though any possible list of works will never be exhaustive.
The link between \rg and data science is obviously no coincidence: the convergence of the general objectives of complex systems physics and those of data analysis, or, in other words, the necessity to extract large-scale regularities from data, naturally establishes the idea of a methodological contamination between these two fields (see, for instance,~\cite{DataPhy1}).
The \rg is probably the most notable example of physics-inspired technique for data analysis, as it provides an explanatory framework for the apparent simplicity of macroscopic physical laws, regardless of the complexity of the underlying microscopic reality of the phenomena.
The \rg thus offers the immense advantage of providing both a general mechanism for understanding the emergence of physical laws and their \enquote{simplification} to the core \dof needed to explain them, as well as a simple and unifying explanation for the many universalities encountered in physics.

Aware that, due to the interdisciplinary nature of this work, the reader may not be familiar with \rg, we shall provide a brief review in the technical sections to present the formalism.
As of today, it is the most efficient mechanism for explaining the apparent simplicity of effective laws and their mysterious insensitivity to the microscopic details of the theory (for an introduction, see~\cite{Zinn1}) as shown in \Cref{fig0} (notations such as $H_1, H_2, \dots, \Gamma, \dots$ will become clear in the following sections).
The simplest and most illuminating example of this idea lies in the fact that the three-dimensional Ising model at equilibrium, near its critical temperature (involving discrete \dof), is essentially indistinguishable from the $\phi^4_3$ field theory (see \Cref{part1}).\footnote{%
    The notation, conventional in field theory literature, gives the degree of interaction in the exponent and the dimension of space in the subscript.
}

\begin{figure}[t]
    \centering
    \begin{tikzpicture}
    \definecolor{spinblue}{RGB}{172,215,230}

    \draw[fill=spinblue!25, draw=black!75, very thick] (4.977, 2.721) ellipse (4.949 and 2.693);
    \node[anchor=west] at (2.0, 1.28) {\textbf{Full theory space}};

    \draw[fill=red!25, draw=red!70, very thick]
    (2.407, 3.686) circle (1.106);

    \fill[black] (2.133, 3.097) node[left] (P1) {$H_1$} circle (2pt);
    \fill[black] (1.567, 3.523) node[above] (P2) {$H_2$} circle (2pt);
    \fill[black] (2.714, 4.376) node[left] (P3) {$H_3$} circle (2pt);

    \filldraw (6.882, 1.545) node[right] (Gamma) {$\Gamma$} circle (2pt);

    \draw[dashed, -{Latex}]
    (2.133, 3.097) .. controls (3.436, 2.886) and (5.135, 1.879) ..
    (6.732, 1.567);
    \draw[dashed, -{Latex}]
    (1.567, 3.523) .. controls (3.244, 3.636) and (3.800, 3.116) ..
    (4.350, 2.734) .. controls (4.899, 2.352) and (5.442, 2.107) ..
    (6.764, 1.604);
    \draw[dashed, -{Latex}]
    (2.714, 4.376) .. controls (4.263, 3.338) and (4.985, 2.544) ..
    (6.787, 1.615);

    \node[anchor=north west] at (3.39, 4.53) {Microscopic theories};
    \node[anchor=center]  at (5.43, 3.09) {RG flow};
    \node[anchor=west]    at (6.65, 2.12) {Effective physics};
    \node[anchor=south west] at (7.44, 2.39) {\textbf{(IR)}};
    \node[anchor=south east] at (5.48, 4.28) {\textbf{(UV)}};
\end{tikzpicture}
    \caption{%
        Illustration of the \rg universality: many microscopic (ultraviolet, \uv) theories describe the same long-distance physics in the sense that they become indistinguishable in the low-energy (infrared, \ir) regime.
    }
    \label{fig0}
\end{figure}
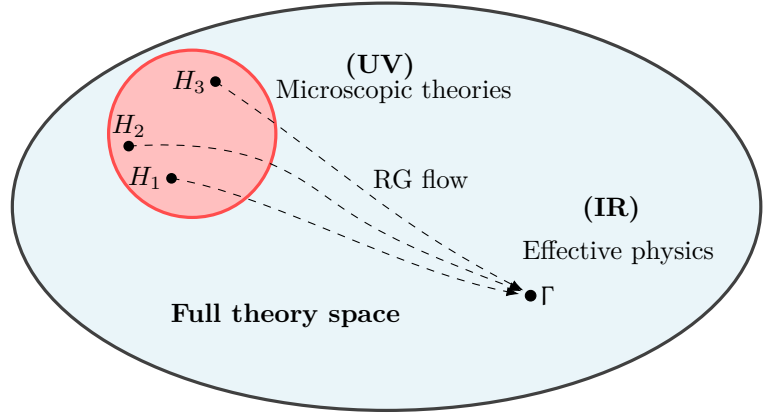

The underlying philosophy of this article follows its predecessors.
We shall consider an analogous field theory in the vicinity of universal noise, built on the hypothesis that it faithfully reproduces the characteristics of maximum entropy estimators for the fundamental (generally unknown) \dof whose correlations are (partially) encoded in the matrix $C$.
This field theory enables power counting that follows the shape of the spectrum, which in turn makes it possible to distinguish relevant interactions from irrelevant ones.\footnote{%
    For the reader unfamiliar with these notions, they will be explained in \Cref{part1}.
}
This power counting notably determines a particular equivalence class for the \mpdistr spectrum, asymptotically analogous to a field theory in dimension~3, as will be discussed in the technical sections.
One can therefore consider determining an objective cutoff between information and noise based on a significant threshold of deviation from this universality class, while ensuring that it does not compete with the natural fluctuations of the statistical ensemble to which the data belongs.
Articles \cite{RG1,RG2,RG3,RG4,RG6} constitute the original works in this field, but the reader may consult the recently updated review \cite{RG5}.
The present work follows directly from \cite{RG7}, whose definitions and conventions it intentionally adopts in their entirety to facilitate access for the data scientists interested in this review.
All the numerical work was produced using the library \href{https://github.com/thesfinox/frg-signal-detection}{frg-signal-detection} freely available on GitHub (physical simulations also use the library \href{https://github.com/ParhamRadpay/Model-A}{ModelA}).

In what follows, we shall adopt the physicist's notation
\begin{equation}
    \left\langle f(X) \right\rangle \defeq \mathds{E}\left[ f(X) \right]
\end{equation}
to indicate the average of $f(X)$ over the distribution of $X \sim \mathrm{P}_X$.
In the case of a continuous variable, this is computed by:
\begin{equation}
    \left\langle f(X) \right\rangle = \int\mathrm{d}P_X\; f(X),
\end{equation}
where, most of the time, $\mathrm{d}P_X = p_X(x) \mathrm{d}x$ with $p_X(x) = \exp(-S(x))$ being the probability density function of $X$.

\paragraph{Outline}
The review is organised as follows:
\begin{itemize}
    \item \Cref{sec_invitation} motivates the \rg approach to signal detection through the concrete example of hyperspectral imaging, introduces the concept of the dimensional phase transition, and previews the key results of the review.

    \item \Cref{part1} introduces the \rg formalism from first principles using the Ising model as a laboratory:
          \begin{itemize}
              \item \Cref{sec13} reviews the Ising model and the mean-field approximation.
              \item \Cref{sec14} presents Kadanoff's block-spin construction and the Ginzburg--Landau $\phi^4$ field theory.
              \item \Cref{sec:wilson_rg} develops the Wilsonian \rg, the concepts of renormalizability and canonical dimension.
              \item \Cref{sec:lesson} summarises the lessons learned: the classification of couplings as relevant, marginal, or irrelevant, the upper critical dimension $D=4$, and the emergence of universality.
          \end{itemize}

    \item \Cref{part2} builds the field theory directly from the empirical spectrum and applies the functional \rg:
          \begin{itemize}
              \item \Cref{section1} constructs the field theory and establishes the continuum limit.
              \item \Cref{sec:frg_dft} presents the Wetterich--Morris formalism.
              \item \Cref{sec:generalized_scaling} analyses the scale dependence of the canonical dimensions.
          \end{itemize}

    \item \Cref{part3} develops and validates the operational detection framework:
          \begin{itemize}
              \item \Cref{sec:gsa} introduces the key concepts of the theoretical framework and describes the methodology.
              \item \Cref{sec:ising_critical_temp} benchmarks the method against the exact Onsager solution of the two-dimensional Ising model.
              \item \Cref{Formalization} formalises the framework through the distribution proxies and the direct and inverse Gaussian distances.
              \item \Cref{sec:dimensional_phase_transition} provides a numerical illustration of the dimensional phase transition.
          \end{itemize}

    \item \Cref{part4} extends the framework through four applications:
          \begin{itemize}
              \item \Cref{Sec4-1} analyses independent noise components and the multimodal structure they induce in the canonical dimension flow, providing an estimator for the number of distinct confounding sources.
              \item \Cref{sec:non_gaussian_noise} tests the method against non-Gaussian and structured noises, including convolutions (blurring) and periodic interference patterns.
              \item \Cref{sec:time_series} extends the method to time-dependent data, measuring the dynamic critical exponent of the two-dimensional Ising model through the spectral scaling of the cutoff index.
              \item \Cref{sec:hyperspectral} confronts the method with real-world sensor data from a hyperspectral image of the Martian surface, detecting residual spectral correlations.
          \end{itemize}

    \item The appendices provide supporting technical material:
          \begin{itemize}
              \item \Cref{App0} derives the maximum entropy estimator that underpins the field-theoretic construction.
              \item \Cref{App1} presents the locality and universality arguments justifying the local field theory at the tail of the spectrum.
              \item \Cref{sec:app1} gives the large-$N$ solution and Hartree approximation for the critical temperature.
              \item \Cref{AppOJK} reviews the Ohta--Jasnow--Kawasaki approximation for domain growth in the coarsening regime.
              \item \Cref{sec:app2} recovers the equilibrium Ising model from the non-equilibrium Model~A dynamics.
              \item \Cref{AppD} provides a short review of spiked matrix models and the Baik--Ben Arous--P\'{e}ch\'{e} phase transition.
          \end{itemize}
\end{itemize}

\section{An Interdisciplinary Invitation}\label{sec_invitation}
Signal detection in high-dimensional datasets is, in its most elementary formulation, an eigenvalue problem: we seek to identify the largest contribution to the variance of the data, as it possibly corresponds to the most significant directions containing the meaningful information.
\pca and its derivatives address this problem directly by isolating the largest eigenvalues and their associated eigenvectors.
When the signal manifests as low-rank, isolated eigenvalues, or \emph{spikes}, that separate cleanly from the quasi-continuous bulk, a cutoff threshold is almost always possible to define, and its validity is backed by the \bbp transition threshold~\cite{math3}, which provides a rigorous criterion for its detectability.
Once the spikes have been extracted, however, \pca has exhausted its diagnostic power.
The residual continuous spectrum is usually treated as noise and discarded.
Yet, in many realistic scenarios, the signal is not a low-rank injection in a continuous bulk: it is actually extensive, distributed across a large number of eigenvectors, and mostly submerged within the bulk.
It manifests not as an isolated eigenvalue but as a subtle, collective deformation of the spectral density.
In this \enquote{post-\pca} regime, where no isolated feature offers a natural place to draw the line, standard methods are no longer useful.
It is precisely here that the \rg provides a new analytical language.

A concrete and instructive illustration of this regime is offered by hyperspectral imaging (\hsi).
Whereas an ordinary \emph{Red-Green-Blue} (\rgb) image records three broadband colour channels, a hyperspectral image captures the complete reflectance spectrum of the scene at each pixel, typically across several hundred narrow and contiguous spectral bands.
Each pixel is therefore a high-dimensional vector whose components encode the interaction of light with matter at different wavelengths.
This becomes particularly useful as different materials leave distinct spectral fingerprints and can be distinguished based on their spectral signatures.
Standard spectral analysis, including \pca-based denoising pipelines such as the Minimum Noise Fraction (\mnf) transform~\cite{green1988MNF}, can isolate the dominant mineral signatures present in a scene.
However, when only traces of a given mineral trace drowned in the compounded noise of the instrument, photon shot noise, and atmospheric residuals, its spectral contribution no longer appears as an isolated eigenvalue.
The signal is, however, not absent and is submerged, distributed across a subpopulation of eigenvalues that, while individually indistinguishable from noise, collectively deform the shape of the bulk.
It is this kind of signal, invisible to any method that relies on a hard eigenvalue cutoff, that the \rg is designed to detect.

The essential difference between \pca and the \rg-based approach lies in how each method conceptualises the boundary between signal and noise.
\pca imposes a hard partition: eigenvalues on one side are retained, those on the other are discarded.
The \rg, by contrast, does not partition the spectrum but analyses its structure.
The central observation, developed in detail in the following parts of this review, is that the eigenvalue distribution can be mapped onto an effective field theory (\eft) whose properties are entirely determined by the shape of the spectrum.
Within this \eft, the degrees of freedom represents the principal directions of the correlation matrix, and their effective interactions encode the statistical dependencies among the original observables.
Small deviations from a universal noise model manifest as modifications of the effective interactions, rendering the \eft sensitive to the same subpopulations of eigenvalues that \pca discards.

The \eft constructed in this manner presents different characteristics, depending on the original data.
In particular, it possesses a remarkable property: its effective spacetime dimension is not fixed, but depends on the data themselves.
For pure noise belonging to the \mpdistr universality class, the asymptotic dimension is $D = 3$, meaning that the theory behaves like a $\phi^4$ \eft in three dimensions, where quartic interactions are relevant and the Gaussian fixed point is unstable.
The presence of a signal modifies the tail of the spectral distribution, and this deformation changes the canonical dimensions of the effective couplings.
As the signal strength increases, the effective dimension grows.
When $D$ crosses the critical value $4$, all interactions become irrelevant and the \rg flow is driven towards the Gaussian fixed point.
This \emph{dimensional phase transition}, which will be formalised in \Cref{part3}, provides an objective and universal definition of the signal-to-noise (\snr) boundary: it measures, in a precise and computable sense, a distance from pure noise.
The signal is thus not identified by where its eigenvalues sit, but by how it alters the dimensional character of the fluctuations that carry it.
In other words, by looking at how the underlying spacetime dimension changes, the \rg enables the detection of signal, recovering the phase transition highlighted recently in the analysis of extensive-rank random matrices~\cite{Landau2023}.

The practical power of this approach comes from its universality.
The \mpdistr class describes an extraordinarily broad family of noise sources: essentially, any large covariance matrix whose entries are sufficiently weakly correlated can be actively described by it.
The \eft, by construction, is blind to the microscopic nature of the data that produced the correlation matrix.
Whether the underlying variables are reflectance spectra from a Martian crater, gene expression levels, financial returns, or spin configurations from a Monte Carlo simulation, the \eft treats them through the same lens: only the eigenvalue distribution matters, as only the eigenvalue distribution is responsible for how the different \dof interact.
This is the same universality that allows the Ising model, originally conceived for ferromagnets, to describe systems as distant as binary alloys and neural networks.
The reader familiar with the \rg will recognise this as the modern understanding of \eft: the microscopic details are irrelevant in the \ir, and only a small number of universality classes survive.

The validity of this framework has been established through extensive numerical experiments, a selection of which is reviewed in the following parts of this article.
A particularly interesting test is provided by the two-dimensional Ising model, whose critical temperature is known exactly from Onsager's solution.
The \snr boundary identified by the dimensional criterion recovers the Onsager temperature to within approximately $2$--$3\%$ (with a central value of $2.2\%$), outperforming standard methods based on Kullback--Leibler (\kl) divergence minimisation (which yield errors of approximately $7\%$).
Additional results on real-world images, non-equilibrium magnetic systems, and the hyperspectral data discussed above confirm that the method is both robust and broadly applicable.
The reader interested in an interdisciplinary approach to data analysis, and in the unexpected dialogue between the renormalization group and practical signal detection, will find in the following pages a self-contained exposition of the formalism, its numerical implementation, and its validation across multiple domains.
The review is structured to be accessible to the data scientist with no prior exposure to field theory, while providing the physicist with a concrete and unfamiliar application of the \rg framework.

\part{The Renormalization Group: a Law of Nature}\label{part1}

Discovered in the second half of the 20th century, \rg is one of the key techniques in the theory of complex systems~\cite{Complex1}.
To put it simply, \rg is a procedure for extracting a simplified description from a more fundamental one, which involves a very large number of interacting \dof, retaining only a reduced number of parameters: the so-called renormalizable couplings.
These couplings play a role analogous to that of relevant features within large datasets in the sense that they define the large-scale regularities of the system.
With this in mind, \rg is indeed a law of nature rather than a tool and represents the channel through which the different scales of the world are interconnected.
It can be found almost everywhere in physics, from high-energy \rg to condensed matter \rg and even cosmology.
The novice reader may consult references~\cite{Zinn1,Zinn2} for a deeper understanding, as well as the historical reference~\cite{Wilson1} and the pedagogical paper~\cite{kadanoff2009more}.
For a reference about condensed matter \rg and especially the physics of magnetic systems, see also~\cite{kittel2018introduction}.

This section develops the formalism and notably introduces essential concepts related to \rg (\rg flow, renormalizable couplings, canonical dimensions, etc.).
For the reader in data science who may not have heard of \rg, we provide a concise presentation, richly referenced, yet sufficiently self-contained for the remainder of the review.
Our presentation of these concepts will primarily be through the lens of the Ising model, which is also well known in data science and will occupy a significant place in the results of this review.
The reader familiar with these concepts can, of course, skip the first section of this part without hindering their understanding.

\section{The Ising Model and the Mean Field Approximation}\label{sec13}
To introduce the principal concepts of the \rg, we shall consider the classic example of ferromagnets and their statistical behaviour.
We shall introduce the basic concepts and notation that will be useful for the rest of the discussion.
We revisit the Ising model because it provides the simplest laboratory for the concepts that will be essential to our signal-detection framework: universality, the upper critical dimension, the failure of mean-field theory, and the need for a systematic treatment of fluctuations (precisely the role that the \rg will play).

\subsection{Ferromagnetic Systems}

The origins of \rg date back to the early days of particle physics in the 1950s, with the pioneering work of Gell-Mann and Low~\cite{gell1954quantum}.
However, it truly took shape in the 1970s with the work of Kadanoff~\cite{kadanoff1966scaling} and, especially, Wilson~\cite{wilson1971renormalization} on the theory of magnetic phase transitions, a field far removed from particle physics.
At that time, the main concern was the behaviour of magnetic systems (conventional magnets) and more specifically what is known as the \emph{ferromagnetic transition}.
The debate centred on the theoretical understanding of Curie's observation that, above a certain temperature $T_c$, the \emph{Curie temperature} (or \emph{critical temperature}), a magnet spontaneously loses its magnetization.
More precisely, it was shown experimentally that the magnetization $M$, i.e.\ the global magnetic moment of a magnet, vanishes above $T_c$ and acquires a non-zero expectation value below $T_c$, with a behaviour that follows a \emph{universal law}, meaning it is independent of the type of magnetic material considered:\footnote{%
    Universality actually extends far beyond magnetic systems.
    Furthermore, the values of the critical exponents given below, $\beta$ and $\nu$, correspond to magnetic systems in the Ising model universality class.
    In particular, they concern systems with strong uniaxial anisotropy; fully isotropic systems belong instead to the Heisenberg universality class, which we shall not discuss~\cite{kittel2018introduction}.
}
\begin{equation}
    M \sim \left(\frac{T_c-T}{T_c}\right)^{\beta},
\end{equation}
with $\beta \approx 0.326$.
The transition point is also marked by the divergence of the correlation length $\xi$, which represents the typical length over which spins can be strongly correlated:
\begin{equation}
    \langle S_i S_j \rangle \sim e^{-\vert \vec{x}_i-\vec{x}_j \vert/\xi},
\end{equation}
where $\vec{x}_i$ denotes the position vector of the spin at site $i$ and $\vert \vec{x}_i-\vec{x}_j \vert$ is the Euclidean distance between sites $i$ and $j$.
The dependence of the correlation length on the temperature near the critical point is given by:
\begin{equation}
    \xi \sim \left \vert \frac{T-T_c}{T_c}\right \vert^{-\nu},
    \label{exponu}
\end{equation}
with $\nu \approx 0.630$.
The existence of universal behaviour, largely independent of the microscopic nature of interactions between atoms and electrons in magnetic materials, suggests the existence of a \enquote{coarse-graining} mechanism for these details.
In other words, there must exist, at a certain scale, a mathematical description common to all materials.
\rg is the mathematical tool that explains why these properties are not unique to magnetic systems: they are common to families of systems that possess a natural hierarchy (microscopic, mesoscopic, and macroscopic) together with intrinsic variability.

By the mid-1920s, it was understood that electrons within atoms possess a spin and behave, to a good approximation, as tiny magnets.
This spin is discrete, and in the case of electrons it can occupy only two states, labelled $S = \pm 1$ in the Ising model.\footnote{%
    For readers with a background in physics, we are obviously simplifying the discussion to get to the essentials.
    Strictly, the electron has spin quantum number $S = 1/2$, and its $z$-projection takes the values $S_z = \pm \hbar/2$.
}.
The atomic origin of ordinary magnetism was soon established beyond doubt, yet understanding the resulting large-scale magnetism remained a challenge for physics.
A simplified model of the microscopic reality of these magnetic systems and the interactions between spins was proposed by Ising, and we shall review it below.\footnote{%
    In reality, the Ising model was proposed in 1920, before the discovery of spin, as a simplification of another mechanism.
    However, its deep justification came with the advent of quantum mechanics and relies largely on the Pauli exclusion principle~\cite{kittel2018introduction}.
}

\subsection{Ising model}

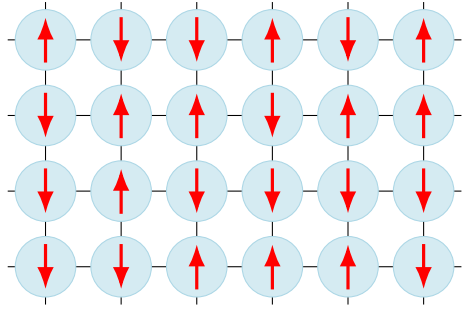
\begin{figure}[t]
    \centering
    \begin{tikzpicture}
    \def\s{1.0}
    \draw[thin, black] (-0.5,-0.5) grid (5*\s+0.5, 3*\s+0.5);

    \definecolor{spinblue}{RGB}{172,215,230}

    \foreach \x/\y/\d in {
    0/3/1,  1/3/-1, 2/3/-1, 3/3/1,  4/3/-1, 5/3/1,
    0/2/-1, 1/2/1,  2/2/1,  3/2/-1, 4/2/1,  5/2/1,
    0/1/-1, 1/1/1,  2/1/-1, 3/1/-1, 4/1/-1, 5/1/-1,
    0/0/-1, 1/0/-1, 2/0/1,  3/0/1,  4/0/1,  5/0/-1
    } {
    \draw[fill=spinblue!50, draw=spinblue] (\x*\s,\y*\s) circle (0.4);

    \ifnum\d=1
        \draw[-{Latex}, red, very thick] (\x*\s,\y*\s-0.3) -- (\x*\s,\y*\s+0.3);
    \else
        \draw[-{Latex}, red, very thick] (\x*\s,\y*\s+0.3) -- (\x*\s,\y*\s-0.3);
    \fi
    }
\end{tikzpicture}
    \caption{%
        A spin configuration on a 2-dimensional lattice, i.e.\ the 2D Ising model. $S = +1$ spins are represented by an upward-pointing arrow, and $S = -1$ spins by a downward-pointing arrow.}
    \label{Ising1}
\end{figure}

The Ising model is based on a regular $D$-dimensional lattice with a lattice spacing of $1$ (mimicking the atomic lattice in matter, see \Cref{Ising1}).
The interactions between the spins, represented as discrete variables $S_i \defeq \pm 1$ at each site $i$ of the lattice, are treated as nearest-neighbour interactions, and the partition function (or \emph{moment-generating} functional) is given by:
\begin{equation}
    Z_{\text{Ising}}[B,J]
    \defeq
    \sum_{\{ S_i= \pm 1 \}} e^{-H_{\text{Ising}}[\{S_i\}]/T},
\end{equation}
where the \emph{Ising's Hamiltonian} is:
\begin{equation}
    H_{\text{Ising}}[\{S_i\}]
    \defeq
    -J \sum_{\langle i,j \rangle} S_i S_j + \sum_{i=1}^{N^D} S_i B_i
\end{equation}
where $J > 0$ characterises the interaction between spins, $T$ is the temperature of the system, $B_i$ is the value of the external magnetic field at site $i$, $N^D$ denotes the total number of sites ($N$ spins per dimension), and the notation $\langle i, j \rangle$ denotes nearest-neighbour pairs.
It is possible to show that, in spatial dimensions $D > 1$, the Ising model always exhibits a phase transition: below a certain critical temperature $T_c$, a spontaneous magnetization appears.
This is a state where the average value of the spins differs significantly from zero as $N \to \infty$.
Schematically, we have:
\begin{equation}
    \begin{tabular}{ccc}
        $T<T_c$ & $\colon$ & $\vcenter{\hbox{\begin{tikzpicture}
    \def\s{1.0}

    \definecolor{spinblue}{RGB}{172,215,230}

    \foreach \x/\d in {
    0/1, 1/1, 2/1, 3/1, 4/1, 5/1
    } {
    \draw[fill=spinblue!50, draw=spinblue] (\x*\s,0) circle (0.4);

    \ifnum\d=1
        \draw[-{Latex}, red, very thick] (\x*\s,-0.3) -- (\x*\s,0.3);
    \else
        \draw[-{Latex}, red, very thick] (\x*\s,0.3) -- (\x*\s,-0.3);
    \fi
    }
\end{tikzpicture}}}$ \\[1em]
        $T>T_c$ & $\colon$ & $\vcenter{\hbox{\begin{tikzpicture}
    \def\s{1.0}

    \definecolor{spinblue}{RGB}{172,215,230}

    \foreach \x/\d in {
    0/1, 1/-1, 2/-1, 3/1, 4/-1, 5/1
    } {
    \draw[fill=spinblue!50, draw=spinblue] (\x*\s,0) circle (0.4);

    \ifnum\d=1
        \draw[-{Latex}, red, very thick] (\x*\s,-0.3) -- (\x*\s,0.3);
    \else
        \draw[-{Latex}, red, very thick] (\x*\s,0.3) -- (\x*\s,-0.3);
    \fi
    }
\end{tikzpicture}}}$
    \end{tabular}
\end{equation}
Note that the choice of the direction in which the magnetization appears breaks the discrete $\mathds{Z}_2$ symmetry ($S_i \to -S_i$) of the model.
This is referred to as \emph{spontaneous symmetry breaking}, or a \emph{broken phase}.

Although a simplified model, the Ising model captures, at least qualitatively, the essence of the physics of the ferromagnetic transition in dimension 3.\footnote{%
    In arbitrary dimension $D$, the difference can also be qualitative.
    For instance, the Ising model predicts a phase transition in dimension 2, whereas the more physical model known as the Heisenberg model does not.
}
It should be noted that there are no analytical solutions to the Ising model in dimension $D > 2$.
However, it is possible to understand the phase transition through various methods.
The simplest and most immediate is \emph{mean-field theory}, which consists in replacing the exact interaction of the model with an effective interaction, such that each spin interacts only with the average of its neighbours.
For an infinite or periodic system that is perfectly homogeneous, this mean field $\mathcal{M}$ must be site-independent.
By expressing this condition for each site, one readily obtains a self-consistent equation for $\mathcal{M}$:
\begin{equation}
    \mathcal{M}
    =
    \tanh \left(\frac{2D J \mathcal{M}}{T} \right)
    \approx
    \frac{2D J \mathcal{M}}{T} - \frac{1}{3}\left(\frac{2D J \mathcal{M}}{T}\right)^3+\mathcal{O}(\mathcal{M}^4).
\end{equation}
From this transcendental equation, one readily sees that the system acquires a spontaneous magnetization $\mathcal{M} \neq 0$ when $T < T_0 \defeq 2DJ$:\footnote{%
    $T_0$ here plays the role of the critical temperature.
    However, as we shall see, this temperature is not the physical critical temperature $T_c$: fluctuations generally lower the critical temperature compared to the value predicted by mean-field theory.
    The change in notation accounts for this distinction.
}
\begin{equation}
    \mathcal{M} \sim (T_0-T)^{1/2}.
\end{equation}
Mean-field theory predicts $\beta = \nu = 1/2$, regardless of the dimension.
This is a form of super-universality that does not correspond to the reality of the Ising model, whose critical exponents depend heavily on the dimension.
Furthermore, mean-field theory predicts a phase transition in dimension 1, whereas a direct calculation shows that the 1D Ising model does not exhibit a phase transition.
The mean-field approximation is, in fact, only correct for dimensions $D > 4$, and only qualitatively correct for $1 < D \leq 4$.
Mean-field theory thus misses an essential part of the physics, which is naturally associated with the fluctuations of the mean field $\mathcal{M}$.
The Ginzburg criterion evaluates the relative size of these fluctuations $\Delta \mathcal{M}$,
\begin{equation}
    \frac{(\Delta \mathcal{M})^2}{\mathcal{M}^2} \sim \vert T-T_0 \vert^{\frac{D-4}{2}}.
\end{equation}
Thus, when $D > 4$, $|T - T_0|^{\frac{D-4}{2}}$ tends toward zero as $T \to T_0$, such that $(\Delta \mathcal{M})^2 / \mathcal{M}^2 \ll 1$ and mean-field theory remains consistent.
In contrast, when $D \leq 4$, fluctuations become stronger and the mean-field approximation breaks down.
A theory capable of taking these fluctuations into account is therefore required.
Note that $D = 4$ will define the \emph{critical dimension}.

The key lesson of this section is that mean-field theory, while intuitive and simple, fails exactly when we most need it, that is near the critical point in dimensions $D \leq 4$, where fluctuations dominate.
The \rg provides the systematic framework for handling these fluctuations, and it does so by organising couplings according to their relevance under scale transformations.
In the next section, we develop this \rg machinery explicitly, building from the concrete example of the Ising model toward the effective field theory that will later describe eigenvalue spectra.

\section{A Field Theory for Fluctuations}\label{sec14}
This section illustrates the core ideas of the \rg through the Ising model.
Kadanoff's block-spin construction introduces the iterative coarse-graining procedure at the heart of the \rg.
Near criticality this discrete picture gives way to the Ginzburg-Landau $\phi^4$ field theory, whose treatment occupies the remainder of this part.

\subsection{Kadanoff's Block Spin}

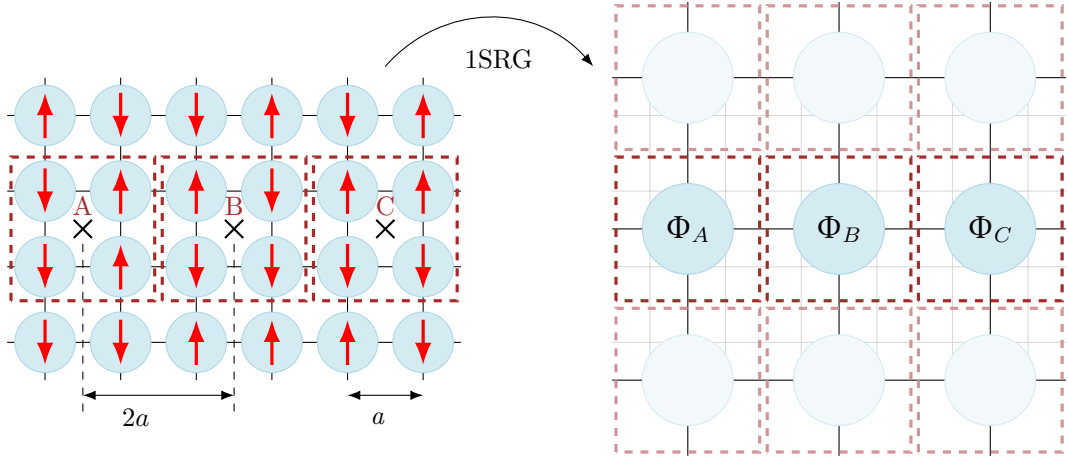
\begin{figure}[t]
    \centering
    \begin{tikzpicture}
    \def\s{1.0}
    \def\shift{8.0}
    \definecolor{spinblue}{RGB}{172,215,230}
    \definecolor{blockred}{RGB}{165,42,42}
    \definecolor{structred}{RGB}{138,0,0}

    \draw[thin, black] (-0.5,-0.5) grid (5*\s+0.5, 3*\s+0.5);
    \draw[thin, black!15] (-0.5+\shift,-0.5) grid (5*\s+\shift+0.5, 3*\s+0.5);
    \draw[thin, black] (-0.5+\shift, -0.5) -- (5.5+\shift, -0.5);
    \draw[thin, black] (-0.5+\shift, 1.5) -- (5.5+\shift, 1.5);
    \draw[thin, black] (-0.5+\shift, 3.5) -- (5.5+\shift, 3.5);
    \draw[thin, black] (0.5+\shift, -1.5) -- (0.5+\shift, 4.5);
    \draw[thin, black] (2.5+\shift, -1.5) -- (2.5+\shift, 4.5);
    \draw[thin, black] (4.5+\shift, -1.5) -- (4.5+\shift, 4.5);

    \foreach \x/\y/\d in {
    0/3/1,  1/3/-1, 2/3/-1, 3/3/1,  4/3/-1, 5/3/1,
    0/2/-1, 1/2/1,  2/2/1,  3/2/-1, 4/2/1,  5/2/1,
    0/1/-1, 1/1/1,  2/1/-1, 3/1/-1, 4/1/-1, 5/1/-1,
    0/0/-1, 1/0/-1, 2/0/1,  3/0/1,  4/0/1,  5/0/-1
    } {
    \draw[fill=spinblue!50, draw=spinblue] (\x*\s,\y*\s) circle (0.4);

    \ifnum\d=1
        \draw[-{Latex}, red, very thick] (\x*\s,\y*\s-0.3) -- (\x*\s,\y*\s+0.3);
    \else
        \draw[-{Latex}, red, very thick] (\x*\s,\y*\s+0.3) -- (\x*\s,\y*\s-0.3);
    \fi
    }
    \foreach \x/\y/\d in {
            0.5/1.5/A, 2.5/1.5/B, 4.5/1.5/C
        } {
            \draw[fill=spinblue!50, draw=spinblue] (\x*\s+\shift,\y*\s) node {\large $\Phi_\d$} circle (0.6);
        }
    \foreach \x/\y in {
            0.5/-0.5, 0.5/3.5, 2.5/-0.5, 2.5/3.5, 4.5/-0.5, 4.5/3.5
        } {
            \draw[fill=spinblue!15, draw=spinblue!50] (\x*\s+\shift,\y*\s) circle (0.6);
        }

    \draw[color=blockred, dashed, very thick] (-0.45, 0.55) rectangle (1.45, 2.45);
    \draw[color=blockred, dashed, very thick] (-0.45+\shift, 0.55) rectangle (1.45+\shift, 2.45);
    \draw[color=blockred!50, dashed, very thick] (-0.45+\shift, 0.55+2.0) rectangle (1.45+\shift, 2.45+2.0);
    \draw[color=blockred!50, dashed, very thick] (-0.45+\shift, 0.55-2.0) rectangle (1.45+\shift, 2.45-2.0);
    \draw[color=blockred, dashed, very thick] (1.55, 0.55) rectangle (3.45, 2.45);
    \draw[color=blockred, dashed, very thick] (1.55+\shift, 0.55) rectangle (3.45+\shift, 2.45);
    \draw[color=blockred!50, dashed, very thick] (1.55+\shift, 0.55+2.0) rectangle (3.45+\shift, 2.45+2.0);
    \draw[color=blockred!50, dashed, very thick] (1.55+\shift, 0.55-2.0) rectangle (3.45+\shift, 2.45-2.0);
    \draw[color=blockred, dashed, very thick] (3.55, 0.55) rectangle (5.45, 2.45);
    \draw[color=blockred, dashed, very thick] (3.55+\shift, 0.55) rectangle (5.45+\shift, 2.45);
    \draw[color=blockred!50, dashed, very thick] (3.55+\shift, 0.55+2.0) rectangle (5.45+\shift, 2.45+2.0);
    \draw[color=blockred!50, dashed, very thick] (3.55+\shift, 0.55-2.0) rectangle (5.45+\shift, 2.45-2.0);

    \node[color=blockred] at (0.5, 1.8) {A};
    \node[color=blockred] at (2.5, 1.8) {B};
    \node[color=blockred] at (4.5, 1.8) {C};

    \draw[black, thick] (0.38, 1.38) -- (0.62, 1.62);
    \draw[black, thick] (0.62, 1.38) -- (0.38, 1.62);
    \draw[black, thick] (2.38, 1.38) -- (2.62, 1.62);
    \draw[black, thick] (2.62, 1.38) -- (2.38, 1.62);
    \draw[black, thick] (4.38, 1.38) -- (4.62, 1.62);
    \draw[black, thick] (4.62, 1.38) -- (4.38, 1.62);

    \draw[thin, dashed, black] (0.50, 1.30) -- (0.50, -1.00);
    \draw[thin, dashed, black] (2.50, 1.30) -- (2.50, -1.00);

    \draw[{Latex}-{Latex}, black, thin] (0.50, -0.7) -- (2.50, -0.7);
    \node at (1.2, -1.0) {2$a$};

    \draw[{Latex}-{Latex}, black, thin] (4.0, -0.7) -- (5.0, -0.7);
    \node at (4.4, -1.0) {\emph{a}};

    \draw[-{Latex}, black, thin] (4.5, 3.65) .. controls (5.25, 4.5) and (6.5, 4.5) .. (7.25, 3.65);
    \node at (6, 3.75) {1SRG};

\end{tikzpicture}
    \caption{%
        The construction of spin blocks: spins are averaged inside blocks (A, B, and C, for example), defining the \emph{block spin} ($\Phi_A$, $\Phi_B$, $\Phi_C$), and the interaction between spins is replaced by an effective interaction between blocks derived from the latter.
    }
    \label{IsingBlock}
\end{figure}

Mean-field theory captures the existence of a phase transition but fails quantitatively: the critical exponents it predicts differ from those of the exact solution of the 2D Ising model, proposed by Lars Onsager in the 1940s~\cite{das2019field,onsager1944crystal}\footnote{%
    The critical exponents of mean-field theory differ from the exact values of Onsager's solution.
}.
No such exact solution exists for the Ising model in dimension $D>2$.
However, an exact solution is not required to understand the transition.
A different approach is needed for higher dimensions.

Kadanoff introduced the \emph{block spin} as an original approach to the problem.
Because only macroscopic behaviour matters, Kadanoff argued that the microscopic details of individual spins are irrelevant.
He proposed to:
\begin{enumerate}
    \item Group spins into blocks (e.g.\ $2 \times 2$ or $3 \times 3$ blocks in 2D).
    \item Replace each block with a single \enquote{block spin} ($\Phi_{\text{block}}$) representing the average orientation of the group.
    \item Redefine a new interaction $J'$ between these new spin blocks after \enquote{rescaling} the original interaction $J$, so that the new lattice resembles the original one.
\end{enumerate}
This step ensures that models at every stage are compared on the same lattice.

The general strategy, summarised in \Cref{IsingBlock}, defines the \emph{one-step renormalization group} (1SRG) transformation.
After one step, the partition function of the Ising model is written in terms of the interactions between block spins $\Phi_I$:\footnote{%
    Here, the integral is formal and should be understood as \enquote{a sum over all permitted values of the block spins $\Phi_I$.}
}
\begin{equation}
    Z_{\text{Ising}} = \int \prod_I \dd \Phi_I\, e^{-H^{(1)}[\Phi_I]/k_BT},
\end{equation}
where, up to an irrelevant global factor, we have:
\begin{equation}
    H^{(1)}[\Phi_I] = \sum_{I,J} V_{IJ} \Phi_I \Phi_J+\sum_{I,J,K,L} V_{IJKL} \Phi_I \Phi_J\Phi_K \Phi_L+\cdots,
    \label{onestepblock}
\end{equation}
where the indices $I,J$ run over the 1SRG block spins.
The original form of the Ising Hamiltonian is not preserved: $H^{(1)}$ involves couplings absent from the microscopic model.

The procedure iterates an arbitrary number of times.
At each step, the different couplings $V_{IJ}$, $V_{IJKL}$, etc., are modified.
The $n$SRG is therefore a transformation that maps the Hamiltonian $H^{(n-1)}$ (from step $n-1$) to another Hamiltonian $H^{(n)}$, defining a discrete trajectory in the abstract space of Hamiltonians $\mathbf{S}$, as shown on the left of \Cref{flow}.
In the example in \Cref{IsingBlock}, at each step, the lattice spacing is multiplied by 2, so that if $a$ is the original spacing, the $n$SRG Hamiltonian describes a lattice with spacing $2^n a$.
Rescaling keeps the lattice spacing $a$ fixed.
Small fluctuations at the block scale are eliminated at each step, yet systems at every stage are compared on the same lattice.
Kadanoff's approach thus changes the model's resolution: at each step both the effective degrees of freedom and the coupling constants are replaced.
At $T_c$ the correlation length diverges.
The system is therefore scale-invariant, and the physics at all scales is identical.
This scale invariance manifests in the block-spin transformation as self-similarity: the block Hamiltonian retains the same form as the original Hamiltonian.
In the coupling space, this defines the \emph{critical surface} $\mathbf{S}_\infty \subset \mathbf{S}$.
Self-similarity means that the image of a Hamiltonian on $\mathbf{S}_\infty$ under an SRG step lies in $\mathbf{S}_\infty$ (see \Cref{flow}, right).
Conversely, if the initial Hamiltonian is not exactly at the critical temperature, it moves away from the critical surface.
Scale invariance implies that each Hamiltonian $H^{(n)}$ describes the same large-scale physics and that only small-scale fluctuations have been erased.

\begin{figure}[t]
    \centering
    \begin{tikzpicture}[
        >=Latex,
        axis/.style={->, thick, black},
        rflow/.style={red, thick, rounded corners=20pt,
                postaction={decorate, decoration={markings,
                                mark=between positions 0.1 and 0.95 step 0.1 with {\arrow[red]{Latex}}}
                    }
            },
        rflowarrow/.style={red, ->, thick},
        dashedline/.style={black, dashed, dash pattern=on 1pt off 3pt},
        scale=0.9
    ]

    \coordinate (O) at (0,0);

    \draw[axis] (O) -- (3.5, 0) node[below right] {$V_{IJ}$};
    \draw[axis] (O) -- (0, 3.0) node[right] {$V_{IJKL}$};
    \draw[axis] (O) -- (-1.4, -1.4) node[below right] {$V_{IJKLMN}$};
    \draw[axis] (O) -- (-2.5, -0.45) node[above left] {$V_\infty$};
    \draw[dashedline] (-1.8, -0.3) .. controls (-1.9, -0.8) and (-1.7, -0.9) .. (-0.9, -0.9);

    \node (I) at (1.0, 0.0) {};
    \fill[black] (I) node[below] {Ising} circle (2pt);
    \node (E) at (-1.75, 1.8) {};
    \fill[black] (E) circle (2pt);

    \draw[rflow] (I.center)
    to (2.5, 1.9)
    to (3.7, 2.4)
    to (2.9, 1.1)
    to (1.3, 0.1)
    to (-0.5, 0.6)
    to(E.center);

    \node at (-2.3, 2.5) {$\mathbf{S}$};
\end{tikzpicture} \quad \begin{tikzpicture}[
        >=Latex,
        flow/.style={black, thick, rounded corners=20pt,
                postaction={decorate, decoration={markings,
                                mark=between positions 0.2 and 0.9 step 0.2 with {\arrow{Latex}}}
                    }
            },
        redline/.style={red, very thick},
        reddashed/.style={red, thick, dashed, dash pattern=on 1pt off 3pt},
        scale=0.9
    ]

    \fill[black!15, draw=black!30, thick] (-1.5, 1.2) -- (3.5, 1.2) -- (1.5, -1.2) -- (-3.5, -1.2) -- cycle;
    \node at (-2.5, -0.9) {$\mathbf{S}_\infty$};
    \node at (1.5, -0.9) {$\mathbf{S}$};

    \node (H) at (1.2, 0.1) {};
    \draw[redline] (1.2, 1.8) -- (H.center);
    \draw[reddashed] (H.center) -- (1.2, -1.2);
    \draw[redline] (1.2, -1.2) -- (1.2, -1.8);
    \fill[black] (H) node[below right] {$H_{*}$} circle (2pt);

    \draw[flow] (-1.2, 1.1) to[in=180, out=235] (H.center);
    \draw[flow] (-0.4, -1.0) to[in=200, out=95] (H.center);
    \draw[flow] (-1.0, 1.6) to (1.125, 0.15) to (1.1, 1.8);
    \draw[flow] (2.9, 1.6) to (1.325, 0.15) to (1.3, 1.8);
    \draw[flow] (3.5, -0.05) to (1.45, -1.1) to (1.4, -1.7);
\end{tikzpicture}
    \caption{%
        (Left) The steps of the renormalization group form a trajectory in the parameter space $\mathbf{S}$.
        (Right) The flow behavior in the vicinity of the critical surface $\mathbf{S}_\infty \subset \mathbf{S}$.
    }
    \label{flow}
\end{figure}
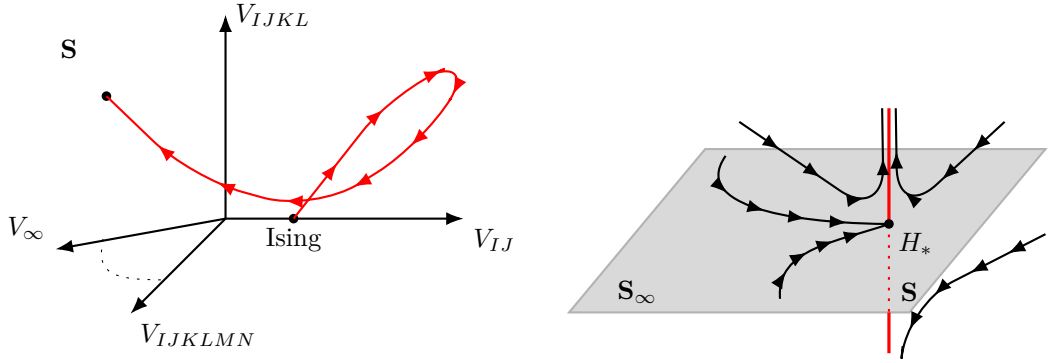

Another characteristic of the Ising model, and of the ferromagnetic transition in general, is universality.
The physics at the critical point, described by the critical exponents, is essentially independent of the microscopic details of the initial model, such as the type of interaction or the shape of the lattice. The simplest account of this universality postulates a \emph{fixed point} $H_* \in \mathbf{S}_\infty$: all trajectories within the basin of attraction of $H_*$ in $\mathbf{S}_\infty$ converge toward $H_*$ (see \Cref{flow}, right).\footnote{%
    In a dynamical system, the basin of attraction of a fixed point is the set of all initial states whose trajectories converge to it.
    Within the critical surface, not every Hamiltonian flows to the same fixed point, but those that flow to $H_*$ constitute its basin of attraction.
    Hamiltonians in the same basin share the same large-scale behaviour: this is the geometric origin of universality.
}

In the Ising model, the critical surface has codimension $1$,\footnote{%
    The codimension of a subspace is the number of independent constraints needed to select it from the ambient space.
    A surface of codimension~1 in an $n$-dimensional space has dimension $n-1$, i.e.\ it is a hypersurface.
    In the \rg context, the critical surface is a hypersurface because a single scalar parameter (the temperature) determines whether a Hamiltonian lies on it: tune $T$ to $T_c$ and you are on the surface.
    Any other value places you off it.
} and the transition line along which trajectories converge toward the fixed point from above and below $T_c$ is shown in red in \Cref{flow}.
The affine parameter of this line is identified with the temperature.\footnote{%
    An affine parameter is a coordinate along a curve that varies linearly with distance along that curve.
    In the \rg flow diagram, temperature $T$ plays this role on the transition line: equal changes in $T$ correspond to equal displacements along the line, and the fixed point $H_*$ sits at $T = T_c$.
}
Its value at $H_*$ is the critical temperature $T_c$.
When $T > T_c$, the trajectories converge toward the high-temperature region (unbroken phase), where the magnetisation vanishes. When $T < T_c$ (broken phase), the magnetisation acquires a macroscopic value.
Even at lowest order, the \rg approach improves upon the predictions of mean-field theory.

Kadanoff's spin-block method does not fully exploit universality: the large-scale theory remains tied to the original Ising model and true global properties are not captured.
We shall therefore define a new object, an \emph{effective field}, which will allow us to describe the large-scale behaviour in a more universal way.

\subsection{From Spins to Field}\label{fromisingTofield}

In the critical region, when $T - T_c \ll 1$, the correlation length $\xi$ is finite but large.
Magnetic order can therefore be maintained within droplets of typical radius $r \sim \xi$.
Inside a droplet, each \rg step changes the couplings as though the system were critical.
This continues until the effective lattice spacing reaches the droplet size, $2^n a \approx \xi$.
Beyond this scale the interactions decay, and, without global magnetic order, the droplets are approximately independent.
The droplet magnetisations $\{\Phi_{\mathcal{B}}\}$ are essentially independent variables.
Their distribution therefore tends toward a Gaussian according to the Central Limit Theorem (\clt).
When the droplet size remains small relative to the system, the droplet scale is mesoscopic.
On the macroscopic scale each droplet is a point $x \in \mathds{R}^D$, and the magnetisation $\Phi(x)$ becomes a continuous function on $\mathds{R}^D$.
The function $\Phi(x)$ is then called a \emph{field}.

We can deduce its behaviour from what we know of its physical properties.
Since it must be nearly Gaussian, we expect that the partition function of the original Ising model can be written as:
\begin{equation}
    Z_{\text{Ising}}=\int \prod_x \dd \Phi(x)\, e^{-H[\Phi]/T},
\end{equation}
where the Hamiltonian is assumed to be nearly Gaussian:
\begin{equation}
    H[\Phi]=\frac{1}{2} \int \dd x \dd y \, \Phi(x) V(x,y) \Phi(y)+\mathcal{O}(\Phi^4).
\end{equation}
Although the block spin $\Phi_A$ is bounded in principle, at a sufficiently large scale $\Phi(x)$ can be extended to all of $\mathds{R}$.
The integral is therefore not bounded.

In the case of the Ising model, the interaction between spins can only depend on the distance, such that $V(x,y) = K(\vert x-y \vert) >0$.
The form of $K$, and the deviations from the Gaussian model, can be deduced directly from the Ising model.
We begin by rewriting the Ising Hamiltonian in a more general form (without a magnetic field for simplicity):
\begin{equation}
    H[S]= - J \sum_{\vec{x},\vec{y}}  S_{\vec{x}}S_{\vec{y}},
\end{equation}
where $\vec{x} \in \mathds{R}^D$ is the coordinate point of the spin $S_{\vec{x}}$ on the lattice (the vector, clearly, decomposes on a canonical basis of $\mathds{R}^D$ $\{ \hat{e}_i \}_{i = 1}^{D}$: $\vec{x} = \sum_{i=1}^D n_i \hat{e}_{i}$).
The Ising model on a regular lattice then corresponds precisely to the choice $K_{\vec{x}\vec{y}} \defeq V(\vert \vec{x}-\vec{y} \vert)$, with:
\begin{equation}
    V(\vert \vec{x} \vert) = J \sum_{\alpha=1}^D \left( \delta_{ \vec{x},a \vec{e}_\alpha}+\delta_{\vec{x}, -a \vec{e}_{\alpha}} \right),
    \label{eq:V_Ising}
\end{equation}
where $\vec{e}_{\alpha}$ is the unit vector along the axis $\alpha$ and $a$ is the lattice spacing (in the initial model, $a=1$).
To decouple the quadratic spin interaction, we introduce an auxiliary field via the Hubbard-Stratonovich transformation, which exploits the Gaussian integral identity
\begin{equation}
    \int \dd x \, e^{k x} e^{-\frac{1}{2 \sigma^2} x^2} \propto e^{\frac{1}{2 } \sigma^2 k^2},
\end{equation}
to write (the index $\vec{x}$ is discrete):
\begin{equation}
    \exp \left(-H[S]/T\right)
    \propto
    \int_{-\infty}^{+\infty} \prod_{\vec{x}} \dd \Phi_{\vec{x}} \,
    e^{-\frac{1}{2} T \sum_{\vec{x},\vec{y}} \Phi_{\vec{x}}  K_{\vec{x}\vec{y}}^{-1} \Phi_{\vec{y}} } e^{\sum_{\vec{x}} \Phi_{\vec{x}} S_{\vec{x}}}.
\end{equation}
The sum over spins can be computed exactly, and we get:
\begin{equation}
    \sum_{S_{\vec{x}}=\pm 1}\, \exp \left({\sum_{\vec{x}} \Phi_{\vec{x}} S_{\vec{x}}}\right)= \prod_{\vec{x}} 2 \cosh \left(\Phi_{\vec{x}} \right).
\end{equation}
The partition function becomes:
\begin{equation}
    Z_{\text{Ising}}\propto \int_{-\infty}^{+\infty} \prod_{\vec{x}} \dd \Phi_{\vec{x}} \, e^{-\frac{1}{2} T \sum_{\vec{x},\vec{y}} \Phi_{\vec{x}}  K_{\vec{x}\vec{y}}^{-1} \Phi_{\vec{y}} } \prod_{\vec{x}} e^{\ln \left(2 \cosh \left(\Phi_{\vec{x}} \right)\right)}.
    \label{partition2}
\end{equation}
To isolate the long-wavelength modes we diagonalise $K$ by Fourier transform.
Because the kernel $K$ is translation invariant, it can be diagonalised by its Fourier transform,
\begin{equation}
    \tilde{\Phi}_{\vec{k}} = \frac{1}{\sqrt{N^D}} \sum_{\vec{x}} e^{i \vec{k} \cdot \vec{x}}\, \Phi_{\vec{x}},
\end{equation}
where, assuming periodic boundary conditions, the momentum is quantized:\footnote{%
    This corresponds to the first Brillouin region given the periodic boundary conditions we chose.
    Another way to see this is to remark that, if we denote the modes as $k_{i,n}=2\pi n/(N a)$, with $n=0,1,\dots, N-1$, then the mode $N-1$ is $k_{i,N-1}\approx 2\pi/a$, i.e.\ it is closely equivalent to the lower mode because two vectors differing by $2\pi p/a$ for $p \in \mathds{Z}$ are physically equivalent.
    The true physical maximum momentum is for $n=N/2$, namely $\pi/a$.
    Note that, by consistency, the integrated functions in the integrals must have the same periodicity as the lattice, which is the case here.
}
\begin{equation}
    k_i= \frac{2\pi}{N a}\, n_i, \quad -N/2 \leq n_i<N/2,
\end{equation}
and:
\begin{equation}
    K_{\vec{q}\vec{q}^\prime}
    \defeq
    \frac{1}{N^D}\sum_{\vec{x},\vec{y}}  e^{i \vec{q} \cdot \vec{x}-i \vec{q^\prime} \cdot \vec{y}} V(\vert \vec{x}-\vec{y}\,\vert) = \delta_{\vec{q}\vec{q^\prime}} \tilde{K}(\vec{q}),
\end{equation}
where we used the integral representation of Kronecker's delta:
\begin{equation}
    \delta_{\vec{q}\vec{q^\prime}}=\frac{1}{N^D} \sum_{\vec{x}} e^{i (\vec{q}-\vec{q^\prime}) \cdot \vec{x}}.
\end{equation}
Then, using~\eqref{eq:V_Ising}, we get:
\begin{equation}
    \tilde{K}(\vec{q}\,)
    \defeq
    \sum_{\vec{x}} V(\vert \vec{x}\,\vert) e^{i \vec{q} \cdot \vec{x}}=2 J \sum_{\alpha=1}^D \cos \left(q_\alpha a\right).
\end{equation}
Assuming that $\Phi_{\vec{x}}$ varies slowly enough from one point to another, it can be replaced by a continuous field $\Phi(x)$.
Consequently $\tilde{\Phi}_{\vec{k}}$ is dominated by small momenta $\vert \vec{k} \vert \ll 1$, and we may expand
\begin{equation}
    \tilde{K}(\vec{q}\,) = 2DJ - J a^2 \vec{q}\,^2+\mathcal{O}(q_\alpha^4),
\end{equation}
where $q^2 = \vert \vec{q}\, \vert^2 = \sum_{\alpha=1}^D q_\alpha^2$.
The Gaussian measure for $\Phi$ is then approximated by:
\begin{equation}
    \frac{1}{2} T \sum_{\vec{x},\vec{y}} \Phi_{\vec{x}}  K_{\vec{x}\vec{y}}^{-1} \Phi_{\vec{y}} \approx \frac{1}{2}\sum_{\vec{k}}\tilde{\Phi}_{\vec{k}} \, \frac{T}{T_0} \left(1+\frac{a^2}{2D}\vec{q}\,^2\right)\, \tilde{\Phi}_{-\vec{k}},
    \label{eq:Gaussian_approx}
\end{equation}
with $T_0=2D J$.

The average spin relates to $\Phi$ via
\begin{equation}
    \langle S_{\vec{x}} \rangle = \langle \tanh \Phi_{\vec{x}} \rangle.
\end{equation}
Thus, in the phase where $T-T_c \ll 1$, we expect the deviation from the Gaussian model of zero global magnetization to remain small, which ensures the compatibility with $\vert \Phi_{\vec{x}} \vert \ll 1$ where we have
\begin{equation}
    \ln \left(\cosh \left(\Phi_{\vec{x}} \right)\right)= \frac{\Phi_{\vec{x}}^2}{2}-\frac{\Phi_{\vec{x}}^4}{12}+\mathcal{O}(\Phi_{\vec{x}}^6).
\end{equation}
In the continuous limit and for a smooth function, we use
\begin{equation}
    \sum_{\vec{k}} f(\vec{k})\approx \left(\frac{Na}{2\pi}\right)^D\, \int \dd k\, f(k),
\end{equation}
where $k \in \mathds{R}^D$.
We thus obtain an effective continuous field $\Phi(x)$, $x \in \mathds{R}^D$, which describes the collective behaviour of many spins, much as the equations of hydrodynamics describe the collective motion of many molecules.
After a rescaling of the field, we find that the partition function can be written in the form of a path integral:\footnote{%
    In the continuous limit, the product $\prod_{\vec{x}} \dd\Phi_{\vec{x}}$ becomes a functional measure, written $[\dd\Phi(x)]$, which integrates over all possible field configurations.
}
\begin{equation}
    Z_{\text{Ising}} \approx \int [\dd \Phi(x)]\, e^{-H_{\text{eff}}[\Phi]}\,,\label{Zeff}
\end{equation}
where $[\dd \Phi] = \prod_x \dd \Phi(x)$.
To leading order, the effective Hamiltonian is given by \emph{Landau theory}:
\begin{equation}
    H_{\text{eff}}[\Phi]=\frac{1}{2}\int \dd x \,\Phi(x) (-\Delta +m^2) \Phi(x) + \frac{u}{4}\int \dd x\, \Phi^4(x)+\mathcal{O}(\Phi^6).
    \label{Heff1}
\end{equation}
In this expression, $\Delta$ is the standard Laplacian on $\mathds{R}^D$; in Fourier space it produces the $q^2$ term in~\eqref{eq:Gaussian_approx}.
The parameter $u > 0$ is a \emph{coupling constant}, and the \emph{mass} $m^2$ is:\footnote{%
    For non-physicist readers, $m^2$ is simply an additional parameter.
}
\begin{equation}
    m^2=r(T-T_0),
    \label{eq:Landau_mass}
\end{equation}
where $r$ is a numerical constant.

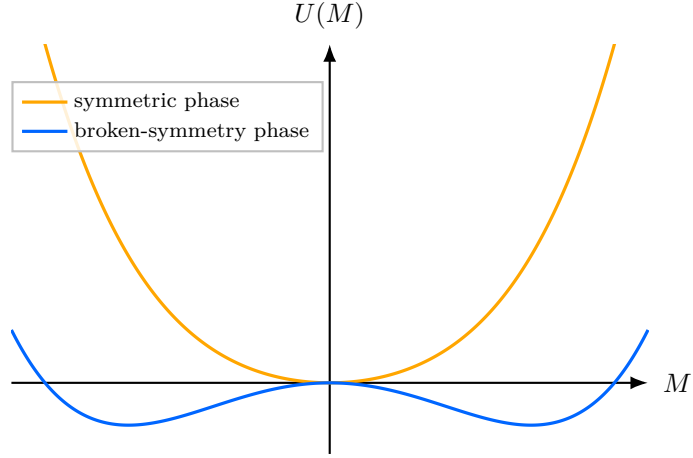
\begin{figure}[t]
    \centering
    \begin{tikzpicture}
    \begin{axis}[
            width=10cm,
            height=7cm,
            axis lines=middle,
            axis line style={-Latex},
            xlabel={$M$},
            ylabel={$U(M)$},
            xlabel style={right, xshift=2pt},
            ylabel style={above, yshift=2pt},
            xmin=-2.5, xmax=2.5,
            ymin=-2.1, ymax=10,
            ticks=none,
            legend pos=north west,
            legend cell align=left,
            legend style={
                    at={(0.0,0.91)},
                    fill=white,
                    fill opacity=0.85,
                    draw=gray!50,
                    text opacity=1,
                    font=\footnotesize
                },
            smooth,
            thick,
            no marks,
            every axis plot/.append style={line width=1.2pt},
        ]
        \addplot[
            domain=-2.5:2.5,
            samples=200,
            color=orange!70!yellow,
        ]
        {x^2 + 0.2*x^4};
        \addlegendentry{symmetric phase}

        \addplot[
            domain=-2.5:2.5,
            samples=200,
            color=blue!60!cyan,
        ]
        {-x^2 + 0.2*x^4};
        \addlegendentry{broken-symmetry phase}
    \end{axis}
\end{tikzpicture}
    \caption{%
        Behaviour of the Landau potential for $T>T_0$ (yellow curve) or for $T<T_0$ (blue curve).
    }
    \label{potentiel}
\end{figure}

To leading order, the integral is dominated by the saddle points of the Hamiltonian.
Assuming the latter is uniform, $\Phi(x)=M$, it is a solution to the equation:
\begin{equation}
    \frac{\partial U(M)}{\partial M}=0,
\end{equation}
the Landau potential $U$ being defined by:
\begin{equation}
    U(M)=\frac{1}{2} r(T-T_0) M^2+\frac{u}{4} M^4 \equiv \frac{H_{\text{eff}}[\Phi]}{(Na)^D}\Big\vert_{\Phi(x)=M}.
    \label{eq:Landau_potential}
\end{equation}
The only stable solution is $M=0$ when $T > T_0$, while two non-zero stable solutions appear when $T < T_0$, giving $M \sim \pm (T_0 - T)^{1/2}$, the zero solution becoming unstable (see \Cref{potentiel}).
At this order, the field theory is equivalent to mean-field theory: the critical exponent is $\beta = 1/2$.

Unlike the mean-field approach, field theory allows the study of fluctuations, which correspond to deviations $\delta \Phi(x)$ from the Landau solution, $\Phi(x) = M + \delta \Phi(x)$.
The first correction is:
\begin{equation}
    Z_{\text{Ising}}
    =
    \underbrace{e^{-H_{\text{eff}}[M]}}_{\text{saddle point}}
    \cdot
    \underbrace{\int [\dd \delta\Phi(x)]\, e^{-\frac{1}{2}\int \dd x\dd y\, \frac{\delta^2 H_{\text{eff}}}{\delta \Phi_0(x)\delta \Phi_0(y)} \delta \Phi(x)\delta \Phi (y)-\mathcal{O}(\delta \Phi^3)}}_{\text{fluctuations}},\label{equationmeanfield}
\end{equation}
where
\begin{equation}
    \frac{\delta}{\delta\Phi_0(x)}
    =
    \eval{\frac{\delta}{\delta\Phi(x)}}_{\Phi(x) = M}
\end{equation}
denotes the functional derivative, evaluated at $\Phi(x)=M$.\footnote{%
    The functional derivative $\delta F[\Phi]/\delta\Phi(x)$ generalises the ordinary derivative to functionals (functions of a function).
    It measures how the functional $F[\Phi]$ changes when the field $\Phi$ is varied at the point $x$, much as $\partial f/\partial x_i$ measures how a function $f(\vec{x})$ changes when the coordinate $x_i$ is varied.
}
For mean-field theory to be consistent, the influence of fluctuations must remain weak.
However, explicit evaluation shows that fluctuations are not always negligible.
The corrections to mean-field theory diverge when $D \leq 4$, which provides another derivation of the Ginzburg criterion:\footnote{%
    The Ginzburg criterion establishes the upper critical dimension $D_c = 4$, below which fluctuations invalidate mean-field theory.
    It states that mean-field theory is self-consistent only when the variance of the order parameter, averaged over a correlation volume $\xi^D$, is small compared to the order parameter squared: $\langle (\delta\Phi)^2 \rangle_\xi \ll M^2$.
    For the $\phi^4$ theory this condition fails for $D \leq 4$.
} in low dimensions, near the transition, correlations are too strong and the conditions for the central limit theorem break down.
This is a direct consequence of critical phenomena: the divergence of the correlation length implies a non-decoupling of physical scales.
Kadanoff's \rg approach thus inspired Wilson's general method, which founded the modern theory of the \rg.

The same logic (introducing an auxiliary field to decouple interactions, expanding around a Gaussian fixed point, and tracking how couplings run with scale) will be applied in Part~2 to the spectral density of the empirical correlation matrix, where the degrees of freedom in the tail of the eigenvalue spectrum replace the spin variables.

\section{Wilsonian Renormalization Group}\label{sec:wilson_rg}
The Gaussian approximation fails near the critical point: below four dimensions, fluctuation corrections diverge and perturbation theory breaks down.
Wilson's renormalization group resolves this by integrating out short-wavelength modes progressively, revealing which interactions survive at large scales.
This section introduces the formalism of the Wilsonian \rg, the concepts of renormalizability and canonical dimension, and the calculation of critical exponents.

\subsection{Failure of the Gaussian Theory at the Transition}\label{Gaussiantheory}

Wilson's renormalization group adopts Kadanoff's spin-block approach within the framework of field theory.
To understand why the \rg is needed, consider how fluctuations modify the critical temperature.\footnote{%
    This is known in physics as the \emph{one-loop} contribution.
}
Field theory allows the computation of these corrections, bypassing the limitations of the mean-field approach.
Among the contributions of the fluctuations $\delta\Phi$, the terms quadratic in $M$ shift the solution of the Landau equation~\eqref{eq:Landau_potential} and therefore the critical temperature itself.
To first order in $u$, one finds (see also \Cref{sec:app1}):
\begin{equation}
    T_0^* = T_0-\frac{3u}{r}\,\int \frac{\dd^D k}{(2\pi)^D} \, \frac{1}{k^2},\label{defTC}
\end{equation}
where $k^2 = \vert \vec{k} \vert^2$, with $k\in \mathds{R}^D$.
The integral diverges, yet the quantity it is meant to correct, $T_0^*$, is finite and measurable: $T_0$ comes from an analytical computation and represents the mean-field critical temperature, not the true transition point of the interacting theory.
The divergence signals that the integral must be regularised (made finite by introducing a cutoff) and that physical quantities should be expressed in terms of $T_0^*$ rather than $T_0$.

The same problem appears when correcting the mass~\eqref{eq:Landau_mass}.
At order $u$, the effective mass becomes:
\begin{equation}
    m^2_{\text{eff}}\approx r (T-T_0^*)\left(1-3 u\,\int \frac{\dd^D k}{(2\pi)^D} \, \frac{1}{(k^2)^2}\right).
\end{equation}
In $D$ dimensions the integral behaves as $\xi^{4-D}$ in the \ir (i.e.\ $k^2 \to 0$) and converges only for $D > 4$.
This prevents a perturbative study of the system near $T_c$.
The lattice spacing imposes a \uv (i.e.\ $k^2 \gg 1$) cutoff at $k = \pi/a$, and the system size $L$ provides a natural \ir cutoff at $k_0 = 2\pi/L$.\footnote{%
    The mode $k=0$ corresponds to a perfectly uniform magnetisation and carries no information about spatial structure.
    The \ir cutoff is therefore the spacing between modes, $\Delta k = 2\pi/L$, not $k=0$ itself.
    This is also confirmed by the maximum wavelength allowed on a lattice of size $L$, that is $\lambda = L$, which implies $\Delta k = 2\pi/L$.
}
The physical \ir cutoff, however, is set by the correlation length: $k_{\text{phys}} \sim 2\pi/\xi$.
For $D > 4$, the integral is finite in both the ultraviolet and the infrared, even as $T \to T_c$ and $\xi \to \infty$.
For $D \leq 4$, the integral diverges in the critical regime because the correlation length diverges.
Perturbation theory breaks down: the fluctuation corrections are no longer small, and the Gaussian approximation fails to describe the system accurately.

\subsection{Wilsonian Renormalization Group}

Wilson noticed that, as shown in the previous paragraphs, the problem arose from the non-decoupling of physical scales at the transition point~\cite{Wilson1,wilson1971renormalization}.
He proposed replacing the integration over all momenta with a partial, progressive integration which, at each step, would not jeopardize the validity of perturbation theory.
To understand Wilson's proposal, it is useful to return to the discretised version of the field theory:
\begin{equation}
    Z_{\text{Ising}}=\int \prod_{\vec{x}}\, \dd \Phi_{\vec{x}}\, e^{-H_{\text{eff}}[\Phi]},
\end{equation}
such that,
\begin{equation}
    H_{\text{eff}}[\Phi]=\frac{1}{2}\sum_{\vec{x}} \,\Phi_{\vec{x}} (-\Delta_{\text{dis}} +m^2) \Phi_{\vec{x}} + \frac{u}{4}\sum_{\vec{x}} \,\Phi^4_{\vec{x}}+ \frac{v}{6}\sum_{\vec{x}} \,\Phi^6_{\vec{x}}+\mathcal{O}(\Phi^8),\label{Heff2}
\end{equation}
where we have retained the sixth-order interactions here for the sake of the discussion, and $\Delta_{\text{dis}}$ denotes the discrete Laplacian (see \Cref{sec:app2}), computed from the discrete gradient:
\begin{equation}
    \left( \nabla_{\text{dis}} \Phi \right)_{\vec{x}}
    =
    \begin{pmatrix}
        \Phi_{\vec{x} + \hat{e}_1} - \Phi_{\vec{x}} \\
        \Phi_{\vec{x} + \hat{e}_2} - \Phi_{\vec{x}} \\
        \vdots                                      \\
        \Phi_{\vec{x} + \hat{e}_D} - \Phi_{\vec{x}}
    \end{pmatrix},
\end{equation}
where $\{ \hat{e}_i \}_{i=1}^D$ denotes the set of unit vectors along each spatial direction.

In Fourier decomposition, the measure is written as
\begin{equation}
    \prod_{\{ \vert k_i \vert \leq \pi/a\}} \dd \tilde{\Phi}_{\vec{k}}.
\end{equation}
Wilson's proposition consists of splitting the field into two parts by introducing a field $\Psi(\vec{x})$ associated with the \enquote{fast modes}, such that:
\begin{equation}
    \Psi(\vec{k})
    =
    \begin{cases}
        \tilde{\Phi}_{\vec{k}} & \text{if}~\exists i \in \{1, 2, \dots, D\}~\text{s.t.}~\pi/a^\prime \leq k_i \leq \pi/a \\
        0                      & \text{otherwise}
    \end{cases}
\end{equation}
The new lattice spacing $a^\prime > a$ is such that $a^\prime = (1+b) a$, with $b\ll 1$, i.e.\ a small dilation of the initial lattice.
We shall designate the complementary field as $\Phi_{a^\prime}(\vec{x}) \defeq \Phi_{\vec{x}}-\Psi({\vec{x}})$, such that:
\begin{equation}
    \sum_{\vec{x}}\, \Phi_{a^\prime}({\vec{x}}) \Psi({\vec{x}}) =0,
\end{equation}
which ensures that there is no overlap of the two fields.
The measure thus decomposes into two terms.
\begin{equation}
    \prod_{\vec{x}} \dd\Phi_{\vec{x}} = \prod_{\vec{x}} \dd \Phi_{a^\prime}(\vec{x}) \dd \Psi(\vec{x}),
\end{equation}
and the initial partition function becomes:
\begin{equation}
    \begin{split}
        Z_{\text{Ising}}
        & =
        \int \prod_{\vec{x}}\, \dd  \Phi_{a^\prime}(\vec{x})\,
        e^{-H_{\text{eff}}[\Phi_{a^\prime}]}
        \\
        & \times \int \prod_{\vec{x}}\,\dd  \Psi(\vec{x})\, \exp \left(-\frac{1}{2}\sum_{\vec{x}} \,\Psi(\vec{x}) (-\Delta_{\text{dis}} +m^2) \Psi(\vec{x}) \right.
        \\
        & - \left. \frac{u}{4}\sum_{\vec{x}} \,(4\Phi_{a^\prime}^3(\vec{x}\,)\Psi(\vec{x})+6 \Phi_{a^\prime}^2(\vec{x}\,)\Psi^2(\vec{x})+4\Phi_{a^\prime}(\vec{x}\,)\Psi^3(\vec{x})+\Psi^4(\vec{x}))+\mathcal{O}(\Phi^6)\right).
    \end{split}
    \label{Scaling1}
\end{equation}
The field $\Psi$ embodies the short-wavelength fluctuations, and therefore no divergence is expected for the corresponding path integral.
After Taylor expanding to the first order in $u$, the integral over $\Psi$, denoted as $Z_{\Psi}$, can be written as:\footnote{%
    As the Gaussian measure is symmetric, the odd terms in $\Psi$ all have a vanishing mean.
    Furthermore, the constant term arising from the mean of the $\Psi^4$ term is unnecessary and simply omitted.
}
\begin{equation}
    Z_{\Psi} = Z_{\Psi}^{(0)} - \frac{3u}{2} \int \prod_{\vec{x}}\,\dd  \Psi(\vec{x})\, e^{-\frac{1}{2}\sum_{\vec{x}} \,\Psi(\vec{x}) (-\Delta_{\text{dis}} +m^2) \Psi(\vec{x})}\sum_{\vec{x}} \, \Phi_{a^\prime}^2(\vec{x}\,)\Psi^2(\vec{x})+\mathcal{O}(\Phi^6,u^2),
    \label{firstorder}
\end{equation}
with:
\begin{equation}
    Z_{\Psi}^{(0)} \defeq \int \prod_{\vec{x}}\,\dd  \Psi(\vec{x})\, e^{-\frac{1}{2}\sum_{\vec{x}} \,\Psi(\vec{x}) (-\Delta_{\text{dis}} +m^2) \Psi(\vec{x})}.
\end{equation}

The rest of the calculation relies on \emph{Wick's theorem}, a well-known property of Gaussian integrals which we recall here:
\begin{theorem}{Wick's Theorem}{theoremWick}
    Let $\vec{x} = (x_1, x_2, \dots, x_n)$ be a vector and let there be a quadratic form defined by an invertible symmetric matrix $A$.
    The Gaussian integral is defined by the measure:
    \begin{equation}
        Z_0 = \int \dd^n x \, e^{-\frac{1}{2} x_i A_{ij} x_j}.
    \end{equation}
    The mean value (the \emph{moment}) of the product of $2n$ variables is given by:
    \begin{equation}
        \begin{split}
            \langle x_{1} x_{2} \dots x_{2n} \rangle
            & =
            \frac{1}{Z_0} \int \dd^n x \, \prod_{i=1}^{2n} x_i \, e^{-\frac{1}{2} x_i A_{ij} x_j} \\
            & =
            \sum_{\text{pairings}} (A^{-1})_{i_1 i_2} (A^{-1})_{i_3 i_4} \dots (A^{-1})_{i_{2n-1} i_{2n}}.
        \end{split}
    \end{equation}
    Where the sum is over all possible pairings of the random variables.
    For instance, we have:
    \begin{equation}
        \langle x_1 x_2 x_3 x_4 \rangle = (A^{-1})_{12}(A^{-1})_{34} + (A^{-1})_{13}(A^{-1})_{24} + (A^{-1})_{14}(A^{-1})_{23}.
    \end{equation}
\end{theorem}
Applying this theorem for the computation of the integral~\eqref{firstorder}, we obtain:
\begin{equation}
    \frac{Z_{\Psi}}{Z_{\Psi}^{(0)}} = e^{-\frac{1}{2}\sum_{\vec{x}} \,\Phi_{a^\prime}(\vec{x}) \Delta m^2 \Phi_{a^\prime}(\vec{x})}+\mathcal{O}(u^2),
    \label{contoneloop1}
\end{equation}
with, in Fourier modes and in the continuum limit, recalling that $a^\prime=(1+b) a$:\footnote{%
We compute an integral of the form $\int dk F(k)$ over the domain $[-\pi/a,\pi/a]^D/[-\pi/a^\prime,\pi/a^\prime]^D$.
The factor $2 \times D$ counts the number of contributions of the same type, once expanded in powers of $b$.
}
\begin{equation}
    \Delta m^2 \defeq \frac{3 u}{(2\pi)^D} (D\times 2)\left(\prod_{i=2}^D \int_{-\pi/a}^{\pi/a} \dd k_i\right)\, \frac{1}{\sum_{i=2}^D k_i^2+(\pi/a)^2+m^2}+\mathcal{O}(b^2).
\end{equation}
This integral can be calculated using an approximation, recalling that in the critical region $m^2 \ll 1$, and specifically $m^2 \ll \pi/a^\prime$.
We can therefore treat $m^2$ as a perturbation.
We now use the Schwinger representation:
\begin{equation}
    \frac{1}{k^2} = \int_{0}^\infty \dd \alpha\, e^{-\alpha k^2}.
\end{equation}
Assuming $a \ll 1$, we obtain:
\begin{equation}
    \int_{-\pi/a}^{\pi/a} \dd k_i\, e^{-\alpha k_i^2} \approx \frac{2 \pi}{a},
\end{equation}
then, setting $\Lambda_{\text{\uv}} = \pi/a$:
\begin{equation}
    \Delta m^2 = \frac{3 D u}{(2\pi)^D} 2^{D} \left(\frac{\pi}{a}\right)^{D-2} b = \frac{3 D u}{\pi^D} \Lambda_{\text{\uv}}^{D-2}b.
\end{equation}
The contribution in \eqref{contoneloop1}, once substituted into \eqref{Scaling1}, yields:
\begin{equation}
    Z_{\text{Ising}}
    =
    \int \prod_{\vec{x}}\, \dd  \Phi_{a^\prime}(\vec{x})\, e^{-H_{\text{eff}}[\Phi_{a^\prime}]} e^{-\frac{1}{2}\sum_{\vec{x}} \,\Phi_{a^\prime}(\vec{x}) \Delta m^2 \Phi_{a^\prime}(\vec{x})}+\mathcal{O}(u^2).
\end{equation}
The mass of the field $\Phi_{a^\prime}$ is thus modified.
The new mass, which we shall call $m^2(a^\prime)$, is related to the original mass \enquote{at scale $a$} (which we will denote as $m^2(a)$ rather than $m^2$) by:
\begin{equation}
    m^2(a^\prime)=m^2(a)+\Delta m^2.
    \label{flowm1}
\end{equation}
This equation summarises the entire Wilsonian \rg philosophy.
The idea that integrating out short-wavelength modes generates effective large-scale dynamics is precisely what the signal-detection framework in \Cref{part2} exploits.
In that setting, the fast modes correspond to the bulk eigenvalues of the empirical correlation spectrum, while the effective large-scale dynamics describe the spectral tail where the signal resides.
The partial integration of the \uv modes $\Psi$ generates effective interactions for the \ir field $\Phi_{a^\prime}$, transforming the study of field theory into a dynamical problem.
The mass is, however, not the only coupling to be modified; at order $u^2$, for example, the quartic coupling at scale $a^\prime$ becomes (using the same notation as before):
\begin{equation}
    u(a^\prime)=u(a)+\Delta u,
\end{equation}
where:
\begin{equation}
    \Delta u =-\frac{9 D u^2 \Lambda_{\text{\uv}}^{D-4}}{\pi ^D}\,b+\mathcal{O}(u^3). \label{flowu1}
\end{equation}
The powers of $\Lambda_{\text{\uv}}$ in \eqref{flowm1} and \eqref{flowu1} will play a crucial role.
These equations can be translated into an infinitesimal version by taking the limit $b \to 0$ and because $a^\prime - a = a b$,
\begin{align}
    \Lambda_{\text{\uv}} \frac{\dd m^2}{d\Lambda_{\text{\uv}}} & = -\frac{3 D u}{\pi^D} \Lambda_{\text{\uv}}^{D-2}     \\
    \Lambda_{\text{\uv}} \frac{\dd u}{d \Lambda_{\text{\uv}}}  & = \frac{9 D u^2 }{\pi ^D} \Lambda_{\text{\uv}}^{D-4}.
\end{align}
If we return to Kadanoff's approach, one element remains: the rescaling of the lattice after the partial integration, $a^\prime \to a^\prime / (1+b)$, so as to always compare theories defined on the same lattice.
One way to account for this rescaling is to define \emph{dimensionless} couplings, which explains why the powers of $\Lambda_{\text{\uv}}$ are essential.
The previous equations depend explicitly on the rescaling through these powers, but it is possible to find a transformation of the couplings that eliminates this explicit dependence.
The second equation for $u$ shows that the choice $u \sim \Lambda_{\text{\uv}}^{4-D}$ eliminates this dependence on both sides of the equation.
By then substituting this choice into the equation for $m^2$, we see that the choice $m^2 \sim \Lambda_{\text{\uv}}^2$ in turn eliminates any explicit trace of $\Lambda_{\text{\uv}}$.
We therefore define the \emph{dimensionless} parameters $\bar{m}^2$ and $\bar{u}$ such that:
\begin{equation}
    m^2 \defeq \bar{m}^2 \Lambda_{\text{\uv}}^2, \qquad u \defeq \bar{u} \Lambda_{\text{\uv}}^{4-D},\label{dimensionless}
\end{equation}
and in terms of these couplings, the flow equations become a self-contained system:
\begin{equation}
    \Lambda_{\text{\uv}} \frac{\dd \bar{m}^2}{d\Lambda_{\text{\uv}}} =- 2\bar{m}^2 -\frac{3 D \bar{u}}{\pi^D},
    \qquad
    \Lambda_{\text{\uv}} \frac{\dd \bar{u}}{d \Lambda_{\text{\uv}}} =(D-4)\bar{u}+\frac{9 D \bar{u}^2}{\pi^D}.
\end{equation}
The exponents in \eqref{dimensionless} are called \emph{canonical dimensions}, and this concept will play an important role in this review.
For an ordinary field theory defined on the manifold $\mathds{R}^D$, these canonical dimensions correspond precisely to the ordinary dimensions of the couplings.
Indeed, $H_{\text{eff}}$ in \eqref{Heff1} must be dimensionless, and since $\dd x$ is dimensionful, this fixes the dimension of the couplings.
In units of the cutoff $\Lambda_{\text{\uv}}$, we then recover the previous result.
However, canonical dimensions are far more fundamental, in the sense that they do not require the existence of a reference dimension external to the problem.

The concept of canonical dimension, and in particular its sensitivity to the shape of the spectrum, lies at the heart of the signal-detection method developed in \Cref{part2}.
There, the empirical eigenvalue distribution plays the role of the spectral measure, and a scale-dependent departure of the canonical dimensions from their pure-noise values signals the presence of information.

\subsection{Relevant, Irrelevant, and Marginal Couplings}

Strictly speaking, the model \eqref{Heff2} is defined on a lattice.
There is, a priori, no \enquote{external} dimension to the problem.
Assigning a dimension to the parameter $a$ indeed fixes the dimension of the Laplacian, and by extension that of $m^2$.
One could then deduce the dimension of $u$ by reversing the previous argument.
But one can just as well consider things abstractly, without fixing a dimension at the outset.
The previous reasoning shows that we can associate a dimension with the coupling directly from the flow itself, that is from the way in which the couplings \enquote{change} under the action of the \rg.
This type of definition is found notably in background-independent field theories, which are of particular interest in quantum gravity \cite{LahocheBeyond,carrozza2014tensorial}.
This canonical dimension depends on only three factors:
\begin{enumerate}
    \item The form of the Gaussian term.
    \item The structure of the interactions (the way in which the fields interact).
    \item The spectral distribution of momenta:\footnote{%
          In the continuum limit, the measure $\dd^D k$ is written in spherical coordinates as $\simeq (k^2)^{\frac{D-2}{2}} \dd k^2$.}
          \begin{equation}
              \rho(k^2) \sim (k^2)^{\frac{D-2}{2}}.
              \label{eq:rhok2}
          \end{equation}
\end{enumerate}

\begin{figure}
    \centering
    \begin{tikzpicture}[
        >=Latex,
        axis/.style={-Latex, black},
        flow/.style={black, very thick, rounded corners=35pt,
                postaction={decorate, decoration={markings,
                                mark=between positions 0.2 and 0.9 step 0.2 with {\arrow{Latex}}}
                    }
            },
        flowarrow/.style={black, -Latex, very thick},
        dashedline/.style={black, dashed, dash pattern=on 1pt off 3pt},
    ]

    \fill[black!5, draw=black!50, thick] (0.0, 0.0) -- (5.0, 0.0) -- (3.0, -4.0) -- (-2.0, -4.0) -- cycle;
    \draw[axis] (0.0, 0.0) -- (3.5, 0.0) node[above] {$\bar{u}$};
    \draw[axis] (0.0, 0.0) -- (-1.5, -3.0) node[left] {$\bar{m}^2$};
    \draw[axis] (0.0, 0.0) -- (0.0, 3.0) node[right] {$\bar{v}$};
    \draw[axis] (0.0, 0.0) -- (-3.0, -1.5) node[above left] {$\bar{u}_\infty$};
    \draw[dashedline] (0.0, 1.5) .. controls (-1.25, 1.0) and (-1.75, 0.0) .. (-1.5, -0.75);

    \coordinate (R) at (2.5, -2.75);
    \node at (3.6, -0.85) {$\mathbf{S}_{R}$};
    \node at (2.1, 1.96) {$\mathbf{S}$};

    \draw[flow, flowarrow] (1.1, 0.9) .. controls (0.5, 0.5) and (0.1, -0.1) .. (R);
    \draw[flow, flowarrow] (-0.5, 2.6) .. controls (-1.5, 1.0) and (-1.0, -0.75) .. (R);
    \draw[flow, flowarrow] (-4.5, -0.25) .. controls (-3.1, 0.75) and (0.8, -3.0) .. (R);
\end{tikzpicture}
    \caption{%
        Illustration of the flow in the neighbourhood of the Gaussian point in dimension $D \leq 4$: the flow converges toward a subspace spanned by essential and marginal interactions, the \enquote{renormalizable} subspace $\mathbf{S}_{R}$.
        In the case where the flow converges toward a stable fixed point, it is called the \enquote{large river effect}.
    }\label{FigLRE}
\end{figure}
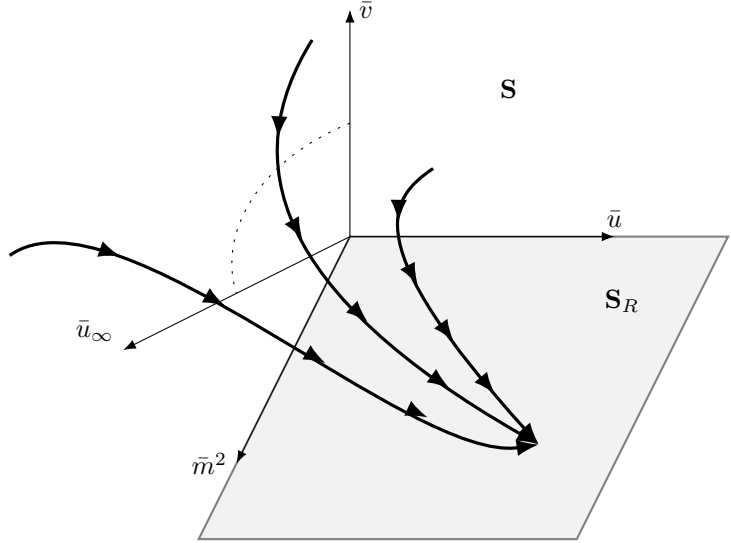

The reason canonical dimensions are so essential is that they determine the large-scale behaviour of the flow in the neighbourhood of the Gaussian region (the restriction to the neighbourhood of the Gaussian region is justified in the following section).
When $D>4$, the equations show that, near the Gaussian point, the quartic interaction $\bar{u}$ is exponentially suppressed in the \ir as $\ln(\Lambda_{\text{\uv}}) \to -\infty$.
For a sextic coupling of the form $(v/6)\Phi_{\vec{x}}^6$, one would similarly find:
\begin{equation}
    \Lambda_{\text{\uv}} \frac{\dd \bar{v}}{d \Lambda_{\text{\uv}}}= (2D-6) \bar{v}+ \mathcal{O}(\bar{v}\bar{u}),
\end{equation}
and here again, $\bar{v}$ is exponentially suppressed in the \ir when $D>4$.
The same would hold true for all couplings, such that in high dimensions, the flow is essentially dominated by the Gaussian contribution (the Gaussian regime is said to be stable).
We find here, as a mirror image, an alternative reading of the Ginzburg criterion: the fact that fluctuations are large in dimension $D\leq 4$ translates into an instability of the Gaussian region, driven by the growth of non-Gaussian couplings.
The flow thus moves irreversibly away from the Gaussian region to reach a non-trivial fixed point, as we shall see in the following section.

Canonical dimensions determine the stability of the Gaussian region.
When the canonical dimension is negative, the coupling is said to be \emph{irrelevant} (or non-renormalizable) and will be suppressed in the \ir.
If the dimension is positive, the coupling is \emph{relevant} and becomes important at low energy.
In the case where the dimension is zero (as is the case for $u$ in dimension $D=4$), we call it a \emph{marginal coupling}, and whether it is relevant or irrelevant then depends on the next order.
However, since the growth or decay is slow (logarithmic), the coupling will still play an important role at large scales (see \Cref{FigLRE}).

Underlying this classification is an abstract notion of scale, defined directly from the spectrum of the kinetic term rather than from any external geometric dimension.
Following the partial integration procedure, small eigenvalues (high momenta) correspond to small scales (\uv), and conversely, large eigenvalues correspond to large scales (\ir).
It is this spectral definition of scale that makes the canonical-dimension analysis applicable to the data-analysis setting of \Cref{part2}, where the empirical eigenvalue distribution replaces the Laplacian spectrum.

\subsection{How Does the Renormalization Group Solve the Critical Regime?}

Let us now see how the Wilsonian \rg solves the problem of calculating critical exponents.
We have seen that mean-field theory, based on a quasi-Gaussian approximation, is only valid in high dimensions, $D>4$.
In low dimensions, $D\le 4$, fluctuations are too strong, and, from the perspective of the flow, this manifests as an instability of the Gaussian region.

Since the upper critical dimension of the $\phi^4$ theory is $D=4$, i.e.\ the dimension at which the quartic coupling becomes marginal, it is natural to expand around this value.
We therefore parametrise the dimension as $D=4-\epsilon$, with $\epsilon >0$.
Retaining only the relevant couplings, the flow equations become:
\begin{align}
    \Lambda_{\text{\uv}} \frac{\dd \bar{m}^2}{\dd \Lambda_{\text{\uv}}}
     & =
    - 2\bar{m}^2 -\frac{12 \bar{u}}{\pi^4} +\frac{12}{\pi^4} \bar{u}\bar{m}^2+\mathcal{O}(\bar{u}\epsilon, \bar{u} (\bar{m}^2)^2), \\
    \Lambda_{\text{\uv}} \frac{\dd \bar{u}}{\dd \Lambda_{\text{\uv}}}
     & =
    -\epsilon \bar{u}+\frac{36 \bar{u}^2}{\pi^4}+\mathcal{O}(\bar{u}^2\epsilon ).
\end{align}
\begin{remark}{Scheme Dependence of Beta Functions}{remarkScheme}
    The values of the beta functions differ from those generally obtained in the literature, which is essentially due to the geometry of momentum space.
    We have chosen a cubic cutoff $[-\pi/a,\pi/a]^D$, though a spherical cutoff $k^2 \le \pi^2/a^2$ is usually the preferred choice.
    The beta function coefficients are not universal: only the critical exponents, when computed to all orders in the loop expansion, are scheme-independent.
    In what follows, we therefore quote the standard universal values for the critical exponents rather than those that would be obtained from a one-loop diagonalisation of the stability matrix constructed from the beta functions above.
    The reader interested in the detailed scheme dependence may consult the standard literature (e.g.\ \cite{Zinn1}).
\end{remark}
Notice that we have included a term proportional to $\bar{u} \bar{m}^2$, which arises from the one-loop correction to the two-point function and is needed to capture the shift of the mass scaling dimension away from its canonical value at the interacting fixed point.
Strictly speaking, these equations are only valid in the neighbourhood of the Gaussian fixed point.
The advantage of the $\epsilon$-expansion is that it improves the convergence of the equations, particularly in the vicinity of the non-trivial fixed point, as we shall see.
Notice also that, following standard notation, we will denote the derivative of a given dimensionless quantity $\bar{X}$ with respect to $\ln \Lambda_{\text{\uv}}$ by $\beta_X$, and we shall refer to these as \emph{beta functions}.

\begin{figure}
    \centering
    \begin{tikzpicture}[
        >=Latex,
        axis/.style={-Latex, black, thick},
        rgtraj/.style={black, rounded corners=25pt,
                postaction={decorate, decoration={markings,
                                mark=between positions 0.3 and 1.0 step 0.3 with {\arrow{Latex[black]}}}
                    }
            },
        redtraj/.style={red,
                postaction={decorate, decoration={markings,
                                mark=between positions 0.3 and 1.0 step 0.3 with {\arrow{Latex[red]}}}
                    }
            },
        criticalline/.style={blue,
                postaction={decorate, decoration={markings,
                                mark=between positions 0.3 and 1.0 step 0.3 with {\arrow{Latex[blue]}}}
                    }},
    ]

    \def\xshift{6.5}
    \def\yshift{-1.75}
    \coordinate (G1) at (0.0, 0.0);
    \coordinate (G2) at (0.0+\xshift, 0.0+\yshift);
    \coordinate (WF) at (2.0, -1.5);

    \draw[draw=none, fill=red!15] (0.0, 0.5) -- (2.0, 0.5) -- (2.0, -4.0) -- (0.0, -4.0) -- cycle;

    \node at (2.0, -4.5) {$\epsilon > 0$};
    \node at (2.0+\xshift, -4.5) {$\epsilon < 0$};

    \draw[axis] (G1) node[left] {$G$} -- (4.0, 0.0) node[above] {$\bar{u}$};
    \draw[axis] (G2) node[left] {$G$} -- (4.0+\xshift, 0.0+\yshift) node[above] {$\bar{u}$};
    \draw[axis] (0.0, 0.5) -- (0.0, -4.0) node[left] {$\bar{m}^2$};

    \draw[criticalline] (WF) node[above right] {WF} -- (2.0, 0.5);
    \draw[criticalline] (WF) -- (2.0, -4.0);
    \draw[criticalline] (G2) -- (0.0+\xshift, 2.25+\yshift);
    \draw[criticalline] (G2) -- (0.0+\xshift, -2.25+\yshift) node[left] {$\bar{m}^2$};

    \draw[redtraj] (G1) -- (WF);
    \draw[redtraj] (4.0, -3.0) -- (WF);
    \draw[redtraj] (3.5+\xshift, -1.75+\yshift) -- (G2);

    \draw[rgtraj] (G1) .. controls (1.0, -0.6) and (1.6, -0.4) .. (1.75, 0.5);
    \draw[rgtraj] (G1) .. controls (1.5, -0.8) and (1.8, -1.35) .. (1.9, 0.5);
    \draw[rgtraj] (G1) .. controls (1.25, -1.5) and (1.75, -3.1) .. (1.95, -4.0);
    \draw[rgtraj] (G1) .. controls (0.5, -1.0) and (1.1, -2.8) .. (1.26, -4.0);
    \draw[rgtraj] (4.0, -2.75) .. controls (2.9, -1.5) and (2.45, -0.8) .. (2.18, 0.5);
    \draw[rgtraj] (4.0, -3.3) .. controls (3.4, -2.9) and (2.25, -1.5) .. (2.1, -4.0);

    \draw[rgtraj] (3.5+\xshift, -1.5+\yshift) .. controls (1.5+\xshift, -0.7+\yshift) and (0.6+\xshift, -0.2+\yshift) ..  (0.3+\xshift, 2.25+\yshift) node[right] {\small $m^2 > 0$};
    \draw[rgtraj] (3.5+\xshift, -1.9+\yshift) .. controls (1.5+\xshift, -1.0+\yshift) and (0.3+\xshift, -0.5+\yshift) ..  (0.3+\xshift, -2.25+\yshift) node[right] {\small $m^2 < 0$};

    \fill[black] (G1) circle (2pt);
    \fill[black] (G2) circle (2pt);
    \fill[blue] (WF) circle (2pt);
\end{tikzpicture}
    \caption{%
        Behaviour of the \rg in the vicinity of the Gaussian fixed point for $\epsilon>0$ (on left) and $\epsilon<0$ (on right).
        The red region corresponds to the Gaussian region.
        Arrows are oriented toward \ir scales.
    }
    \label{WilsonFisherFig}
\end{figure}
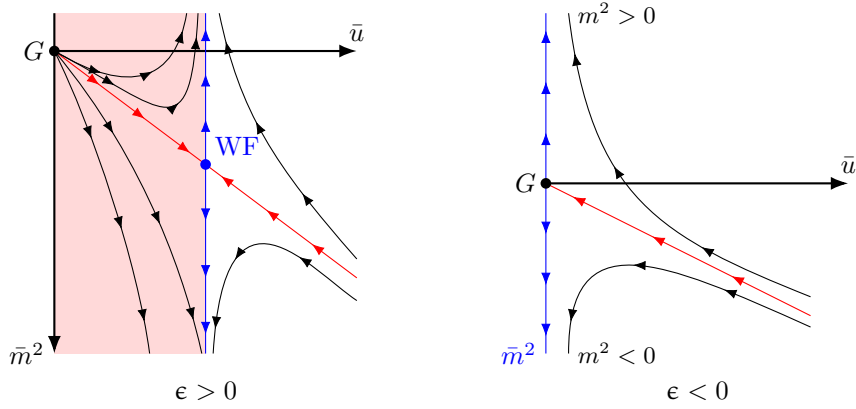

Let us look for the fixed points of these flow equations.
By solving $\beta_{m^2}=\beta_u=0$, we find two solutions:
\begin{itemize}
    \item The Gaussian fixed point $\bar{m}^2=\bar{u}=0$.
    \item An interacting fixed point of order $\epsilon$ : $\bar{m}^2_*=-\epsilon/6$, $\bar{u}_*=\epsilon \pi^4/36$.
\end{itemize}
Notice that when $\epsilon < 0$ ($D > 4$), the quartic coupling of the non-Gaussian fixed point $\bar{u}_*$ becomes negative, rendering the theory unstable: only the Gaussian fixed point exists in dimension $D > 4$.
This interacting fixed point is called, for historical reasons, the \emph{Wilson-Fisher} (\wf) fixed point.
To understand its properties, it is useful to linearise the flow around this fixed point:
\begin{equation}
    \begin{pmatrix}
        \beta_{m^2}(\bar{m}^2,\bar{u})
        \\
        \beta_u(\bar{m}^2,\bar{u})
    \end{pmatrix}
    = \mathbf{M}
    \begin{pmatrix}
        \bar{m}^2-\bar{m}^2_*
        \\
        \bar{u}-\bar{u}_*
    \end{pmatrix},
\end{equation}
where $\mathbf{M}$ is the \emph{stability matrix}:
\begin{equation}
    \mathbf{M}
    =
    \eval{%
        \begin{pmatrix}
            \partial_{\bar{m}^2} \beta_m(\bar{m}^2,\bar{u}) & \partial_{\bar{u}} \beta_m(\bar{m}^2,\bar{u}) \\
            \partial_{\bar{m}^2} \beta_u(\bar{m}^2,\bar{u}) & \partial_{\bar{u}} \beta_u(\bar{m}^2,\bar{u}) \\
        \end{pmatrix}
    }_{\text{fixed point}}.
\end{equation}
We can diagonalise this matrix, and by calling $\theta_1$ and $\theta_2$ the opposite of its eigenvalues, and $\mathbf{O}$ the matrix of the change of basis,
\begin{equation}
    \mathbf{O}^{-1} \mathbf{M}\mathbf{O}
    =
    \begin{pmatrix}
        -\theta_1 & 0         \\
        0         & -\theta_2
    \end{pmatrix}.
\end{equation}
We can then define the couplings in the eigenbasis:
\begin{equation}
    \begin{pmatrix}
        \bar{h}_1 \\ \bar{h}_2
    \end{pmatrix}
    =
    \mathbf{O}^{-1}
    \begin{pmatrix}
        \bar{m}^2 \\ \bar{u}
    \end{pmatrix},
\end{equation}
we have at the linear order:
\begin{equation}
    \beta_{h_1}=-\theta_1 \bar{h_1},
    \qquad
    \beta_{h_2}= -\theta_2 \bar{h}_2.
\end{equation}
The eigenvalues $\theta_1$ and $\theta_2$ are \emph{critical exponents}, and we shall see that they determine the ordinary critical exponents, such as $\nu$ or $\beta$.
Around the Gaussian fixed point, we find $\theta_1=2$ and $\theta_2=\epsilon$: the critical exponents are indeed the canonical dimensions.
Around the \wf fixed point, we get:
\begin{equation}
    \theta_1=  2-\frac{\epsilon}{3},
    \qquad
    \theta_2=-\epsilon.
\end{equation}

Critical exponents generalize canonical dimensions around the non-Gaussian fixed point.
They are associated with the effective couplings along the principal directions spanned by the eigenvectors of $\mathbf{M}$.
Note that the same nomenclature used for canonical dimensions is applied to critical exponents.
Thus, a direction is considered relevant if the exponent is positive, and irrelevant otherwise.

\Cref{WilsonFisherFig} summarises the properties of the flow in the neighbourhood of the fixed points, for $D < 4$ (left) and $D > 4$ (right):
\begin{itemize}
    \item When $D < 4$, the Gaussian fixed point is unstable and the transition is controlled by the non-Gaussian fixed point: the critical surface (red line) separates the phase space into two distinct regions, one where the flow converges toward positive masses (symmetric phase, without magnetisation) and the other where the flow converges toward negative masses (broken phase, with macroscopic magnetisation).
          The transition line corresponds to the blue line, where the parameter ${h}_1$ determines the critical temperature, fixed at the \wf fixed point.
    \item When $D > 4$, the Gaussian point is stable and controls the transition.
          The mass axis then identifies with the transition line.
          In this regime, Gaussian theory remains a good approximation.
\end{itemize}

Let us see through an example how the \rg allows us to determine a critical exponent, say $\nu$ (see equation \eqref{exponu}).
Let us place ourselves along the blue line, to which the flow converges in the \ir, so that the irrelevant parameter $h_2$ can be neglected at large scales.
The correlation length $\xi$ then depends only on the relevant parameter $h_1$.
Dimensional analysis requires the argument of $\xi$ to be dimensionless.
Since $h_1$ has scaling dimension $\theta_1 = 2 - \epsilon/3$, the dimensionless combination is:
\begin{equation}
    \xi = \Lambda_{\text{\uv}}^{-1}\, f\left(\Lambda_{\text{\uv}}^{-\theta_1} \bar{h}_1\right),
\end{equation}
where the prefactor $\Lambda_{\text{\uv}}^{-1}$ accounts for the canonical dimension of $\xi$.
For the function $f$ to cancel the explicit $\Lambda_{\text{\uv}}$ dependence of the prefactor, we must have $f(x) \propto x^{-1/\theta_1}$, yielding:
\begin{equation}
    \xi \propto \vert h_1\vert^{-\frac{1}{2}-\frac{\epsilon}{12}+\mathcal{O}(\epsilon^2)}.
\end{equation}
Since the critical parameter $h_1 \propto (T-T_c)$, we can deduce that:
\begin{equation}
    \nu=\frac{1}{2}+\frac{\epsilon}{12}+\mathcal{O}(\epsilon^2).
\end{equation}
By setting $\epsilon=1$, we find $\nu= 0.58$, to be compared with the empirical value $\nu \approx 0.63$.
The agreement remains quite qualitative here, however, since by imposing $\epsilon=1$ we are effectively moving out of the perturbative regime, since the fixed point then becomes of order $1$, invalidating the power series expansion in $u$.
Recourse to non-perturbative methods, such as those we will see in the following sections, shows excellent agreement with empirical results.
This obviously goes well beyond the scope of this modest introduction, and the curious reader may consult the references provided at the beginning of this part.

\section{The Lesson Learnt}\label{sec:lesson}
Before moving on, let us summarise what we have learnt from the Ising model and the \rg.
We have seen that at an intermediate (mesoscopic) scale, the complex collective behaviour of spins can be effectively described by a continuous field $\Phi(x)$ on $\mathds{R}^D$, corresponding to an average of the spins over a mesoscopic region of space, large relative to the microscopic scale of atoms, but very small compared to the scale of the material.
In the Ising model of ferromagnetism, the field $\Phi$ represents the local magnetisation.
However, this is not the only interpretation of the Ising field and its universality, and its meaning depends on the model.
The key point is that $\Phi_D^4$ theory (the quartic scalar field theory in $D$ spacetime dimensions) is the universal effective description of a broad family of microscopic models near criticality.
On this view, ordinary field theories are effective descriptions of an underlying discrete microscopic reality: they encode long-range correlations and forget microscopic details.
The \rg supports this interpretation directly: as short-wavelength modes are averaged out (the procedure known as \emph{coarse-graining}), the microscopic details decouple from the effective description.
In the Ising example in dimension $D=4$, the mass term is relevant, the quartic coupling is marginal, and every sextic and higher-order interaction is irrelevant (see \Cref{tab:couplings} for the precise criterion).
The \rg flow drives the irrelevant couplings to zero, so the \ir effective theory is dominated by one relevant and one marginal operator, with all higher-order interactions integrated away.
\Cref{fig0} illustrates this mechanism: a large family of microscopic Hamiltonians $H_1$, $H_2$, \dots are well described at large scales by an effective Hamiltonian, generally denoted $\Gamma$ (the \emph{effective action}).

\begin{table}[t]
    \centering
    \caption{%
        Classification of couplings based on their canonical dimensions, considering $D$ spacetime dimensions and a $\Phi^n$ interaction term.
    }\label{tab:couplings}
    \begin{tabular}{@{}lc@{}}
        \toprule
        \textbf{Classification} & \textbf{Dimension}   \\
        \midrule
        Relevant                & $n < \frac{2D}{D-2}$ \\
        Marginal                & $n = \frac{2D}{D-2}$ \\
        Irrelevant              & $n > \frac{2D}{D-2}$ \\
        \bottomrule
    \end{tabular}
\end{table}

The \rg recasts the field theory as a dynamical system: the flow of the couplings determines the large-scale behaviour of the theory.
Scale is set by the spectrum of the Gaussian kernel (the Laplacian, in the cases treated here), where small eigenvalues correspond to large distances, i.e.\ the \ir regime.
The \rg thus constructs effective models by keeping the large-scale physics fixed while modifying the effective interactions.
The influence of microscopic fluctuations is absorbed into the coupling constants at each \rg step.
The flow near the Gaussian fixed point (the free, non-interacting theory around which interactions are perturbations) is controlled by the canonical dimensions, which depend on the structure of the interactions and on the radial measure of $D$-dimensional momentum space \eqref{eq:rhok2}.
This follows from dimensional analysis.
In natural units $\hbar = c = 1$, the action $S = \int \dd^D x\, \mathcal{L}$ is dimensionless and $[\dd^D x] = -D$, so $[\mathcal{L}] = +D$.
The canonical kinetic term $\mathcal{L}_{\text{kin}} = \frac{1}{2}(\partial_\mu \Phi)(\partial^\mu \Phi)$ has dimension $2 + 2[\Phi]$ in mass units (since $[\partial_\mu] = +1$), and equating this to $[\mathcal{L}] = +D$ gives
\begin{equation}
    [\Phi] = \frac{D-2}{2}.
\end{equation}
For a monomial coupling $g_n$ of $\Phi^n$, the action term $(n!)^{-1}\,g_n \Phi^n$ must also have dimension $D$, so
\begin{equation}
    [g_n] = D - n\,[\Phi] = D - n\,\frac{D-2}{2}.
\end{equation}
A coupling is thus relevant if $[g_n] > 0$, marginal if $[g_n] = 0$, and irrelevant if $[g_n] < 0$ (see \Cref{tab:couplings} for a summary).
In words: a coupling with positive mass dimension grows in importance at low energies and is therefore relevant; a coupling with zero mass dimension changes only logarithmically with scale and is marginal; a coupling with negative mass dimension is suppressed in the \ir regime and is irrelevant.

As shown in~\cite{RG1,RG2,RG3,RG4,RG5,RG6}, the problem of separating signal from noise in the quasi-continuous tail of the spectrum reduces to an \rg analysis of an effective field theory (\eft).
In practice, this \eft is a $\Phi^4$ theory whose effective spacetime dimension is set by the shape of the spectrum: for pure noise in the \mpdistr class, the asymptotic dimension is $D = 3$ (this identification, derived from the square-root edge behaviour of the \mpdistr density, is established in the technical sections), where the quartic coupling is relevant and the Gaussian fixed point is unstable.
The effective field plays the same role as $\Phi(x)$ in $\Phi^4_d$ theory, which describes the critical behaviour of the $d$-dimensional Ising model.
Like the Ising field, $\Phi(x)$ does not resolve individual spins: it captures the long-range correlations (in particular the two- and four-point functions) between mesoscopic regions of the ferromagnet.
As established by the \rg pioneers~\cite{Zinn2,Wilson1}, an \eft parametrises the dominant correlations between degrees of freedom and reconstructs the underlying probability distribution from measured data.

As in the Ising case, this \eft does not depend on the microscopic details of the data, provided the noisy \dof follow a universal distribution such as the \mpdistr law~\eqref{MP}.
This insensitivity to microscopic detail is not an additional assumption: it follows directly from universality.
If an \eft can be constructed for one spectrum in a universality class and shown to separate signal from noise, universality guarantees that the same \eft separates them for every other spectrum in the class.
This strategy is developed in this section and in \Cref{App1}.
The \eft follows from the maximum entropy principle applied to the empirical spectrum~\cite{RG7,Jaynes1,Jaynes2,Berman_2024,berman2022dynamicsinferencelearning}.
Following this interpretive reading, the \emph{Standard Model} of particle physics can also be viewed as a maximum-entropy inference: its action is in fact fixed, to within a small number of parameters, by the required gauge symmetries, Lorentz invariance, and the observed particle content.
On the same reading, quantum field theory appears as a mathematical language for the correlations between initial and final states in scattering processes.
Still on the same reading, statistical physics is a special case of statistical inference under the maximum entropy principle (see~\cite{Jaynes1,Jaynes2} and \Cref{App0}, as well as~\cite{Berman_2024,berman2022dynamicsinferencelearning}).

\part{Field Theory Framework for Signal Analysis}\label{part2}

The framework established in \Cref{part1} introduced the mathematical machinery of the \rg and its physical interpretation.
Though abstract, the concepts developed in \Cref{part1} are directly applicable to the analysis of empirical spectra, which is the focus of this part.
We shall build an \eft directly from the empirical spectrum and compute the evolution of the couplings with a functional (Wetterich) \rg approach.
The resulting scale dependence of the canonical dimensions will then become the detection criterion.
In data analysis, this construction goes by the name \emph{Data Field Theory} (\dft), the \eft whose propagator reproduces the empirical spectral density and whose spectral edge provides the natural \uv cutoff.
The reference for \enquote{no signal} is the \mpdistr universality class, whose edge behaviour fixes the asymptotic spacetime dimension of the \dft at $D = 3$ (see \Cref{sec:lesson}).
The presence of a signal perturbs the spectrum away from this reference and drives the system through the \emph{dimensional phase transition} at $D = 4$ that is the central result of this review~\cite{Landau2023}.

The remainder of this part develops the construction in three stages.
We first build the \dft and identify the \mpdistr class as the natural noise reference.
We then implement the Wetterich equation, derive the flow of the canonical dimensions, and locate the dimensional phase transition.
Finally, we turn the phase transition into a detection criterion via the \emph{Generalised Scale Analysis} (\gsa) of \cite{RG1,RG2,RG3,RG4,RG5,RG6,RG7}.

\section{Data Field Theory}\label{section1}
We construct the \dft and the structure of the theory space, and define the \mpdistr universality class.
The construction is heuristic and relies on universality arguments.
The rigorous treatment is in \Cref{App1}.

\subsection{Continuum Limit and Marchenko--Pastur Universality Class}\label{sec:universality_class}

We first clarify the terms \enquote{\mpdistr universality class} and \enquote{continuum limit}.
To avoid any ambiguity with the data-matrix dimensions $N$ and $P$ introduced in \Cref{sec1}, we denote by $T$ the (large) size of a generic symmetric matrix in this generic discussion.
The connection to the empirical correlation matrix $C \in \mathds{R}^{P \times P}$ of this paper is made explicit at the end of the paragraph.
Random Matrix Theory (\rmt) rests on the result that the intrinsically random eigenvalue spectrum of a large matrix ensemble converges to a deterministic distribution as the matrix dimension goes to infinity.
The spectrum is typically bounded, so the density of eigenvalues grows with the matrix size.
We focus on symmetric matrices.
A symmetric $T \times T$ matrix has $T$ real eigenvalues, with typical spacing $\mathcal{O}(1/T)$ between adjacent ones.
The continuum limit is formalised as follows.
For sufficiently large $T$, the empirical eigenvalue distribution $\{\lambda_\mu\}_{1\leq \mu \leq T}$, defined by
\begin{equation}
    \mu_{\text{em}}(\lambda)=\frac{1}{T}\sum_{\mu=1}^T\delta(\lambda-\lambda_\mu),
\end{equation}
converges to a deterministic distribution $\mu_{\infty}(\lambda)$.
Examples include the \mpdistr theorem (see \Cref{thm:thMP}) and Wigner's semicircle law (see the background distribution in \Cref{AppD}).
Concretely, for every continuous bounded test function $f$, this convergence reads:
\begin{equation}
    \frac{1}{T} \sum_{\mu=1}^T f(\lambda_\mu) \to \int \dd \lambda\, \mu_{\infty}(\lambda) f(\lambda),
\end{equation}
which is the form of the continuum limit we use throughout.
The passage from a sum over discrete eigenvalues to an integral over $\lambda$ is the matrix analogue of the standard condensed-matter transition from a sum over discrete momentum modes in the Brillouin zone to an integral over $\dd^D k$.
In both cases the continuum limit replaces a finite discrete sum by a smooth integral over the corresponding variable.
In the remainder of this section, the role of the generic $T$ is played by the dimension $P$ of the empirical correlation matrix $C$ (the number of features, or degrees of freedom), and the role of the $T$ eigenvalues is played by the $P$ eigenvalues $\{\lambda_\mu\}_{\mu=1}^P$.

The \mpdistr spectrum defines the disordered phase of the \dft, acting as a universal attractor for stochastic noise.
The \mpdistr universality class is the set of all random matrices whose empirical spectrum converges to the \mpdistr distribution in the continuum limit:
\begin{definition}{Marchenko--Pastur Universality Class}{definitionMPclass}
    Let $X \in \mathds{R}^{N \times P}$ whose entries are \iid random variables following a centred normal distribution: $X_{ij} \sim \mathcal{N}(0, \sigma^2)$ for $i = 1, 2, \dots, N$ and $j = 1, 2, \dots, P$.
    A white Wishart matrix is $C = N^{-1} X^{\mathsf{T}} X$.
    Let $C$ depend on parameters $(\alpha_1, \dots, \alpha_k) \in I_k \subset \mathds{R}^k$.
    Then $C$ belongs to the \mpdistr universality class at $p \in I_k$ if:
    \begin{enumerate}
        \item The empirical spectrum of $C$ converges to the \mpdistr distribution as $P \to \infty$ with $q = P/N$ fixed.
        \item For any two covariance matrices $C(p)$ and $C(p')$ in the class, the distance between their limiting spectral densities is controlled by the parameter distance. For any open neighbourhood $U_p$ of $p$, there exists a constant $K > 0$ such that
              \begin{equation}
                  \operatorname{TVD}\!\left(\mu_p, \mu_{p'}\right)
                  \;\leq\;
                  K \left\Vert p - p' \right\Vert
                  \quad \text{for all } p' \in U_p.
              \end{equation}
    \end{enumerate}
    where $\operatorname{TVD}$ is the \emph{Total Variation Distance} between the limiting spectral densities:
    \begin{equation}
        \operatorname{TVD}\!\left(\mu_p, \mu_{p'}\right)
        =
        \int \dd\lambda\, \left| \mu_p(\lambda) - \mu_{p'}(\lambda) \right|,
    \end{equation}
    and $\left\Vert p - p' \right\Vert$ is the Euclidean distance on $I_k$.
    This condition is the operational criterion of the universality class: it ensures that any two matrices in the class have sufficiently close spectra (in the TVD sense) that the same \eft analysis applies to both.
\end{definition}
In words, the \enquote{\mpdistr universality class} is the set of covariance matrices whose limiting spectrum can be continuously deformed into the \mpdistr distribution, with the deformation cost in TVD controlled by the parameter distance.
The Lipschitz condition is the operational translation of the standard \enquote{all members of the class share the same universal behaviour} statement, and it is what allows us to apply the same \eft to every member of the class.

\begin{remark}{Ising Model and Random Matrix theory}{RqIsing}
    The Ising model discussed in the first part belongs to the \mpdistr universality class.
    Above the critical temperature, the \ecm is well described by \rmt, a point discussed in detail in \Cref{IsingCritical}.
    Define the entries of the \ecm as:
    \begin{equation}
        C_{ij} \defeq \frac{1}{N}\sum_{n=1}^N s_i^{(n)}s_j^{(n)}\,,
    \end{equation}
    where $\{s_i^{(1)}, s_i^{(2)}, \dots, s_i^{(N)}\}$ denotes $N$ samples for site $i$.
    At high temperatures, the correlation length is small, and the \clt ensures that large-scale fluctuations are Gaussian.
    Consequently, the matrix $C$ is approximately a white Wishart matrix, and \Cref{thm:thMP} states that the spectrum of $C$ follows the \mpdistr law for sufficiently large $N$ and $P$.
\end{remark}

\subsection{Motivation and Universality Argument}

The fundamental challenge of high-dimensional data analysis lies in the regime where the signal is weak or strongly entangled with the noise.
In this regime, the signal does not manifest as isolated eigenvalues (the famous spikes of the \bbp transition, discussed in \Cref{AppD}) but merges into the tail of the continuous bulk of the spectrum.
To extract this hidden information, we adopt a complementary viewpoint: we treat the empirical spectrum through the lens of the \eft that encodes the collective behaviour of the underlying microscopic interactions.
These interactions, much like the \dof that carry them, depend strongly on the system at hand, whether the data originate from financial markets, the activity of genes in a cell, the spins of a magnetic material, or the activity of biological or artificial neural networks.
All these systems nevertheless share features in their large-scale correlations, which are remarkably well described by \rmt~\cite{Bouchaud1,Bouchaud2}.
If we can construct an \eft that determines whether a specific spectrum belongs to the \mpdistr universality class, the same \eft will determine this for any other system in the neighbourhood of that universality class, even though the microscopic nature of the underlying \dof is entirely different.\footnote{%
    This is the transferability afforded by universality (see, for instance, \cite{vtyurina2016hysteresis} for an application of the Ising model to biology).
}

Following the lesson drawn from the study of the Ising model in \Cref{part1}, we postulate the existence of a field, which we denote by $\Phi$.
This field materialises the collective behaviour of the \dof, much as the Ising field of \Cref{part1} materialised the large-scale collective behaviour of discrete spins.
The construction of this field theory rests on three conceptual steps:
\begin{itemize}
    \item \textbf{Universality and the \mpdistr class.}
          In the absence of a signal (i.e.\ in pure noise), the spectrum of the empirical correlation matrix converges to the \mpdistr distribution in the limit $P, N \to \infty$ at fixed ratio $q = P/N$.
          The power of this law lies in its universality: it does not depend on the microscopic details of the data.
          Whether the variables are pixels of an image, gene expression levels, or financial price fluctuations, their purely random interactions converge to the same macroscopic spectral behaviour.
    \item \textbf{Minimum-knowledge prescription.}
          The universality of the noise spectrum makes the maximum-entropy inference the natural baseline.
          Since the noise is universal, we need not formulate an exact microscopic theory for each dataset (see \Cref{App0} for a short review).
    \item \textbf{The emergence of a separation scale between noise and signal.}
          Casting the data matrix into an \eft model turns the signal-detection problem into a field-theory problem.
          The \dof associated with the bulk of the spectrum (small eigenvalues, the \uv\ region in the spectral sense) provide a natural notion of scale that allows us to index the fluctuations of the field and to construct a partial-integration procedure leading to a Wilsonian \rg flow.
          Inspecting this flow in the tail of the spectrum (large eigenvalues) reveals a possible breakdown of universality at a characteristic scale.
          That scale marks the boundary beyond which the pure-noise hypothesis breaks down, and thereby signals the presence of structured information.
\end{itemize}

The use of an \eft is therefore not a mathematical analogy but a direct consequence of the universality of random matrices.
The reader wishing to study the formal construction of this partition function and the integration of the microscopic \dof may consult \Cref{App1}.
In what follows, we take such an \eft as a starting point and show how the \rg scans these scales, detecting the signal's signature in the form of a dimensional phase transition.
In the context of data analysis, we refer to the corresponding \eft as the \dft.

\begin{remark}{High-Temperature Ising Model Universality Class}{RqIsing2}
    The high-temperature Ising model belongs to the \mpdistr universality class (see \Cref{rem:RqIsing}).
    For this system, an explicit field theory was constructed in \Cref{part1} (see \Cref{fromisingTofield}, \Cref{Zeff,Heff1}).
    However, although effective, that field theory remains too closely tied to the underlying physical space $\mathds{R}^D$ in which the system is embedded, and the regime in which it is valid differs from the one of interest here.
    The framework described by \eqref{Heff1} captures large spatial scales, where the theory becomes nearly Gaussian above the critical point.
    At the spectral level of the \ecm, this large-distance regime maps to the \uv region, i.e.\ the core of the bulk (in other words,the spectral convention of small eigenvalues = \uv is the opposite of the real-space convention of large distances = \ir, see \Cref{IsingCritical}).
    In contrast, the data-analysis problem studied in this review focuses on the tail of the spectrum, a distinct mesoscopic regime that we shall analyse in \Cref{IsingCritical}.
    The \eft \eqref{Heff1} remains, however, a valuable source of inspiration, in particular for the structure of its interactions.
\end{remark}

\subsection{Local Field Theory}\label{LocalFieldTheory}

We now construct the \dft explicitly.
We shall first build the Gaussian baseline whose two-point function reproduces the empirical correlation matrix, and then determine, by universality and a maximum-entropy argument, the smallest local interaction that can be added to it.

\subsubsection{Nearly Gaussian Estimator}

As shown in the previous section, the least structured distribution to describe the empirical data, requiring the least prior information, is the maximum entropy distribution (see \Cref{App0}).
This distribution is expected to account for the empirical data, namely the empirical correlation matrix $C_{ij}$ (see \Cref{def:definitionECM}), assumed to be $P \times P$ from a sample of size $N$.

\begin{figure}[t]
    \centering
    \includegraphics[width=0.9\textwidth]{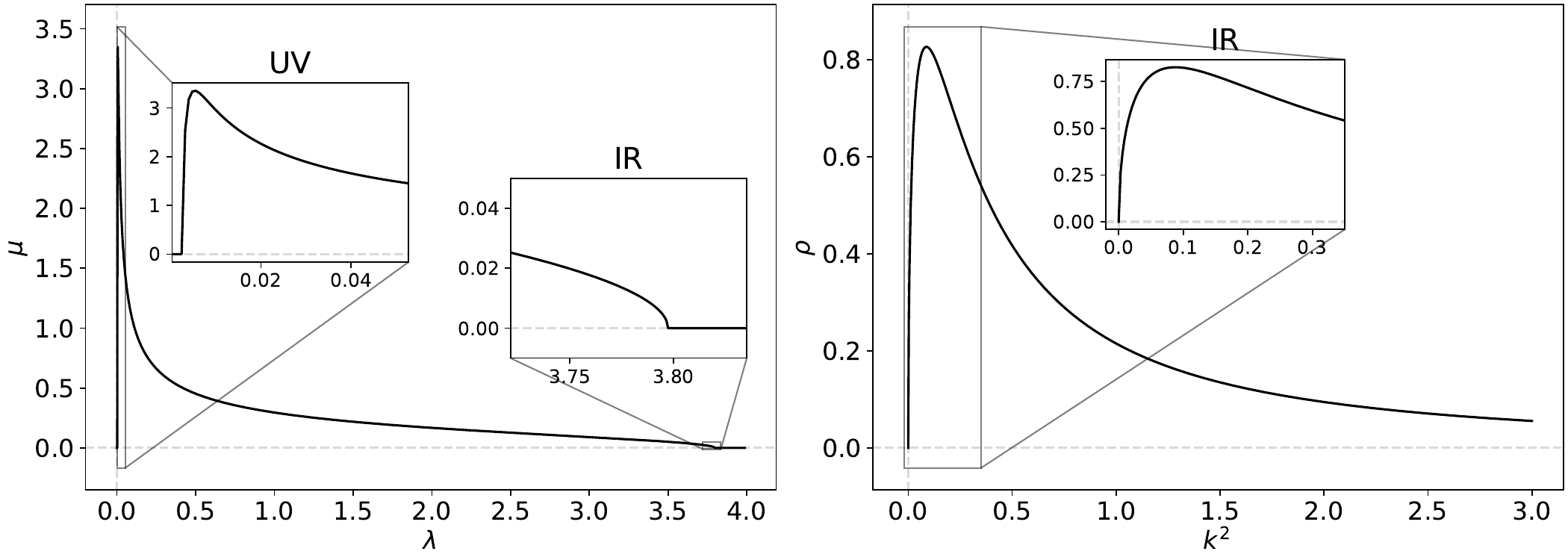}
    \\
    \includegraphics[width=0.9\textwidth]{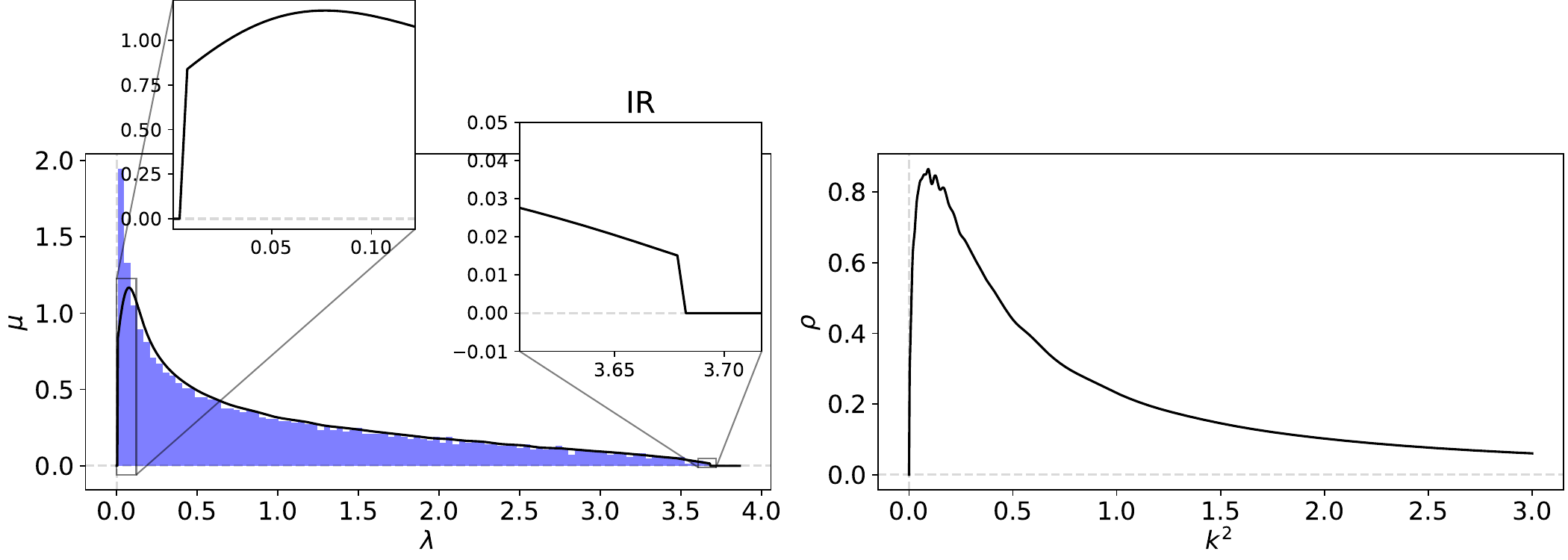}
    \caption{%
        Deep \ir and deep \uv details of the eigenvalue distribution (left) and of the momenta distribution (right).
        The analytic \mpdistr distribution is shown on top, some empirical distribution with sample size $N = 2500$ and ratio $q = 0.9$ at the bottom.
        The black line is the numerical interpolation used to construct the empirical inverse distribution.
    }
    \label{fig:figtail}
\end{figure}

Let us consider a field $\Phi$ as a set of $P$ discrete variables $\Phi=\{\phi_1,\phi_2,\dots, \phi_P \}$ (similar to the Ising model in \Cref{fromisingTofield}).
Let $G$ be its $2$-point function.
Then, the inference constraint translates as the equality between the components of the \ecm $C$:
\begin{equation}
    G_{ij} \defeq \langle \phi_i \phi_j \rangle \overset{\text{emp.}}{\equiv\!\equiv} C_{ij},
    \quad
    i, j = 1, 2, \dots, P.
    \label{eqconst}
\end{equation}
In this context, the \uv and \ir scales are defined in terms of the eigenvalues of the correlation matrix $C$:
\begin{definition}{UV and IR Energy Scales}{UV/IR}
    The spectrum of the correlation matrix $C$, assumed to be nearly continuous, provides a canonical notion of scale.
    \uv scales (high-energy, small position-space scale) correspond to small eigenvalues and \ir scales (low-energy, large position-space scale) to large eigenvalues.

    Let the width of the eigenvalue distribution $\{ \lambda_\mu \}_{\mu = 1}^P$ be defined as
    \begin{equation}
        \lambda_0 \defeq \lambda_+ - \lambda_-,
    \end{equation}
    where $\lambda_+ = \max_\mu \{ \lambda_\mu \}_{\mu=1}^P$ and $\lambda_- = \min_\mu \{ \lambda_\mu \}_{\mu=1}^P$.
    Then, $\lambda_0$ can be used to define a canonical correlation length
    \begin{equation}
        \xi=\sqrt{\lambda_0}.
    \end{equation}
\end{definition}
\Cref{fig:figtail} illustrates the regions corresponding to \uv and \ir scales for a purely Gaussian matrix.

Imposing the constraint \eqref{eqconst}, the maximum entropy estimator takes the standard form:
\begin{equation}
    p(\Phi)
    =
    \frac{1}{Z} \exp \left(-\frac{1}{2} \sum_{i,j=1}^P \phi_i D_{ij} \phi_j \right),
    \label{GaussianModel}
\end{equation}
where $Z=(2\pi)^{P/2}(\det D)^{-1/2}$ (the integral runs over the $P$ components $\phi_1, \dots, \phi_P$).
The maximum entropy estimator \eqref{GaussianModel} satisfies the constraint \eqref{eqconst} provided that
\begin{equation}
    C_{ij}=D^{-1}_{ij}.
    \label{eq:inference_constraint_corr_mat}
\end{equation}
Note that, since we operate at a sufficiently large scale, the matrix $C$ is expected to be invertible, with zero eigenvalues systematically removed during the coarse-graining procedure, as discussed at the end of \Cref{fromisingTofield} (see also \Cref{App1}).
The Gaussian model \eqref{GaussianModel} fixes the form of higher-order correlation functions via Wick's theorem (\Cref{thm:theoremWick}): odd-point functions vanish, whereas even-point functions decompose into sums of products of two-point correlation functions:
\begin{equation}
    \langle \phi_i \phi_j \phi_k \phi_l \cdots \rangle
    =
    C_{ij} \times C_{kl}\times \cdots + \text{permutations},
    \label{corrhigher}
\end{equation}
where the permutations run over all allowed pairings of indices.

Since the matrix $C$ is symmetric and positive-definite, it can be diagonalized, and its eigenvectors $u^{(\mu)}$ can be chosen to form an orthonormal basis:
\begin{equation}
    \sum_{j=1}^P C_{ij} u_{j}^{(\mu)}
    =
    \lambda_\mu u_i^{(\mu)},
    \qquad
    \sum_{i=1}^P u_i^{(\mu)} u_i^{(\mu^\prime)}=\delta_{\mu\mu^\prime}.
\end{equation}
As we keep proceeding in the derivation, recall that for purely random matrices the components of the eigenvectors are uniformly distributed on the sphere $\mathcal{S}_{P-1}$ (see \Cref{sectionPhaseTrans} below).
As it will become clear later, it is particularly convenient to work in the spectral representation, defining:
\begin{equation}
    \phi_\mu \defeq \sum_{i=1}^P \phi_i u_i^{(\mu)}.
    \label{eigenfunctions}
\end{equation}
For purely noisy data, this is equivalent to applying a random rotation to the vector $\Phi$: $\Phi \to \mathbf{O}\Phi$, where $\mathbf{O}$ is Haar-distributed over the orthogonal group $\mathrm{O}(P)$.
In this respect, the maximum entropy estimator \eqref{GaussianModel} reads:
\begin{equation}
    P(\Phi)=\frac{1}{Z}\exp \left(-\frac{1}{2} \sum_{\mu=1}^P \, \phi_\mu (\lambda_\mu)^{-1} \phi_\mu\right).
\end{equation}
Now, let us examine the inference constraint \eqref{eqconst}.
For the positive-definite matrices focused on in this paper, all eigenvalues satisfy $\lambda_\mu > 0$.
Let $\lambda_+$ denote the largest eigenvalue.
It is convenient to express the propagator in the standard field-theoretic form (see \Cref{part1}):
\begin{equation}
    \lambda_\mu=\frac{1}{p^2_\mu+m^2},
    \label{inferenceZero}
\end{equation}
such that the generalized momentum $p_\mu$ interpolates between the \uv and \ir scales.
In the \ir limit ($p_\mu \to 0$), which corresponds to the tail of the spectrum, the mass must be defined as
\begin{equation}
    m^2 \defeq \lambda_+^{-1}.
    \label{eq:mass_largest_eigenvalue}
\end{equation}
Conversely, in the \uv regime, the maximum momentum cutoff $p = \Lambda_{\text{\uv}}$ corresponds to the smallest eigenvalue of the spectrum:
\begin{equation}
    \lambda_- \defeq \dfrac{1}{\Lambda_{\text{\uv}}^2+m^2}.\label{equationLow}
\end{equation}
This maximum value, $p_\mu=\Lambda_{\text{\uv}}$, plays the role of a \uv cutoff in field-theoretic parlance, representing the microscopic scale.
This observation provides the primary motivation for \Cref{def:UV/IR}.

Note that in the following, it will be useful to systematically shift the spectrum so that $\lambda_- = 0$ (i.e.\ $C \to C - \lambda_-\mathds{I}$, where $\mathds{I}$ is the identity matrix).
This convention amounts to replacing the largest eigenvalue with $\lambda_0 \defeq \lambda_+ - \lambda_-$, setting the mass scale to $m^2 \defeq \lambda_0^{-1}$, and choosing $\Lambda = \infty$.
Under this convention, the spectrum of the generalized momentum $p$ spans the entire positive real axis in the continuous limit, i.e.\ $p^2 \in (0, +\infty)$.
Note that the value $p = 0$ is strictly excluded from this interval.
Indeed, since the typical spacing between eigenvalues is of order $\mathcal{O}(1/P)$, the smallest accessible value of $p$ is expected to be of order $\mathcal{O}(1/\sqrt{P})$, which represents the final discrete mode before reaching the mass scale $m^2$.\footnote{%
    In a classical spatial field theory within a finite volume $V = L^D$, the lowest momentum (excluding the zero mode) is dictated by the boundary conditions: $p_{\text{min}} \sim 2\pi/L$.
    See the discussion of \Cref{Gaussiantheory}.
}

\subsubsection{Local Field Theory}

The Gaussian model is remarkably simple, yet significantly inadequate, as it assumes that all higher-order correlations are entirely reducible to two-point correlations encoded in the \ecm via Wick's theorem.
Furthermore, we could also have some doubts whether the Gaussian distribution, characterized by the spectrum of the generalized momentum $\rho(p^2)$ induced by the empirical distribution $\mu_{\text{em}}(\lambda)$, is stable under coarse-graining.
In other words, does a non-Gaussian perturbation remain small along the \rg flow? In \Cref{part1}, we showed that for the field theory naturally associated with the Ising model, the Gaussian fixed point is unstable in low dimensions ($D \leq 4$).
What about the data-driven \dft we just constructed?

A pragmatic way to evaluate the stability of the Gaussian regime is to introduce a small interaction term, a non-Gaussian monomial, and analyse its effect on the flow.
However, determining the appropriate form of this interaction is not straightforward, as its behavior strongly depends on how the fields couple within the interaction terms.
Concretely, this amounts to adding a new constraint to define the maximum entropy estimator (see \Cref{App0}), for instance of the form:\footnote{%
    One could certainly consider odd-power couplings, but they complicate the integrability of the distribution.
    Thus, a quartic interaction appears to be a simpler and more sensible choice.
}
\begin{equation}
    G_4 = \sum_{ijkl} W_{ijkl} \phi_i\phi_j\phi_k\phi_l,
\end{equation}
where the tensor $W_{ijkl}$ acts as a Lagrange multiplier that physically encodes a four-point correlation.

Universality provides the most reliable path to choose the form of tensor $W$.
Such tensor possesses $P^4$ \dof, though the interaction must necessarily reduce this number.
We know that the Ising model, which involves discrete microscopic variables, falls into the \mpdistr universality class.
Furthermore, we have explicitly shown in \Cref{fromisingTofield} that this model can be mapped onto a discrete field $\Phi_i$ (analogous to ours, see \Cref{partition2}), whose defining characteristic is the emergence of local interactions, namely, even powers of the form $\sim \sum_i \Phi_i^{2n}$.
Here, $i$ plays the role of the lattice site index $\vec{x}$.
As universality dictates a common form for the minimal model in the vicinity of a random matrix universality class, the local interactions emerge naturally.

The minimal model is then the following:
\begin{equation}
    p(\Phi) = \frac{1}{Z} \exp \left(-\frac{1}{2} \sum_{i,j=1}^P \phi_i D_{ij} \phi_j -\frac{g_4}{4!} \sum_{i=1}^P \phi_i^4+\mathcal{O}(\phi_i^6)\right).
    \label{GaussianModelBeyond}
\end{equation}
This amounts to choosing $W_{ijkl} \propto \delta_{ij}\delta_{kl}\delta_{il}$, which reduces the $P^4$ \dof of the tensor to one.
It specifically presents an obvious $\mathds{Z}_2$ symmetry ($\Phi \to -\Phi$), proposed in particular in \cite{RG0} in a similar context.
For the data-scientist reader, this symmetry reflects the empirical fact that the empirical correlation matrix $C$ does not distinguish a feature from its negative: the matrix of correlations is unchanged under a global sign flip of the centred variables, and the maximum-entropy estimator inherits this invariance.

For practical reasons, the model \eqref{GaussianModelBeyond} is not yet optimal.
It would be more convenient to operate in momentum space, in the eigenbasis of the symmetric positive-definite matrix $D$.
Following the previous construction, we parameterize its eigenvalues as $(p^2_\mu + m^2)$, thereby defining the generalized momentum $p^2$ (we shall return to the definition of the mass and the distribution of this generalized momentum in the next section).
In this basis, the explicit expression of non-Gaussian interactions becomes complicated, depending heavily on the specific eigenvectors.
A simple way to write it, as proposed in \cite{RG1}, is to double the number of degrees of freedom by indexing the fields by $p$ ($p \in \mathds{R}$ in the continuous limit) rather than by $p^2$.
We thus introduce a new field, $\varphi(p)$, such that:
\begin{itemize}
    \item In the Gaussian limit:
          \begin{equation}
              p(\varphi) =\frac{1}{Z} \exp \left(-\frac{1}{2}\sum_p \varphi(p) (p^2+m^2) \varphi(-p) \right),
          \end{equation}
          faithfully reproduces the correlations provided that \eqref{inferenceZero} holds true.
    \item The non-Gaussian interactions preserve the $\mathds{Z}_2$ reflection symmetry and conserve the generalized momentum.
          For a typical interaction of order $2K$:
          \begin{equation}
              U_{2K}(\varphi)\sim \sum_{\{p_i\}} \delta\left(\sum_{\ell=1}^{2K} p_\ell \right) \prod_{\ell=1}^{2K} \varphi(p_\ell).
          \end{equation}
\end{itemize}
By doubling the \dof in this manner, the field $\varphi(p)$ mimics the algebraic structure of a genuine spatial Fourier transform.
The non-Gaussian couplings, which were previously constrained by the specific shape of the eigenvectors, can now be rewritten in terms of a much simpler conservation law $\delta(p_1 + p_2 + p_3 + p_4)$ (i.e.\ total momentum conservation).
This maps the problem from a pure spectral geometry onto a real-space-like geometry, rendering the evaluation of summation functions and \rg loop integrals analytically tractable.
For the physicist reader, the Feynman diagrams generated by a typical local interaction of the form $\sum_i \Phi_i^4$ are structurally identical to those generated by this generalized-momentum-conserving term.
Consequently, both representations yield an identical \rg flow.
Note that this is not an additional assumption, but rather a standard observation from the ordinary Fourier transform.

In order to construct the statistical inference framework, we consider an \eft whose structure mirrors that of a standard Euclidean equilibrium field theory.
We use a field $\varphi(p)$, which depends on a generalised \emph{momentum variable} $p \in \mathds{Z}_P$.
The properties of the field are governed by the partition function:
\begin{equation}
    Z[j]=\int [\dd \varphi] \, \exp \left(- H[\varphi]+ \sum_p j(-p) \varphi(p) \right).
\end{equation}
The corresponding classical Hamiltonian, $H[\varphi]$, is defined by the quadratic (Gaussian) term and the interaction potential $U[\varphi]$:
\begin{equation}
    H[\varphi]=\frac{1}{2}\sum_p \varphi(p) (p^2+m^2) \varphi(-p)+U[\varphi].\label{classicalaction}
\end{equation}
The interaction term, $U[\varphi]$, is expanded in higher-order field powers, resembling conventional local interactions in momentum space ($\delta$ is the ordinary Kronecker Dirac delta):
\begin{equation}
    U[\varphi] \defeq \sum_{n=3}^\infty\, \frac{u_{2n}}{(2n)!P^{n-1}}\, \delta \left(\sum_{i=1}^n p_i\right) \prod_{i=1}^{2n} \varphi(p_i).
    \label{ultralocalint2}
\end{equation}

\subsection{More on the Inference Constraint}

In the context of \dft, we often work in a pseudo-continuum limit.
While the matrix size $P$ is large and fixed, the squared momentum $p^2$ is restricted to $P$ discrete values with spacing of order $1/P$, as explained above.
These values are determined by an \emph{a priori} unknown distribution $\rho(p^2)$, which exhibits weak convergence toward a true continuous distribution as $P$ tends to infinity.
The inference condition serves to constrain the parameters of the classical action.
Specifically, we impose the requirement that the two-point correlation function, $G(p) \defeq \langle \varphi(p) \varphi(-p) \rangle$, must coincide with the empirical correlation matrix values in the eigenbasis.
In the symmetric phase, this propagator is explicitly given by:
\begin{equation}
    G(p^2) \defeq \frac{1}{Z[j=0]} \int [\dd \varphi] \, e^{-{H}[\varphi]}\varphi(p)\varphi(-p).
\end{equation}
In the Gaussian model, the spectral correspondence is established by \eqref{inferenceZero}.
The mass parameter $m^2$ is fixed by the inverse of the largest eigenvalue, $m^2 = \lambda_0^{-1}$.
Recall that, with the convention adopted at the end of \Cref{LocalFieldTheory}, the spectrum is shifted by $\lambda_- \to 0$, so that the largest eigenvalue is $\lambda_0 = \lambda_+ - \lambda_-$.
The induced momentum distribution $\rho_{\text{G}}(p^2)$ (where the index $G$ stands for \enquote{Gaussian}) is consequently derived from the empirical eigenvalue distribution $\mu_{\text{em}}(\lambda)$.
It is important to note that a perturbative analysis indicates that the Gaussian fixed point is unstable, implying that a perturbation, such as a quartic term, will be relevant under the standard \rg flow.
The canonical dimension typically dictates the linear term in the flow equation near the Gaussian fixed point.

Beyond the Gaussian limit, the bare propagator acquires quantum corrections that are collectively resummed by the well-known Dyson series in quantum field theory.
In \eqref{inferenceZero}, the bare mass $m^2 = \lambda_+^{-1}$ was the parameter of the Gaussian action.
In the Wetterich formalism of the next section, this bare parameter is replaced by the \emph{running} two-point vertex $u_2(k) \defeq \Gamma_k^{(2)}(0,0)$, which depends on the \rg scale $k$.
At the ultraviolet scale $k = \Lambda$, $u_2(\Lambda) = m^2$.
The Dyson resummation is therefore most naturally written in terms of the running $u_2$ rather than the bare $m^2$ (we keep $u_2$ throughout, since the bare $m^2$ is a \emph{measured} input whereas $u_2$ is the \emph{theoretical} coupling that will flow under the \rg).
The resummation is the standard field-theoretic one, i.e.\ a geometric series $\sum_n x^n = 1/(1-x)$ when $|x| < 1$, here applied with $x = \Sigma(p^2)/(p^2 + u_2)$:
\begin{equation}
    \begin{split}
        G(p^2)
        & =
        \frac{1}{p^2 + u_2} + \frac{1}{p^2 + u_2}\,\Sigma(p^2)\,\frac{1}{p^2 + u_2} + \dots
        \\
        & =
        \frac{1}{p^2 + u_2}
        \sum_{n = 0}^{\infty} \left(\frac{\Sigma(p^2)}{p^2 + u_2}\right)^n
        \\
        & = \frac{1}{p^2 + u_2 - \Sigma(p^2)},
    \end{split}
    \label{eq:dyson}
\end{equation}
where $\Sigma(p^2) < p^2 + u_2$ is the \emph{self-energy} at the scale of the running coupling $u_2$.
Solving $G(p^2_\mu) = \lambda_\mu(p^2_\mu)$ exactly is non-trivial and generally leads to $\rho(p^2) \neq \rho_{\text{G}}(p^2)$.
However, the analysis is greatly simplified when focusing on the tail of the spectrum, corresponding to the \ir regime where the generalized momentum is small ($p^2 \ll 1$).
In this low-momentum limit, where standard techniques such as the derivative expansion and the \emph{Local Potential Approximation} (\lpa, discussed in the following sections) are robust enough, the propagator can be approximated as
\begin{equation}
    G(p^2) \approx \frac{1}{Z p^2 + m^2_{\text{eff}}},
\end{equation}
where
\begin{equation}
    Z \defeq 1 - \Sigma'(0)
\end{equation}
is the field-strength renormalization and
\begin{equation}
    m^2_{\text{eff}} \defeq m^2 - \Sigma(0)
\end{equation}
is the effective mass.
Within the strict \lpa, the field-strength renormalization effects ($Z$ corrections) may be ignored, as explicitly validated in \cite{RG2,RG5} and expected to hold in the deep \ir.
This simplification restores the exact Gaussian correspondence, albeit subject to a global shift in the mass term:
\begin{equation}
    m^2 - \Sigma(0) = \lambda_0^{-1},
\end{equation}
leading in particular to the approximation:
\begin{equation}
    \rho(p^2) \approx \rho_G(p^2).
    \label{eqrho}
\end{equation}

The \dft is thus built as the \eft whose propagator reproduces the empirical spectrum of the empirical correlation matrix $C$.
We identified the \mpdistr universality class as the disordered phase of the theory: a Gaussian fixed point perturbed by the smallest local interaction compatible with universality.
Two distinct continuum limits are at play in this section.
The first is the \emph{spectral} continuum limit $P \to \infty$ at fixed $q = P/N$, which controls the smoothness of the empirical density $\rho(p^2)$ and justifies replacing the discrete sum over $P$ eigenvalues by a smooth integral, which is the matrix analogue of the Brillouin-zone sum-to-integral transition introduced at the beginning of the section.
The second is the \emph{field-theoretic} continuum limit in which the discrete set of momentum modes is replaced by a continuous integral weighted by $\rho(p^2)$ and $\rho(p^2) p \dd p$.
This limit is what makes the \rg loop integrals well defined.
The first limit justifies the use of a continuum \eft, while the second limit defines the \emph{unusual geometry} in which the flow lives, and the central place where the data field theory departs from a textbook Euclidean field theory.
In this geometry, the canonical dimension of a coupling is not a fixed number but a function of the running scale $k$, because the radial measure $\rho(p^2)$ is not a power law.
The presence of a signal deforms $\rho(p^2)$ in the tail of the spectrum, thereby changing the canonical dimensions.
This scale dependence of the canonical dimensions is the central object of the next section, where the Wetterich equation is used to track the flow of the couplings $u_2, u_4, u_6, \dots$ and to locate the dimensional phase transition that separates signal from noise.

\section{Functional Renormalization Group for Data Field Theory}\label{sec:frg_dft}
The Functional Renormalization Group (\frg) is an approach to the \rg specifically suited for non-perturbative studies.
Although the core of this review focuses on dimensional analysis of the flow behaviour in the vicinity of the Gaussian fixed point, the non-perturbative nature of the flow, suggested by the canonical dimensions, requires a non-perturbative formalism.
In this section, we present the theoretical foundations of the \frg applied to \dft.
We begin with a brief review of the non-perturbative Wetterich-Morris formalism and discuss the construction of approximations justified by the local nature of the \dft interactions.
For a general review of the \frg, the reader may refer to the standard reference~\cite{Delamotte}, as well as~\cite{RG5}, which is more specific to \dft.

\subsection{Wetterich's Formalism in a Nutshell}

In \Cref{part1}, we saw that, in the Wilsonian approach, \dof are integrated out step-by-step, causing the field couplings to change.
The \rg thus transforms the Hamiltonian into another Hamiltonian at each step, tracing a trajectory in the configuration space of physical theories, i.e.\ the flow.
From this perspective, the Hamiltonian describes long-distance physics for the \ir modes that have survived the partial integration of \uv modes.
The \frg approach, and more specifically that of Wetterich-Morris, follows Wilson's philosophy but focuses on a different object: the \emph{effective action}, denoted by $\Gamma$ (see \Cref{fig0}).

Consider a Euclidean field theory for a field $\Phi(x)$, where $x \in \mathcal{X}$ is either a discrete or continuous variable.
Suppose the partition function is defined by a path integral of the type:
\begin{equation}
    Z[J] = \int [\dd \Phi]\, e^{-H[\Phi]+\int J(x)\, \Phi(x) },
\end{equation}
where $[\dd \Phi] \defeq \prod_x \dd \Phi(x)$ is the usual integration measure of the path integral (see \Cref{part1}) and $\int J(x)\, \Phi(x)$ formally represents the sum of the \emph{source couplings} (discrete or almost continuous)\footnote{%
    Almost continuous, in this context, means that, in a given limit, discrete sums should be replaced by integrals, see \Cref{part1} and below.
}
such that the cumulants of the distribution are generated by a power series expansion in $J$ (i.e.\ cumulants are higher-order derivatives of $\ln Z[J]$ with respect to $J$).
The Gaussian part $H_{\text{kin}}[\Phi]$ is assumed to take the general form:
\begin{equation}
    H_{\text{kin}}[\Phi]
    \defeq
    \frac{1}{2}\int \Phi(x) K_{xy} \Phi(y)+\frac{1}{2}\int m^2 \Phi^2(x),
\end{equation}
where the last term is the ordinary mass term, and the coupling $m^2$ is called \emph{the mass}.
The symmetric kernel $K_{xy}$, on the other hand, has positive spectrum $\Lambda = \{\zeta_{\mu}\}$.
It is useful to work in a basis where $K_{xy}$ is diagonal.
Defining the normalized eigenvector $u_x^{(\mu)}$ (i.e.\ the standard Fourier basis in physics, though we shall keep it as general as possible in this section) such that:
\begin{equation}
    \int_y K_{xy} u_y^{(\mu)} = \zeta_\mu u_x^{(\mu)}.
\end{equation}
Then, because of the normalization of eigenvectors, the kinetic action becomes (remember the shift \eqref{eq:mass_largest_eigenvalue}):
\begin{equation}
    H_{\text{kin}}[\Phi]
    =
    \frac{1}{2} \sum_\zeta \Phi(\zeta) (\zeta+m^2) \Phi(\zeta).
\end{equation}
with:
\begin{equation}
    \Phi(\zeta_\mu) \defeq \int_x \Phi(x) \,u_x^{(\mu)}.
\end{equation}
The spectrum of the Gaussian kernel, which is essentially the inverse of the Fisher information metric when interactions (i.e.\ the non-Gaussian terms in the Hamiltonian) are set to zero, provides a non-trivial notion of scale.
In particular, a Wilsonian renormalization group can be constructed via a partial integration procedure over \uv modes, leaving the long-distance \ir physics unchanged.
Thus, the spectrum of the Gaussian kernel defines a preferential ordering, which enables the formal definition of this partial integration procedure.

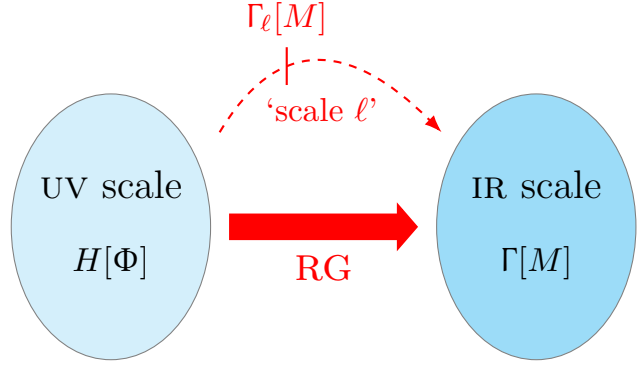
\begin{figure}[t]
    \centering
    \begin{tikzpicture}[scale=1.25, every node/.style={scale=1.25}]
    \draw[draw=black!50, fill=cyan!15] (0, 0)
    node {\begin{tabular}{c}{\large \uv scale} \\[1em] $H[\Phi]$\end{tabular}}
    ellipse (3em and 4em);
    \draw[draw=black!50, fill=cyan!35] (4.5, 0)
    node {\begin{tabular}{c}{\large \ir scale} \\[1em] $\Gamma[M]$\end{tabular}}
    ellipse (3em and 4em);

    \draw[-{Triangle[width = 18pt, length = 8pt]}, draw=red, fill=none, line width=10pt] (1.25, 0) -- (3.25, 0)
    node[midway, below, text=red] {\Large \rg};
    \draw[dashed, -Latex, thick, red] (1.15, 1) .. controls (1.75, 2) and (2.5, 2) .. (3.5, 1);

    \draw[red, thick] (1.86, 1.5) node[below right, xshift=-10pt, text=red] {\enquote{scale $\ell$}} -- (1.86, 1.9) node[above, text=red] {$\Gamma_{\ell}[M]$};
\end{tikzpicture}
    \caption{%
        Illustration of the \rg map from \uv physics described by the Hamiltonian $H$ to long range physics described by the effective action $\Gamma$.
    }
    \label{figRG2}
\end{figure}

To construct the Wetterich-Morris formalism, we must first define the effective action $\Gamma$.
The effective action is defined as the Legendre transform of the free energy (the generating functional of the cumulants, or connected correlation functions) $W[J] \defeq \ln Z[J]$:
\begin{equation}
    \Gamma[M] = \int J(x)\, M(x) - W[J].
    \label{GammaDef}
\end{equation}
The effective action, also known as the \emph{generating functional of 1-particle irreducible} (\onepi) diagrams in physics, physically corresponds to the classical action of the classical field
\begin{equation}
    M(x) \defeq \langle \Phi(x) \rangle,
\end{equation}
though one should rather speak here of a \enquote{classical Hamiltonian}:\footnote{%
    In all these relations, if $x$ is a continuous variable, the formal partial derivatives must be replaced by functional derivatives.
}
\begin{equation}
    M(x)
    \defeq
    \frac{1}{Z[J]}\frac{\partial }{\partial J(x)} Z[J] \equiv \frac{\partial }{\partial J(x)} W[J],
\end{equation}
in the sense that:
\begin{equation}
    \frac{\partial}{\partial M(x)} \Gamma[M]= J(x),
    \label{eqstate}
\end{equation}
a relation equivalent to the one that fixes the saddle point of the path integral:
\begin{equation}
    \frac{\partial}{\partial \Phi(x)} H[\Phi]= J(x).
    \label{eqbare}
\end{equation}
Formally, \eqref{eqstate} and \eqref{eqbare} are similar: the latter determines a solution that does not account for fluctuations and will only be physically correct if these fluctuations are negligible.
In contrast, \eqref{eqstate} integrates all quantum fluctuations.
It provides an exact effective description, which is why we call it an \enquote{effective action} (or effective Hamiltonian).
In \Cref{part1}, we saw for example that the saddle-point approximation for the $\Phi_D^4$ theory is only valid in high dimensions $D>4$, whereas in low dimensions, fluctuations are too large.
The \emph{complete} physics is therefore captured by $\Gamma$ in that case.

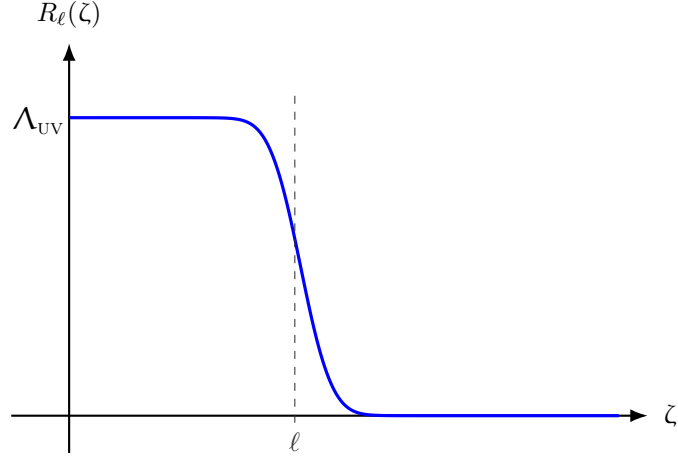
\begin{figure}[t]
    \centering
    \begin{tikzpicture}
    \begin{axis}[
            width=10cm,
            height=7cm,
            axis lines=middle,
            axis line style={-Latex},
            xlabel={$\zeta$},
            ylabel={$R_{\ell}(\zeta)$},
            xlabel style={right, xshift=2pt},
            ylabel style={above, yshift=2pt},
            xmin=-0.5, xmax=5.0,
            ymin=-0.5, ymax=5.0,
            ticks=none,
            legend pos=north west,
            legend cell align=left,
            legend style={fill=white, fill opacity=0.85, draw=gray!50, text opacity=1, font=\footnotesize},
            smooth,
            thick,
            no marks,
            every axis plot/.append style={line width=1.2pt},
            declare function={
                    erf_approx(\x) = 1 - 1 / (1 + 0.278393*\x + 0.230389*\x^2 + 0.000972*\x^3 + 0.078108*\x^4)^4;
                    erf(\x) = (\x >= 0 ? erf_approx(\x) : -erf_approx(abs(\x)));
                }
        ]
        \addplot[
            domain=0.0:4.75,
            samples=200,
            color=blue,
        ]
        {2*erf(-(x-2)/0.3)+2};
    \end{axis}

    \node[anchor=center] at (0.36, 4.4) {\large $\Lambda_{\text{\uv}}$};
    \draw[dashed, black!80] (3.75, 0.4) node[below] {$\ell$} -- (3.75, 4.79);
\end{tikzpicture}
    \caption{Typical behaviour of the cut-off $R_{\ell}(\zeta)$.}
    \label{Rkshape}
\end{figure}

Instead of focusing on the Hamiltonian of the \enquote{remaining} \ir \dof, Wetterich's approach focuses on the effective action of the \uv \dof, i.e.\ those integrated out by the \rg procedure, by introducing a scale-dependent modification.
Wetterich's modified effective action, $\Gamma_{\ell}$, depends on an IR scale $\ell \in [0, \Lambda_{\text{\uv}}]$ and interpolates between the initial \emph{microscopic} (\uv) Hamiltonian $H[M]$ (where $M \equiv \Phi$) as $\ell \to \Lambda_{\text{\uv}}$, and the effective action $\Gamma[M]$, given by \eqref{GammaDef}, as $\ell \to 0$ (see \Cref{figRG2}).
The scale $\Lambda_{\text{\uv}}$ is an intermediate \uv scale where the field theory is assumed to be a good approximation of a more fundamental reality (a mesoscopic scale in the example of the Ising model discussed in \Cref{part1}).
To explicitly construct this effective action, the \frg modifies the Hamiltonian by adding a mass-like term $\Delta H_{\ell}[\Phi]$:
\begin{equation}
    H[\Phi] \to H[\Phi] +
    \underbrace{\frac{1}{2} \sum_{\zeta} \phi(\zeta) R_{\ell}(\zeta) \phi(\zeta)}_{\defeq\Delta H_{\ell}[\Phi]}.
\end{equation}
The function $R_{\ell}(\zeta)$ freezes large-scale fluctuations ($\zeta > \ell$), which acquire a large mass, whereas the \uv modes ($\zeta < \ell$) maintain a small mass and are integrated out.\footnote{%
    Physically, larger masses suppress fluctuations, as can be seen from the Gaussian distribution $e^{-m^2 x^2/2}$, which becomes more and more concentrated around $x=0$ as $m^2$ increases.
}
More precisely, the following conditions must hold:
\begin{enumerate}
    \item $R_{\ell}(\zeta) \sim \Lambda_{\text{\uv}} \gg 1$ as $\ell \to \Lambda_{\text{\uv}}$.
    \item $R_{\ell}(\zeta) \to 0$ as $\ell \to 0$.
    \item $R_{\ell}(\zeta) \ll 1$ as $\zeta \gg \ell$.
    \item $R_{\ell}(\zeta) \gg \ell$ as $\zeta \ll \ell$.
\end{enumerate}
The typical shape of $R_{\ell}(\zeta)$ is shown in \Cref{Rkshape}.
The last two conditions ensure that \ir modes are frozen out from the long-distance physics (corresponding to small eigenvalues $\zeta$) while \uv modes are integrated out.
The first two conditions, moreover, ensure that the boundary conditions:
\begin{enumerate}
    \item $\Gamma_{\ell}[M] \to H[M]$ for $\ell \to \Lambda_{\text{\uv}}$.
    \item $\Gamma_{\ell}[M]\to \Gamma[M]$ for $\ell \to 0$.
\end{enumerate}

We can formally show these conditions as follows.
To begin, let us provide the exact expression for the functional $\Gamma_{\ell}[M]$ in the eigenbasis of $K$:
\begin{equation}
    \Gamma_{\ell}[M]+\Delta H_{\ell}[M]=\sum_\zeta \, J(\zeta) M(\zeta)-W_{\ell}[J],
    \label{defGammak}
\end{equation}
where $W_{\ell}[J] \defeq \ln Z_{\ell}[J]$, the partition function $Z_{\ell}$ being calculated with the modified Hamiltonian $H+\Delta H_{\ell}$.
The second condition is evident: as $\ell\to 0$, $R_{\ell}(\zeta)\to 0$ and $\Gamma_{\ell}[M]$ reduces to $\Gamma[M]$.
For the \uv limit, it should be noted that the definitions imply, by expressing $W_{\ell}[J]$ in terms of $\Gamma_{\ell}[M]$:
\begin{equation}
    \exp \left({\sum_\zeta \, J(\zeta) M(\zeta)-\Delta H_{\ell}[M]-\Gamma_{\ell}[M]}\right)= \int [\dd \Phi]\, e^{-H[\Phi]-\Delta H_{\ell}[\Phi]+ \sum_\zeta \Phi(\zeta) J(\zeta)}.
\end{equation}
Using the fact that, because of the definition \eqref{defGammak}, we have:
\begin{equation}
    J(\zeta)=\frac{\partial \Gamma_{\ell}}{\partial M(\zeta)} + R_{\ell}(\zeta) M(\zeta),
\end{equation}
and the previous equation can be rewritten as:
\begin{equation}
    e^{-\Gamma_{\ell}(M)}=\int [\dd \Phi] e^{-H[\Phi]-\frac{1}{2}\sum_\zeta R_{\ell}(\zeta) (\Phi(\zeta)-M(\zeta))^2+\sum_\zeta \frac{\partial \Gamma_{\ell}}{\partial M(\zeta)}(\Phi(\zeta)-M(\zeta))}.
\end{equation}
$R_{\ell}(\zeta)$ becomes large in the limit $\ell \to \Lambda_{\text{\uv}}$.
The term $\frac{1}{2}\sum_\zeta R_{\ell}(\zeta) (\Phi(\zeta)-M(\zeta))^2$ enforces $\Phi=M$ in the integral (it behaves like a functional Dirac delta).
Hence, we have:
\begin{equation}
    \lim_{\ell\to \Lambda_{\text{\uv}}} e^{-\Gamma_{\ell}(M)} \propto e^{-H[M]},
\end{equation}
and the proposition follows.

The functional $\Gamma_{\ell}$ satisfies a first-order integro-differential equation, known as the \emph{Wetterich's equation}:
\begin{equation}
    \dot{\Gamma}_{\ell}[M]=\frac{1}{2}\sum_\zeta \,\dot{R}_{\ell}(\zeta) \, G_{\ell}(\zeta,\zeta),
    \label{Wett}
\end{equation}
where $\dot{X} \defeq \ell \frac{\dd X}{\dd \ell}$ and $G_{\ell}(\zeta,\zeta^\prime)$ is the two-point function:
\begin{equation}
    G_{\ell}(\zeta,\zeta^\prime) \defeq \frac{\partial^2 W_{\ell}}{\partial J(\zeta) \partial J(\zeta^\prime)},
\end{equation}
which, given the properties of the Legendre transform, corresponds to the inverse of the operator $\Gamma^{(2)}_{\ell}+R_{\ell} \mathds{I}$,
\begin{equation}
    G_{\ell}(\zeta,\zeta^\prime) \defeq (\Gamma^{(2)}_{\ell}+R_{\ell} \mathds{I})^{-1}(\zeta,\zeta^\prime),
\end{equation}
where $\mathds{I}$ is the identity matrix with elements $\delta_{\zeta\zeta^\prime}$ and $\Gamma^{(2)}_{\ell}$ is the second derivative of $\Gamma_{\ell}$:
\begin{equation}
    \Gamma^{(2)}_{\ell}(\zeta,\zeta^\prime) \defeq \frac{\partial^2 \Gamma_{\ell}}{\partial M(\zeta) \partial M(\zeta^\prime)}.
\end{equation}
Notice that in the Wetterich's equation~\eqref{Wett}, the dot derivative is calculated at fixed $M$ rather than at fixed $J$, the two derivatives being related by:
\begin{equation}
    \dot{\Gamma}_{\ell} \vert_{J~\text{fixed}}
    =
    \dot{\Gamma}_{\ell} \vert_{M~\text{fixed}} + \left. \frac{\partial \Gamma_{\ell}}{\partial M} \dot{M}\right|_{J~\text{fixed}}.
\end{equation}
While this is an exact equation, it is practically impossible to solve exactly.
We should remark that this is not a specific feature of the \rg.
Most equations in physics cannot be solved exactly and require approximations.
In the following, we shall discuss the form of these approximations in the specific context of \dft.

\subsection{The Local Potential Approximation}

\dft is a local theory that is distinguished from ordinary field theories by the features of its momentum distribution.
The structure of its interactions justifies the use of standard approximation methods to construct approximate solutions to the Wetterich's equation~\eqref{Wett}.
The simplest and most widespread is the \lpa (\emph{Local Potential Approximation}).
The choice of approximation is intimately linked to the scaling behaviour of the couplings, that is, to their canonical dimensions.
We shall see that from this perspective, in the vicinity of the \mpdistr universality class, \dft essentially behave like ordinary Euclidean field theories in three dimensions, justifying a posteriori the choice of the \lpa.

\subsubsection{The Principles}

\Cref{Wett} is a functional equation describing trajectories in an infinite-dimensional space.
Finding exact solutions for problems of actual interest is notoriously difficult, and the \dft makes no exception.
The most general strategy consists in constructing approximate solutions within different regimes.
These solutions, commonly referred to as \emph{truncations}, essentially project the flow onto a finite-dimensional or otherwise manageable subspace.
The first approximation relies on a derivative expansion, that is, an expansion in powers of the momenta $p_\mu$.
This approximation is notably assumed to be valid in the deep IR ($k \ll 1$) where the \emph{effective average action} (\eaa) becomes:
\begin{equation}
    \Gamma_k[M]
    \defeq
    \frac{1}{2}\sum_{\mu=1}^P Z(k)p^2_\mu\, M(p_\mu)  M(-p_\mu)+\mathcal{U}_k[M]+\mathcal{O}(p_\mu^2 M^3) + \mathcal{O}(p_\mu^3 M^2),
    \label{eq:LPA_basic_exp}
\end{equation}
where the form of the coupling is constrained by the form of the \uv theory space, many interaction terms being forbidden by the perturbative expansion of the \rg.
Clearly, this is a strong assumption.
To rigorously verify it, one would in principle need to construct an explicit trajectory starting from the deep \uv.
However, as we shall see, canonical dimensions diverge in this limit, causing most standard approximations to fail.
We must therefore accept that the flow exists, even if it is not explicitly computable, and instead construct a tractable \ir approximation.
The deep \ir is thus not only the region of interest (where the signal resides), but it also represents the region where most standard approximations apply.
Note moreover that, in a large part of this work, investigations focus on the Gaussian region, and the above argument can be sidestepped, though it is important to identify the limitations of the approach.

The functional $\mathcal{U}_k[M]$ does not involve any powers of $p_\mu$, and following the common terminology in the literature, we shall call it the \emph{local potential}.
The coupling $Z(k)$ denotes the \emph{wave function renormalization}.
The \lpa essentially consists in neglecting this effect by setting $Z(k) \equiv 1$ (which is justified in the vicinity of the Gaussian fixed point and will be discussed in \Cref{part3}), so that within this approximation:
\begin{equation}
    \Gamma_k[M]
    \defeq
    \frac{1}{2}\sum_{\mu=1}^N p^2_\mu\, M(p_\mu)  M(-p_\mu)+\mathcal{U}_k[M].
\end{equation}
In the following, we will specifically consider two approximation schemes for the functional $\mathcal{U}_k[M]$, valid in two different regimes: the Gaussian regime, within the \emph{symmetric} phase, and far from the Gaussian regime, within the \emph{broken symmetry} phase, using the standard terminology related to the physics of symmetry breaking, as discussed in \Cref{part1}.

Let us conclude this section with some remarks on the continuum limit.
With the newly defined expansion, \Cref{Wett} involves a discrete sum over the momentum modes $p_\mu$.
Since our analytical calculations and numerical simulations focus on the \enquote{large-$P$ limit} ($P \gg 1$), where the number of \dof is large, one can replace this discrete sum with a continuous integral by introducing the mode density $\rho(p^2)$, via the substitution
\begin{equation}
    \sum_\mu \to \int \dd p\, \rho(p^2) p.
\end{equation}
However, this continuous approximation disregards the contribution of the isolated mode associated with the eigenvalue $\lambda_+-\lambda_- = 1/m^2$ (the inverse of the renormalized mass).
Within this framework, the \rg flow is best understood as describing the evolution of the effective couplings governing these zero modes.
As a consequence, the infrared scale $k$ does not flow strictly to zero but terminates at the eigenvalue immediately above the zero modes.
Given that the spacing of typical eigenvalues scales as $\mathcal{O}(1/P)$ (which translates to a momentum scale of order $1/\sqrt{P}$), the running scale is bounded from below, meaning that $k \in [\sim 1/\sqrt{P}, \Lambda_{\text{\uv}}]$.
Let $P_c \le P$ denote the effective number of \dof in the nearly continuous spectrum, then the flow equation becomes:
\begin{equation}
    \dot{\Gamma}_k = P_c \int_0^\infty \dd p\, \rho(p^2) p\, \dot{R}_k(p^2) \left( \Gamma^{(2)}_k + R_k \right)^{-1}(p,-p),
    \label{Wett2}
\end{equation}
where the upper integration limit has been extended to $+\infty$, assuming that the momentum window selected by the regulator derivative $\dot{R}_k$ is highly localized.
Note moreover that the $2$-point function is defined as:
\begin{equation}
    G_{k}(p,-p) \defeq \frac{\partial^2 W_k}{\partial J(p) \partial J(-p)} = \left( \Gamma^{(2)}_k + R_k \right)^{-1}(p,-p),
\end{equation}
the $\pm p$ structure reflecting the momentum conservation.

This setup is reminiscent of field theories in finite geometries, which is expected since the total integrated mode density remains bounded: $\frac{1}{P_c} \sum_\mu 1 \sim \int \dd p\, \rho(p^2) p = \mathcal{O}(1)$.
To illustrate this analogy, consider a one-dimensional lattice theory with periodic boundary conditions, where the momentum is quantized as $p_n = \frac{2\pi n}{N}$, and $N$ denotes the number of lattice sites.
The spacing between consecutive modes is then equal to $2\pi/N$, while the total number of modes is $\sum_n 1 = N$, representing the \enquote{volume} of the network, which does not changes under \rg transform.

\subsubsection{Symmetric Phase and Vertex Expansion}

The condition $J=0$ is generally imposed on the source terms after computations.
The symmetric phase is defined as the region of the phase space where an expansion around the solution $M=0$ of the exact \Cref{eqstate} is consistent.
In this scenario, the local potential $\mathcal{U}_k[M]$ takes the form of an expansion in powers of $M$, where each monomial represents a local interaction in the sense defined in \Cref{section1}:
\begin{equation}
    \mathcal{U}_k[M]
    \defeq
    \frac{1}{2} \sum_p u_2 M(p) M(-p) +
    \sum_{n=3}^\infty\, \frac{u_{2n}}{(2n)!P_c^{n-1}}\,
    \sum_{\{p_i \}} \delta \left(\sum_{i=1}^n p_i\right)
    \prod_{i=1}^{2n} M(p_i).
    \label{ultralocalint}
\end{equation}
Note that the power of $P_c$ is chosen such that $1/P_c$ exist, and in particular such that the $P_c \to \infty$ limit can be suitably constructed.

The cancellation of the classical field $M$ ensures the cancellation of all the odd vertex function $\Gamma_k^{(2n+1)}$, with the definition:
\begin{equation}
    \Gamma_k^{(n)}(p_{\mu_1}, p_{\mu_1}, \cdots, p_{\mu_n})
    \defeq
    \frac{\partial^n \Gamma_k[M]}{\partial M(p_{\mu_1})\partial M(p_{\mu_2})\cdots \partial M(p_{\mu_n})}.
    \label{npoints}
\end{equation}
Taking the second derivative for $M$ of the Wetterich-Morris equation \eqref{Wett}, setting $M=0$ on both sides of the equation and using the condition above, we get:
\begin{equation}
    \dot{\Gamma}^{(2)}_{k}(p_{\mu_1},-p_{\mu_1})
    =
    -\frac{1}{2}\sum_{p_\mu} \dot{R}_k(p_\mu^2) G_{k}^2(p_\mu^2) \Gamma_{k}^{(4)}(p_\mu,-p_\mu,p_{\mu_1},-p_{\mu_1}).
    \label{equflow1}
\end{equation}

Let us recall that our inference prescription \eqref{eq:inference_constraint_corr_mat} enforces the condition for the \ir mass $m^2$:
\begin{equation}
    m^2 \defeq \Gamma^{(2)}_{k=0}(0,0),
\end{equation}
as the mass is defined as the zero-mode (largest eigenvalue) of the 2-point function (i.e.\ the \ecm).
Moreover, $1/m^2$ must behave like any eigenvalue under a global dilation of the spectrum.
In other words, $m^2$ must have the same scaling behaviour as the cut-off $\Lambda^2_{\text{UV}}$.
Hence, it follows that the 2-point function at zero momenta,
\begin{equation}
    u_2(k) \defeq \Gamma^{(2)}_{k}(0,0),
    \label{defu2}
\end{equation}
which is the \emph{effective mass at scale} $k$ (notice that $k \ge 0$ in this notation), must have the same scaling as $k^2$ itself, and we define the \emph{dimensionless} mass $\bar{u}_2$ as:
\begin{equation}
    \bar{u}_2(k) \defeq k^{-2} u_2(k).
    \label{defu2bar}
\end{equation}
Following the standard definition in field theory used in \Cref{part1}, this means that the scale dimension of the mass $u_2$ is $2$.

We shall now return to \Cref{equflow1}.
Assuming that $\dot{R}_k$ allows only a narrow window of momenta around $k^2$, and that $k^2 \ll 1$ following the fact that we are only interested in the tail of the spectrum, we expect to be able to neglect the dependence in $p_\mu$ in $\Gamma_k^{(4)}$, replacing $p_\mu$ by $0$:
\begin{equation}
    \dot{\Gamma}^{(2)}_{k}(p_{\mu_1},-p_{\mu_1})
    \approx
    -\frac{1}{2}\Gamma_{k}^{(4)}(0,0,p_{\mu_1},-p_{\mu_1})\sum_{p_\mu} \dot{R}_k(p_\mu^2) G_{k}^2(p_\mu^2).
    \label{equflow1Bis}
\end{equation}
Finally, in the local potential approximation, $\Gamma^{(4)}_k$ must have the form of a local vertex (interaction):
\begin{equation}
    \Gamma^{(4)}_k(p_1,p_2,p_3,p_4)=\frac{u_4(k)}{P_c} \delta (p_1+p_2+p_3+p_4),\label{quarticvert}
\end{equation}
according to the definition \eqref{ultralocalint}.

For the rest of this work, we choose the standard optimised Litim regulator \cite{Litim}:
\begin{equation}
    R_k(p^2) \defeq (k^2-p^2)\theta(k^2-p^2).
    \label{Litim}
\end{equation}
Hence, setting $p_{\mu_1}=0$, and using the definitions \eqref{defu2} and \eqref{defu2bar}, we find:
\begin{equation}
    \dot{u}_2= -\frac{1}{2}\frac{2k^2}{(k^2+u_2)^2} \sum_{p_\mu} \theta(k^2-p^2_\mu) \left. \Gamma^{(4)}_{k}(p_\mu,-p_\mu,p_{\mu_1},-p_{\mu_1})\right|_{p_{\mu_1}=0},
\end{equation}
or, in the continuum limit, considering the overall factor $1/P_c$ in front of the right hand side (coming from the quartic vertex, equation \eqref{quarticvert}):
\begin{equation}
    \dot{\bar{u}}_2
    =
    -2\bar{u}_2-\frac{2u_4}{(1+\bar{u}_2)^2} \frac{1}{k^4} \int_0^k \dd p \, \rho(p^2) p.
    \label{eqfloWu2}
\end{equation}
Similar equations for $u_4$, $u_6$, etc.\ can be obtained by taking higher derivatives of \eqref{Wett}.
Higher-order flow equations for the vertex function $\Gamma_k^{(2n)}$, then, involve $\Gamma_k^{(2n+2)}$, for every order $n$.
To close the system of equations, we must thus truncate the flow, i.e.\ project it into a finite-dimensional subspace, by imposing:
\begin{equation}
    \Gamma_k^{(2\bar{P})}=0,
\end{equation}
up to a given $\bar{P}$.
In what follow, we will consider explicitly the truncation around $\bar{P}=3$, taking into account only local sextic effective interactions (i.e.\ we keep only the relevant couplings in the deep \ir).

\subsubsection{Beyond the Symmetric Phase}

In this section, we extend the formalism to the case of a non-zero vacuum.
This formalism is particularly suitable for investigations in the deep \ir, and we will assume that the vacuum only affects the zero-mode, neglecting the momentum dependence of the classical field $M(p_\mu)$:
\begin{equation}
    M_0(p_\mu)\sim M \delta_{\mu,0}.
    \label{eq:M_zeromode}
\end{equation}
This approximation works well at a large scale, where a symmetry breaking scenario is expected, requiring an expansion around a non-vanishing vacuum $M\neq 0$.
For this reason, we consider the following parameters:
\begin{equation}
    {U_k[\chi]}
    =
    \frac{u_4(k)}{2!} \left(\chi-\kappa(k)\right)^2+\frac{u_6(k)}{3!} \left(\chi-\kappa(k)\right)^3+\dots,
    \label{truncationpotential}
\end{equation}
where the broken-phase couplings $u_4(k)$ and $u_6(k)$ denote the second and third derivatives of $U_k$ with respect to $\chi$, evaluated at the running vacuum $\chi = \kappa(k)$. They differ from the symmetric-phase quartic coupling of \eqref{quarticvert} by a factor that depends on $N$ and on the change of variable $M^2 = 2 N \chi$ introduced below.
where
\begin{equation}
    N \chi \defeq M^2/2.
\end{equation}
We call $\kappa(k)$ the \emph{running vacuum}.
The global normalisation is chosen such that, under condition \eqref{eq:M_zeromode}, we have:
\begin{equation}
    N{U_k[\chi]} \defeq \Gamma_k[M=M_0].
\end{equation}
The $2$-point vertex $\Gamma_k^{(2)}$ is then defined as:\footnote{%
    Notice that in \eqref{eq:LPA_basic_exp}, we introduced the wavefunction renormalization $Z(k)$, which presents a non-vanishing flow equation for $\kappa\neq 0$.
    The curious reader should have a look at the full derivation~\cite{RG5}.
}
\begin{equation}
    \Gamma^{(2)}_{k}(p_\mu,p_\mu^\prime)
    =
    \left.\left(p^2+\frac{\partial^2 U_k}{\partial M^2 }\right)\right|_{M^2=2N\chi} \delta(p_\mu+p_{\mu^\prime}),
    \label{2points2}
\end{equation}
where the second derivative of the potential plays the role of an effective mass.
The flow equation for $U_k$ can be deduced from \eqref{Wett}, setting $M=M_0$ on both sides.
For a large number of \dof, taking the continuum limit, we get:
\begin{equation}
    \dot{U}_k[M]
    =
    \frac{1}{2}\, \int \dd p^2 \, k\partial_k (R_k(p^2)) \rho(p^2) \left( \frac{1}{\Gamma^{(2)}_k+R_k} \right)(p,-p).
    \label{exactRGEbis}
\end{equation}
Then, using \eqref{2points2} and the Litim regulator \eqref{Litim}, we arrive to the formula of the flow of the effective potential~\cite{RG5}:
\begin{equation}
    \dot{U}_k[\chi]=\left(2 \int_0^k \dd p \, \rho(p^2)p \right)\, \frac{k^2}{k^2+\partial_\chi U_k (\chi)+ 2\chi\, \partial^2_{\chi}U_k(\chi)}.
    \label{exactRGEbis_2}
\end{equation}
The flow equations can be deduced using the same strategy as before.
Because the couplings $\kappa$, $u_4$, $u_6$, etc., satisfy the conditions:
\begin{align}
    \left. \frac{\partial U_k}{\partial \chi}\right|_{\chi=\kappa}     & = 0,\label{renmass}        \\
    \left. \frac{\partial^2 U_k}{\partial \chi^2}\right|_{\chi=\kappa} & = u_4(k),\label{rencoupl1} \\
    \left. \frac{\partial^3 U_k}{\partial \chi^3}\right|_{\chi=\kappa} & = u_6(k),\label{rencoupl2}
\end{align}
we can obtain the flow equation for $u_{2n}$ by simply differentiating $n$ times both sides of the equation and setting $\chi=\kappa(k)$.
The flow equation for $\kappa$ is slightly trickier, though we notice:
\begin{equation}
    \frac{\partial \dot{U}_k}{\partial \chi}=-u_4 \dot{\kappa}.
\end{equation}
We will return to the explicit derivation of the flow equation in the next section.

\section{Generalized Scaling and Dimensions}\label{sec:generalized_scaling}
In this section, we analyse the heart of \dft: the dependence of the canonical dimensions on the renormalization scale $k$.
This dependence, inherent to the structure of the data, can be interpreted as a scale dependence of the effective dimension of the background space, which is itself shaped by the data.
This scaling forms the basis of the signal detection approach that we develop in \Cref{part3} of this work.

\subsection{Canonical Dimensions}

The flow equations derived above involve a scale factor taking the form of an integral over the distribution $\rho(p^2)$ in the continuous limit (see, for instance, \eqref{eqfloWu2}):
\begin{equation}
    L(k)=\int_0^k \dd p \, \rho(p^2) p.
\end{equation}
For a power law distribution $\rho(p^2)\sim (p^2)^{\alpha}$, we recover $L(k)\propto k^{2\alpha+2}$.

We also know the ordinary \enquote{time} of the flow
\begin{equation}
    t \defeq \ln k \Rightarrow \dd t = \frac{1}{k} \dd k.
\end{equation}
We can generalise this definition using $L(k)$, replacing the ordinary logarithmic time by a more general notion that absorbs the shape of the spectral measure:
\begin{equation}
    \dd \tau \defeq \dd \ln \left(\int_0^k \dd p \, \rho(p^2) p \right).
    \label{dtau}
\end{equation}
For a power law distribution we get $\dd\tau=(2\alpha+2) \dd t$.
Requiring the explicit $k$-dependence on both sides of the flow equations to cancel, and using the $t$-time scaling dimension of $u_2$ (which is $2$), the quartic coupling must scale as $u_4\sim k^{2-2\alpha}$ in ordinary $t$ time, that is $\dim_t(u_4)=2-2\alpha$.
In $\tau$ time, with $t^{\prime} = \dd t/\dd\tau = 1/(2\alpha+2)$, this becomes $\dim_\tau(u_4) = (1-\alpha)/(1+\alpha)$.
For an ordinary field theory in dimension $D$, $\alpha=(D-2)/2$ and $t'=1/D$, so the $\tau$-time dimension becomes $\dim_\tau(u_4) = (4-D)/D$.
Multiplying by $D$ recovers the standard $t$-time canonical dimension $4-D$ of a local quartic coupling in $\mathds{R}^D$.

In our case, however, we do not deal with power law distributions.
We suppose that the distribution of $p_\mu$ is given by a generic $\rho(p^2)$, which has an immediate impact as if the effective dimension of the space depended on the scale.
Under these conditions, we cannot remove the explicit scale dependence completely.
The optimal strategy is to absorb the residual scale dependence into the linear term of the coupling, allowing the canonical dimension itself to become scale-dependent.
Replacing $t$ by $\tau$ rescales the canonical dimension of $u_2(k)$ by the Jacobian $t^{\prime} \defeq \dd t/\dd\tau$.
We denote the resulting $\tau$-scaling dimension by $\dim_\tau(u_2)$, which by construction is:
\begin{equation}
    \dim_\tau(u_2) \defeq \frac{\dd \ln k^2}{\dd \tau} = 2 t^\prime.
\end{equation}
This way, the contribution proportional to $u_4$ receives the scale-dependent factor $\rho(k^2) k^{-2} (t^\prime)^2$ (the equation will be derived in full detail in the next section).
Getting rid of the explicit scale dependence in the loop term thus amounts to defining the \emph{locally dimensionless coupling} as:
\begin{equation}
    \bar{u}_4
    \defeq
    u_4 \frac{{\rho}(k^2)}{k^2} \left(\frac{dt}{d\tau}\right)^2.
    \label{baru4}
\end{equation}
Using $\bar{u}_4$ instead of $u_4$ in the flow equation for $u_4$ thus introduces a linear contribution
\begin{equation}
    -\dim_\tau(u_4) \bar{u}_4,
\end{equation}
with the dimension obtained by differentiating \eqref{baru4} w.r.t.\ $\tau$ and using $\dd t = t^{\prime} \dd \tau$:
\begin{equation}
    \dim_\tau(u_4)
    \defeq
    -2\left(\frac{t^{\prime\prime}}{t^\prime}+t^\prime\left(\frac{1}{2} \frac{\dd \ln{\rho}}{\dd t}-1\right)\right).
    \label{canonicalu4}
\end{equation}
The flow equation for $u_4$ in turn fixes the $\tau$-dimension of $u_6$.
The local dimensionless sextic coupling is constructed so that the two explicit factors of $\bar{u}_4$ that enter the quartic loop in the $u_6$ flow are absorbed:
\begin{equation}
    \bar{u}_6
    \defeq
    u_6\,k^2 \left( \frac{{\rho}(k^2)}{k^2} \left(\frac{dt}{d\tau}\right)^2\right)^2.
    \label{u6bar}
\end{equation}
Replacing $u_6$ with $\bar{u}_6$ in its own flow equation introduces the linear term $-\dim_\tau(u_6) \bar{u}_6$, with:
\begin{equation}
    \dim_\tau(u_6)
    \defeq
    -\dim_\tau(u_2) + 2 \dim_{\tau}(u_4).
\end{equation}
The same argument generalises to higher couplings: a coupling $u_{2p}$ with $2p$ field insertions is made dimensionless by absorbing the corresponding powers of $\rho(k^2)/k^2$ and $t^{\prime}$, which yields the recurrence
\begin{equation}
    \dim_\tau(u_{2p})
    \defeq
    (2-p) \dim_{\tau}(u_2) + (p-1) \dim_\tau(u_4).
\end{equation}

\subsection{Dimensionless Flow Equations}

In this section, we show how the previous definitions allow us to rewrite the flow equations in an explicitly dimensionless form, both in the symmetric phase and in the broken phase.
In turn, this will be the starting point for numerical investigations of the flow equations, which will be the subject of \Cref{part3,part4}.

\subsubsection{Vertex Expansion in the Symmetric Phase}

We truncate the tower of couplings at order $u_6$ by setting $u_8=0$ in the flow of $\bar{u}_6$, retaining only the quadratic, quartic, and sextic terms in the effective action.
The truncation order is denoted by $\bar{P}=3$ since three couplings ($u_2, u_4, u_6$) are kept, in agreement with the truncation introduced in the previous section:
\begin{equation}
    \begin{split}
        \Gamma_k[M]
        & = \frac{1}{2}\sum_p M(-p) \left( p^2+u_2(k)\right) M(p)
        \\
        & + \frac{u_4(k)}{4!\,P_c}\, \sum_{\{p_i\}} \, \delta\left(\sum_i p_i \right) \prod_{i=1}^4 M(p_i)
        \\
        & +\frac{u_6(k)}{6!\,P_c^{2}}\, \sum_{\{p_i\}}\,  \delta\left(\sum_i p_i \right) \prod_{i=1}^6 M(p_i).
    \end{split}
    \label{truncation1}
\end{equation}
The powers of $P_c$ are chosen so that the couplings $u_4$ and $u_6$ have a well-defined $P_c\to\infty$ limit, with $P_c$ the number of effective \dof in the nearly continuous spectrum.
Using the time flow \eqref{dtau}, the flow equation for $u_2$ becomes:
\begin{equation}
    \bar{u}_2^{\prime}
    =
    -2 t^{\prime} \bar{u}_2
    -
    \frac{2u_4}{(1+\bar{u}_2)^2}
    \frac{{\rho}(k^2)}{k^2}
    \left(t^{\prime}\right)^2,
\end{equation}
and the definition \eqref{baru4} gives us:
\begin{equation}
    \bar{u}_2^{\prime}
    =
    -2 t^{\prime} \bar{u}_2
    -
    \frac{2\bar{u}_4}{(1+\bar{u}_2)^2}.
    \label{flow1}
\end{equation}
The flow equation for the coupling $u_4$ can be computed following the same strategy.
Taking the fourth derivative of the flow equation \eqref{Wett} with respect to $M(p)$ and discarding the odd functions which vanish in the symmetric phase, we obtain:
\begin{equation}
    u_4^{\prime}
    =
    -\frac{2u_6}{(1+\bar{u}_2)^2} {\rho}(k^2)\left(t^{\prime}\right)^2
    +
    \frac{12u_4^2}{(1+\bar{u}_2)^3} \frac{{\rho}(k^2)}{k^2} \left(t^{\prime}\right)^2.
\end{equation}
Then, using \eqref{u6bar}, we get:
\begin{equation}
    \bar{u}_4^{\prime}
    =
    -\dim_\tau(u_4) \bar{u}_4
    -
    \frac{2\bar{u}_6}{(1+\bar{u}_2)^2}
    +
    \frac{12\bar{u}^2_4}{(1+\bar{u}_2)^3}.
    \label{flow2}
\end{equation}
Finally, the truncation closes by setting $u_8=0$ in the flow of $\bar{u}_6$, giving:
\begin{equation}
    \bar{u}_6^{\prime}
    =
    -\dim_\tau(u_6)\bar{u}_6
    +
    60\, \frac{\bar{u}_4\bar{u}_6}{(1+\bar{u}_2)^3}
    -
    108\, \frac{\bar{u}_6^3}{(1+\bar{u}_2)^4}.
    \label{flow3}
\end{equation}

\subsubsection{Beyond the Symmetric Phase}

Using the flow parameter $\tau$ in \eqref{exactRGEbis_2}, we obtain:
\begin{equation}
    {U}_k^\prime[\chi]=k^2 \rho(k^2) \left( t^{\prime} \right)^2\, \frac{k^2}{k^2+\partial_{\chi}U_k(\chi)+ 2\chi \partial^2_{\chi}U_k(\chi)}.
    \label{eq:flowforUk_nobar}
\end{equation}
Then, according to the definitions adopted in the symmetric phase, the scaling of the effective potential has to verify (globally, the denominator must scale as $k^2$).
Comparing \eqref{eq:flowforUk_nobar} with the expression at fixed $\bar\chi$ obtained from the chain rule, we see that the denominator $k^2+\partial_\chi U_k+2\chi\partial^2_\chi U_k$ must be invariant under the rescaling of $U_k$ and $\chi$, which gives:
\begin{equation}
    \partial_{\chi}U_k (\chi)k^{-2}
    =
    \partial_{\bar\chi}\bar{U}_k (\bar{\chi}),
    \qquad
    \chi \partial^2_{\chi}U_k(\chi) k^{-2}
    =
    \bar{\chi}\partial^2_{\bar\chi} \bar{U}_k(\bar{\chi}),
    \label{scaling1}
\end{equation}
leading to:
\begin{equation}
    {U}_k^\prime[\chi]
    =
    \left(t^{\prime}\right)^2\, \frac{k^2\rho(k^2)}{1+\partial_{\bar\chi}\bar{U}_k (\bar{\chi})+ 2\bar{\chi}\partial^2_{\bar\chi} \bar{U}_k(\bar{\chi})}.
    \label{flowforUk}
\end{equation}
\Cref{scaling1} fixes the relative scaling of $U_k$ and $\chi$.
The previous relation also fixes the absolute scaling, which we use to define the dimensionless potential:
\begin{equation}
    {U}_k[\chi]
    \defeq
    \bar{U}_k[\bar{\chi}] k^2\rho(k^2) \left(t^{\prime}\right)^2.
\end{equation}

Finally, we fix the dimension of the field $\chi$.
Let us introduce a scale-dependent rescaling $\chi \defeq A(k) \bar{\chi}$, with the prefactor $A(k)$ to be determined below.
Substituting it in the rescaled potential gives:
\begin{equation}
    {U}_k[\chi]
    \defeq
    \bar{U}_k[A^{-1}{\chi}]k^2 \rho(k^2) \left(t^{\prime}\right)^2.
\end{equation}
Expanding in power of $\chi$ on both sides, we find for the linear term:
\begin{equation}
    \chi \partial_{\chi}U_k[\chi=0]
    =
    \bar{\chi}
    \partial_{\bar{\chi}}\bar{U}_k[\bar{\chi}=0]
    k^2\rho(k^2) \left(t^{\prime}\right)^2,
\end{equation}
or, from \eqref{scaling1}:
\begin{equation}
    \chi \partial_{\chi}U_k[\chi=0]
    =
    \chi \partial_{\chi}U_k[\chi=0] A^{-1} \rho(k^2) \left(t^{\prime}\right)^2.
\end{equation}
Then, using \eqref{scaling1} for the linear term, the two sides differ by the factor $A^{-1}\rho(k^2)(t^{\prime})^2$, and assuming $\chi \partial_{\chi}U_k[\chi=0] \neq 0$ (the linear term does not vanish by hypothesis), this factor must equal $1$, which fixes:
\begin{equation}
    A(k) \defeq \rho(k^2) \left(t^{\prime}\right)^2,
\end{equation}
and:
\begin{equation}
    \bar{\chi} \defeq \frac{1}{\rho(k^2) \left(t^{\prime}\right)^2} \chi.
\end{equation}

Notice that these equations hold in particular for $\chi = \kappa$.
To derive the flow equations for the dimensionless couplings, it is convenient to work with a flow equation at fixed $\bar{\chi}$.
However, the flow equation \eqref{flowforUk} is defined for a fixed $\chi$.
Converting the flow at fixed $\chi$ to a flow at fixed $\bar{\chi}$ introduces two extra terms: the explicit $k$-dependence of the rescaling factor in $\bar U_k$ (giving $\dim_\tau(U_k)\bar U_k$) and the rescaling of the field (giving $-\dim_\tau(\chi)\bar\chi\partial_{\bar\chi}\bar U_k$).
Explicitly, we get:
\begin{equation}
    {U}_k^\prime[\chi]
    =
    \rho(k^2) \left(t^{\prime}\right)^2
    \left[
    {\bar{U}}_k^\prime[\bar{\chi}]+\dim_\tau(U_k)\bar{U}_k[\bar{\chi}] -\dim_\tau(\chi) \bar{\chi} \partial_{\bar{\chi}} \bar{U}_k[\bar{\chi}]
    \right],
    \label{eqLagrangebis}
\end{equation}
where $\dim_\tau(U_k)$ and $\dim_\tau(\chi)$ respectively denote the canonical dimensions of $U_k$ and $\chi$.
The factor $k^2\rho(k^2)(t^{\prime})^2$ in the definition of $\bar U_k$ enters $\dim_\tau(U_k)$, and the factor $\rho(k^2)(t^{\prime})^2$ in the definition of $\bar\chi$ enters $\dim_\tau(\chi)$.
Their explicit expressions are obtained by differentiating the corresponding logarithms:
\begin{equation}
    \dim_\tau(U_k)= t^\prime \frac{\dd}{\dd t} \ln \left(k^2\rho(k^2) \left( t^{\prime}\right)^2 \right),
\end{equation}
and
\begin{equation}
    \dim_\tau(\chi)= t^\prime \frac{\dd}{\dd t} \ln \left(\rho(k^2) \left( t^{\prime}\right)^2 \right).
\end{equation}
The final expression for the effective potential \rg equation is then:
\begin{equation}
    \bar{U}_k^\prime[\bar{\chi}]
    =
    -\dim_\tau(U_k)\bar{U}_k[\bar{\chi}]
    +\dim_\tau(\chi) \bar{\chi} \partial_{\bar{\chi}} \bar{U}_k[\bar{\chi}]
    +\frac{1}{1+\partial_{\bar\chi}\bar{U}_k (\bar{\chi})
    + 2\bar{\chi}\partial^2_{\bar\chi} \bar{U}_k(\bar{\chi})}.
    \label{potentialflow}
\end{equation}

From this expression, it is straightforward to work out the explicit expressions for the coupling constants.
Using the definition \eqref{renmass}, we have
\begin{equation}
    \partial_{\bar\chi}{\bar{U}}^{\prime}_k[\bar{\chi}=\bar{\kappa}]=-\bar{u}_4\, {\bar{\kappa}}^{\prime}.
\end{equation}
Therefore, differentiating equation \eqref{potentialflow}, we obtain the following expression for ${\bar{\kappa}}^{\prime}$:
\begin{equation}
    {\bar{\kappa}}^{\prime}
    =
    -\dim_\tau(\chi) \bar{\kappa}
    +2\frac{3+2\bar{\kappa} \frac{\bar{u}_6}{\bar{u}_4}}{(1+ 2\bar{\kappa}
        \bar{u}_4)^2}
\end{equation}
Taking second and third derivatives, and from the conditions \eqref{rencoupl1} and \eqref{rencoupl2}, we obtain the flow equations of the dimensionless broken-phase couplings:
\begin{align}
    {\bar{u}_4}^{\prime}
     & =
    -\dim_\tau(u_4) \bar{u}_4 +\dim_\tau(\chi) \bar{\kappa} \bar{u}_6-\,\frac{10\bar{u}_6}{(1+2\bar{\kappa} \bar{u}_4)^2}+4\,\frac{(3\bar{u}_4+2\bar{\kappa} \bar{u}_6)^2}{(1+ 2\bar{\kappa} \bar{u}_4)^3},
    \\
    {\bar{u}_6}^{\prime}
     & =
    -\dim_\tau(u_6) \bar{u}_6 -12\, \frac{(3\bar{u}_4+2\bar{\kappa}\bar{u}_6)^3}{(1+2\bar{\kappa} \bar{u}_4)^4}+40\bar{u}_6\,\frac{3\bar{u}_4+2\bar{\kappa}\bar{u}_6}{(1+ 2\bar{\kappa} \bar{u}_4)^3},
\end{align}
where the term $-\dim_\tau(\chi)\bar\kappa\bar u_6$ in $\bar u_4^\prime$ arises from the $\tau$-dependence of the running vacuum $\bar\kappa$, which contributes to the second derivative of the potential through the chain rule.

\subsection{Wave Function Renormalization and Anomalous Dimension}

The previous formalism can be extended to take into consideration the wave function renormalization (\wfr) factor $Z(k)$, which we neglected so far.
This coupling affects the scaling dimensions of the renormalized couplings, which are henceforth defined by:
\begin{equation}
    u_{2n} = Z^n(k) k^{\dim_\tau(u_{2n})} \bar{u}_{2n}.
    \label{defcouplingRenprime}
\end{equation}
Notice that, at a fixed point (asymptotic, in this case) $p_*$, $Z(k)$ obeys a scaling law $Z_*(k) = k^{\eta(p_*)}$, which introduces a new explicit $k$-dependence into the flow equations around this point.
The definition \eqref{defcouplingRenprime} ensures that the flow equations remain autonomous, as per requirement formulated in \Cref{part1}.

The quantity $\eta(p_*)$ is called the \emph{anomalous dimension}.
The general definition (extending beyond the neighbourhood of a fixed point) is:
\begin{equation}
    \eta(k)
    \defeq
    k\frac{\dd}{\dd k} \ln Z(k),
\end{equation}
and one can verify by explicit substitution that this definition ensures that the equations remain autonomous.
The anomalous dimension generally yields a non-trivial contribution to the calculation of critical exponents.
\Cref{defcouplingRenprime} ensures that, for the \emph{renormalized field}, we have:
\begin{equation}
    \bar{\Phi}_R \defeq Z^{-1/2}(k)k^{-\dim_\tau(\phi)} \Phi,
\end{equation}
and the effective propagator becomes:
\begin{equation}
    \langle \bar{\Phi}_R(0) \bar{\Phi}_R(0) \rangle = \frac{1}{\bar{u}_{2}}.
\end{equation}

To discuss the \wfr, it is suitable to modify the standard optimised Litim regulator \eqref{Litim} introduced in the previous section, following~\cite{Litim,RG5}:
\begin{equation}
    R_k(p^2) = Z(k) (k^2-p^2)\theta(k^2-p^2).
    \label{eq:modified_litim}
\end{equation}
The truncation, which we denote by \lpa\!$^{\prime}$ (i.e.\ \lpa plus wave-function renormalization), includes $Z(k)$ explicitly:
\begin{equation}
    \Gamma_k[M]
    \defeq
    \frac{1}{2}\sum_{\mu=1}^{P_c}
    Z(M,k)\,p^2_\mu\, M(p_\mu)  M(-p_\mu)+\mathcal{U}_k[M],
    \label{LPAprime}
\end{equation}
where the function $Z(M,k)$ includes an explicit dependence on $M$.
Assuming the vacuum of the effective potential is $\sqrt{2\kappa}$, we have:
\begin{equation}
    Z(M=\sqrt{2\kappa},k)
    \equiv
    \left.\frac{\dd}{\dd p^2}\Gamma^{(2)}_k(p,-p)\right|_{M=\sqrt{2\kappa},\,p=0}.
\end{equation}
In the symmetric phase, $\kappa=0$ and the vacuum sits at the origin, so the second functional derivative of the effective action at $p=0$ is the second derivative of the potential at the origin, namely $u_2$ (the kinetic term, which is proportional to $p^2$, vanishes at $p=0$).
For $p^2 < k^2$, the modified Litim regulator is active and the regulated propagator factors as $G_k(p^2) = (Z(k)k^2)^{-1}(1+\bar{u}_2)^{-1}$, where the prefactor $(Z(k)k^2)^{-1}$ is absorbed into the rescaling of the dimensionless coupling.
The factor $(1+\bar{u}_2)^{-1}$ then appears explicitly in the loop term of the dimensionless flow equation, as in \eqref{flow1}.
Then, since the internal loop has no dependency with respect to the external momenta, the anomalous dimension must have a vanishing flow:
\begin{equation}
    \dot{Z}=0.
\end{equation}
Assuming that $Z$ depends only on the value of the vacuum, we obtain the running anomalous dimension by differentiating the flow of the two-point function.
The standard derivation takes the $p^2$-derivative of the Wetterich flow at fixed $M=\sqrt{2\kappa}$ and identifies the coefficient of $p^2$ as the flowing $Z(k)$:
\begin{equation}
    \eta(k)
    \defeq
    \frac{1}{Z(k)}k\frac{\dd Z}{\dd k}
    =\frac{1}{Z(k)} \left.\frac{\dd}{\dd p^2}\dot{\Gamma}^{(2)}_k(p,-p)\right|_{M=\sqrt{2\kappa},\,p=0}.
\end{equation}
The flow equation for $\Gamma_k^{(2)}$ can be computed from \eqref{Wett} taking the second derivative with respect to the classical field.
Because the effective vertex are momentum independent in the \lpa, the contributions involving only $\Gamma^{(4)}_k$ have to be discarded from the flow equation for $Z$, as it does not depend on the external generalized momenta $p$.
Therefore:
\begin{equation}
    \dot{Z}
    \defeq
    (\Gamma^{(3)}_{k}(0,0,0))^2 \left. \frac{\dd}{\dd p^2}\sum_{q}\dot{R}_k(q^2) G^2(q^2)G((q+p)^2) \right|_{M=\sqrt{2\kappa}, p=0},
\end{equation}
where, according to the \lpa, we evaluated the right-hand side over uniform configurations.
Computing the $p^2$-derivative of the Wetterich flow and identifying the coefficient of $p^2$ in the limit $\kappa \neq 0$ yields the following closed form:
\begin{proposition}{Anomalous Dimension}{propanomalous}
    The anomalous dimension $\eta(k)$ for $\kappa\neq 0$ is given by:
    \begin{equation}
        \eta(k)
        =
        2(t^\prime)^{-2}\frac{(3\sqrt{2\bar{\kappa}} \bar{u}_4 + (2\bar{\kappa})^{3/2} \bar{u}_6)^2}{(1+2\bar{\kappa} \bar{u}_4)^4} .
        \label{anomalousdim}
    \end{equation}
\end{proposition}

\begin{proof}
    We have that
    \begin{equation}
        G_k(p,p^\prime) \defeq G_k(p^2)\delta(p+p^\prime)
        =
        \left[\left(\Gamma^{(2)}_k(p,p^\prime) + R_k(p^2) \right) \delta(p+p^\prime)\right]^{-1}.
    \end{equation}
    The expression of $\Gamma^{(3)}_{k}(0,0,0)$ can be obtained from the definition \eqref{npoints}, taking the third derivative of the effective potential for $M$ and setting $M = \sqrt{2\kappa}$:
    \begin{equation}
        \Gamma^{(3)}_{k}(0,0,0)= 3 u_4 \sqrt{2\kappa} + u_6 (2\kappa)^{3/2}.
    \end{equation}
    Next, as the Litim regulator was modified in \eqref{eq:modified_litim}, we get:
    \begin{equation}
        \dot{R}_k(p^2) = \eta(k) R_k(p^2)+2Z(k) k^2\theta(k^2-p^2)\,,
    \end{equation}
    and
    \begin{equation}
        \frac{\dd}{\dd p^2}\, R_k(p^2) = -Z(k)\theta(k^2-p^2).
    \end{equation}
    In this framework, the diagonal components of the effective propagator reads:
    \begin{equation}
        G_k(p^2)
        =
        \frac{1}{Z(k)p^2+Z(k)(k^2-p^2)\theta(k^2-p^2)+\mathcal{M}^2(\kappa)},
    \end{equation}
    where $\mathcal{M}^2$ denotes the effective mass, i.e.\ the second derivative of the effective potential:
    \begin{equation}
        \mathcal{M}^2(\kappa)
        \defeq
        \partial_{{\chi}}\mathcal{U}_k(\kappa)+ 2{\kappa} \partial^2_{{\chi}}\mathcal{U}_k({\kappa}).
    \end{equation}

    The computation then involves integrals of the form:
    \begin{equation}
        I_n(k,p)
        =
        \int_{-k}^k \dd q \, \rho(q^2) q (q^2)^n G_k\left((p+q)^2 \right).
    \end{equation}
    We focus on small and positive $p$.
    The integral thus decomposes as
    \begin{equation}
        I_n(k,p)=I_n^{(+)}(k,p)+I_n^{(-)}(k,p),
    \end{equation}
    where:
    \begin{equation}
        I_n^{(\pm)}(k,p)=\pm\int_{0}^{\pm k} \dd q \, \rho(q^2) q (q^2)^nG_k\left((p+q)^2 \right) .
    \end{equation}
    Since $p>0$, in the negative branch, we have $(q+p)^2<k^2$ and:
    \begin{equation}
        I_n^{(-)}(k,p)=\frac{1}{Z(k) k^2+\mathcal{M}^2}\; \int_{-k}^0 \dd q \, \rho(q^2) q (q^2)^n,
    \end{equation}
    which is independent of $p$.
    In the positive branch however:
    \begin{equation}
        I_n^{(+)}(k,p)= \frac{1}{Z(k)k^2+\mathcal{M}^2}\, \int_0^{k-p} \dd q \, \rho(q^2) q (q^2)^n + \int_{k-p}^k \dd q \, q (q^2)^n\frac{\rho(q^2)}{Z(k)(q+p)^2+\mathcal{M}^2}.
    \end{equation}
    Hence, taking the first derivative with respect to $p$, we get:
    \begin{equation}
        \begin{split}
            \frac{\dd}{ \dd p}I_n^{(+)}(k,p) =
            & -\frac{1}{Z(k) k^2+\mathcal{M}^2} \rho(q^2) q (q^2)^n\vert_{q=k-p}                            \\
            & + \rho(q^2) q (q^2)^n  \frac{1}{Z(k)(q+p)^2+\mathcal{M}^2}\vert_{q=k-p}                       \\
            & -2Z(k) \int_{k-p}^k \dd q \, \rho(q^2) q (q^2)^n \frac{(q+p)}{(Z(k)(q+p)^2+\mathcal{M}^2)^2}.
        \end{split}
    \end{equation}
    The first two terms cancel exactly.
    We can then consider:
    \begin{equation}
        \frac{\dd}{ \dd p}I_n^{(+)}(k,0)
        =
        -2Z(k) \int_{k-p}^k \dd q \, \rho(q^2) q (q^2)^n\frac{(q+p)}{(Z(k)(q+p)^2+\mathcal{M}^2)^2}.
    \end{equation}
    Taking the second derivative and setting $p=0$, the integration range $[k-p,k]$ collapses to a point, so the second derivative of the integral reduces, by Leibniz' rule, to $-2\,\partial f/\partial q(k,0)$ where $f(q,p)$ is the integrand.
    Computing this derivative, the result is
    \begin{equation}
        \frac{1}{2}\frac{ \dd^2}{\dd p^2} I_n(k,0)
        = - \frac{Z(k) \rho(k^2) (k^2)^{n+1}}{(Z(k) k^2+\mathcal{M}^2)^2}
        \defeq
        I_n^{\prime\prime}(k,0).
    \end{equation}
    Therefore, we get:
    \begin{equation}
        \begin{split}
            Z(k)\eta(k)
            & = \frac{(3 u_4 \sqrt{2\kappa} + u_6 (2\kappa)^{3/2})^2}{(Z(k) k^2+\mathcal{M}^2)^2} \left(2 Z(k) k^2 I_0^{\prime\prime}(k,0) \right. \\
            & + \left. Z(k)\eta(k) (k^2 I_0^{\prime\prime}(k,0)-I_1^{\prime\prime}(k,0)) \right).
        \end{split}
    \end{equation}
    To introduce $\tau$-dimensionless quantities, we notice that both $u_4 \kappa$ and $u_6 \kappa^2$ have the same scaling dimension.
    Finally, using renormalized and dimensionless quantities $\bar{u}_{2p}$, and replacing the effective mass by its value:
    \begin{equation}
        \bar{\mathcal{M}}^2
        =
        \partial_{\bar{\chi}}\bar{U}_k(\bar\kappa)+ 2\bar{\kappa} \partial^2_{\bar{\chi}}\bar{U}_k(\bar{\kappa})=2\bar{\kappa} \bar{u}_4,
    \end{equation}
    we arrive to the expression \eqref{anomalousdim}.
\end{proof}

Notice that to derive this expression we took into account the additional rescaling coming from $Z(k)$ for the renormalized vacuum, which modifies $\bar{\kappa} \to Z(k)^{-1} \bar{\kappa}$ with respect to the strict \lpa definition.
Due to the factors of $Z(k)$ in the definition of the barred quantities, $\eta(k)$ appears in the flow equations.
It corresponds in particular to a redefinition of the canonical dimensions
\begin{equation}
    \dim_{\tau}(u_{2n}) \to \dim_{\tau}(u_{2n})-n t^{\prime} \eta(k)
\end{equation}
with respect to the equations obtained previously for $Z=1$.
$\eta$ thus enters the effective loops because the regulator $R_k(p^2) = Z(k)(k^2-p^2)\theta(k^2-p^2)$ depends on $Z(k)$, so $\dot R_k$ contains a term proportional to $\eta(k)$ that feeds back into the loop integrals.

\begin{remark}{Regulator Bound}{remarkregulator}
    Although the computation is thereby greatly simplified, the factor $Z$ in front of the regulator \eqref{Litim} must not affect the boundary conditions at $k \to 0$ and $k \to \infty$, namely:
    \begin{equation}
        \Gamma_{k \to \infty} \to H, \qquad \Gamma_{k \to 0} \to \Gamma.
    \end{equation}
    In particular, the first of these conditions requires that $R_{k \gg 1} \sim k^r$ for some positive exponent $r$.
    This requirement is automatically satisfied for $Z=1$, where $R_{k \gg 1} \sim k^2$.
    However, the scale dependence of $Z(k)$ can potentially violate this condition.
    This situation may occur when the flow reaches a fixed point with a non-zero anomalous dimension, $\eta_* \neq 0$.
    In this regime, the \wfr scales as $Z(k) \sim k^{\eta_*}$, implying that $R_{k \gg 1} \sim k^{2+\eta_*}$.
    Hence, the requirement $r > 0$ imposes the following constraint:
    \begin{equation}
        \eta_* > -2,
    \end{equation}
    which we refer to as the \emph{regulator bound}.
    This restriction is a limitation entirely connected to the choice of regulator, rather than to the method itself.
    Since the flow equation is intrinsically non-autonomous, exact fixed points are not expected, implying that this criterion requires a more refined definition.
    Generally, the \lpa, which represents the lowest order in the derivative expansion, is reliable only in regimes where $\eta$ remains sufficiently small.
    For $\vert \eta \vert \gtrsim 1$, higher-order corrections become necessary~\cite{balog2019convergence}.
\end{remark}

\subsection{Dimensional Phase Transition}\label{sectionPhaseTrans}

The fact that canonical dimensions are scale-dependent has the immediate consequence that a global fixed point cannot exist.
However, fixed directions may exist, along which the beta functions vanish.
There are also asymptotic fixed points, arising from the fact that, in the deep \ir, the behaviour of the distribution $\rho(p^2)$ follows a power law $\rho(p^2) \sim (p^2)^{\alpha}$.
In other words, in this limit, the flow behaves like an ordinary field theory in the asymptotic dimension $D_0$.
Since moments in $\mathds{R}^D$ are distributed according to $\rho_{\mathds{R}^D}(p^2) \sim (p^2)^{\frac{D-2}{2}}$, this motivates the definition of the \emph{asymptotic dimension}:
\begin{definition}{Asymptotic Dimension of a Distribution}{asymptoticdim}
    For a universal distribution $\mu(\lambda)$ behaving as a power law $\mu(\lambda)\sim(\lambda_{+}-\lambda)^{\alpha}$ in the vicinity of $\lambda_+$, we call
    \begin{equation}
        D_0 \defeq 2\alpha+2
    \end{equation}
    the asymptotic dimension of the distribution.
\end{definition}
The thickness $\delta \lambda \defeq \vert \lambda_{\text{max}}-\lambda_+\vert$ between the largest eigenvalue and $\lambda_+$ can be estimated by the observation that $\mu(\lambda_{\text{max}})\delta \lambda$ must equal $1/P$, the minimal separation from which we can distinguish two eigenvalues \cite{Bouchaud3}.
Setting $\mu(\lambda_+ - \delta\lambda) \sim (\delta\lambda)^\alpha$ and equating to the typical $1/P$ spacing of eigenvalues at the edge, since $D_0 = 2(1+\alpha)$ one obtains:
\begin{equation}
    \delta \lambda \sim P^{-1/(1+\alpha)} = P^{-\frac{2}{D_0}}.
\end{equation}
This is closely related to the concept of \emph{sparsity}, on which we shall return in \Cref{part3}.
In a nutshell, sparsity measures the number of vanishing entries of the matrix, or equivalently the number of clusters in its correlation structure.
A purely random matrix has statistically independent entries, which in graph language corresponds to a complete graph: every pair of variables is correlated, in the sense that the off-diagonal entries of the correlation matrix are non-zero with probability $1$.
The associated spectral density has $D_0 = 3$, as already mentioned in the previous sections.
In the presence of structures, the connectivity of the correlation graph is reduced, clusters are formed, and $D_0$ increases.

The sparsity property is directly linked to eigenvector localization.
For a purely random matrix belonging to the \mpdistr class, each component $u_i$ of a given normalised eigenvector carries approximately the same weight, scaling as $\sim 1/\sqrt{P}$.
The entropy of this vector is thus maximized, and the distribution of its components follows a Gaussian law, the \emph{Porter-Thomas distribution} (see \cite{porter1956fluctuations}):
\begin{equation}
    p(x) = \sqrt{\frac{P}{2\pi}} e^{-P x^2/2}.
    \label{PTdist}
\end{equation}
An increase in sparsity leads to a progressive localisation of the eigenvectors $w_i$ ($i = 1, 2, \dots, P$), with the bulk of their norm (close to $1$) concentrating on a very small number of components.\footnote{%
    This is a phenomenon akin to \emph{Anderson's localisation}, a quantum effect occurring in the physics of disordered conductors where the electronic wave function localises under the influence of randomly distributed impurities.
}
The localisation can be measured using the \emph{Inverse Participation Ratio} (\ipr):
\begin{equation}
    \operatorname{\ipr}(w)
    \defeq
    \sum_{i=1}^P \vert w_i \vert ^4,
\end{equation}
which scales as $1/P$ for delocalised states and remains of order $1$ if the vector localises on a small number of components.
We shall return to the delocalisation issue in \Cref{part3}.

In this limit, we recover the usual result:
\begin{equation}
    \dd \tau = (2\alpha+2) \dd t = D_0\, \dd t.
\end{equation}
A direct calculation further yields:
\begin{equation}
    \dim_\tau(u_4) = \frac{1-\alpha}{1+\alpha} = \frac{4-D_0}{D_0}.
\end{equation}
Since this dimension is a derivative with respect to $\tau$, the canonical dimension in ordinary $t$ time is recovered by multiplying by $D_0$, giving $4-D_0$.
We thus recover the standard canonical dimension of a local quartic coupling in $\mathds{R}^{D_0}$.

\begin{figure}[t]
    \centering
    \includegraphics[width=0.7\textwidth]{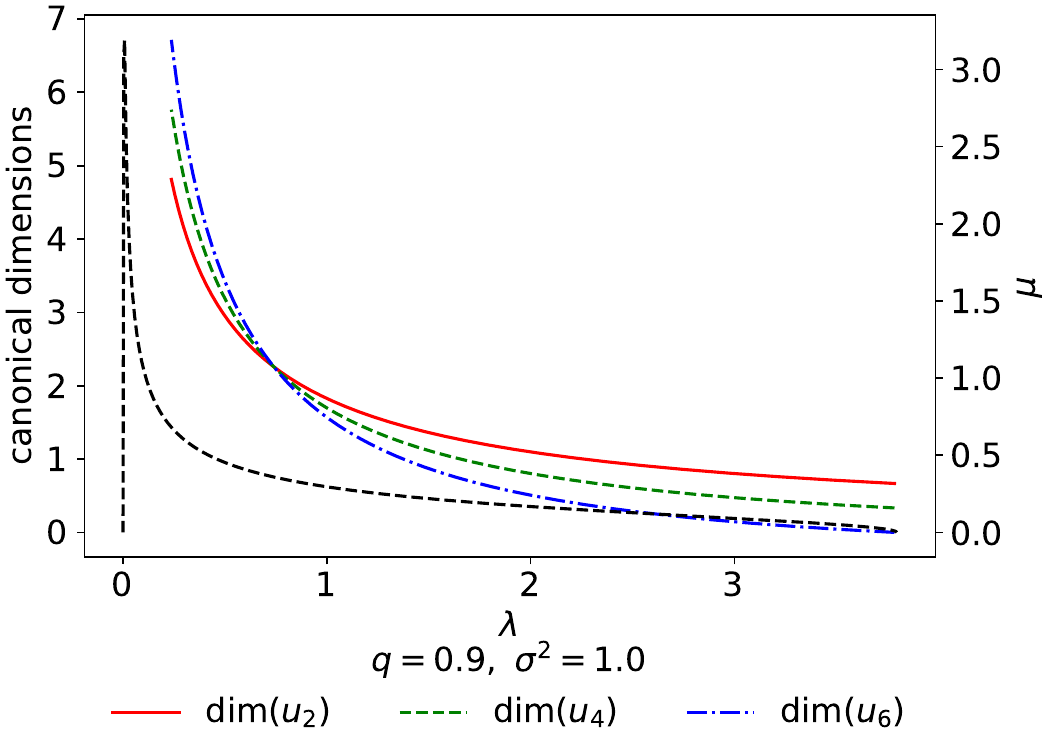}
    \caption{
        Behaviour of the canonical dimensions for the \mpdistr distribution with $\sigma^2 = 1$ and $q=0.9$ (black dashed curve, see \Cref{thm:thMP}).
    }
    \label{MP_fig}
\end{figure}

For the analytical \mpdistr law, we have
\begin{equation}
    \alpha=1/2, \qquad
    \delta \lambda\sim P^{-2/3}, \qquad
    D_0=3.
\end{equation}
Asymptotically, the flow behaves like that of a field theory in three dimensions (note that $D_0=3$ is actually a general feature for convex potentials except at the critical points, see~\cite{Bouchaud3}).
This implies, in particular, that quartic interactions are relevant and sextic interactions are marginal.
Higher-degree interactions are irrelevant in the \ir.
The renormalizable sector is therefore three-dimensional for pure noise.
This is quite counter-intuitive, as one might have expected pure noise to correspond to a Gaussian theory, but this is not the case.
The Gaussian fixed point is unstable because the quartic coupling is relevant in $D_0=3$, so the flow runs away from the Gaussian region, similar to what is encountered in $\Phi^4$ field theory in dimension $D_0 \leq 4$ (see \Cref{part1}).
The overall behaviour of the canonical dimensions for \mpdistr is presented in \Cref{MP_fig}.
We observe that starting from a relatively deep scale in the UV, only quartic and sextic interactions survive within the renormalizable sector (positive dimensions).
Note that in the deep \uv, all interactions are relevant, and their canonical dimensions diverge.
Asymptotically, the theory involves only a small number of relevant parameters and can be studied using standard \frg methods.
For pure noise, non-perturbative methods are required, analogous to those used to describe $\Phi^4$ field theory in $D_0=3$.
Conversely, in the deep \uv, more specific non-perturbative methods must be developed, which are beyond the scope of this review.
Here, we will consistently focus on the \ir, where the signal is assumed to be located.
This \ir zone is roughly bounded by the point where $\dim_\tau(u_8)=0$.
This marks the limit of our analysis in the \uv.

\begin{figure}[t]
    \centering
    \includegraphics[width=0.7\textwidth]{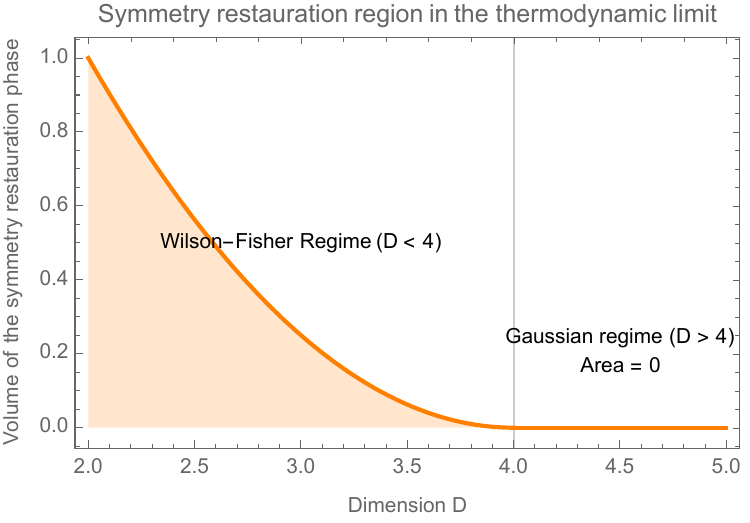}
    \caption{Symmetric phase restoration as a function of the dimension $D_0$ in the thermodynamic limit.}\label{fig_size_RS}
\end{figure}

\begin{figure}[t]
    \centering
    \includegraphics[width=0.95\textwidth]{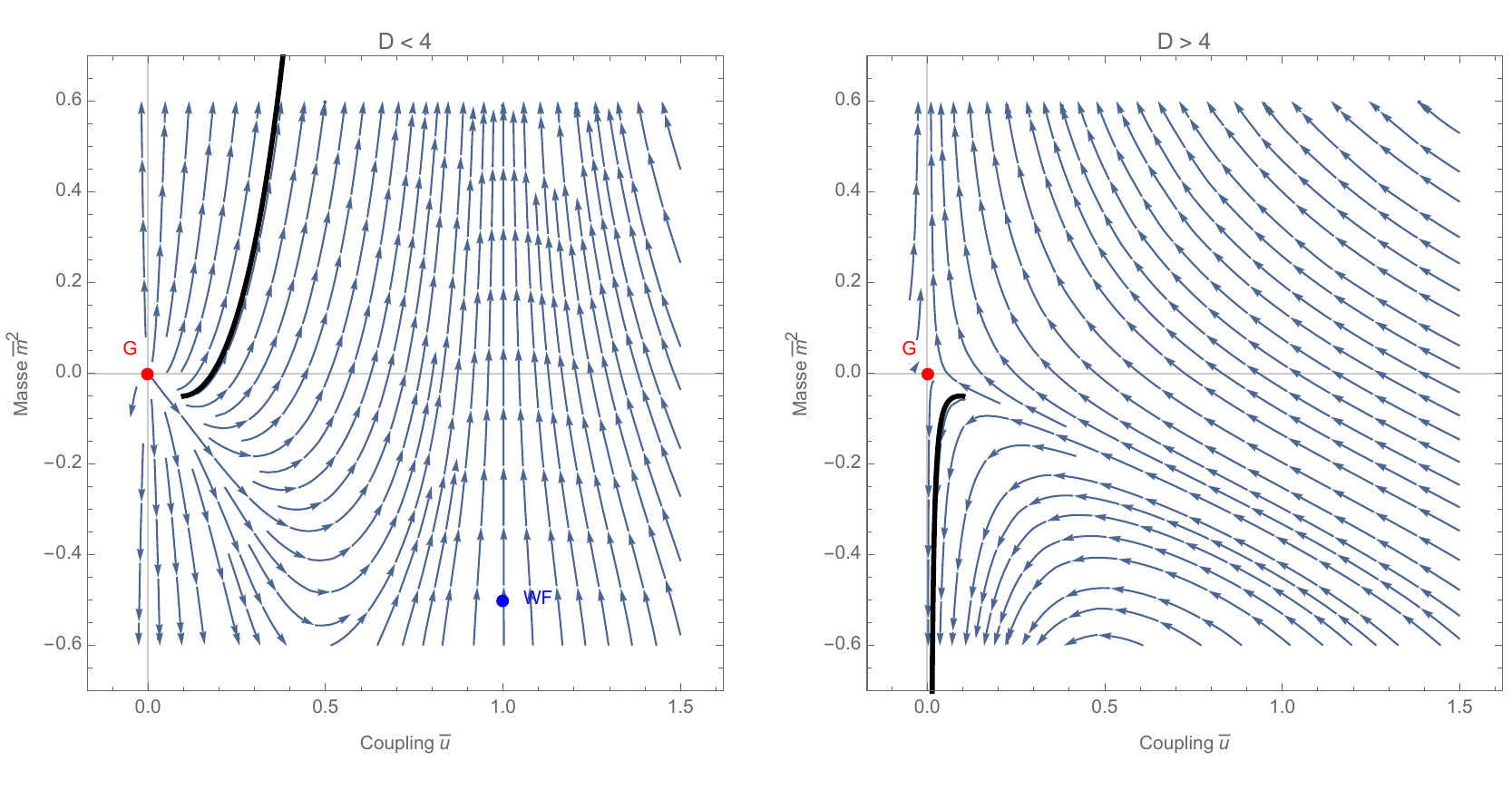}
    \caption{Flow behaviour and symmetry restoration for $D_0<4$ (left) and $D_0>4$ (right). The same trajectory, with initial conditions $(-0.05, 0.1)$, ends up in the symmetric phase in the first case but not in the second.}\label{size_effect}
\end{figure}

As explained in \Cref{sec_invitation} and as will be further discussed in \Cref{part3}, the presence of a signal makes interactions \enquote{less relevant}, notably by lowering the values of the asymptotic canonical dimensions of the sextic and quartic couplings.
This has a consequence on the behaviour of certain \rg flow trajectories by affecting a very specific region of the phase space: the symmetry restoration region.
These are trajectories which start with a negative mass in the \uv and end up in the positive mass region at sufficiently large scales due to interactions.
Thus, on the left side of \Cref{WilsonFisherFig} in \Cref{part1}, some trajectories starting from the $m^2<0$ region (broken symmetry region) end up in the $m^2>0$ zone (symmetric restoration region).
However, the existence of this region depends heavily on the dimension of the space.
Flow equations can be analytically continued in $D_0$, allowing us to follow the behaviour of the flow from $D_0<4$ to $D_0>4$:
\begin{itemize}
    \item In $D_0<4$, there exists a non-trivial \wf fixed point (visible in \Cref{WilsonFisherFig}), which controls the ferromagnetic transition.
          The red line in the figure thus divides the flow into two distinct regimes: one where trajectories end up in the positive mass region (symmetric region), and another where they end up in the negative mass region (broken symmetry region).\footnote{%
              If you do not see to what the symmetric and broken regimes correspond, please read \Cref{part1}.
          }
    \item As $D_0 \to 4$, the upper critical dimension, the \wf fixed point converges toward the Gaussian fixed point, such that in $D_0>4$, the two fixed points have merged, and the Gaussian fixed point now controls the transition.
          In this region, the quartic coupling is irrelevant and thus converges rapidly toward zero at large scales, so that the size of the critical region suddenly tends toward zero, as shown in \Cref{fig_size_RS}.
\end{itemize}
The reason is that even though the dimensionless couplings show the existence of a critical direction in the negative mass region (as shown on the right side of \Cref{WilsonFisherFig}), the effect of the quartic coupling vanishes rapidly, and the critical line shifts along the zero-mass axis in the thermodynamic limit (at large scales).
Thus, trajectories that end up in the symmetric phase when $D_0<4$ end up in the non-symmetric phase when $D_0>4$ (\Cref{size_effect}).
In this specific case, the extent of the symmetry restoration region is fixed by the Gaussian fixed point.
If $m^2_{c}(u)$ denotes the value of the mass along the critical line, the area of the critical region is:
\begin{equation}
    \operatorname{Area}(D_0) = \int_0^{u_*} \dd u \, m^2_{c}(u),
\end{equation}
where $u_*$ denotes the value of the quartic coupling at the \wf fixed point.
Increasing the \snr pushes the system from the noise universality class ($D_0=3$) toward higher effective dimension: the flow, parameterised by the canonical dimension of the couplings, crosses the upper critical dimension $D_0=4$ at which the quartic coupling becomes irrelevant.
The signal, in this language, is the deviation of the canonical dimension from its \mpdistr value, evaluated at the location where $\dim_\tau(u_4)=0$.
Gradually, trajectories that were located in the symmetry restoration phase exit it, inducing a \emph{dimensional phase transition}.

\begin{remark}{Dimensional Phase Transition and Quantum Gravity}{remarkdimtrans}
    The observations of this section can be rephrased in the language of information theory.
    The transition from a Gaussian matrix model to a model featuring an interacting (non-Gaussian) potential is, in this reading, akin to shifting from \enquote{white noise} to \enquote{structured} (or \emph{coloured}) noise.
    The Gaussian action, which maximises entropy under a variance constraint, describes a system of fundamentally uncorrelated matrix elements whose macroscopic state is constrained by the rigidity of the semicircle law.
    Introducing a non-linear matrix coupling into the joint probability distribution (such as $P\, \gamma \mathrm{Tr}(M^4)/4$) breaks this local independence and generates topological correlations between elements.
    This restriction of the accessible phase space corresponds to a decrease in Shannon entropy and, consequently, a gain in structural information.
    A concrete realisation is provided by the quartic matrix model with coupling $\gamma$: at the critical point $\gamma = -1/12$, the system undergoes a phase transition affecting the universality of the spectral edge, and the edge exponent changes discontinuously from $\alpha = 1/2$ for $\gamma > -1/12$ to $\alpha = 3/2$ at $\gamma = -1/12$.
    The value $\alpha = 3/2$ yields an asymptotic dimension of $D_0 = 5$, which, according to the criteria of \Cref{sectionPhaseTrans}, corresponds to the emergence of a signal.
    In this regime, the matrix noise ceases to be a mere statistical fluctuation.

    In some sense analogously, the Feynman diagrams of such interacting random matrix models generate random surfaces, as established in the 1990s by Di Francesco, Kazakov, Parisi, and Zuber~\cite{di19952d}.
    At the critical point $\gamma = -1/12$, these surfaces converge toward a continuous, non-smooth geometry identified with 2D Einstein quantum gravity (or, more precisely, Liouville quantum gravity).
    According to this analogy, the underlying matrix structure manifests as a (Brownian) geometry, and the dimensional analysis presented in this section is consistent with such a geometric interpretation (see the next part).
\end{remark}

\subsection{Learnable Region}

Field theory and the standard approximations for the \eaa $\Gamma_k$, and specifically the vertex expansion or the \lpa, are valid at sufficiently low-energy (\ir) scales, where the tail of the spectrum resides.
\Cref{MP_fig} illustrates the cause of this limitation: in the deep \uv, all couplings become relevant and, worse, diverge.
This corresponds to a \emph{dimensional crisis}, meaning, in \rg language, that an arbitrarily large number of couplings become relevant at high-energy.
In such a regime the flow is no longer predictive since, in other words, an arbitrarily large number of initial conditions for the couplings must be fixed, and the truncations become arbitrarily large.

The theory is then assumed to be valid until some \uv scale, the cutoff $\Lambda_{\text{\uv}}$ introduced in \Cref{def:UV/IR}.
The power counting allows us to fix this limit less arbitrarily than heuristically, and strictly from \rg arguments.
The frontier between the \enquote{good} \ir regime and the \enquote{bad} \uv regime can be conservatively estimated by the scale $\Lambda_0$ where the canonical dimension of $u_8$ vanishes:
\begin{equation}
    \left[ t^{\prime} - \frac{3}{4} \dim_\tau(u_4) \right]_{t=\ln(\Lambda_0)}=0.
    \label{eqcritique}
\end{equation}
The field theory was not designed to be more than an effective model valid at large scales, yet the canonical dimension analysis determines its own domain of validity: the truncation at $\bar{P}=3$ (i.e.\ setting $u_8=0$) is consistent only as long as $u_8$ remains irrelevant, and breaks down at the scale $\Lambda_0$.
Numerically, this limit corresponds for the \mpdistr distribution to the eigenvalue domain $\lambda \sim \lambda_+/3 \equiv \Lambda_{\text{\uv}}$.

The flow thus exhibits two regions:
\begin{itemize}
    \item The \emph{learnable region}, for $k<\Lambda_{\text{\uv}}$.
          In this region, only $u_4$ and $u_6$ are relevant and the theory is effectively predictive: the presence of signal can be learnt.
    \item The \emph{deep noisy region}, for $k \gg \Lambda_{\text{\uv}}$, in which the number of relevant couplings becomes arbitrarily large and the values taken by the dimensions diverge: inference is no longer possible.
\end{itemize}

\part{Generalised Scale Analysis: The Proof of Concept}\label{part3}

This section develops the principles of \gsa as introduced in \Cref{sec_invitation} and numerically demonstrates its validity.
For a family of datasets generated by perturbing a real-world image with additive Gaussian noise, we investigate how the canonical dimensions depend on the \snr.
This construction provides precise control over the \snr while allowing us to isolate the effects of finite-size fluctuations associated with the dimensionality of the data.
We show that a sufficiently strong signal near the \mpdistr class induces a distinct shift in the universality class of the asymptotic \ir flow, accompanied by a breakdown of Porter--Thomas eigenvector statistics.
We demonstrate that this phenomenon can be analysed as a phase transition, in which the signal reduces the phase-space volume within which the symmetry of the \mpdistr class is restored.
Recall from \Cref{part2} that the canonical dimensions quantify the scale dependence of the effective couplings: the dimensional phase transition occurs when all couplings become irrelevant, signalling Gaussian behaviour dominated by the mass term.

These conclusions are supported by robust empirical evidence from an independent benchmark: mapping the signal detection problem onto the estimation of the critical temperature in a non-equilibrium magnetic system.
The critical temperature estimated via \gsa for the 2D Ising model agrees closely with the exact Onsager solution (to within 2\%).
It also matches predictions from standard methods such as Binder cumulant analysis.
\gsa thus outperforms standard signal analysis methods, notably those based on \kl divergence minimisation.

\section{Generalised scale analysis}\label{sec:gsa}
This section introduces the key concepts of \gsa and the signal detection thresholds.
We also describe in detail the methodology used in our experiments.
The Python libraries used for the numerical experiments are freely available at \href{https://github.com/thesfinox/frg-signal-detection}{frg-signal-detection} and \href{https://github.com/ParhamRadpay/Model-A}{Model-A}.

\subsection{Marchenko--Pastur Universality Class and Signal Detection}

The proposed signal-detection method focuses, by construction, on the vicinity of a random matrix universality class, specifically the \mpdistr universality class.
In the previous sections, we showed how the empirical spectral distribution induces a natural power counting for the effective field describing the large-scale correlations of the data, following an \enquote{energy} scale fixed by the spectrum of the \ecm.
Within the \mpdistr universality class, as well as most random matrix classes whose edge behaviour follows the same power law, only a finite number of couplings (the quartic and sextic interactions) survive in the renormalisable sector.
\gsa relies on the observation that the presence of a signal alters the canonical dimensions, thereby modifying the behaviour of the large-scale \rg flow.
In particular, for a sufficiently strong signal (above the critical dimension $D = 4$), all couplings become irrelevant and the flow is then dominated by the mass term.
This behaviour might appear counter-intuitive, as noise is associated with non-Gaussian dynamics while signal pushes the system towards Gaussianity.
In what follows, we show how this dimensional interpretation of universality establishes a boundary between signal and noise within a quasi-continuous spectrum.

\subsection{Empirical Methodology}

For numerical experiments, we rely on the standard Python scientific libraries \texttt{numpy}~\cite{Harris:2020xlr} and \texttt{scipy}~\cite{Virtanen:SciPy10Fundamental:2020}.
Our goal is to solve numerically differential equations of the form:
\begin{equation}
    \dv{f(s, x)}{s} = \mathcal{D}_{\qty{\partial_x, \partial^2_x}}\qty[f](s, x),
\end{equation}
where $\mathcal{D}$ is a differential operator of at most second order acting on the function $f$.

It can be shown that the canonical dimensions are independent of the prior \rg flow history.
They are determined entirely by the position within the spectrum $\rho\qty(p^2)$ (or $\rho_G\qty(p^2)$).
The \rg flow equations \Cref{flow1,flow2,flow3} are integrated numerically using a finite element method:
\begin{equation}
    u_{2n}\qty(k^2 - \Delta k^2)
    =
    u_{2n}\qty(k^2) - \Delta k^2\, \mathcal{R}\qty(u_{2n}\qty(k^2), u_{2(n+1)}\qty(k^2)),
\end{equation}
where $\mathcal{R}$ denotes the right-hand side of the relevant flow equation and $\Delta k^2$ is a small integration step over the spectrum.
The negative sign arises naturally from the direction of the \rg evolution, since the \dof are integrated out from the \uv scales towards the deep \ir limit ($k^2 \to 0$).

\begin{algorithm}[t]
    \caption{Construction of the samples}\label{alg:sample_constr}

    \DontPrintSemicolon
    \SetKwInOut{Input}{input}
    \SetKwInOut{Init}{init}
    \SetKwInOut{Let}{let}
    \SetKwInOut{Output}{output}

    \Input{size of the sample $N > 0$, and ratio $q \in \qty[0, 1]$}
    \Let{$P = \left\lfloor q\, N \right\rfloor$}
    \Input{$\beta \ge 0$}
    \Input{$Z \in \mathds{R}^{N \times P}$ where $Z \sim \mathcal{N}\qty(0, \sigma^2)$}
    \Input{an image $S \in {\qty[0, 255]}^{H \times W \times C}$}

    \If{$C > 1$}{$S_{ij} \gets C^{-1}\, \sum\limits_{c = 1}^C S_{ijc}$ \tcp*{convert to greyscale}}
    $S \gets \frac{S - \left\langle S \right\rangle}{\sqrt{\mathrm{Var}\qty(S)}}$ \tcp*{standardise the image}
    $S \gets \mathrm{resize}(S) \in \mathds{R}^{N \times P}$ \tcp*{interpolate to sample size}
    \Let{$X = \beta\, S + Z \in \mathds{R}^{N \times P}$}
    $\qty(\Sigma, W) = \mathrm{SVD}(X)$ \tcp*{compute singular values and right eigenvectors}
    $\mathrm{E} \gets \mathrm{flatten}\qty(\Sigma^2 / (N - 1))$ \tcp*{convert to covariance eigenvalues}
    $\mathrm{E}^{\prime} = \text{remove\_spikes}\qty(\mathrm{E})$ \tcp*{PCA --- isolated spikes removal}
    $\tilde{\mu}_G \gets \text{histogram}\qty(\mathrm{E}^{\prime})$

    $\mu_G \gets \operatorname{\text{\kde}}\qty(\tilde{\mu}_G)$ \tcp*{interpolation by kernel density}

    \Output{$\rho_G \gets \mu_G\qty(\frac{1}{k^2 + m^2} + \lambda_-) / {\qty(k^2 + m^2)}^2$ \tcp*{momenta distribution}}
    \Output{$W^{\mathsf{T}}$ \tcp*{covariance eigenvectors}}
\end{algorithm}

We build the samples for the analysis using a simple additive model for normally distributed noise with a \snr $\beta \ge 0$:
\begin{equation}
    X = \beta\, S + Z,
    \label{eq:additive_model}
\end{equation}
where $Z \sim \mathcal{N}\qty(0, \sigma^2)$ is an $N \times P$ matrix with \iid normally distributed entries and $S \in \mathds{R}^{N \times P}$ is the centred signal matrix.
Unless specified otherwise, we set $\sigma^2 = 1$ in all our numerical experiments, and
\begin{equation}
    N = \num{2.0e4},
    \quad
    P = \num{1.8e4},
    \quad
    \text{s.t.}
    \quad
    q = \frac{P}{N} = 0.9.
\end{equation}
\begin{remark}{Scale of the Additive Noise Model}{scaleadditivenoise}
    As already pointed out in the introduction, it is worth emphasising that the definition of a \enquote{signal} is fundamentally a matter of scale.
    Prior to any artificial corruption, the data intrinsically contain a certain level of noise.
    Consequently, \eqref{eq:additive_model} is strictly valid only in a regime where the added Gaussian noise dominates this intrinsic noise.
    We assume this throughout this section and return to this point in \Cref{Sec4-1}.
\end{remark}

\begin{figure}[t]
    \centering
    \includegraphics[width=0.15\textwidth]{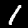}
    \hfill
    \includegraphics[width=0.15\textwidth]{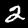}
    \hfill
    \includegraphics[width=0.15\textwidth]{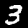}
    \hfill
    \includegraphics[width=0.15\textwidth]{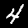}
    \hfill
    \includegraphics[width=0.15\textwidth]{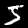}
    \caption{%
        Samples extracted from the \textsc{mnist} dataset~\cite{LeCun:2010:MNIST} and used for numerical evaluations.
    }\label{fig:exp_images}
\end{figure}

\begin{figure}[t]
    \centering
    \includegraphics[width=0.45\textwidth]{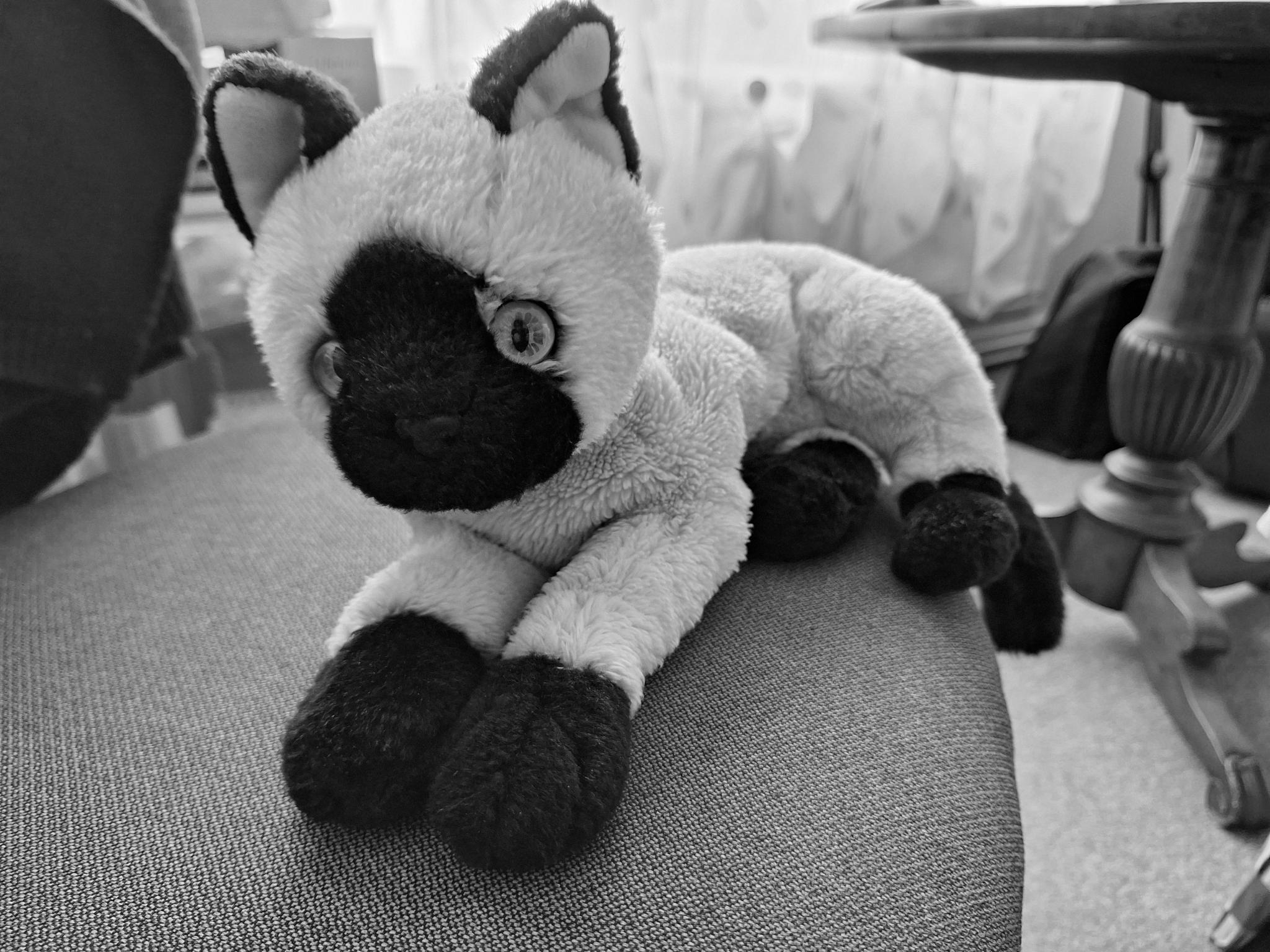}
    \caption{%
        Realistic scenario considered in the analysis: a photo of a (plush) cat with a non-trivial background.
        For simplicity, we consider a monochrome version.
    }\label{fig:gianduja}
\end{figure}

The numerical integration of the \rg equations and the evaluation of the canonical dimensions rely on the empirical momentum distribution generated by \Cref{alg:sample_constr}.
Since our analysis operates on finite-size data (such as images), the intrinsic scale of the empirical distributions requires careful consideration.
These spectra consist of a finite, ordered set of eigenvalues, from which normalised histograms are constructed to approximate the continuous distributions $\rho_G(p^2)$.
The discrete binning used to estimate the density of empirical momenta introduces a discrete \enquote{energy step} into the \rg flow, corresponding to a physical scale difference:
\begin{equation}
    \Delta_{\text{phys}} = P^{-\alpha},
    \label{eq:time_step}
\end{equation}
where the parameter $\alpha \in [0.5, 1)$ can be determined by analysing the distance between isolated spikes and the bulk momentum distribution.
In our numerical experiments, we set $\alpha = 0.5$, matching the theoretical arguments presented in the previous sections.
Below this resolution, the eigenvalues no longer form a dense continuum and start to emerge as isolated states.
The histogram is therefore constructed with a bin width of $\Delta_{\text{phys}}$ and subsequently smoothed via kernel density estimation (\kde) to yield a continuous, analytically tractable curve.
Finally, the momentum density function is derived by computing the probability distribution of the inverse variable and shifting it to the origin:
\begin{equation}
    \rho_G\qty(k^2)
    =
    \frac{1}{{\qty(k^2 + m_{\text{eff}}^2)}^2}\,
    \mu_G\qty(\frac{1}{k^2 + m_{\text{eff}}^2} + \lambda_-),
    \label{eq:empirical_momentum_density}
\end{equation}
where $m_{\text{eff}}^2$ denotes the inverse of the largest eigenvalue, as established in the previous sections.
For our numerical experiments, we use images from the standard \textsc{mnist} dataset~\cite{LeCun:2010:MNIST}, chosen for their inherent simplicity and well-defined structure (see \Cref{fig:exp_images}), together with a real-world image (depicted in \Cref{fig:gianduja}).
The corresponding empirical distributions for various \snr are presented in \Cref{fig:distgianduja}.

\begin{figure}[t]
    \centering
    \includegraphics[width=0.45\textwidth]{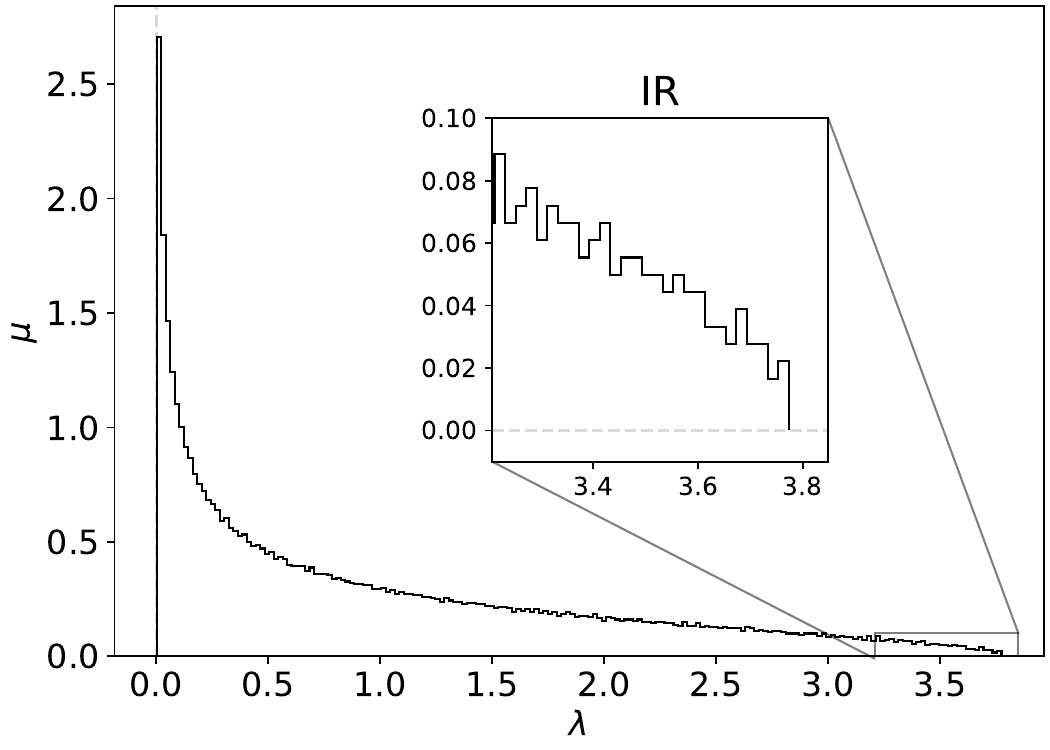}
    \hfill
    \includegraphics[width=0.45\textwidth]{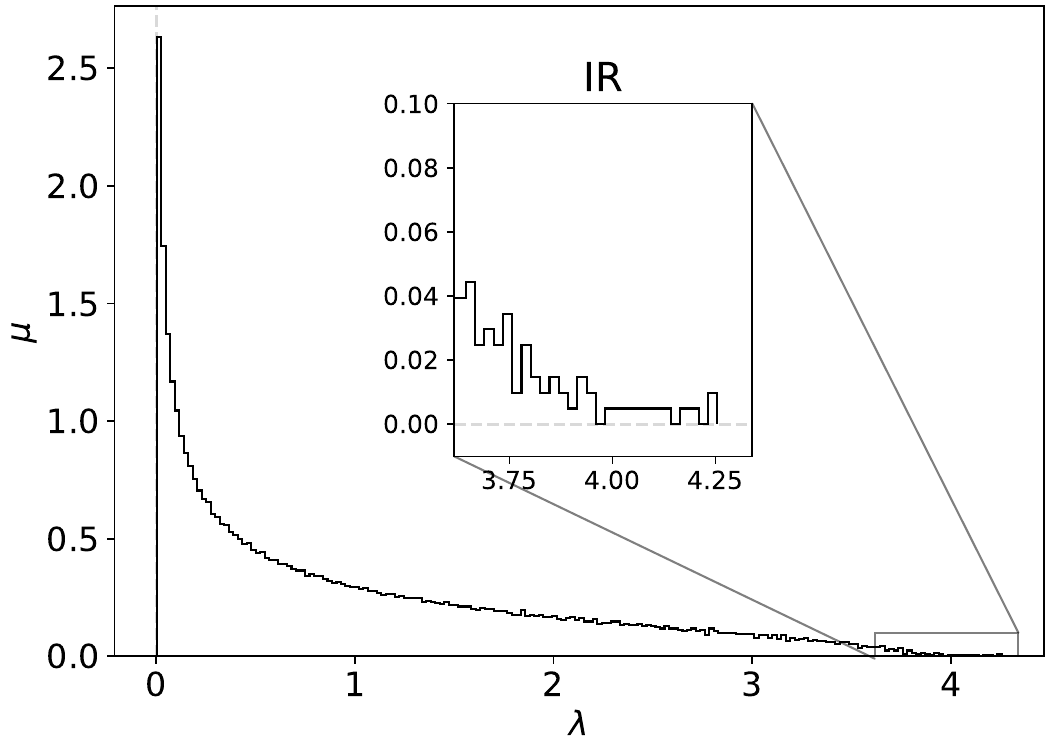}
    \caption{%
        Empirical distribution corresponding to \Cref{fig:gianduja} for $\beta=0$ (no signal, on the left), and for $\beta=0.4$ (on the right).
    }\label{fig:distgianduja}
\end{figure}

Due to the eigenvalue gap \eqref{eq:time_step}, it becomes numerically infeasible to track the evolution of the flow equations or to evaluate the canonical dimensions in the deep \ir limit ($k^2 \to 0$).
Computations must therefore stop at a finite infrared cutoff $k^2_{\text{IR}}$.
This scale is chosen to be as close as practicable to the minimum attainable resolution $\Delta_{\text{phys}}$ while remaining safely shielded from numerical instabilities.
In practice, this stopping scale is set to the midpoint:
\begin{equation}
    k^2_{\text{IR}}
    =
    \frac{k^2_{*} - k^2_{0.5}}{2},
\end{equation}
where $\rho_G\qty(k^2_{0.5}) = 0.5$ and $k^2_{*} = \underset{k^2}{\text{argmax}}\, \rho_G\qty(k^2)$ serve as a pragmatic definition of the near-IR region within the bulk momentum distribution.
A more rigorous justification stems from established results concerning eigenvalue densities~\cite{Zinn2}.
Indeed, since the eigenvalues are confined to a compact support, any structural variation in the distribution propagates across the entire spectrum, coupling the \uv and \ir regimes.
Consequently, for signal detection in nearly continuous spectra, the exact choice of the evaluation scale is largely arbitrary.
Because of the natural density of the eigenvalues, the \enquote{signal signature} in the \ir affects the entire distribution, although it is most readily measured (and best supported by theory) in the \ir region, where the spectral deformations are most pronounced.

In practice, this implies that only relative detection against a chosen \enquote{background} is feasible.
While the flow equations and canonical dimensions can be computed deterministically for the asymptotic \mpdistr distribution, finite-size effects of the empirical data introduce statistical fluctuations.
Nevertheless, as we discussed in the context of universality classes (\Cref{sec:universality_class}), this limitation reveals a powerful feature of the framework.
By evaluating the couplings and canonical dimensions at any energy scale for a \emph{blank} reference (e.g.\ a signal-less sample in chemometrics), one can establish a robust baseline.
The presence of a relevant signal or any structural deformation in the eigenvalue spectrum can then be quantified via the \rg framework as a \emph{distance} from this pure noise background.
The \frg thus acts as a highly sensitive tool for the relative detection of signals against a chosen background distribution.

\begin{remark}{Signal Detection in Extensive-Rank Spectra}{extensiverank}
    Since our primary focus is the behaviour of nearly continuous spectra, we must carefully address the presence of isolated spikes in the momentum distribution.
    Our framework is specifically designed to uncover extensive-rank signals that remain hidden within the (nearly-continuous) bulk spectrum, even after standard \pca has extracted the isolated outliers.
    We consider these spikes to represent the readily accessible, low-rank component of the signal, whereas our approach targets the non-trivial component deeply embedded within the noise.
    Indeed, increasing the signal-to-noise ratio $\beta$ inevitably delocalises the specific eigenvectors that carry the majority of the signal, a phenomenon that we examine in detail in the later sections of this article.
\end{remark}

To isolate the continuous distribution of eigenvalues, we systematically prune the momentum spectrum by discarding these spikes.
This is achieved by scanning the eigenvalues from the \ir to the \uv regime and computing the spacing between adjacent eigenvalues.
This procedure allows us to pinpoint the exact index corresponding to the spectral edge of the bulk:\footnote{%
    As emphasised in the previous sections, our computations rely on dimensionless parameters, such as $\Delta_{\text{phys}}$.
    Throughout these sections, we use this parameter interchangeably to quantify distances between both eigenvalues and momenta, since its numerical value, defined in \Cref{eq:time_step}, remains identical in both contexts.
}
\begin{equation}
    \mu_{\text{bulk}}
    =
    \underset{\mu \in \qty[0, P-1]}{\text{argmin}}
    \qty{\max\qty(0, \qty(\lambda_{\mu} - \lambda_{\mu+1}) - \tilde{\Delta}_{\text{phys}})},
\end{equation}
where $\tilde{\Delta}_{\text{phys}}$ may differ from $\Delta_{\text{phys}}$, depending on the specific continuum limit under consideration.
Standard results from \rmt suggest a scaling $\tilde{\Delta}_{\text{phys}} = \mathcal{O}\qty(1/P)$, since the typical eigenvalue spacing for Wigner matrices scales as $1/P$.
Under this assumption, the spectrum generally partitions into distinct continuous components, the most prominent being the bulk.
Varying the definition of the continuum limit directly impacts the identification of both the bulk and the isolated spikes.
Our numerical experiments demonstrate that setting $\tilde{\Delta}_{\text{phys}} = P^{-0.8}$ is empirically appropriate and aligns with the rationale behind \Cref{eq:time_step}.
We also explored alternative definitions, including a linear dependence on $\beta$ (introduced conservatively, since $\beta$ is typically unknown in practice), but these variations had no significant effect on our overall conclusions.

Following the discussion above, we restrict the set of eigenvalues to:
\begin{equation}
    \Lambda = \qty{\lambda_{\mu_{\text{bulk}}} \ge \lambda_{\mu_{\text{bulk}}+1} \ge \dots \ge \lambda_P},
\end{equation}
since $\mu_{\text{bulk}} \ge 0$ by construction.
While reducing the cardinality to $|\Lambda| \le P$ could in principle introduce additional finite-size effects, we find no such artefacts for $P \gg 1$ beyond those already addressed by the previous procedure.
Indeed, across the range of $\beta$ values explored in this study, we numerically observe only a few tens of isolated spikes.
Removing them has a negligible impact on the bulk distribution, which still comprises $P = \num{1.8e4}$ \dof.
Conceptually, this procedure is equivalent to analysing the residual eigenvalue spectrum after the dominant spikes have been extracted via standard \pca (see \Cref{rem:extensiverank}).

\subsection{Detection Thresholds}\label{detectionTh}

We establish a hierarchy of thresholds ($\beta_t < \beta_c < \beta_O$) that defines the \emph{limit of detection} (\lod) for extensive-rank signals.
The notion of \lod plays a central role in many practical applications, where identifying the presence of a signal is essential for calibrating instruments or interpreting the results of subsequent analyses~\cite{Armbruster2008}.
Such a detection threshold is typically defined via Gaussian statistics for univariate calibration models.
However, unlike traditional univariate definitions, this field-theoretic criterion is model-agnostic, relying solely on the universal breakdown of the noise sector rather than on specific priors regarding signal structure.
We contrast this regime with the standard \bbp transition~\cite{math3} for isolated spiked models, demonstrating that the field-theoretic approach remains sensitive within the spectral bulk, where conventional methods fail.
This provides a dual perspective on the bimodal connected phase identified in~\cite{Landau2023}.

\begin{figure}[!ht]
    \centering
    \begin{minipage}[t]{0.45\textwidth}
        \centering
        {Analytic computation (\mpdistr)} \\[0.25em]
        \includegraphics[width=\linewidth]{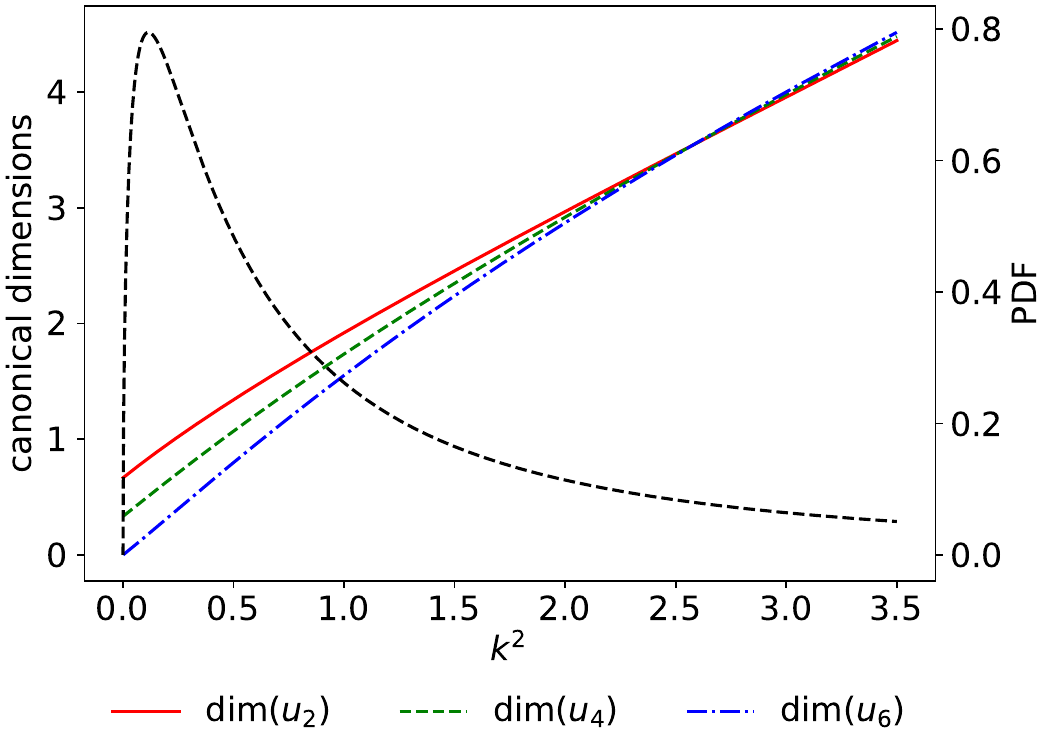}
    \end{minipage}
    \hfil
    \begin{minipage}[t]{0.45\textwidth}
        \centering
        {Numerical simulation ($\beta = 0$)} \\[0.25em]
        \includegraphics[width=\linewidth]{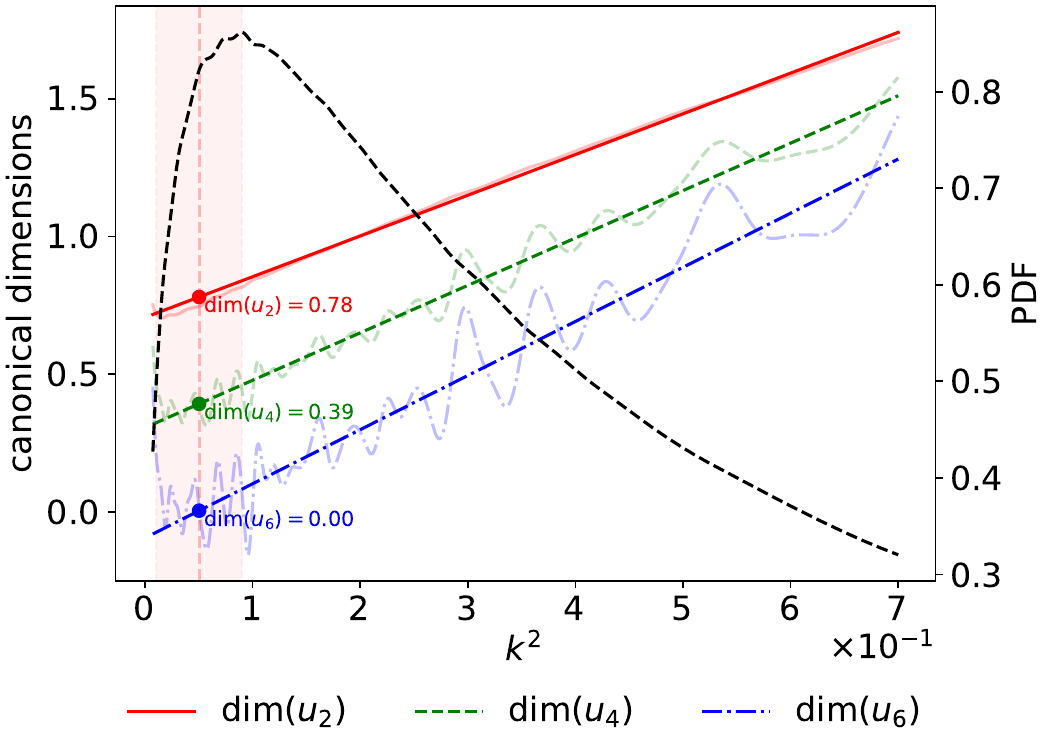}
    \end{minipage}
    \\[0.25em]

    \begin{minipage}[t]{0.45\textwidth}
        \centering
        {Numerical simulation ($\beta = 0.1$)} \\[0.25em]
        \includegraphics[width=\linewidth]{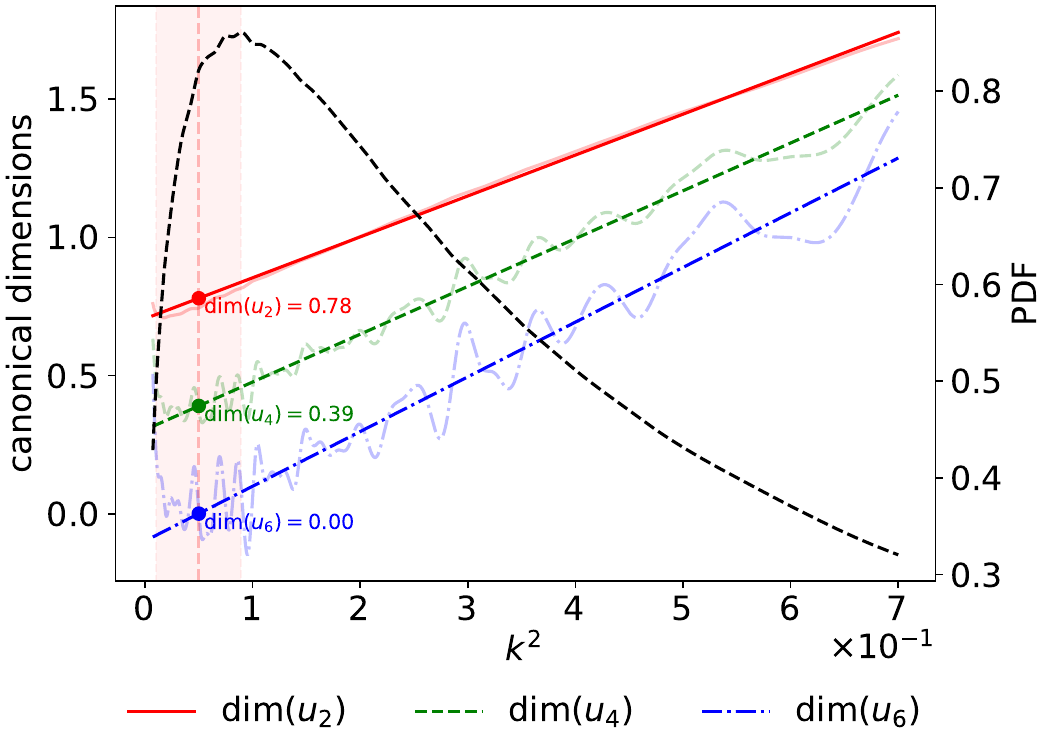}
    \end{minipage}
    \hfil
    \begin{minipage}[t]{0.45\textwidth}
        \centering
        {Numerical simulation ($\beta = 0.2$)} \\[0.25em]
        \includegraphics[width=\linewidth]{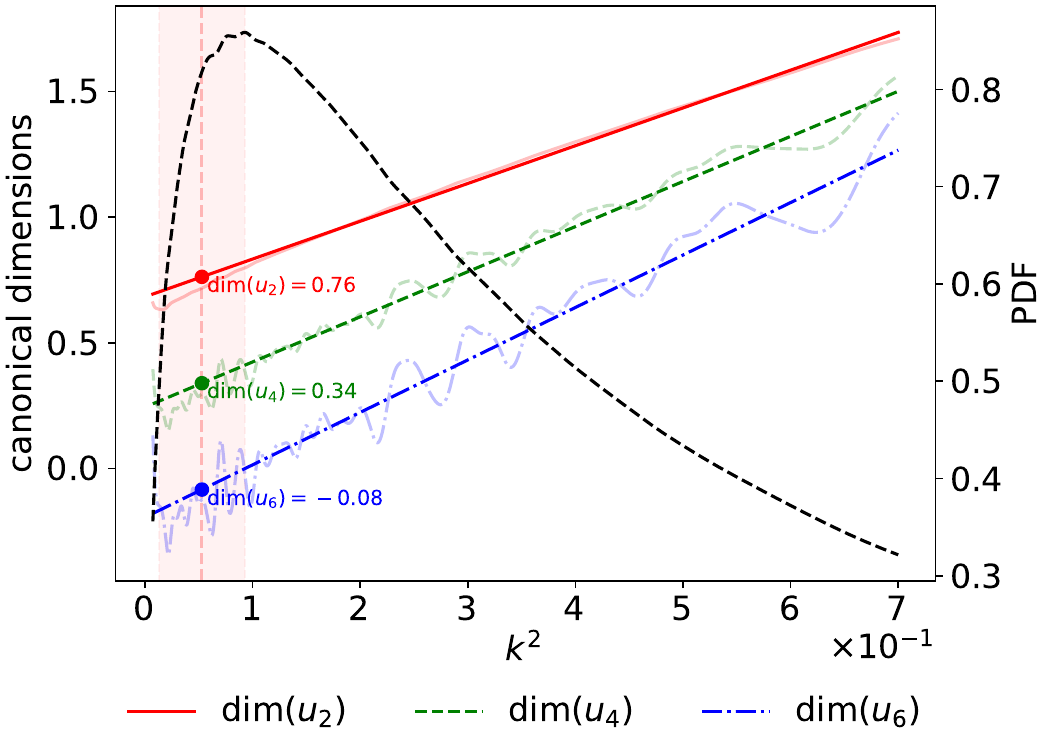}
    \end{minipage}
    \\[0.25em]

    \begin{minipage}[t]{0.45\textwidth}
        \centering
        {Numerical simulation ($\beta = 0.3$)} \\[0.25em]
        \includegraphics[width=\linewidth]{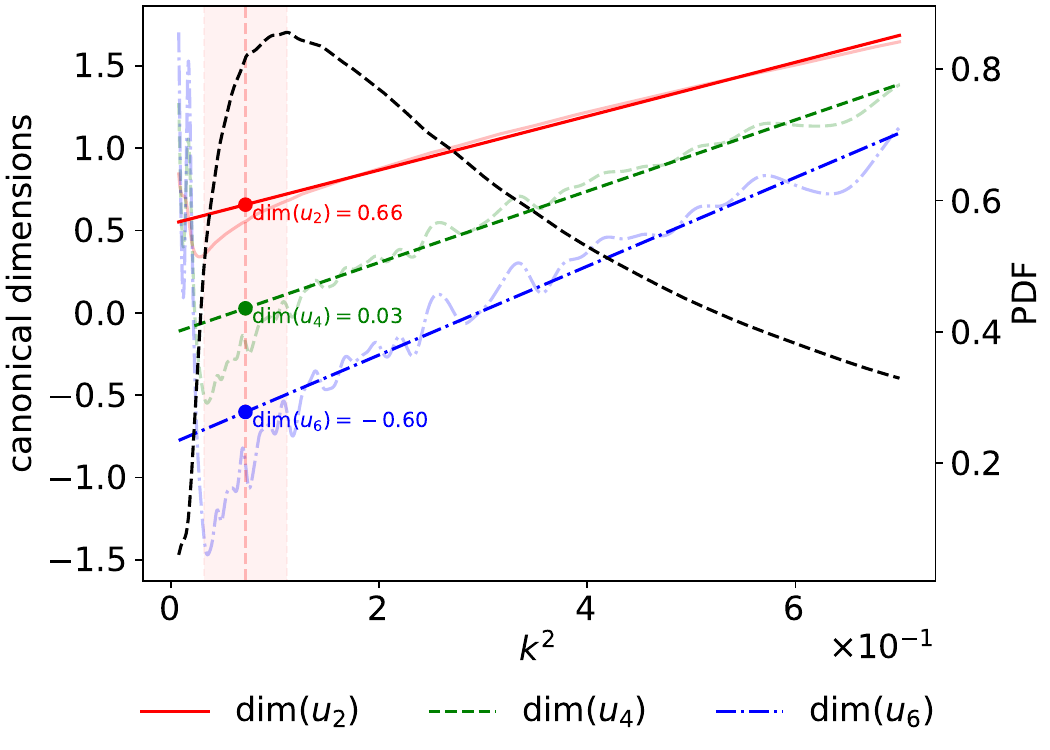}
    \end{minipage}
    \hfil
    \begin{minipage}[t]{0.45\textwidth}
        \centering
        {Numerical simulation ($\beta = 0.4$)} \\[0.25em]
        \includegraphics[width=\linewidth]{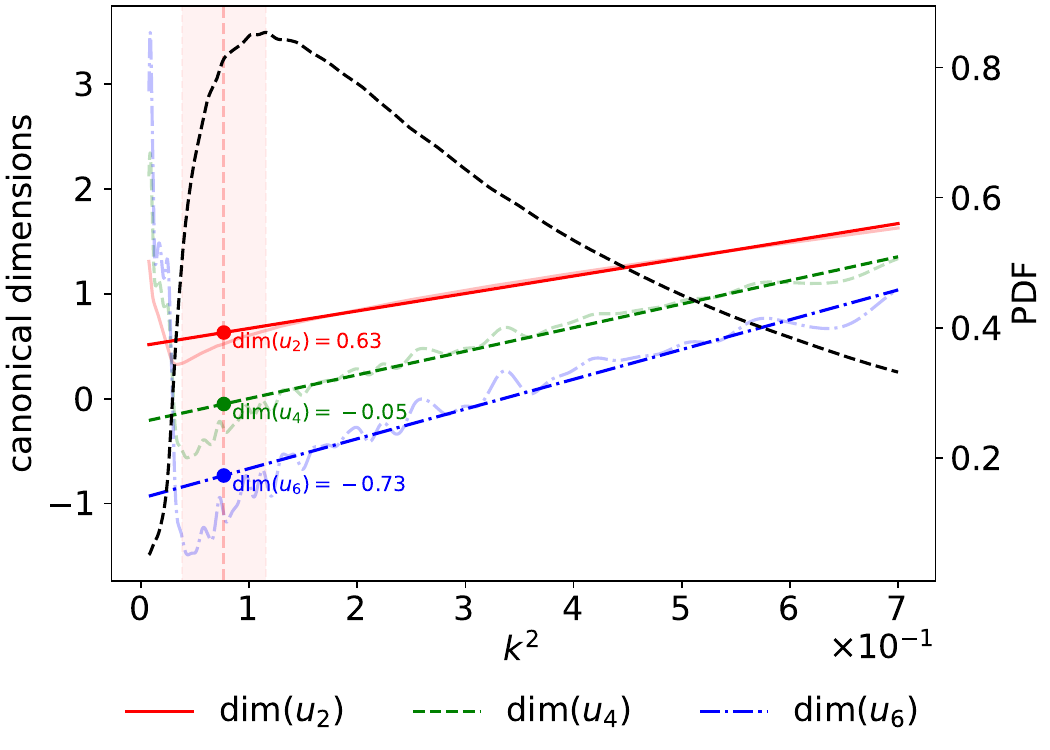}
    \end{minipage}

    \caption{%
        Behaviour of the canonical dimensions of \Cref{fig:gianduja} as the \snr $\beta$ increases.
        The top-left panel provides a comparative analysis with the \mpdistr distribution.
        Canonical dimensions are represented on the left y-axis, while the momentum distribution values are on the right.
    }\label{fig:figplotgianduja}
\end{figure}

First, we analyse the realistic sample depicted in \Cref{fig:gianduja} as a function of the \snr $\beta$.
The evolution of the canonical dimension with respect to $k^2$ is illustrated in \Cref{fig:figplotgianduja,fig:figplotgianduja2}.
Due to numerical instabilities in the spectral tail, we restrict our analysis to the \ir scale $k^2_{\text{IR}}$ defined previously, which is marked by a vertical dashed red line in each plot.

\begin{figure}[t]
    \centering
    \includegraphics[width=0.7\textwidth]{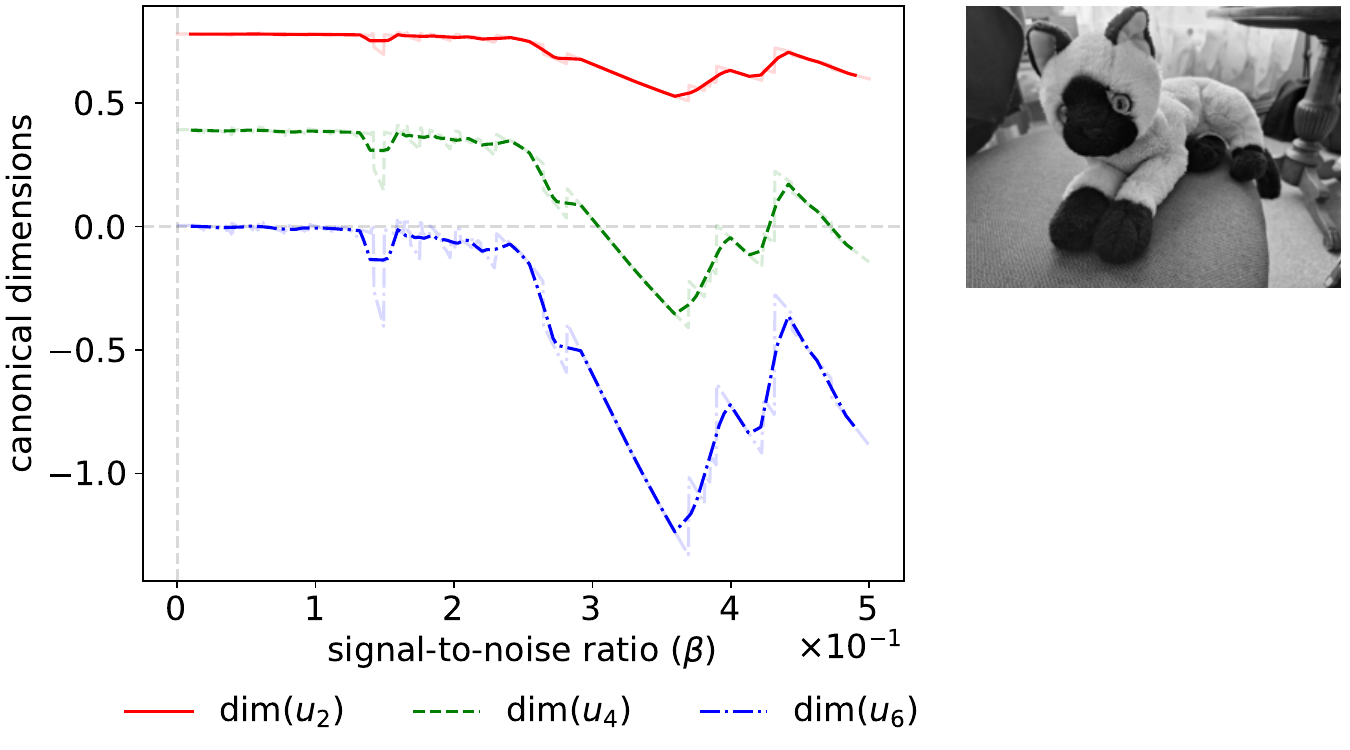}
    \caption{%
    The behaviour of the canonical dimension at the scale $k^2_{\text{IR}}$ is presented as a function of $\beta$.
    Computations were performed with a step size of $\Delta \beta = \num{5e-4}$.
    Solid lines represent a moving average with a window width of $\Delta \beta_{\text{w}} = \num{2e-2}$, while the raw experimental values are rendered with reduced opacity.
    }\label{fig:figplotgianduja2}
\end{figure}

\begin{figure}[t]
    \centering
    \includegraphics[width=0.45\textwidth]{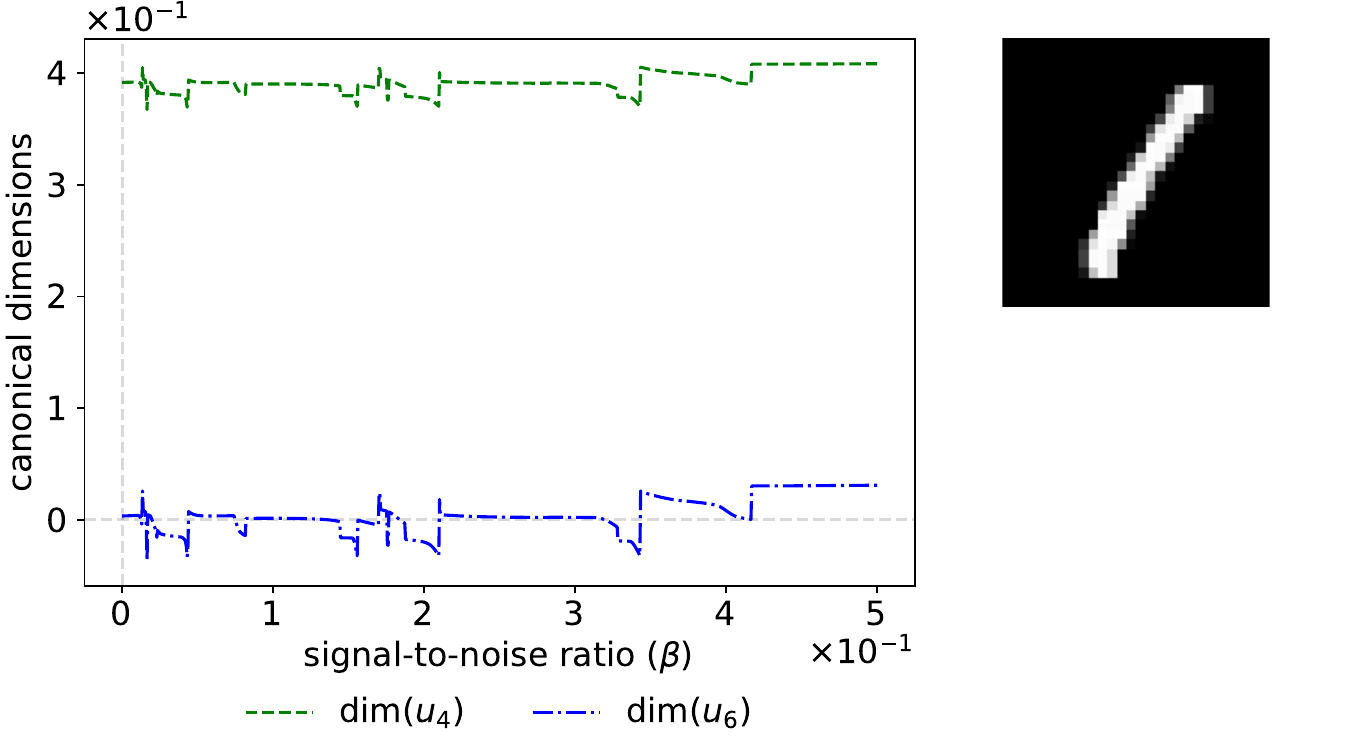}
    \hfil
    \includegraphics[width=0.45\textwidth]{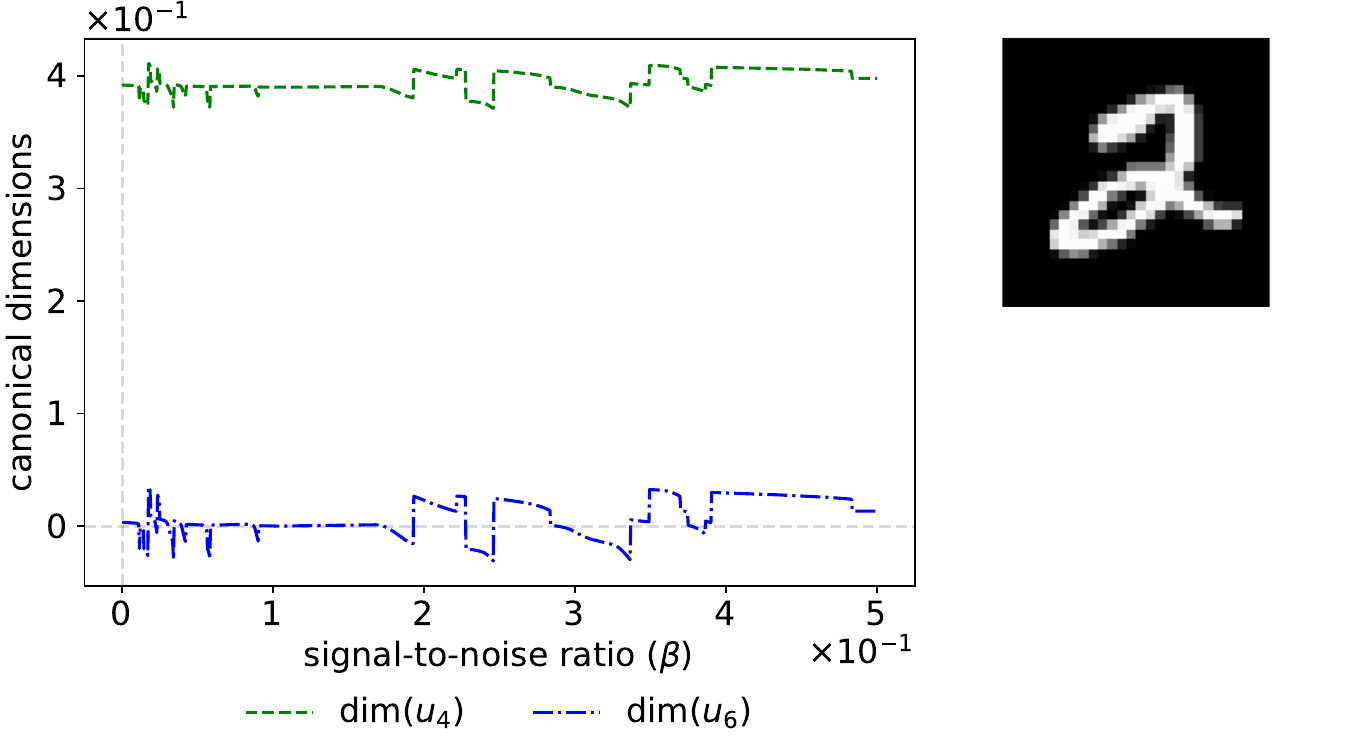}\\[0.5em]
    \includegraphics[width=0.45\textwidth]{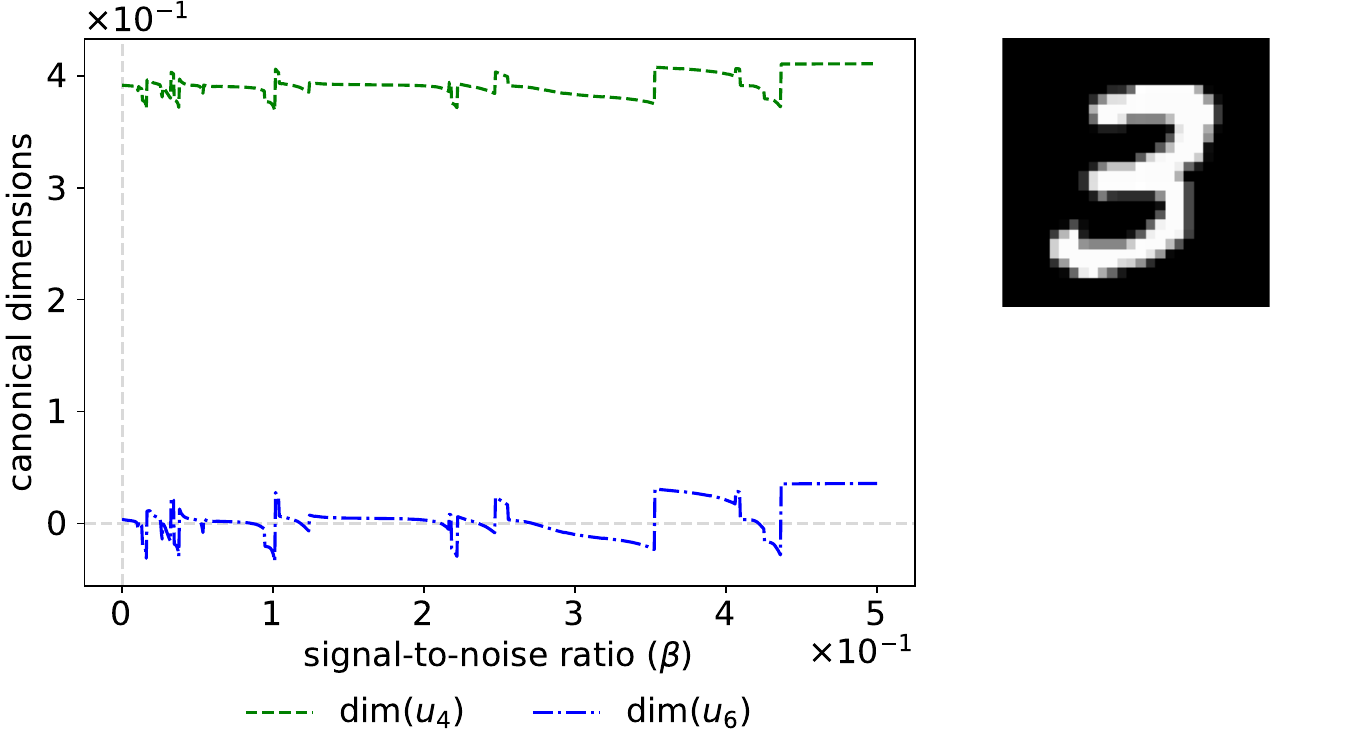}
    \hfil
    \includegraphics[width=0.45\textwidth]{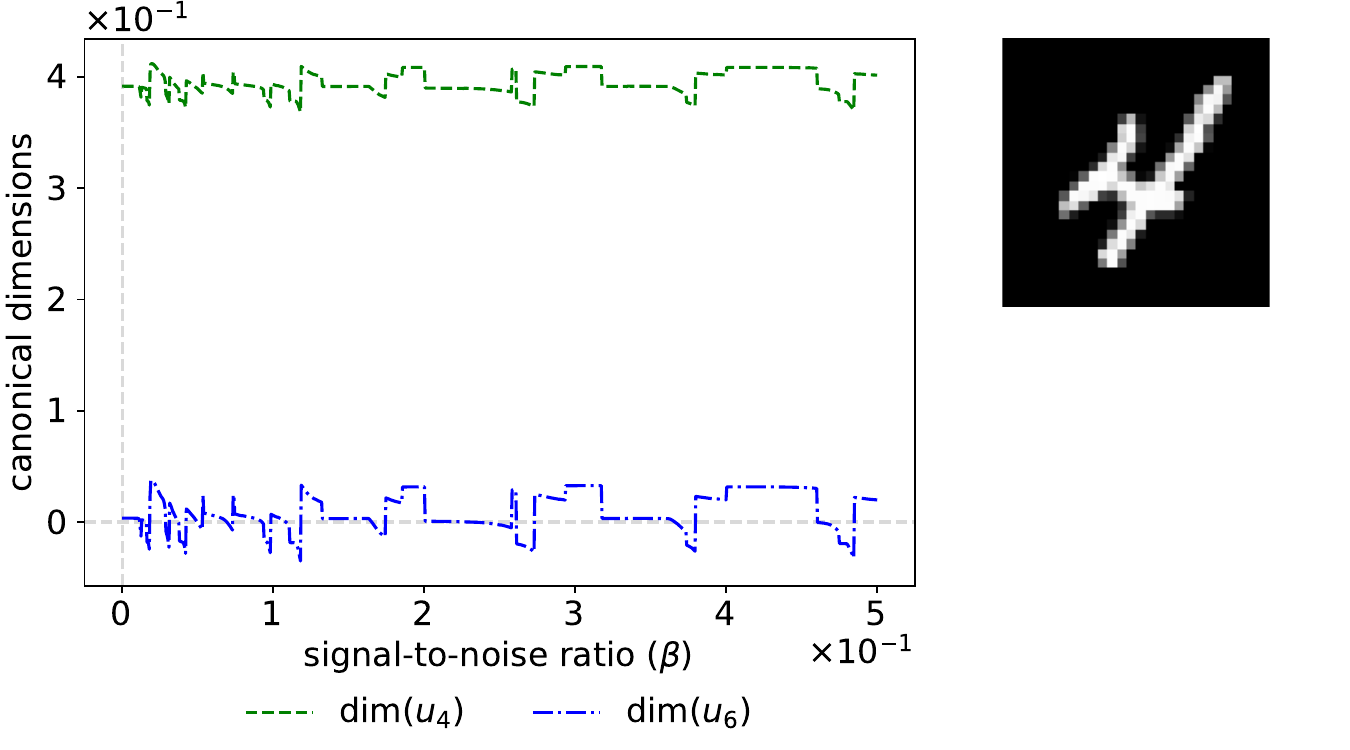}
    \caption{%
        Typical behaviour of the canonical dimension in the \textsc{mnist} set.
    }\label{fig:mnistplot}
\end{figure}

\Cref{fig:figplotgianduja} illustrates the global behaviour of the canonical dimensions as the \snr magnitude $\beta$ is progressively increased.
In each panel, the empirical inverse distribution $\rho(k^2)$ is shown as a dashed black curve, except in the top-left plot, which displays the analytical \mpdistr distribution for comparison.
We highlight several key features of these results:
\begin{enumerate}
    \item The canonical dimensions exhibit fluctuations that intensify with the interaction rank $n$, scaling approximately by a factor of $(n-1)$.
    \item The optimal linear interpolations, excluding the dataset endpoints which are affected by numerical instabilities originating from residual spikes, show excellent agreement with the analytical predictions at $\beta = 0$.
    \item These plots demonstrate the rigidity property of the distribution: the canonical dimension remains largely invariant for low values of $\beta$, undergoing a significant transition only upon reaching a threshold identified as the \lod $\beta_t$ (in this instance, $\beta_t \approx 0.15$).
\end{enumerate}

\begin{table}[t]
    \caption{Summary of the detection thresholds identified by the \frg flow.}\label{tab:thresholds}
    \centering
    \begin{tabular}{@{}cp{0.6\textwidth}@{}}
        \toprule
        \textbf{Symbol} & \textbf{Physical Interpretation}                                                                                                                                               \\
        \midrule
        $\beta_t$       & \textbf{Limit of detection (\lod).} The \enquote{rigidity} threshold at which the canonical dimensions begin to depart from the noise baseline.                                \\
        $\beta_c$       & \textbf{Critical threshold.} The value at which the canonical dimension of $u_4$ vanishes ($\dim_{\tau}(u_4) = 0$), thereby defining the effective critical point $\lambda_c$. \\
        $\beta_O$       & \textbf{Optimal threshold.} The first local minimum of $\dim_{\tau}(u_4)$, which corresponds to the point of maximum signal contrast prior to the emergence of cyclic effects. \\
        \bottomrule
    \end{tabular}
\end{table}

We compare these detection thresholds with the standard \bbp transition (see \Cref{AppD} for a review), which governs the detection of isolated, low-rank spikes in random matrices~\cite{math3}.
In the standard spiked covariance model, the \bbp transition occurs when the signal strength becomes large enough to eject the largest eigenvalue from the \mpdistr bulk.
For the additive noise model \eqref{eq:additive_model} with $q = P / N$, the critical \bbp threshold is defined by $\beta_{\text{BBP}} = q^{1/4}$~\cite{Bouchaud3}.
Given our simulation parameters ($N=\num{2.0e4}$, $q=0.9$), we obtain $\beta_{\text{BBP}} \approx 0.97$.
Our detection threshold of $\beta_t \approx 0.15$ is much lower ($\beta_t \ll \beta_{\text{BBP}}$), confirming that the \frg framework detects signals while they remain deeply embedded within the spectral bulk.
This regime corresponds to the bimodal connected phase recently identified in the extensive spike model, characterised by the emergence of a bimodal spectral distribution~\cite{Landau2023}.

We define two additional detection thresholds motivated by physical principles.
The critical detection threshold ($\beta_c$) corresponds to the value where the asymptotic canonical dimension of $u_4$ at scale $k^2_{\text{IR}}$ vanishes, signifying a local critical dimension of exactly 4 (in this example, $\beta_c \approx 0.32$).
This defines the critical scale $\lambda_c$ (in the eigenvalue spectrum), marking the boundary between signal and noise:
\begin{equation}
    \eval{\dim_{\tau}(u_4)}_{\lambda_c} = 0.
    \label{Lambdac}
\end{equation}
Finally, we define the optimal threshold ($\beta_O$) as the first local minimum of $\dim_{\tau}\qty(u_4)$ for $\beta < \beta_c$:
\begin{equation}
    \eval{\dv{\beta} \dim_{\tau}(u_4)}_{\beta_O} = 0.
\end{equation}
For our data, $\beta_O \approx 0.37$.
By construction, these thresholds satisfy the hierarchy $\beta_t \le \beta_c < \beta_O$.
A summary of these detection thresholds is provided in \Cref{tab:thresholds}.

\Cref{fig:figplotgianduja2} illustrates the behaviour of the canonical dimension at the scale $k^2_{\text{IR}}$, providing a pragmatic visualisation of the transition between two distinct physical regimes.
In the rigid regime ($\beta \in [0, \beta_t]$), the \ir canonical dimensions remain essentially invariant, exhibiting only minor fluctuations arising from the intrinsic variability of the noise relative to the analytical asymptotic spectra (see the next \Cref{Intrinsic}).
Beyond this threshold, the variations grow and the canonical dimension begins to drift.
The \rg interpretation of this transition is direct: a sufficiently strong signal pushes the flow towards Gaussianity, driven primarily by the flow of the mass, which emerges as the only relevant parameter.
Recalling that, for a suitably rescaled spectrum, the physical mass corresponds to the inverse of the largest eigenvalue, $\lambda_0^{-1}$, we conclude that, near the Gaussian fixed point, a signal-induced variance robustly approximates the microscopic \dof encoded in the correlation matrix spectrum.
We focus on the couplings $u_2$, $u_4$, and $u_6$, since they are the only relevant or marginal interactions for the \mpdistr universality class in the deep \ir~\cite{RG5}.
Within the \rg framework, higher-order interactions ($u_{2n}$ for $n > 3$) are strictly irrelevant: they are quickly suppressed as the flow evolves toward the \ir, thereby carrying no relevant information concerning signal-induced spectral deformations.

Note that for $\beta > \beta_O$, the canonical dimensions exhibit a secondary increase before decreasing again.
This implies that the quartic coupling may regain relevance before its canonical dimension becomes negative once more.
While we examine this cyclic behaviour and its interpretation in detail in \Cref{Sec4-1}, we restrict our current analysis to the neighbourhood of $\beta_O$, where our approximations, in particular those concerning the flow and the use of the \lpa, remain valid.
The behaviour of the canonical dimension in \Cref{fig:figplotgianduja2} can be contrasted with that of the \textsc{mnist} handwritten digits shown in \Cref{fig:mnistplot}.
It is well known that \textsc{mnist} samples are sufficiently low-dimensional for their primary features to be captured by algorithms such as \pca.
Consequently, no significant signal is expected within the spectral bulk, as isolated spikes account for nearly all relevant information.
Nevertheless, we can still detect residual signatures of a nearly continuous deformation, arising from the non-trivial nature of the signal.
This is evident in the variation of the sextic coupling behaviour, even though the quartic coupling remains relevant.

\subsection{Intrinsic Variability}\label{Intrinsic}

\begin{figure}[t]
    \centering
    \includegraphics[width=0.49\textwidth]{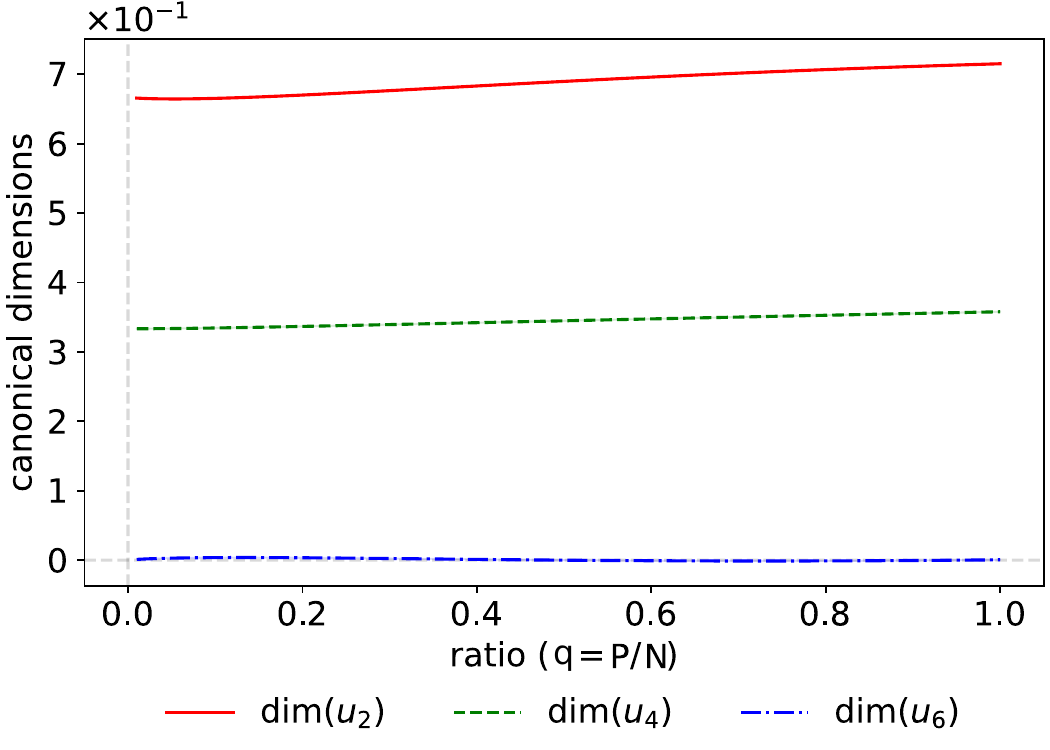}
    \hfil
    \includegraphics[width=0.49\textwidth]{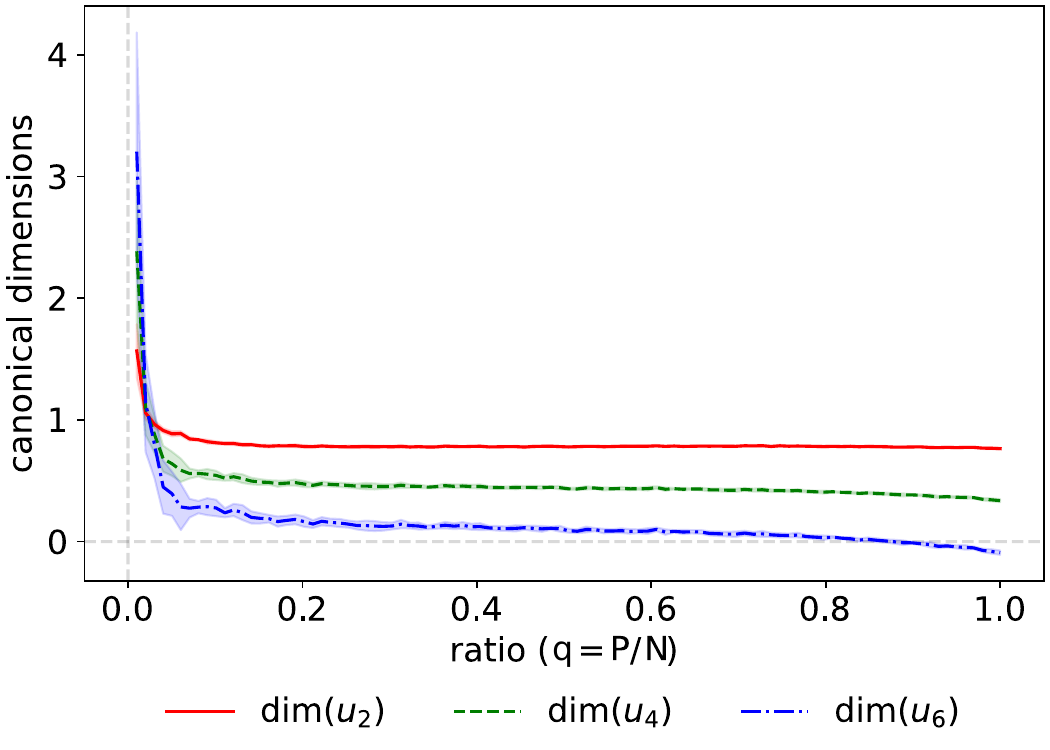}
    \caption{%
        (Left) Behaviour of the canonical dimension based on the analytical \mpdistr law. (Right) Behaviour of the empirical canonical dimension in the IR with respect to $q$ ($N$ fixed).
    }\label{fig:figPvar}
\end{figure}

\begin{figure}[t]
    \centering
    \includegraphics[width=0.49\textwidth]{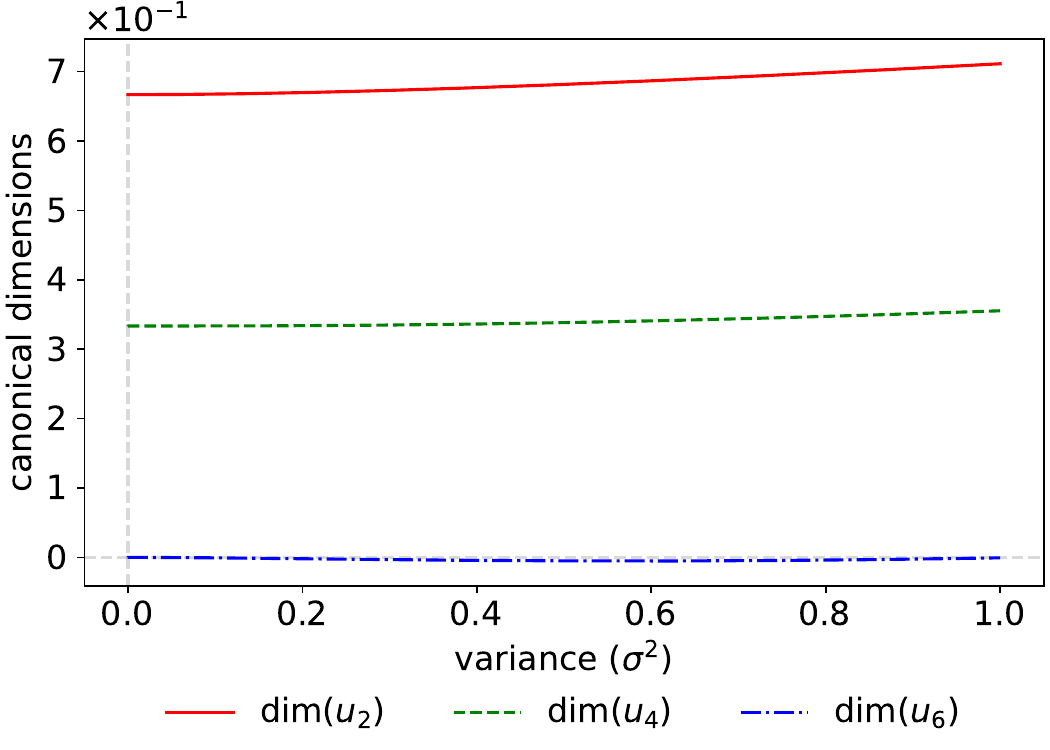}
    \hfil
    \includegraphics[width=0.49\textwidth]{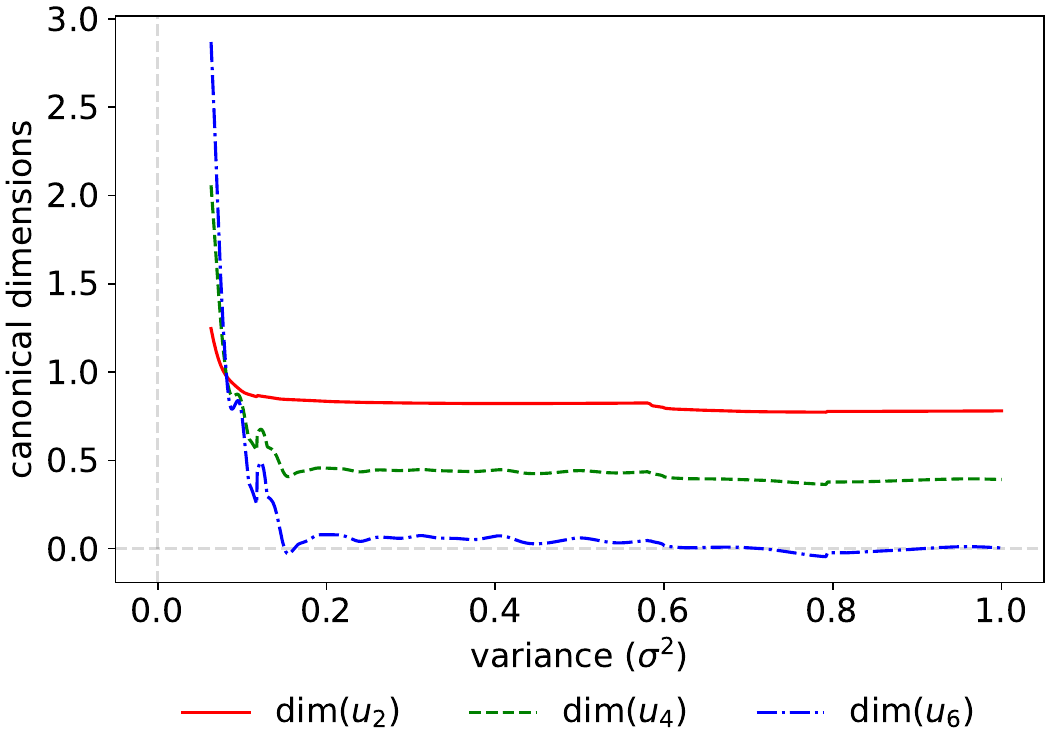}
    \caption{%
        (Left) Behaviour of the canonical dimension based on the analytical \mpdistr law. (Right) Empirical behaviour of the canonical dimension in the IR as a function of $q$ (with fixed $N$).
    }\label{fig:deviation}
\end{figure}

In what follows, we establish that the impact of intrinsic data fluctuations on the canonical dimensions is several orders of magnitude smaller than that induced by the signal.
This disparity identifies a specific detection threshold, which we term the variability threshold.
We elaborate on this criterion in \Cref{Formalization}, where we propose a formal method for its evaluation.
For the present discussion, we compare the magnitude of these noise-induced effects with those stemming from the signal in the vicinity of $\beta_O$.

In this context, \emph{variability} refers to the aleatoric uncertainty stemming from finite-size effects intrinsic to random matrix realisations.
Even with fixed macroscopic parameters ($N, P, q, \sigma^2$), distinct microscopic realisations of the noise matrix $Z$ introduce fluctuations into the empirical eigenvalue spectrum.
Quantifying this baseline variance is essential for distinguishing genuine signal-induced symmetry breaking from statistical noise.

\Cref{fig:figPvar,fig:deviation} quantify the variability of the \ir canonical dimensions across realisations for varying distribution parameters, such as $q$ (averaged over 100 random realisations with $N$ fixed) and the variance of the distribution.
\Cref{fig:Pvarbetaneq0} further details the dependence of the canonical dimension on $q$ for $\beta > 0$ (averaged over 100 realisations for a large, fixed $N$).
These findings clearly show that signal-induced effects, in particular the sign of the dimension of the quartic coupling, substantially exceed the intrinsic fluctuations inherent in the data.

\rmt provides a quantitative framework to characterise the fluctuations induced by finite-sample effects, especially at the spectral tail.
It is well established that the top eigenvalue exhibits fluctuations governed by the Tracy--Widom distribution, typically scaling as $\sim P^{-2/3}$~\cite{Bouchaud3}.
Consequently, signal-induced effects on the power counting become comparable in magnitude to these fluctuations when $\beta \simeq \beta_0 \defeq P^{-2/3}$.
Although this threshold sets an intrinsic limitation of our approach, we find that, for the values of $P$ considered here, $\beta_0 \approx \num{1.5e-3}$, which is substantially smaller than our typical detection scale ($\beta \sim 0.3$).
Although this detection scale depends on the data set, the clear separation of scales ($\beta_t \gg \beta_0$) confirms that the dimensional phase transition observed at $\beta_t \approx 0.15$ is a robust macroscopic phenomenon, distinct from the finite-size fluctuations of the spectral edge.

\begin{figure}[t]
    \centering
    \includegraphics[width=0.7\textwidth]{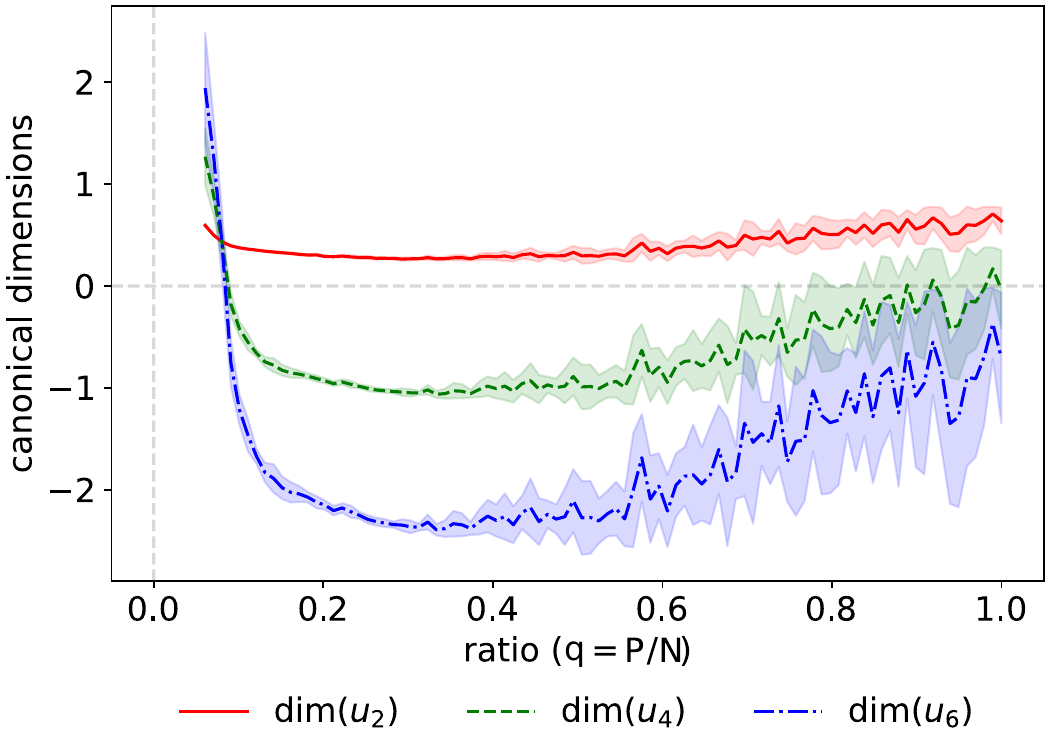}
    \caption{%
        Empirical canonical dimension in the \ir as a function of $q$ for a fixed \snr ($\beta=1.25$) and constant sample size ($N = \num{2e4}$), using the dataset \Cref{fig:gianduja}.
    }\label{fig:Pvarbetaneq0}
\end{figure}

\subsection{Small Breaking of the Porter--Thomas Universality}

To conclude our presentation of the results near the detection threshold, we now turn our attention to the statistical properties of the eigenvectors.
We begin by examining eigenvectors in the \mpdistr class.
In the case of purely stochastic noise, these eigenvectors
\begin{equation}
    u_\lambda \defeq (u_\lambda^1,u_\lambda^2,\dots, u_\lambda^P),
\end{equation}
are delocalised, with individual entries bounded by $\simeq 1/\sqrt{P}$~\cite{bouferroum:hal-00835504,bogomolny2017modification}.
Furthermore, as $P \to \infty$, the matrix of right eigenvectors is asymptotically Haar-distributed over the orthogonal group $\mathrm{O}(P)$.
In the absence of additional information, the distribution of the components $s = u_{(\mu)}^{i}$ is well approximated by the maximum entropy distribution subject to the normalisation constraint $\sum_i (u_\lambda^i)^2 = 1$.
This yields the Porter--Thomas distribution~\eqref{PTdist}.

This behaviour is empirically validated for $\beta = 0$ and illustrated in \Cref{fig:distbeta0}.
For this comparison, we analyse eigenvectors associated with the 100 smallest eigenvalues of the \uv regime, and those at the \mpdistr mass scale ($k^2 = \qty{(\lambda_+ - \lambda_-)}^{-1}$) of the \ir regime, to ensure consistency across varying $\beta$ values.
The empirical distribution shows strong agreement with the Porter--Thomas maximum entropy estimator, reinforcing the consensus that the underlying distribution lacks additional structure beyond that captured by the Porter--Thomas law.

\begin{figure}[t]
    \centering
    \includegraphics[width=0.7\textwidth]{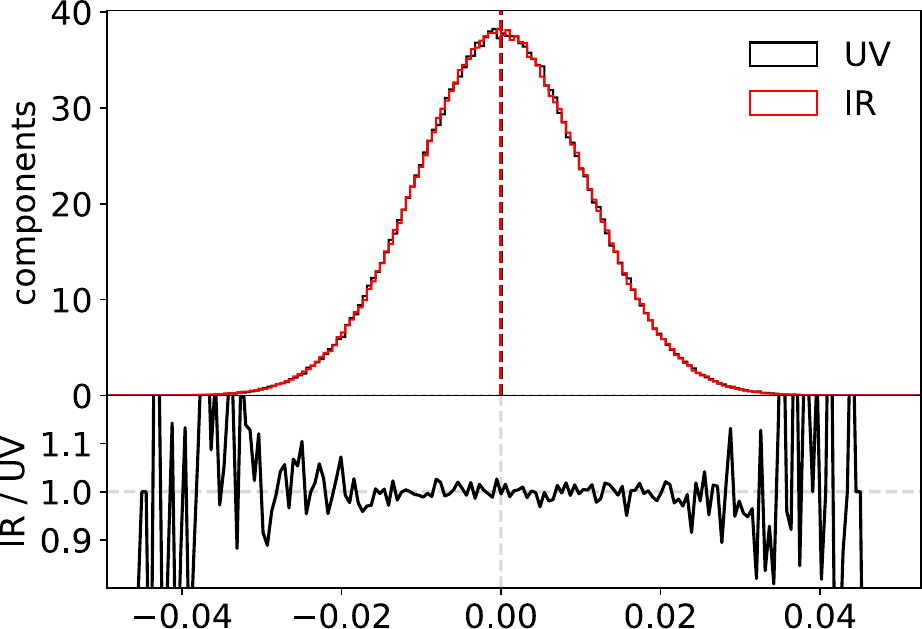}
    \caption{%
        Distribution of eigenvector components at $\beta = 0$ (pure noise regime), shown for the 100 eigenvectors associated with the smallest eigenvalues (\uv) and the 100 eigenvectors corresponding to the \mpdistr bulk scale (\ir).
    }\label{fig:distbeta0}
\end{figure}

\begin{figure}[t]
    \centering
    \begin{minipage}{0.49\textwidth}
        \centering
        {$\beta = 0.2$} \\[0.25em]
        \includegraphics[width=\linewidth]{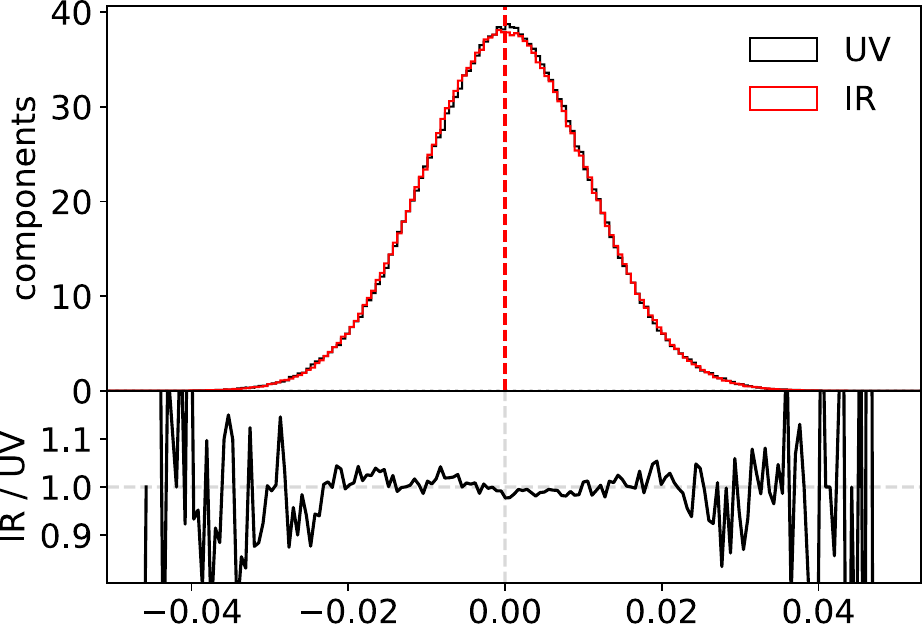}
    \end{minipage}
    \hfil
    \begin{minipage}{0.49\textwidth}
        \centering
        {$\beta = 0.3$} \\[0.25em]
        \includegraphics[width=\linewidth]{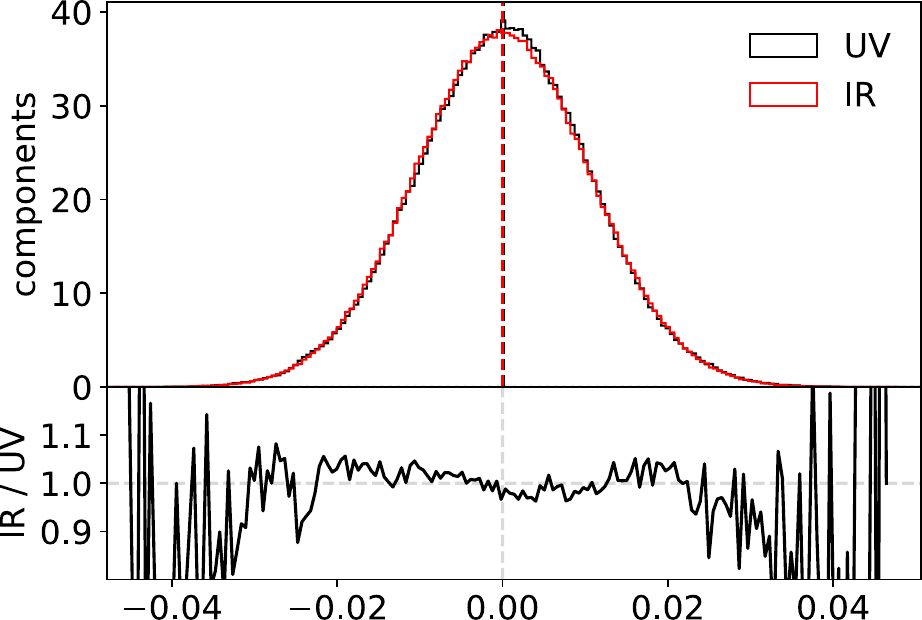}
    \end{minipage}
    \\[0.25em]

    \begin{minipage}{0.49\textwidth}
        \centering
        {$\beta = 0.4$} \\[0.25em]
        \includegraphics[width=\linewidth]{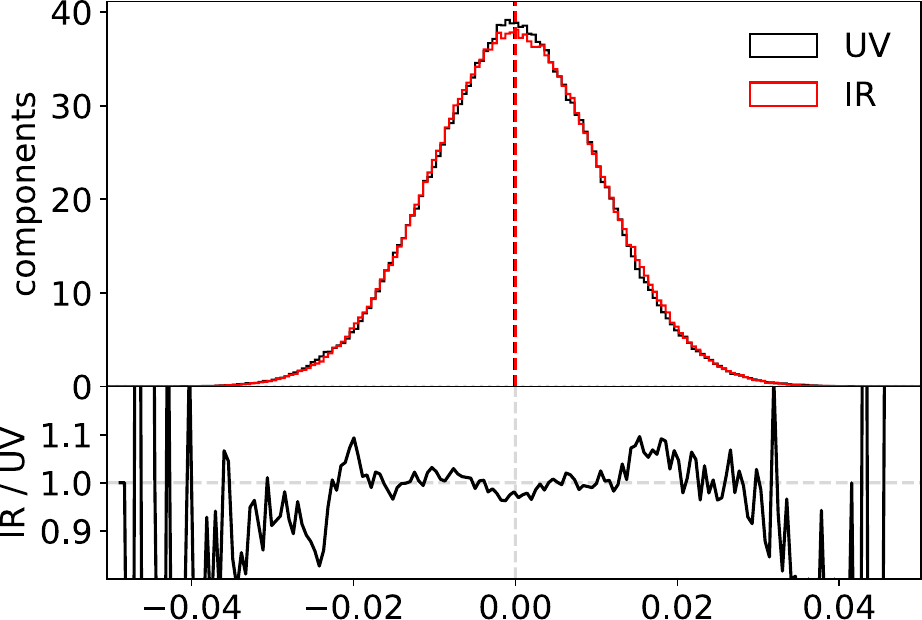}
    \end{minipage}
    \hfil
    \begin{minipage}{0.49\textwidth}
        \centering
        {$\beta = 0.5$} \\[0.25em]
        \includegraphics[width=\linewidth]{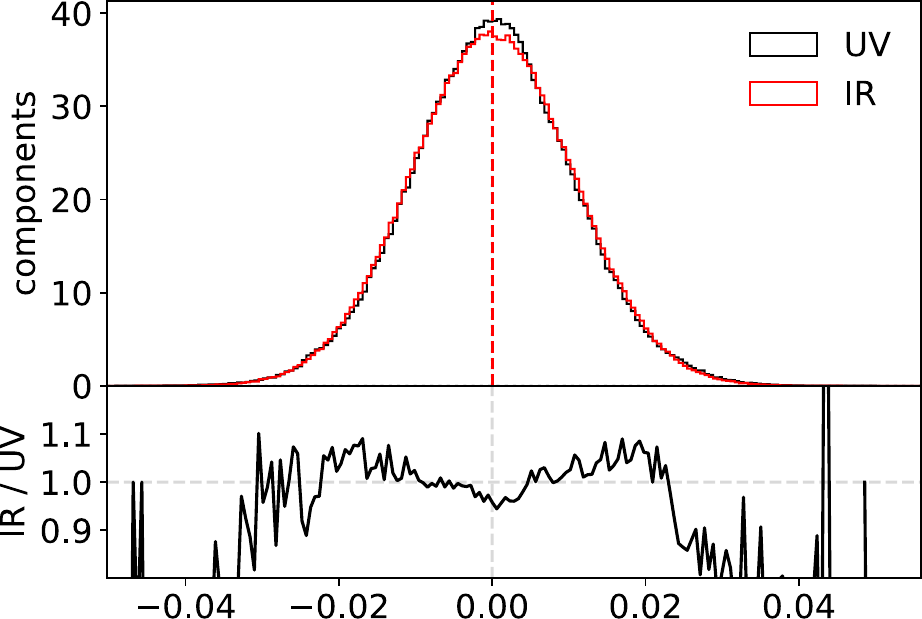}
    \end{minipage}

    \caption{%
        \uv and \ir distributions of eigenvector components for different values of $\beta$.
    }\label{fig:distbetabeta}
\end{figure}

\Cref{fig:distbetabeta} illustrates the evolution of eigenvector statistics with increasing \snr, with the relevant properties summarised in \Cref{fig:distbetabeta2}.
These results demonstrate a clear transition in the statistics at values of $\beta$ corresponding to significant changes in the canonical dimensions, thereby validating $\beta_t$, $\beta_c$, and $\beta_O$ as robust indicators of signal presence.
The most pronounced change is observed in the distribution's standard deviation, which increases with signal intensity.
In particular, the ratio of the \ir to \uv standard deviations exhibits a peak at each local minimum of the canonical dimensions, particularly near $\beta_O$.
The presence of a signal also induces a shift in the mean of the distribution.
These numerical findings are consistent with the theoretical results of~\cite{bogomolny2017modification} concerning the single-spike regime.

\begin{figure}[t]
    \centering
    \includegraphics[width=\textwidth]{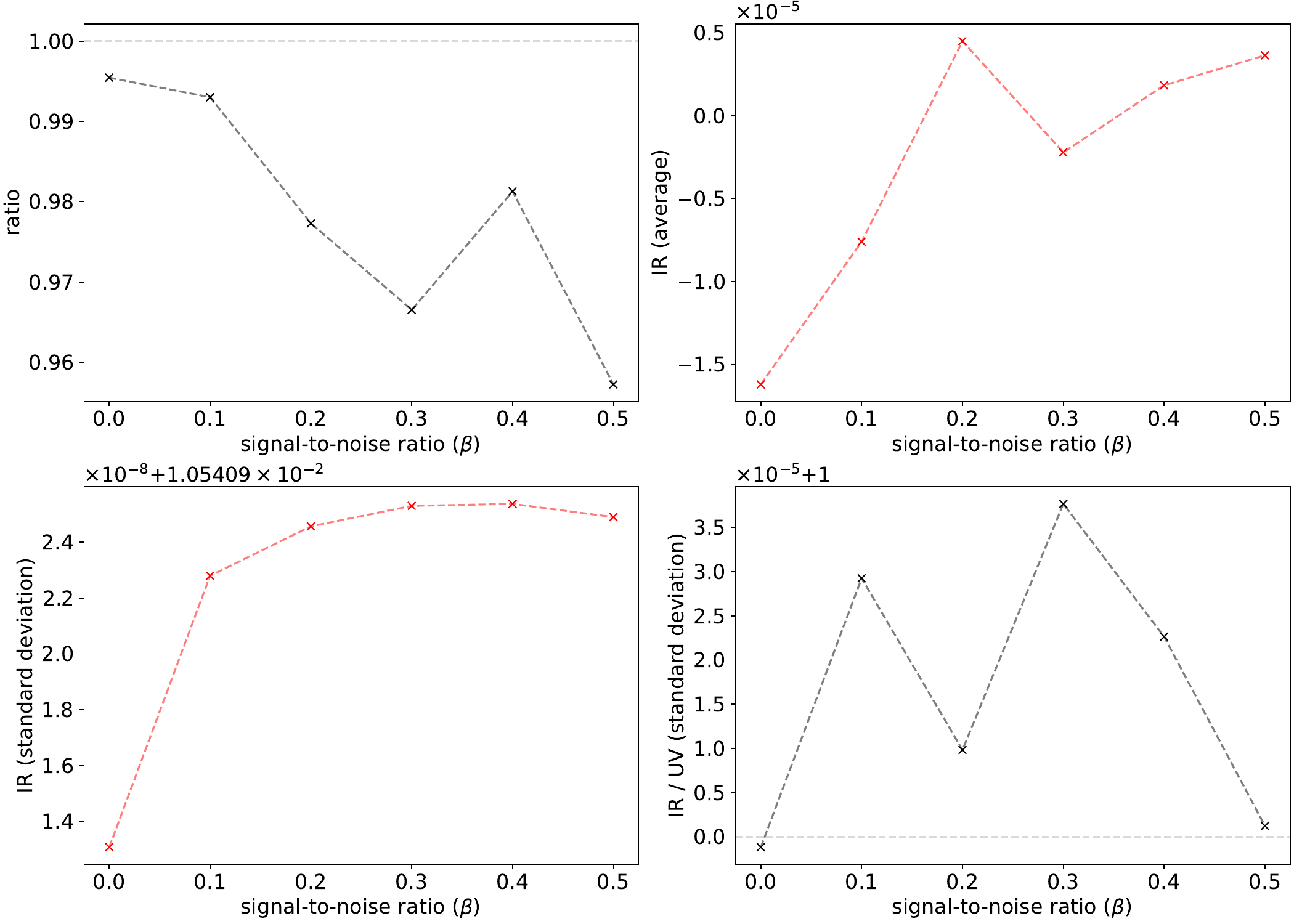}
    \caption{%
        Summary of the statistical properties characterising the eigenvector distributions: the ratio of the components at the origin, the mean value, the standard deviation of the \ir components, and the ratio of the \ir-to-\uv standard deviations.
    }\label{fig:distbetabeta2}
\end{figure}

\section{Ising Critical Temperature Estimation}\label{sec:ising_critical_temp}
In this section, we present a first benchmark analysis that anchors the theoretical claims of the previous sections in physical reality, by examining the non-equilibrium dynamics of an Ising system in the quenched regime.
This allows us to connect the results of the \dft to known physical phenomena and to test the effectiveness of \gsa for signal detection.

In a simple 2D non-equilibrium Ising model, spin domains emerge during coarsening after a \emph{quench} of the dynamics from a high-temperature phase, forming correlated structures of varying sizes in which all spins remain aligned~\cite{livi2017nonequilibrium}.
Such correlations mirror those found in realistic settings such as computer-vision images, where the distinction from white noise rests on the fact that pixels remain correlated over extended spatial domains.
From a heuristic standpoint, the temperature that controls the emergence of these domains below the critical value $T_c$ plays a role analogous to the inverse of the \snr.
A noise-dominated image therefore exhibits essentially no correlation among its pixels.

We focus on a non-equilibrium field theory model for this analysis.
Specifically, we work with \emph{Model~A}, which can be viewed as the non-equilibrium (time-dependent) counterpart of the ordinary $\Phi_D^4$ theory (see \Cref{part1}) for the Ising model.
This model exhibits a dynamic phase transition at the same critical temperature as the corresponding equilibrium theory.
The section serves two purposes.
First, it illustrates a methodological connection with the physics of phase transitions.
Second, it demonstrates the relevance of \gsa on a well-known example.
Indeed, Onsager's exact solution~\cite{onsager1944crystal} of the 2D Ising model provides an excellent benchmark, with the critical temperature known to very high precision.\footnote{%
    The precise value depends weakly on the system size, provided the system is sufficiently large.
}
The study of these systems is taken further in \Cref{part3}; we provide here an overview of the relevant background before presenting the results.
The physical principles summarised in \Cref{part1} of this review are taken as prerequisites, although the key concepts needed to follow this section are recalled briefly where they first appear.

\subsection{Critical Dynamics and Coarsening}\label{critical}

In contrast to equilibrium Euclidean field theory, which is well understood and naturally arises in condensed matter physics~\cite{Zinn1,Zinn2}, non-equilibrium Euclidean field theory and non-equilibrium phenomena remain an active research frontier in theoretical physics~\cite{kamenev2023field,calzetta2009nonequilibrium,livi2017nonequilibrium}.
These phenomena, which break time-reversal symmetry, include the dynamics of phase transitions (critical dynamics), glassy systems, reaction-diffusion processes, and quantum transport, among others.
They also serve as a source of inspiration for optimisation algorithms in deep learning~\cite{agoritsas2018out}.
In this review, we focus on critical dynamics and phase-ordering kinetics~\cite{Bray_1994,tauber2014critical}.

The term \emph{quench} denotes the rapid change in control parameters by which a thermodynamic system such as the Ising model (see \Cref{part1}) is driven abruptly from a disordered phase to an unstable phase in which order subsequently develops progressively and locally.
A representative example, central to this section, is a ferromagnetic system at high temperature that is rapidly quenched in a low-temperature bath.
When the bath temperature lies below the critical temperature $T_c$, nucleation associated with ageing effects drives the appearance of domains that grow slowly over time according to a characteristic scaling law
\begin{equation}
    L(t) \sim t^n,
\end{equation}
where $L(t)$ also represents the typical correlation length of the system at time $t$ (see \Cref{fig2}).
For a system of dimension $D$, the two-time correlation function for $t>t_0$ decays as a power law, exhibiting \emph{weak ergodicity breaking} (i.e.\ the system retains a memory of its initial state in the long-time limit, see~\Cref{AppOJK})~\cite{Bray_1994}:
\begin{equation}
    C(t, t_0) \sim \left(\frac{t}{t_0}\right)^{-D/4},
    \label{timecorr}
\end{equation}

This coarsening phenomenon, characterised by slow growth governed by the coarsening exponent $n$, physically reflects the fact that the system takes an arbitrarily long time to reach equilibrium as the system size increases.
Before this phase, a comparatively short stage known as \enquote{critical coarsening} occurs immediately after the quench, while the system is still influenced by critical fluctuations.
Contrary to what happens at equilibrium, the behaviour of the system, and in particular the value of $n$, does not depend on the spatial dimension for stochastic models (those subject to random thermal fluctuations); instead, it depends on other characteristics, such as whether the dynamics is conservative.
For further details, the reader may consult the references cited above and the monograph~\cite{livi2017nonequilibrium}.

At the moment of the transition, a phenomenon linked to the divergence of the correlation length occurs, known as \emph{critical slowing down}.
This poses serious numerical difficulties, which are further accentuated in the ordinary Ising model where the variables are discrete (for instance, the energy cost of \enquote{flipping} a spin).
This phenomenon, for which an elementary approach is provided in \Cref{sec:app1}, arises from the fact that near the critical temperature, strong long-range correlations between spins make it highly improbable for a spin to change its sign.
The system thus appears frozen on a long time scale $\tau$, typically:
\begin{equation}
    \tau \sim \xi^z,
\end{equation}
where $\xi \sim |T-T_c|^{-\nu}$ is the correlation length and $\nu, z > 0$ are the associated critical exponents.

In this paper, we describe this phenomenon through the lens of \rg considerations applied to the \ecm of time series generated by a time-dependent discretised Ginzburg--Landau theory, namely \emph{Model~A}, for which $n = 1/2$ below the critical temperature.
We also focus on the behaviour observed above $T_c$, specifically as the system approaches the critical point from above.
In this regime, the spectrum of the \ecm, which corresponds to a \mpdistr distribution for $T = \infty$, begins to exhibit a heavy-tailed, power-law form as the transition temperature is approached.
We show below that this crossover conceals a phase transition in the corresponding field theory, which itself coincides with the maximum-entropy estimator of the spectral data (see \Cref{App0}).

Model~A is a kinetic Ising model that emerges in the coarse-graining limit of Glauber dynamics for Ising spins.
In $D$ dimensions, the coarse-grained description of the discrete spins is captured by a \emph{continuous field} $\varphi(x,t)$ following Langevin-type dynamics~\cite{Zinn1,livi2017nonequilibrium}:
\begin{equation}
    \frac{\partial}{\partial t} \varphi(x,t)
    =
    - \frac{1}{T} \frac{\delta H}{\delta \varphi(x,t)}+ \eta (x,t),
    \label{eqLangevin}
\end{equation}
where $T$ is the temperature of the equilibrium system\footnote{In units where the Boltzmann constant $k_B$ is set to unity.} and $H$ is the Ginzburg--Landau Hamiltonian:
\begin{equation}
    H = \int \dd x \dd t\,
    \left(\frac{1}{2} (\nabla \varphi)^2 + V(\varphi) \right).
\end{equation}
The potential $V$ is taken to be a polynomial $V(\varphi) = \frac{1}{2} a \varphi^2 + \frac{b}{4} \varphi^4 + \cdots$, and the term $\eta(x,t)$ represents spatiotemporally white Gaussian noise with zero mean:\footnote{%
    The factor of $2$ ensures that the equilibrium probability distribution matches the Boltzmann distribution $P(\varphi) \sim e^{-H/T}$, see~\cite{Zinn1}.
}
\begin{equation}
    \langle \eta(x,t) \rangle = 0,
    \qquad
    \langle \eta(x,t) \eta(x^\prime,t^\prime) \rangle_\eta = 2 \delta(x-x^\prime) \delta(t-t^\prime).
\end{equation}
\Cref{eqLangevin} describes a stochastic gradient descent toward a minimum of the Hamiltonian.
For the numerical simulations in this review, we focus on a 2D model on a regular grid, governed by the explicit Euler--Maruyama discretisation:
\begin{equation}
    \varphi_{i,j}^{n+1}
    =
    \varphi_{i,j}^{n} + \Delta t \left[ \frac{1}{T} (\Delta_{\text{dis}}\varphi)^{n}_{i,j} - V^{\prime}[\varphi^{n}]\right]+\sqrt{2\Delta t}\,\xi_{i,j}^{\,n}
    \label{discrete}
\end{equation}
where $\xi_{i,j}^{\,n}\sim\mathcal{N}(0,1)$ is a spatiotemporally uncorrelated standard Gaussian variable.
We use the standard 5-point discrete Laplacian:
\begin{equation}
    (\Delta_{\text{dis}}\varphi)_{i,j}
    =
    \frac{\varphi_{i+1,j} + \varphi_{i-1,j}+\varphi_{i,j+1} + \varphi_{i,j-1}-4\varphi_{i,j}}{ \ell^2}.
    \label{DiscreteLaplacian}
\end{equation}
The hyperparameter $\ell$ is the lattice spacing, fixed at $\ell = 1$.
We further truncate the potential at the quartic term:
\begin{equation}
    V(\varphi) = \frac{1 - 1/T}{2}\,\varphi^2 + \frac{b}{4T}\varphi^4.\label{potentialV}
\end{equation}
This parametrisation is numerically advantageous compared with the standard $(T-T_0)$ form, with which it is equivalent near the transition, because it avoids the appearance of an excessively large mass term.
The bare transition temperature $T_0$ is then fixed at $T_0 = 1$.
When $T<1$, the potential exhibits a double-well structure and $\varphi=0$ becomes an unstable solution.
Note, however, that $T_0$ does not coincide with the physical critical temperature $T_c$: the quartic interaction generates additional rigidity through fluctuations, which shifts the bare value of the transition (see \Cref{part1} and \Cref{sec:app1})~\cite{Hoenberg1977}.

\begin{figure}[t]
    \centering
    \includegraphics[height=0.18\textheight]{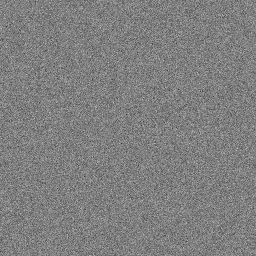}\,
    \includegraphics[height=0.18\textheight]{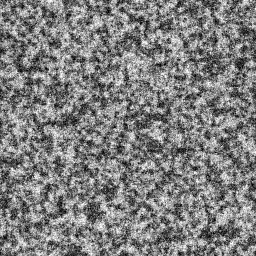}\,
    \includegraphics[height=0.18\textheight]{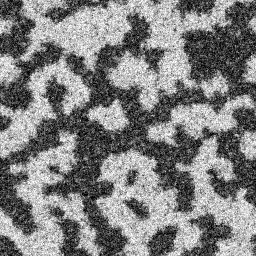}\\[0.25em]
    \includegraphics[height=0.18\textheight]{img/IM1.png}\,
    \includegraphics[height=0.18\textheight]{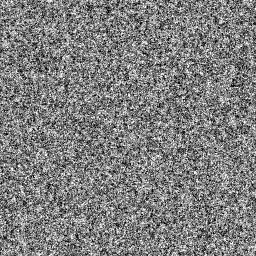}\,
    \includegraphics[height=0.18\textheight]{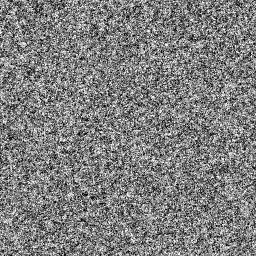}\
    \caption{%
        Illustration of the phase ordering kinetics for the Ising-like stochastic model (Model~A) considered in this paper.
        The three top images represent the evolution following a quench from $T=\infty$ at $t<0$ (first image on the left) to $T=0.5\, T_c$ at $t=0$.
        The next two images correspond respectively to time steps $\num{3e3}$ and $\num{3.9e4}$ of the evolution.
        From left to right, the three bottom images correspond to the same situation for $T=1.5\, T_c$.
    }\label{fig2}
\end{figure}

The correspondence with the standard 2D Ising model can be established heuristically using the background material in \Cref{sec13,sec14}.
The critical temperature in mean-field theory is $T_0=4\,J$, where $J>0$ is the ferromagnetic Ising coupling, defined such that aligned neighbouring spins lower the energy. The Hamiltonian is therefore:
\begin{equation}
    H_\text{Ising} = -J \sum_{\langle i,j \rangle} S_i S_j,
    \label{IsingModel}
\end{equation}
where $\langle i,j \rangle$ denotes a sum over nearest neighbours.
For $D=2$, Onsager obtained~\cite{onsager1944crystal}:
\begin{equation}
    T_c = \frac{2\,J}{\ln(1 + \sqrt{2})} \approx 2.27\, J.
\end{equation}
Setting $T_0=1$ fixes $J=1/4$, which yields $T_c \approx 0.57$.
The formal correspondence at the critical point then requires
\begin{equation}
    b = \frac{1}{3 T_c^3} \approx 1.80.
\end{equation}
This value of $b$ is determined by a fourth-order truncated Ginzburg--Landau expansion obtained from a Hubbard--Stratonovich transformation (see \Cref{sec14}).
Although this mapping is exact at the level of the effective action (up to interactions involving powers of the field higher than six), numerical convergence toward the Onsager temperature is achieved only in the limit of an infinitely deep potential well ($b \to \infty$), which suppresses transitions between the two wells.
Numerical coincidences for finite values of $b$ correspond to Ising-like systems, but they remain heuristic at best.
A more robust correspondence with the standard 2D Ising model can be obtained by considering the potential (see \Cref{sec:app2})
\begin{equation}
    V(\varphi)=-\frac{1}{2} \varepsilon\varphi^2+\frac{1}{4}\varepsilon \varphi^4
    \label{VIsing}
\end{equation}
with the limit $\varepsilon \to \infty$, which enforces the constraint $\varphi=\pm 1$.
The choice $J=1/4$ then requires rescaling the Laplacian in~\eqref{discrete} by a factor of $1/4$, as detailed in the appendix.
In what follows, the model defined by~\eqref{VIsing} is referred to as the \emph{Ising model}, while \emph{Ising-like} denotes the heuristic correspondence derived from the quartic Model~A.

\subsection{Critical Temperature Estimation}\label{IsingCritical}

When approaching $T_c$ from above, the spectrum of the \ecm of the Ising model shows a deformation of the tail of the distribution, which begins to follow a power law.
Below, we show that this crossover conceals a phase transition in the scaling behaviour of the operators, which manifests as a discontinuity in the cutoff $\lambda_c$ determined by the condition \eqref{Lambdac}.
This transition occurs at a temperature that coincides with experimental estimates of the critical temperature.

We also reproduce the same analysis for the exact Ising model, defined as Model~A in the $\varepsilon \to \infty$ limit according to the parametrisation \eqref{VIsing}.
The key advantage of this case is the existence of an exact analytical result, namely the Onsager temperature ($T_c \approx 0.57$ in the chosen parametrisation), which allows a direct comparison with the \gsa estimate, independent of any experimental measurement.

\subsubsection{Model~A, Method and Result}

The data are represented by a large matrix $X \in \mathds{R}^{t_{\text{fin}} \times N^2}$, where $N$ is the linear size of the 2D grid and $t_{\text{fin}} = 4 \times 10^4$ is the number of time steps.
At high temperatures the model is nearly Gaussian, and the spectrum of the $N^2 \times N^2$ correlation matrix $C$ closely follows the distribution predicted by \Cref{thm:thMP} (left panel of \Cref{fig3}).
At lower temperatures, and particularly when crossing the critical point $T=T_c$, long-range correlations emerge and distort the spectrum, which then exhibits a characteristic heavy-tailed behaviour.
This phenomenon is linked to the appearance of low-mass \ir modes, associated with long-range correlations and characteristic power-law scaling~\cite{vinayak2014spectral} (right panel of \Cref{fig3}).

\begin{figure}[t]
    \centering
    \includegraphics[width=0.49\textwidth]{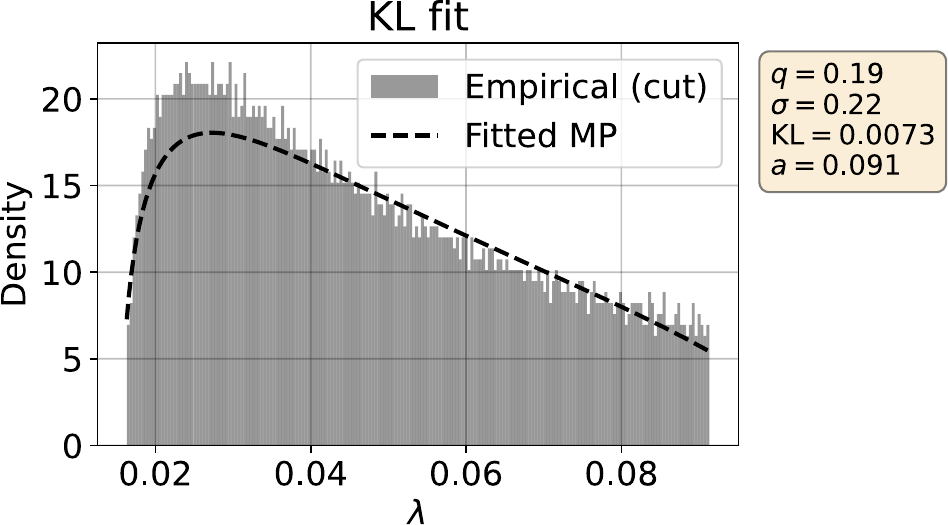}
    \hfil
    \includegraphics[width=0.49\textwidth]{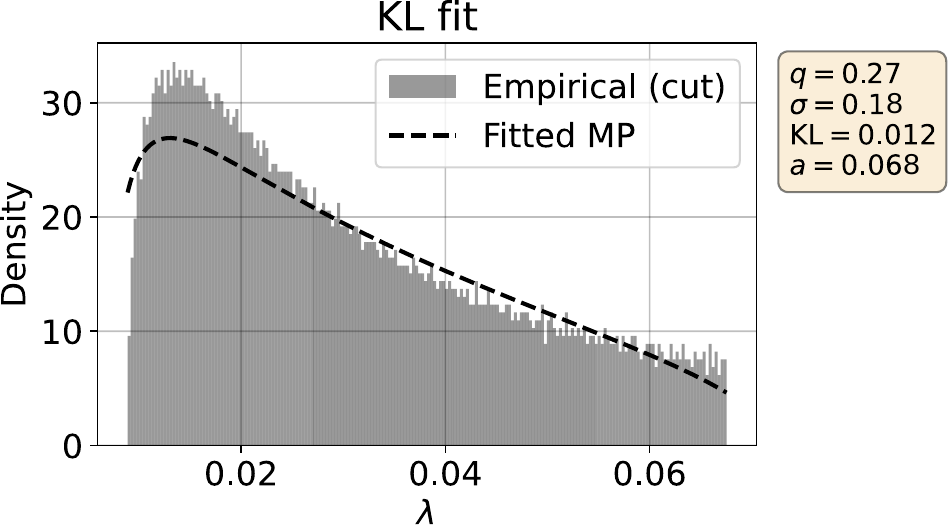}
    \hfil
    \includegraphics[width=0.49\textwidth]{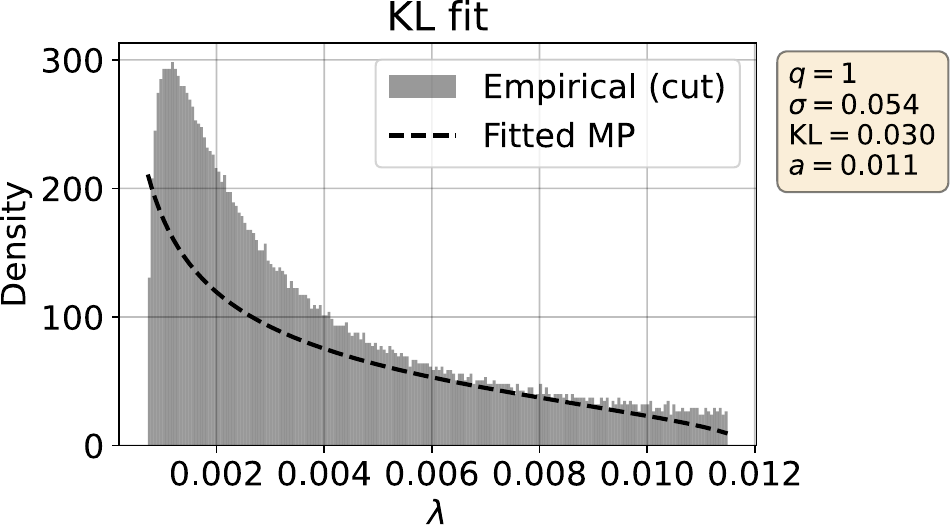}
    \caption{%
        (Top left) Behaviour of the spectrum found by the \kl-proxy for high temperature, $T=5$.
        (Top right) Behaviour in the vicinity of the critical regime $T=0.648$.
        (Bottom) Behaviour below the critical temperature $T=0.350$ for $b=0.8$.
    }\label{fig3}
\end{figure}

\begin{figure}[t]
    \centering
    \includegraphics[width=0.7\textwidth]{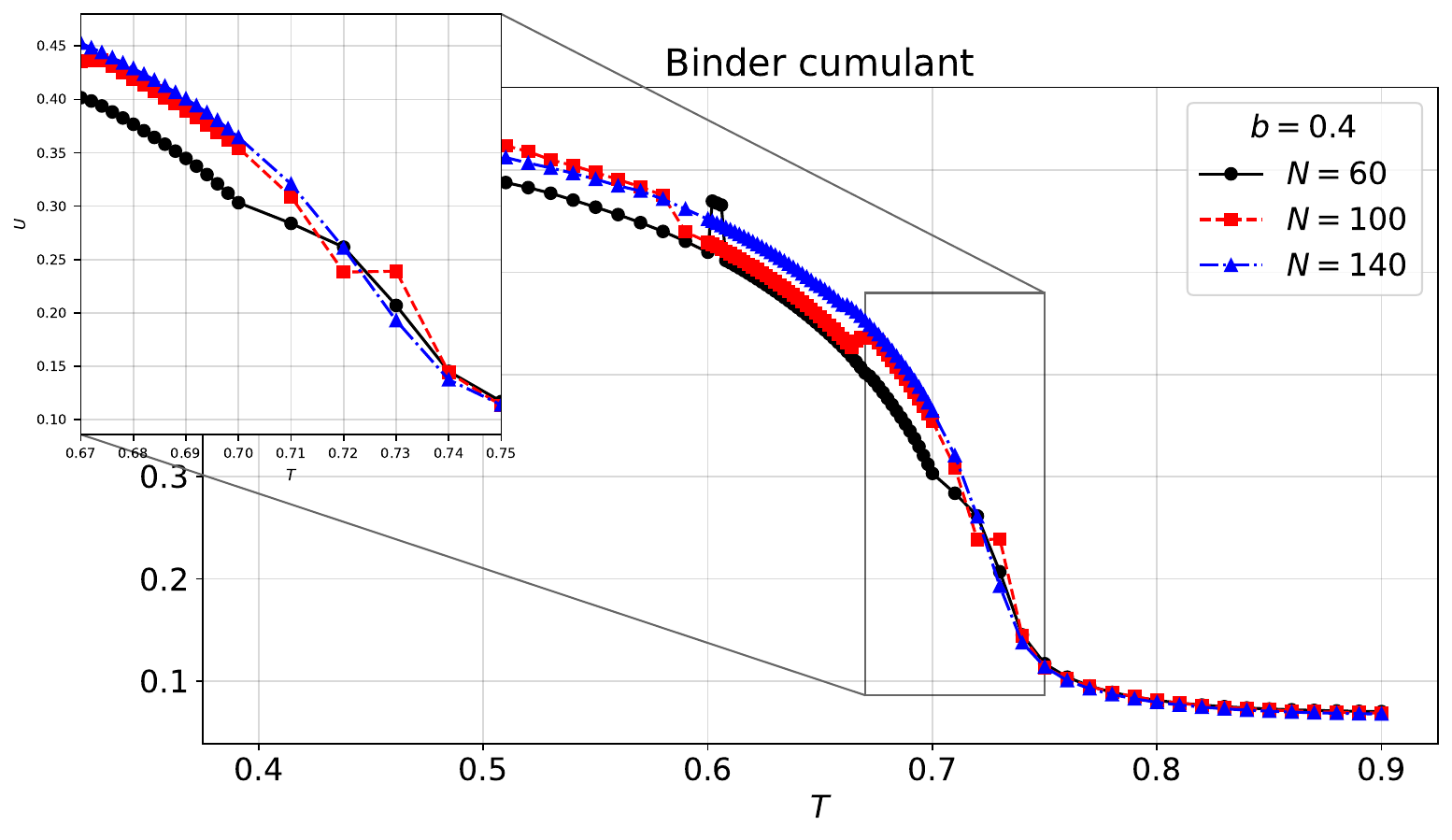} \\[0.25em]
    \includegraphics[width=0.7\textwidth]{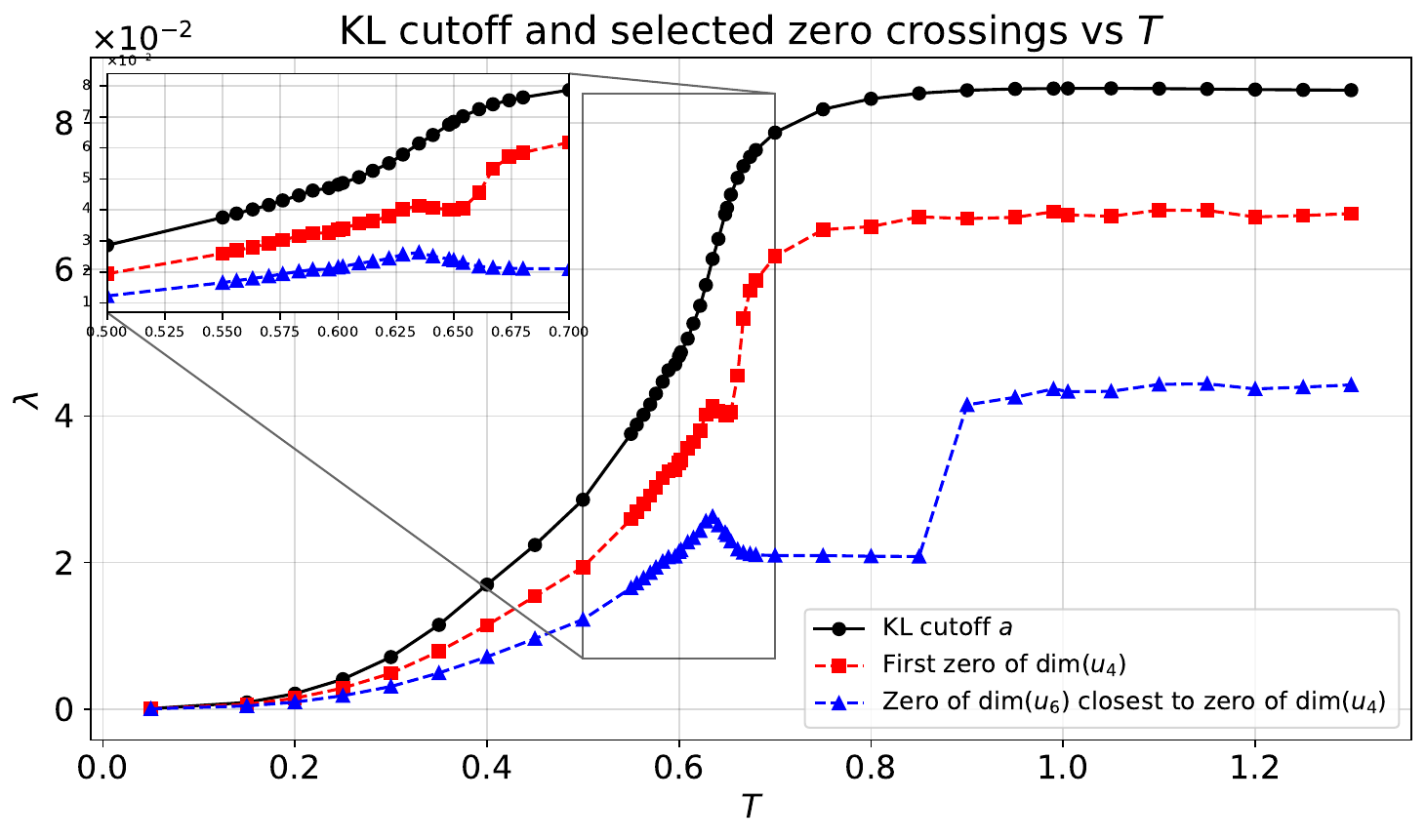}
    \caption{%
        (Top) Behaviour of the Binder cumulant for different grid sizes with $b=0.4$.
        (Bottom) Empirical behaviour of $\lambda_c$ for different temperatures with different methods for $b=0.8$.
    }  \label{fig6}
\end{figure}

To validate the \gsa results, we need a reliable way to estimate the critical temperature.
A common method relies on the behaviour of the magnetic susceptibility, $\chi \defeq \partial M / \partial B$, using the notation established in \Cref{part1}.
This method has a major drawback: while the susceptibility exhibits a singularity for an infinite system, this singular behaviour disappears for finite systems.
Furthermore, the maximum values of the susceptibility depend on the system size, which reduces the precision of the critical-temperature estimate.
Another widely used method exploits \emph{Binder cumulants}~\cite{selke2006critical}:\footnote{%
    The Binder cumulant is a dimensionless quantity that characterises the shape of the probability distribution of the order parameter.
    In statistical terms, it equals the negative excess kurtosis of that distribution divided by three: $K = -\tilde{\mu}_4/3$, where the excess kurtosis $\tilde{\mu}_4 = \mu_4/\sigma^4 - 3$ is the standardised fourth central moment in excess of the Gaussian value (a Gaussian distribution has $\tilde{\mu}_4 = 0$ and therefore $K = 0$).
}
\begin{equation}
    K = 1 - \frac{\langle \varphi^4 \rangle}{3 \langle \varphi^2 \rangle^2}.
\end{equation}
As a function of temperature $T$, $K$ interpolates between two limiting behaviours.
At high temperatures ($T \gg 1$), the field is essentially Gaussian, and Wick's theorem (see \Cref{thm:theoremWick}) yields the exact four-point function $\langle \varphi^4 \rangle$:
\begin{equation}
    \langle \varphi^4 \rangle \approx 3\langle \varphi^2 \rangle,
\end{equation}
such that $\lim_{T \gg 1} K(T)\approx 0$.
At low temperatures, the field is frozen at a non-zero value $|\phi_0|$, so that
\begin{equation}
    \langle \varphi^2 \rangle \approx \phi_0^2, \qquad \langle \varphi^4 \rangle \approx \phi_0^4.
\end{equation}
Thus, $\lim_{T \ll 1} K(T)\approx 2/3$.
At the transition point, the correlation length becomes singular and the absolute system size becomes irrelevant (in the \rg sense): the physical transition temperature is therefore identified as the crossing point of $K(T)$ for grids of different sizes.
The top panel of \Cref{fig6} illustrates this for grids of sizes $N \in \{ 60, 100, 140 \}$ at fixed coupling $b$ (see \Cref{potentialV}).
On varying $b$, we obtain an experimental value $T_c^{(\text{exp})}(b)$ and study the behaviour of the eigenvalue distribution of the matrix $C$ across temperatures.

The \snr boundary $\Lambda_{\text{exp}}(T)$, previously defined by the \gsa (see \Cref{sec:gsa}), also depends on temperature: a qualitative change in the distribution between the high-temperature Gaussian regime and the low-temperature regime, in which large spin clusters dominate, can be interpreted as the emergence of a signal associated with the appearance of highly correlated macroscopic regions (mirroring what is observed in real-world images).
The transition region around $T_c$ therefore marks the precise \snr transition point, which is reflected in the empirical behaviour of $\Lambda_{\text{exp}}(T)$ (see the bottom panel of \Cref{fig6}).
The transition point highlighted in the figure corresponds to the inflection point, at which a discontinuity emerges that becomes increasingly pronounced as the grid size increases.
This point marks the transition from a high-temperature regime, characterised by a distribution very close to the \mpdistr law, to a low-temperature regime in which an increasing fraction of the \dof is captured by the emerging order (such that $\Lambda_{\text{exp}}(T) \to 0$ as $T \to 0$).
At the transition, a macroscopic fraction of the \dof suddenly becomes ordered and therefore abruptly exits the noise-dominated regime.
The \snr boundary $\Lambda_{\text{exp}}(T)$ decreases in the low-temperature regime until it vanishes as $T \to 0$.
In the high-temperature regime, the distribution stabilises as the typical scale of fluctuations tends towards a constant.

As a complementary benchmark, we compare the \gsa-based estimate of the \snr boundary with a more standard method in the literature, in the spirit of the approach of Bouchaud and Potters~\cite{Bouchaud1,Bouchaud2}.
We use the \kl divergence:\footnote{%
    Although alternative metrics, such as the Wasserstein distance, are in principle available (and theoretical connections between them exist~\cite{Belavkin_2018}), the \kl divergence is particularly advantageous here owing to its high sensitivity to the \snr.
    It does, however, strictly require both distributions to be defined over the same sample space $\mathds{X}$.
}
\begin{definition}{Kullback-Leibler Divergence}{KLdiv}
    For two (probability) distributions $\mathcal{P} \colon \mathds{X} \to [0,1]$ and $\mathcal{Q} \colon \mathds{X} \to [0,1]$ defined on the same probability space $\Omega$, the \kl divergence is defined as:
    \begin{equation}
        \mathrm{D}_{\text{KL}}(\mathcal{P} \lVert \mathcal{Q})
        =
        \sum_{x\in \mathds{X}}\, \mathcal{P}(x)\,\ln \left( \frac{\mathcal{P}(x)}{\mathcal{Q}(x)}\right).
        \label{eq:KLdiv}
    \end{equation}
    From an information-theoretic perspective, the \kl divergence quantifies the additional information required to encode data when using an approximating distribution $\mathcal{Q}$ instead of the true distribution $\mathcal{P}$.
\end{definition}
In this work, $\mathcal{P}$ represents the analytic momentum distribution induced by the asymptotic \mpdistr law (i.e.\ the distribution of the generalised momentum $p$ obtained by inverting the eigenvalue-momentum correspondence $\lambda = (p^2+m^2)^{-1}$), while $\mathcal{Q}$ is the empirical momentum density $\rho_G$ defined in~\eqref{eq:empirical_momentum_density}.
To compare the empirical distribution with the baseline theoretical model, we minimise \eqref{eq:KLdiv}.\footnote{%
    Strictly speaking, \kl is not a true metric, as it lacks symmetry and does not satisfy the triangle inequality.
}
The upper edge of the empirical momentum distribution $\rho_G$ defines an estimator for $\lambda_c(T)$, whose behaviour is shown in the bottom panel of \Cref{fig6} and closely matches that derived from the \gsa.
In the same manner, we extract a measurement of the critical temperature $T_c^{(KL)}(b)$ for different values of $b$.
\Cref{fig7} summarises our findings, comparing the three estimates: $T_c^{(\text{exp})}(b)$, $T_c^{(\text{\gsa})}(b)$, and $T_c^{(\text{\kl})}(b)$.

\begin{figure}[t]
    \centering
    \includegraphics[width=0.7\textwidth]{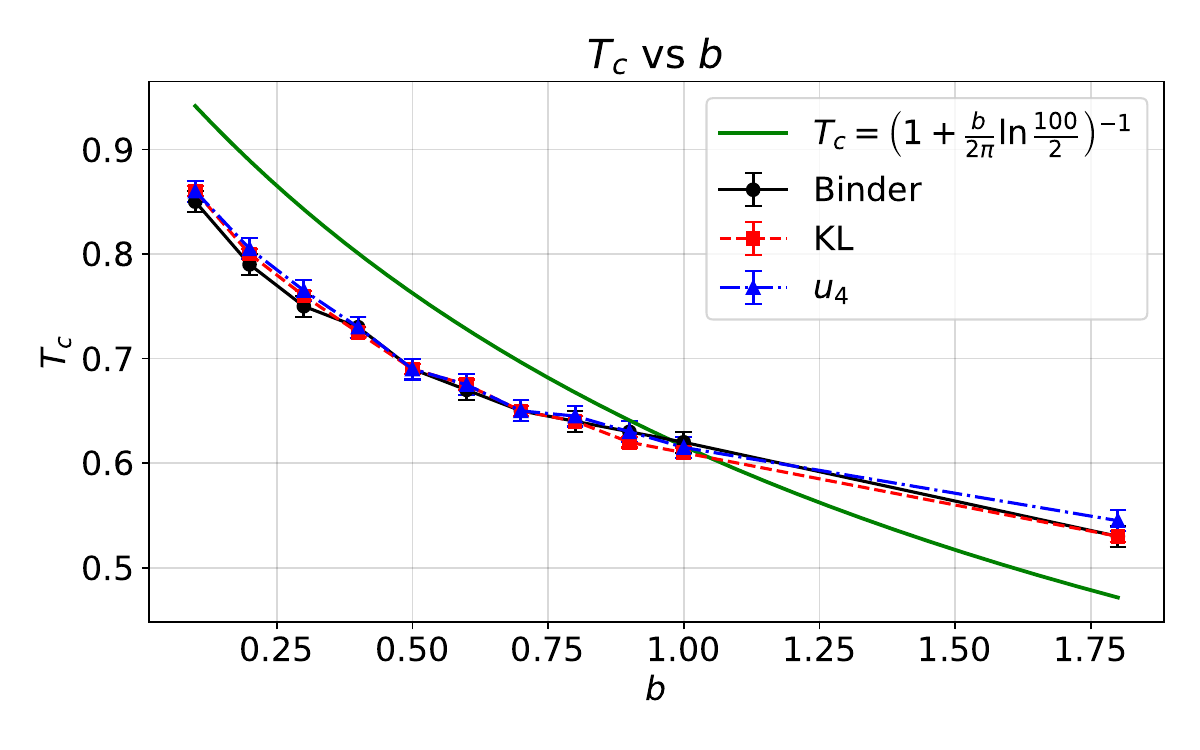}
    \caption{%
        Comparison of the critical temperatures $T_c(b)$ obtained by the various methods, for different values of $b$ and $N=100$.
        The final point at $b \approx 1.8$, corresponding to the \enquote{Ising-like} case, is in good agreement with Onsager's temperature ($\approx 0.57$).
        The mean-field approximation obtained through the Hartree method is given by the green curve (see \Cref{sec:app1}).
    }\label{fig7}
\end{figure}

By formulating these transitions as an anomaly-detection task and using \gsa to perform the detection, we have empirically shown that the \snr boundary proposed by the \gsa, in its most conservative version, aligns with experimental observations.
It competes with standard methods, such as those based on the minimisation of the \kl divergence, while offering a far more readable signature of the transition.
As shown in the case of the Ising model in the previous subsection, this qualitative assessment can be translated into quantitative results.
These findings establish the \gsa as a robust and reliable method, with a precision at least comparable to established techniques in the literature for the estimation of the critical temperature.
The results also validate the method against its original objective, making the \gsa a high-precision approach for anomaly detection in low-\snr regimes.
From a physics perspective, the \gsa is particularly interesting because it replaces a non-perturbative out-of-equilibrium study with an analysis of the \rg flow behaviour around the Gaussian fixed point (within the perturbative regime), associated with an effective equilibrium theory.
This correspondence offers a direct practical advantage over the Binder method, as it only requires varying the temperature rather than the system size.
As a consequence, the \gsa could be especially relevant for systems that lack self-averaging (i.e.\ systems in which sample-to-sample fluctuations do not vanish in the thermodynamic limit), notably those with quenched disorder such as spin glasses~\cite{de2006random,Bouchaud4,lahoche2024intriguing}.

\subsubsection{Ising Model}

\begin{figure}[t]
    \centering
    \includegraphics[width=0.49\textwidth]{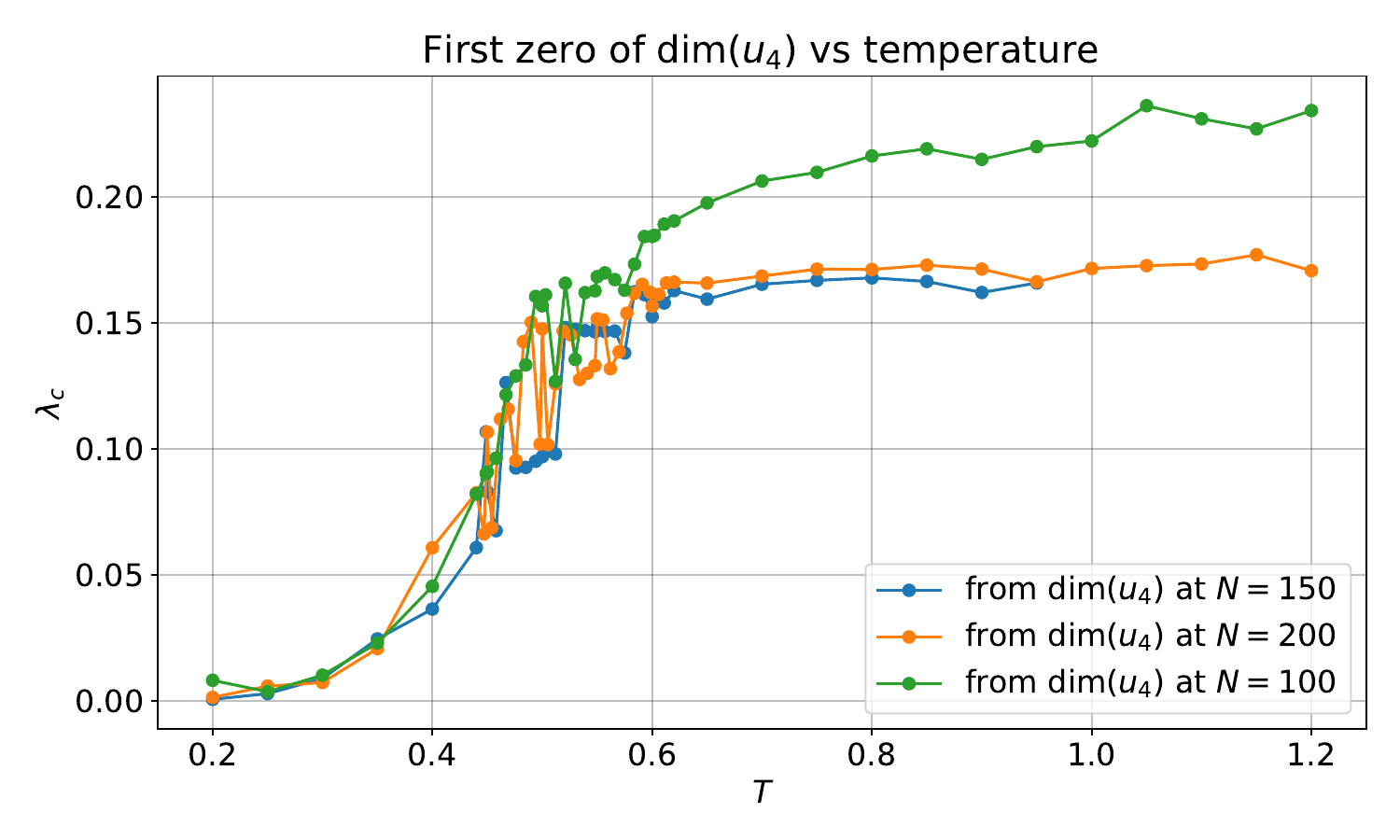}
    \hfil
    \includegraphics[width=0.49\textwidth]{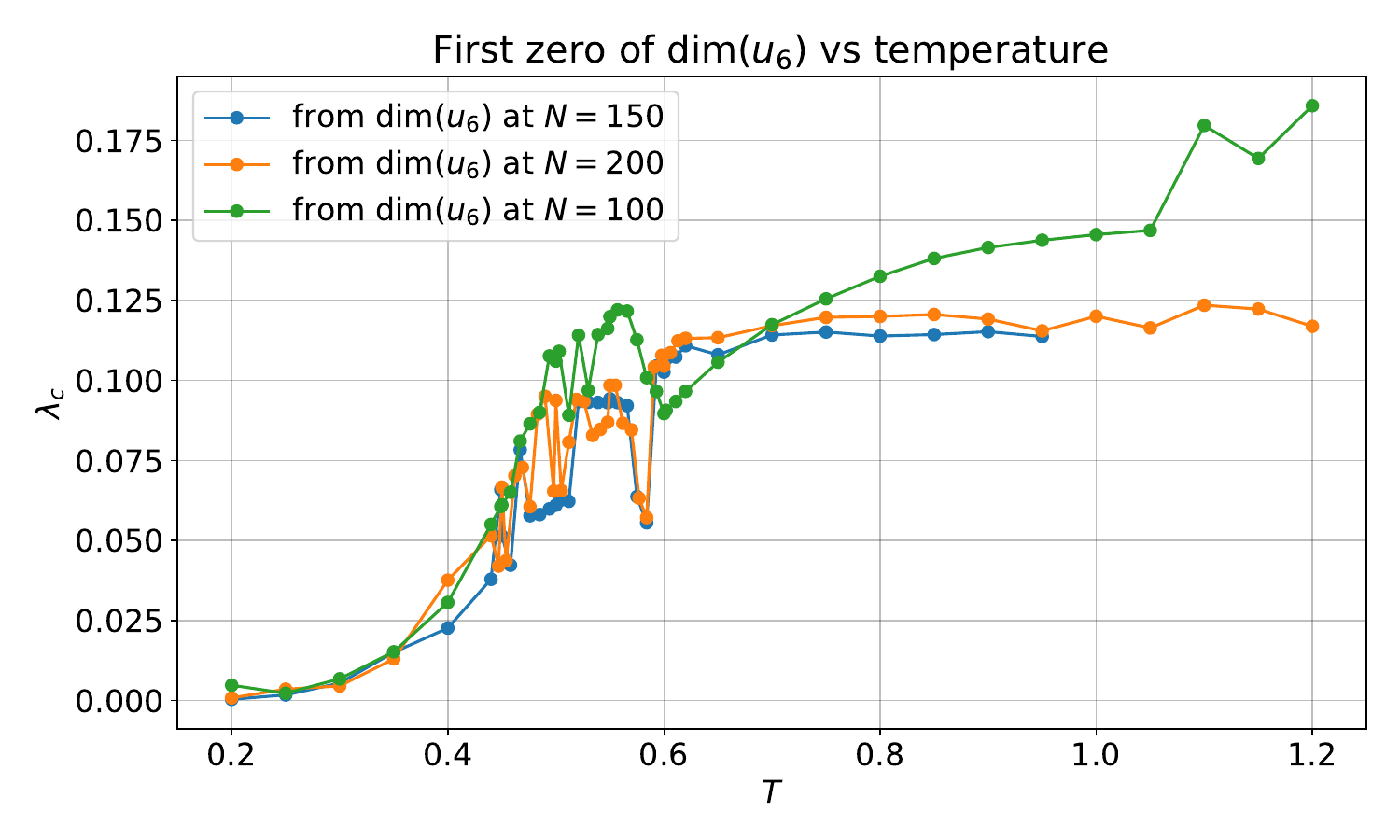}
    \hfil
    \includegraphics[width=0.49\textwidth]{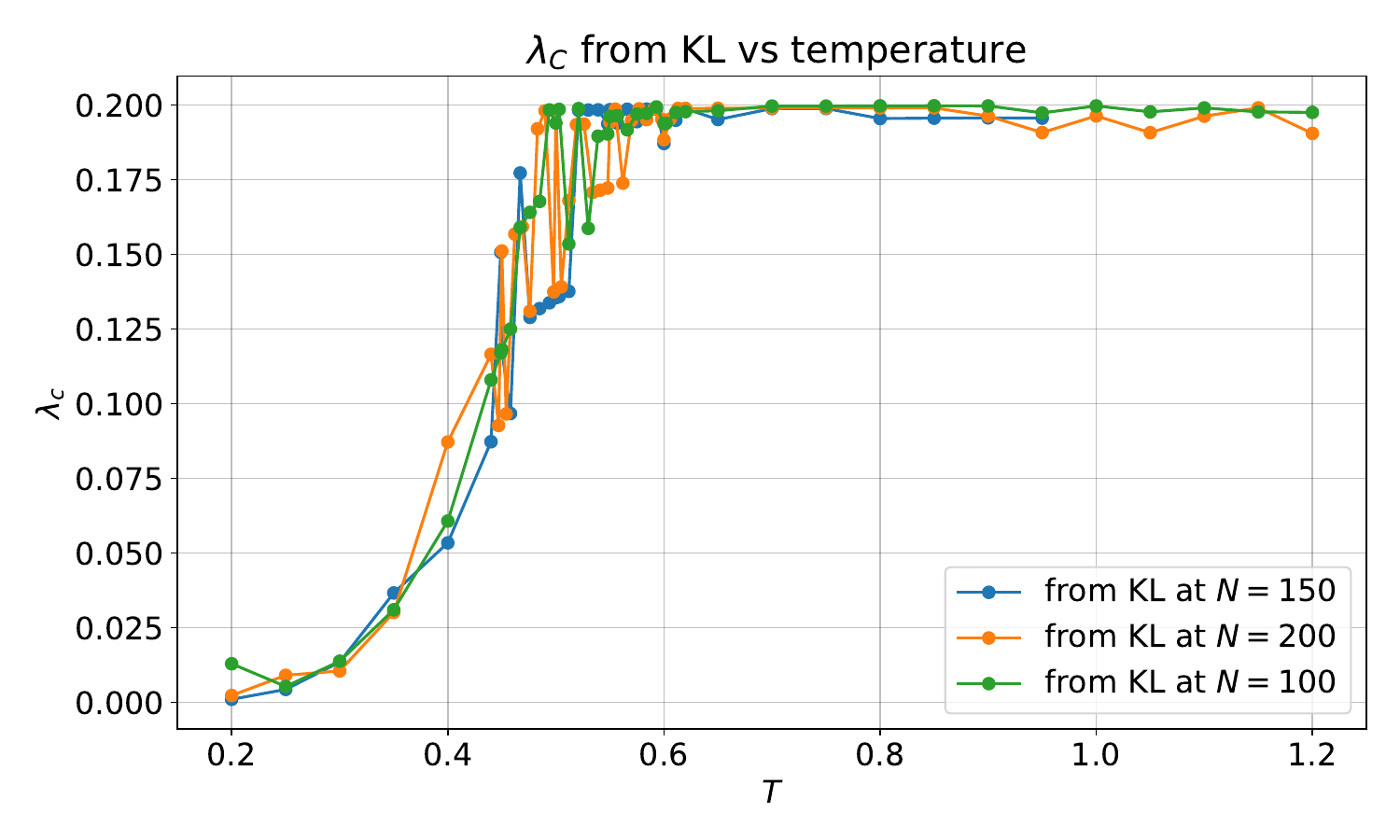}
    \caption{%
        Position of $\lambda_c$ as a function of $T$ with different methods: $\dim_{\tau}(u_4)$ (top left), $\dim_{\tau}(u_6)$ (top right), and \kl divergence (bottom).
    }\label{fig44}
\end{figure}

In this subsubsection, we analyse the phase transition of the pure Ising model, corresponding to the parametrisation in \eqref{VIsing} as $\varepsilon \to \infty$, using a Monte Carlo method with $J=1/4$.
We no longer describe the dynamics of a continuous field but rather that of discrete variables, taking values $S_i = \pm 1$ at each lattice site, with the Hamiltonian given by \eqref{IsingModel} (details on the correspondence with Model~A are provided in \Cref{sec:app2}).
The experimental method differs from the previous subsubsection in that we use \emph{importance sampling} via the Metropolis--Hastings algorithm.
The algorithm consists of the following steps:
\begin{enumerate}
    \item \textbf{Initialization:} We start from a configuration where all spins are randomly distributed, mimicking a high-temperature regime.
    \item \textbf{Selection:} A site $i$ is chosen at random on the lattice.
    \item \textbf{Energy variation calculation ($\Delta E$):} We determine the energy change that would result from flipping the spin ($S_i \to -S_i$).
          The corresponding energy variation is:
          \begin{equation}
              \Delta E = E_{final} - E_{initial} = 2 J S_i \sum_{j \in \text{neighborhood}} S_j,
          \end{equation}
          where the sign is consistent with the ferromagnetic convention $H = -J \sum S_i S_j$ of \eqref{IsingModel}: a flip aligned with the local field (\mbox{$S_i \sum_j S_j > 0$}) raises the energy (\mbox{$\Delta E > 0$}) and is therefore disfavoured, while a flip against the local field is energetically favoured.
\end{enumerate}
The results are summarised in \Cref{fig44,fig45}.
The Binder cumulant method must be adapted to the discrete case by replacing the continuous field $\varphi$ of \eqref{eqLangevin} with the total magnetisation $M$ of a spin configuration.

The magnetisation for a given configuration is defined as:
\begin{equation}
    M \defeq \sum_{i=1}^{N^2} S_i
\end{equation}
and the Binder cumulant becomes:
\begin{equation}
    K \defeq 1-\frac{\langle M^4 \rangle}{3 \langle M^2 \rangle^2}.
\end{equation}
Here, to compute the moments, we make use of ergodicity, and, for the quantity $\mathcal{O}_k$, we consider an average over a large enough ensemble of $N_s$ realisations:
\begin{equation}
    \langle \mathcal{O} \rangle = \frac{1}{N_s} \sum_{k=1}^{N_s} \mathcal{O}_k.
\end{equation}
Typically, $N_s=10^4$ configurations around the transition.
This value has to be large enough because of the critical slowing down (see \Cref{sec:app1} and below).

The Binder cumulant method provides an estimate of the critical temperature in the interval $T_c \in [0.58, 0.60]$, corresponding to a deviation of $+2.2\%$ to $+5.8\%$ from Onsager's analytical prediction ($T_c \approx 0.5673$ for $J=1/4$, as derived in \eqref{IsingModel}).
The behaviour of the magnetisation yields a similar, if slightly less precise, result with central estimate $T_c \approx 0.59$ in the same interval.

The figures, which are comparable to those obtained for Model~A, clearly show three distinct regimes:
\begin{enumerate}
    \item A high-temperature regime for $T > 0.6$, where the curves form a plateau.
          This is a consequence of the chosen parametrisation of the potential: as $T \to \infty$, the distribution of the magnetisation tends toward a Gaussian regime in which the variance is fixed by the parametrisation and the fluctuations are statistically suppressed.
    \item A transition regime at $0.45 \le T < T_c$ marked by strong instability: the correlation length $\xi$ becomes so large that the noise--signal separation becomes numerically unstable.
    \item A low-temperature regime for $T < 0.45$, where the cutoff $\Lambda$ decreases because the signal (order) occupies the entire space and the \snr becomes low, with the majority of available degrees of freedom captured by the emerging order.
\end{enumerate}
The presence of oscillations and instability in the critical region is numerically expected.
It is primarily a consequence of the divergence of the correlation length and the associated critical slowing down.

\begin{figure}[t]
    \centering
    \includegraphics[width=0.7\textwidth]{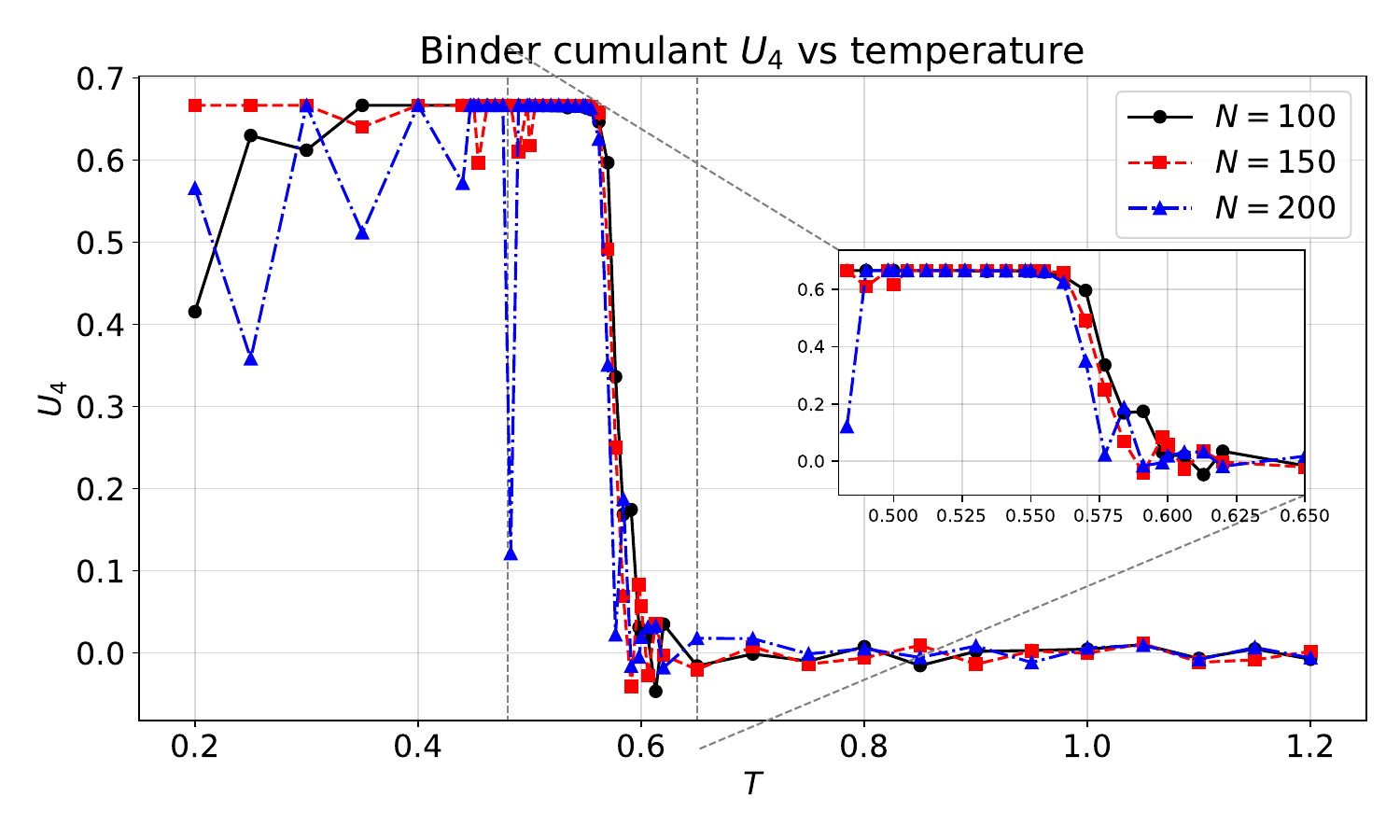} \\[0.25em]
    \includegraphics[width=0.7\textwidth]{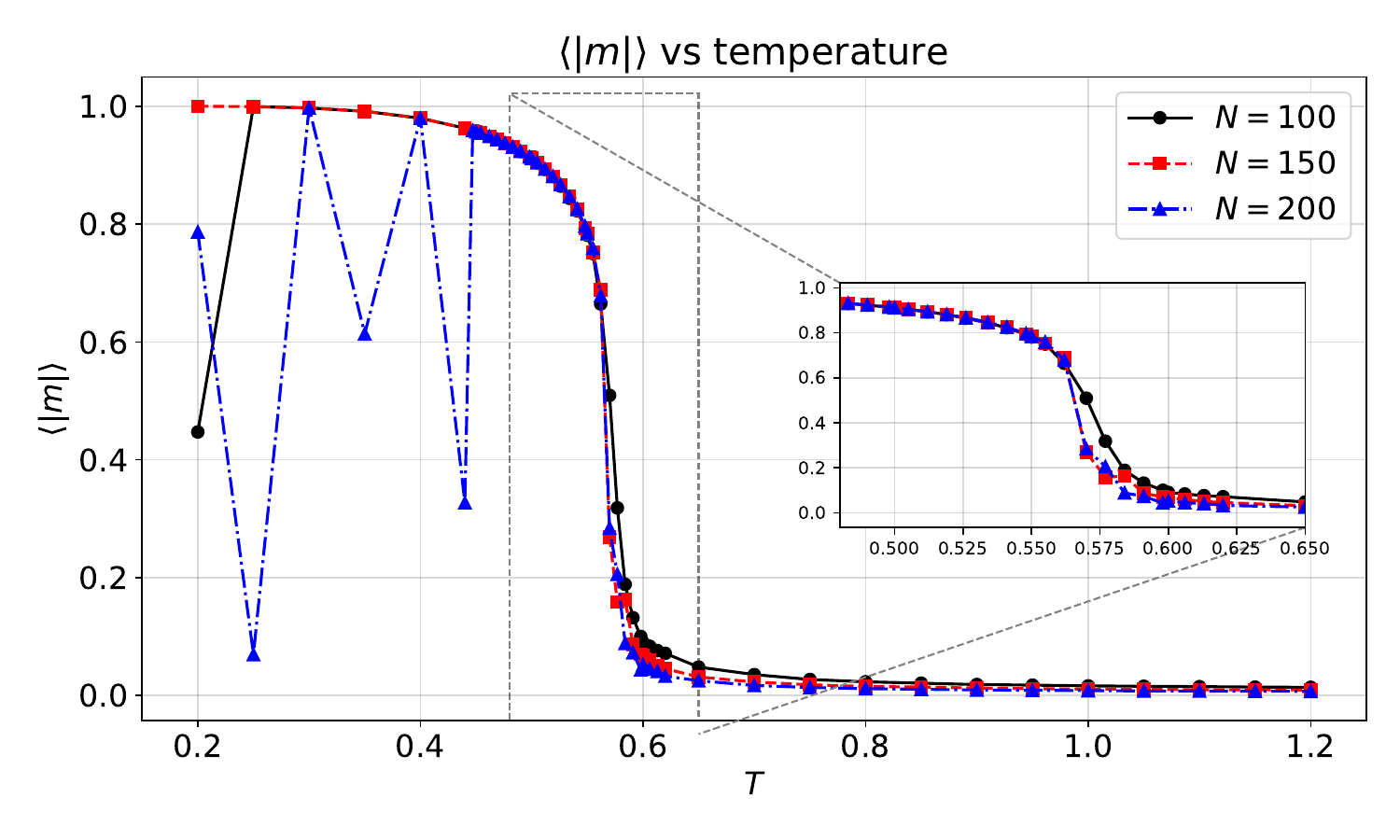}
    \caption{%
        (Top) Binder cumulant of the Ising model for different system sizes.
        (Bottom) Absolute magnetisation for the Ising model for various values of $N$.
    } \label{fig45}
\end{figure}

This critical slowing down causes instabilities to persist on the time scales accessible to the Monte Carlo simulation: due to strong correlations, the flip of a single spin is rejected by the Metropolis criterion with high probability, and the configuration evolves only through rare collective rearrangements.
At $T_c$ itself, the correlation length diverges, so the system hosts fluctuations on every length scale, and a single-spin flip can in principle trigger a cascade across the entire lattice: thermal noise and domain rearrangements become indistinguishable on any finite time scale.

Notably, the transition is significantly clearer via \gsa, specifically through the quartic dimension, but even more so through the sextic dimension: the transition toward the high-temperature plateau (where the sextic coupling reaches its marginality threshold and \mpdistr noise dominates) occurs around $T = 0.58 \pm 0.01$, which is in good agreement with the theoretical prediction.
This constitutes a stringent validation of \gsa: the canonical dimension of the sextic coupling detects the Ising transition to within approximately $2\%$ (the central value of the estimate, $0.58$, deviates from Onsager's $0.5673$ by $2.2\%$), simply by observing when the interaction couplings reach the marginality threshold.
The \kl method, by contrast, predicts a lower value, around $0.53$, corresponding to a deviation of approximately $6.6\%$ (or, equivalently, a factor of about three times worse than the sextic-criterion \gsa estimate at the central value).
This suggests that the \kl method is significantly more impacted by the discrete nature of the Ising transition than the \gsa approach.

Finally, numerical instabilities appear stronger in this case than in Model~A.
In the Ising model, this is primarily caused by the fact that the system can only overcome energy barriers through discrete jumps.
In contrast, the continuous nature of the Langevin dynamics in Model~A provides a smoother evolution than the discrete Monte Carlo spin-flip dynamics.
For Model~A, the results from the \kl method were comparable to those from \gsa.
The results presented here demonstrate that \gsa not only reproduces but also improves the prediction of a known physical behaviour, by using signal detection as the underlying mechanism.
Finally, the asymptotic scale (for high temperature) of the cutoff is of particular interest: based on the canonical dimension of the quartic coupling, the \gsa cutoff value remains below the one provided by the minimisation of the \kl method at all temperatures we have explored.

\section{The Formalisation of Generalised Scale Analysis}\label{Formalization}
This section formalises \gsa, refining the dimensional criteria introduced in \Cref{sec:gsa} and centred on the critical cut-off $\lambda_c$ of \Cref{Lambdac}.
The empirical \rg flow and the universal \mpdistr baseline are compared through a family of spectral distances, which act as order parameters quantifying the deviation of the empirical flow from the reference.
We begin with the \mpdistr \emph{distribution proxy} and the corresponding \emph{direct} and \emph{inverse} distances, and close with a discussion of how higher-order couplings serve as sensitive probes of \uv structure.

\subsection{The Marchenko--Pastur Distribution Proxy and Concordance Index}

We first define the \mpdistr \emph{distribution proxy}.
Let $\mu(\lambda)$ denote the spectral density of the \ecm, and let $\mu_{\Delta_{\text{phys}}}(\lambda)$ denote the bulk distribution depending on the discretisation scale $\Delta_{\text{phys}}$ introduced in \Cref{eq:time_step}.
Let $\lambda_{\pm}$ denote the empirical spectral edges.
As a conservative numerical estimate, we identify $\lambda_+$ with the first global minimum of $\dim_\tau(u_4)$ in the \ir, deliberately excluding local fluctuations.
\begin{definition}{Adherent Set of Marchenko--Pastur Distributions}{adherentMP}
    Let $\mathcal{D}$ denote the one-parameter family of \mpdistr distributions $\nu_{\sigma_*^2, q} \equiv \nu(q)$ of \Cref{thm:thMP}, parametrised by the ratio $q$ and the variance
    \begin{equation}
        \sigma_*^2 \defeq \frac{\lambda_+ - \lambda_-}{4\sqrt{q}},
    \end{equation}
    fixed by the empirical edges $\lambda_{\pm}$.
\end{definition}
The expression for $\sigma_*^2$ is a direct consequence of the MP edges $\lambda_\pm = \sigma^2 (1 \pm \sqrt{q})^2$ of \Cref{thm:thMP}: expanding the difference $\lambda_+ - \lambda_-$ yields $4\sigma^2\sqrt{q}$, which inverts to the formula above.
For each $q$, the resulting $\sigma_*^2(q)$ is the unique variance that places the bulk of the corresponding MP law exactly on the empirical support $[\lambda_-, \lambda_+]$, and the family $\mathcal{D}$ traces out all such distributions.

We define the direct Gaussian distance as follows:\footnote{%
    The term \enquote{Gaussian} refers to the fact that the distance is built from the canonical dimensions, themselves derived from the Gaussian power counting of \Cref{sec:lesson} and \eqref{canonicalu4}.
}
\begin{definition}{Direct Gaussian Distance}{directGaussianDistance}
    Let $\mu$ denote the empirical spectral density, and let $\nu \in \mathcal{D}$ be a reference \mpdistr distribution.
    The direct Gaussian distance between $\mu$ and $\nu$ is
    \begin{equation}
        G_\mathcal{D}\qty(\mu,\nu(q))
        \defeq
        \max_{\lambda \in (\lambda_-, \lambda_+)} \,
        \left| \eval{\dim_{\tau}\qty(u_4)}_{\mu(\lambda)} - \eval{\dim_{\tau}\qty(u_4)}_{\nu(\lambda)} \right|.
    \end{equation}
\end{definition}
The maximum is taken over the common spectral support, on which both $\mu$ and $\nu$ are defined.
The proxy is the member of $\mathcal{D}$ minimising this distance.
\begin{definition}{Marchenko--Pastur Distribution Proxy}{MPdistributionProxy}
    The \mpdistr distribution proxy $\nu_{*}(\lambda) \in \mathcal{D}$ is the analytical \mpdistr distribution satisfying
    \begin{equation}
        G(\mu,\nu_{*})
        =
        \min_{q} \,
        G_\mathcal{D}(\mu, \nu(q)).
    \end{equation}
\end{definition}
If the minimiser lies in the interior of the admissible range of $q$, the proxy satisfies the first-order necessary condition
\begin{equation}
    \eval{\frac{\mathrm{d}}{\mathrm{d} q}\, G_{\mathcal{D}}(\mu,\nu(q))}_{q=q_*} = 0,
    \label{eq:FOC}
\end{equation}
where $q_*$ denotes the value of $q$ corresponding to $\nu_*$.
The reduced distance $G(\mu, \nu_*)$ then quantifies how far the empirical distribution lies from the adherent set, and we call it the direct concordance index.
\begin{equation}
    \eta(\mu) \defeq  G(\mu,\nu_{*}).
    \label{eq:eta}
\end{equation}
The set $\mathcal{D}$ is, in topological terms, the closure of the family of \mpdistr distributions consistent with the empirical edges: it contains the limiting distributions at the boundary of the admissible range of $q$ (in particular the $\sigma^2 \to 0$ and $q \to 0$ limits), beyond which no \mpdistr law can reproduce the empirical support.
The concordance index $\eta(\mu)$ therefore measures the global agreement between the empirical and reference \rg flows, with $\eta(\mu) = 0$ when the empirical spectrum is itself of the \mpdistr form.

We finally define the direct absolute global adherence $\zeta(\mu)$ of $\mu$ as
\begin{equation}
    \zeta(\mu) \defeq
    \min_{\lambda \in (\lambda_-, \lambda_+)} \,
    \left|  \eval{\dim_{\tau}\qty(u_4)}_{\mu(\lambda)} - \eval{\dim_{\tau}\qty(u_4)}_{\nu_*(\lambda)} \right|.
    \label{eq:zeta}
\end{equation}
Whereas $\eta(\mu)$ measures the global deviation, $\zeta(\mu)$ quantifies the typical amplitude of local fluctuations of the empirical canonical dimension around the proxy, providing a single number characteristic of the bulk.
Both indices are nevertheless global order parameters: they do not single out the spectral region in which the signal is expected to be most visible.
As established in \Cref{sec:gsa}, fluctuations dominate the \uv regime, while signal-induced effects and large deviations are concentrated in the \ir.
The sign of these deviations moreover carries diagnostic information.
We therefore refine the definitions as follows.

\begin{definition}{Local Direct Concordance Index and Direct Relative Adherence}{localDirectConcordance}
    The local direct concordance index at scale $\lambda$, denoted $\eta_{<\lambda}(\mu)$, and the direct relative adherence $\zeta_{<\lambda}(\mu)$ of the distribution $\mu$ are defined by
    \begin{equation}
        \eta_{<\lambda}(\mu)
        \defeq
        \max_{\lambda' \in (\lambda, \lambda_+)} \,
        \left|  \eval{\dim_{\tau}\qty(u_4)}_{\mu(\lambda')} - \eval{\dim_{\tau}\qty(u_4)}_{\nu_*(\lambda')} \right|,
    \end{equation}
    \begin{equation}
        \zeta_{<\lambda}(\mu)
        \defeq
        \min_{\lambda' \in (\lambda, \lambda_+)} \,
        \left( \eval{\dim_{\tau}\qty(u_4)}_{\nu_*(\lambda')} - \eval{\dim_{\tau}\qty(u_4)}_{\mu(\lambda')} \right).
    \end{equation}
\end{definition}
The relative sign in $\zeta_{<\lambda}(\mu)$ is essential: the analysis of \Cref{sec:gsa} indicates that a positive value is a necessary condition for the presence of a signal within the spectrum.

\begin{figure}[t]
    \centering
    \includegraphics[width=0.95\textwidth]{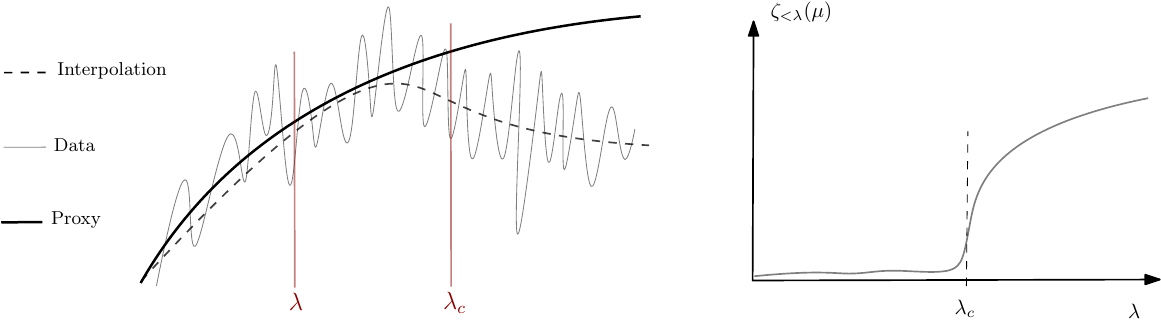}
    \caption{%
        Qualitative illustration of the behaviour of the direct relative adherence $\zeta_{<\lambda}(\mu)$.
    }\label{fig:illustration}
\end{figure}

\begin{figure}[t]
    \centering
    \begin{minipage}[b]{0.49\textwidth}
        \centering
        {$\beta = 0.21$} \\[0.25em]
        \includegraphics[height=0.25\textheight]{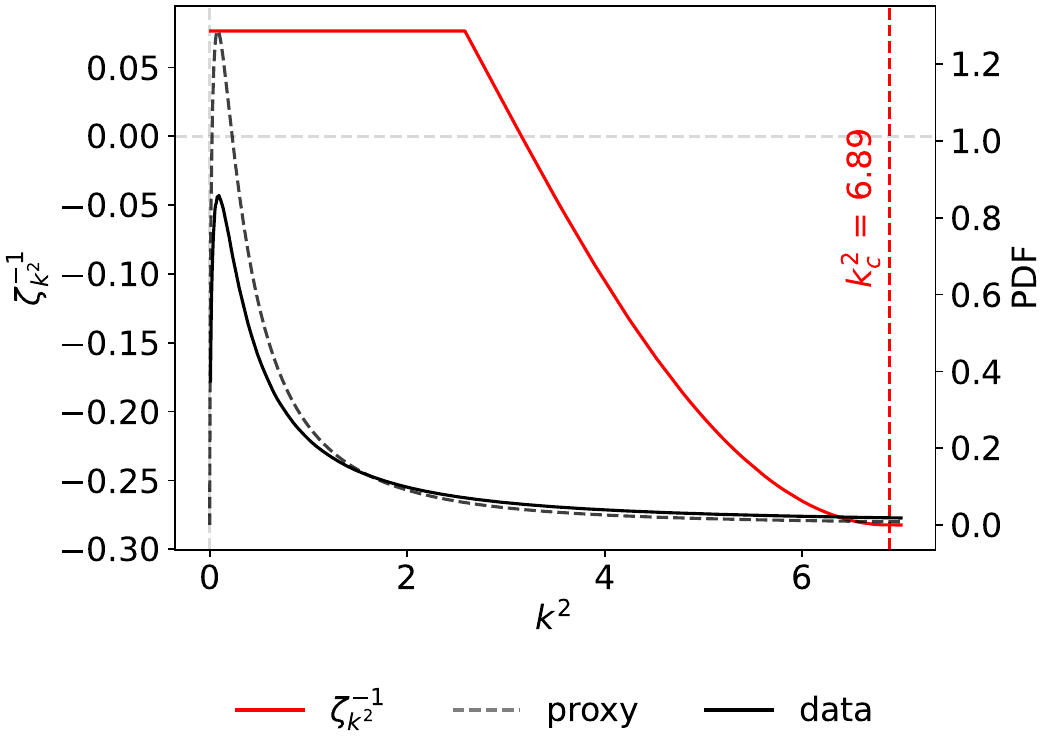}
    \end{minipage}
    \hfill
    \begin{minipage}[b]{0.49\textwidth}
        \centering
        {$\beta = 0.37$} \\[0.25em]
        \includegraphics[height=0.25\textheight]{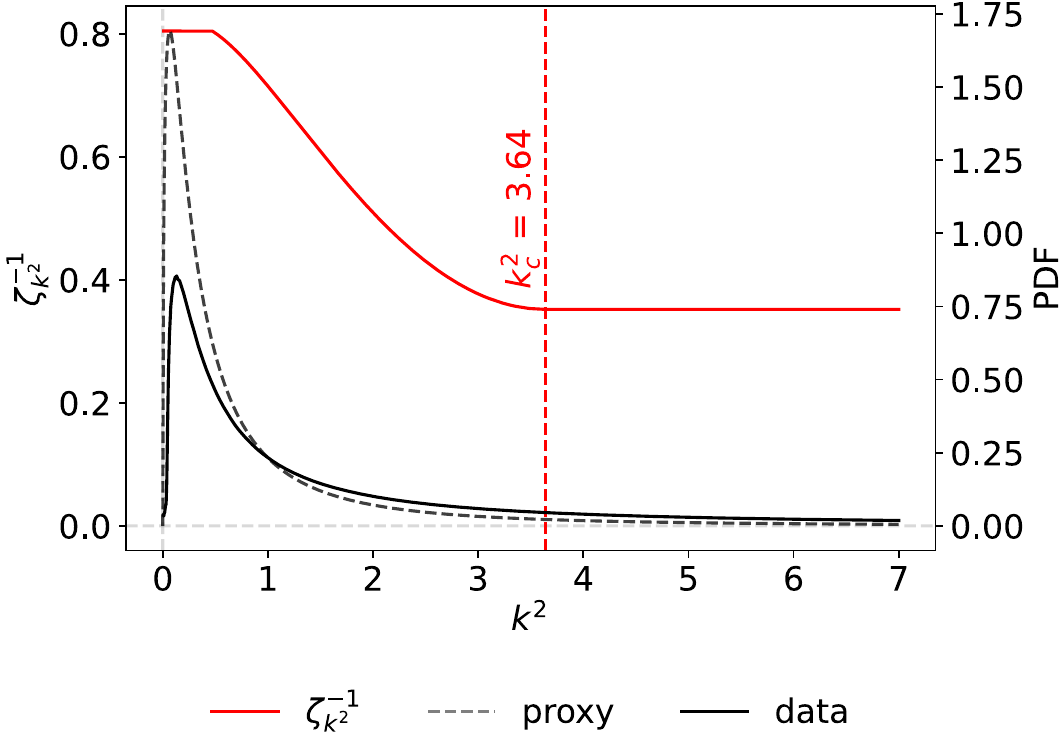}
    \end{minipage}
    \\[1em] 

    \begin{minipage}[b]{0.49\textwidth}
        \centering
        {$\beta = 0.50$} \\[0.25em]
        \includegraphics[height=0.25\textheight]{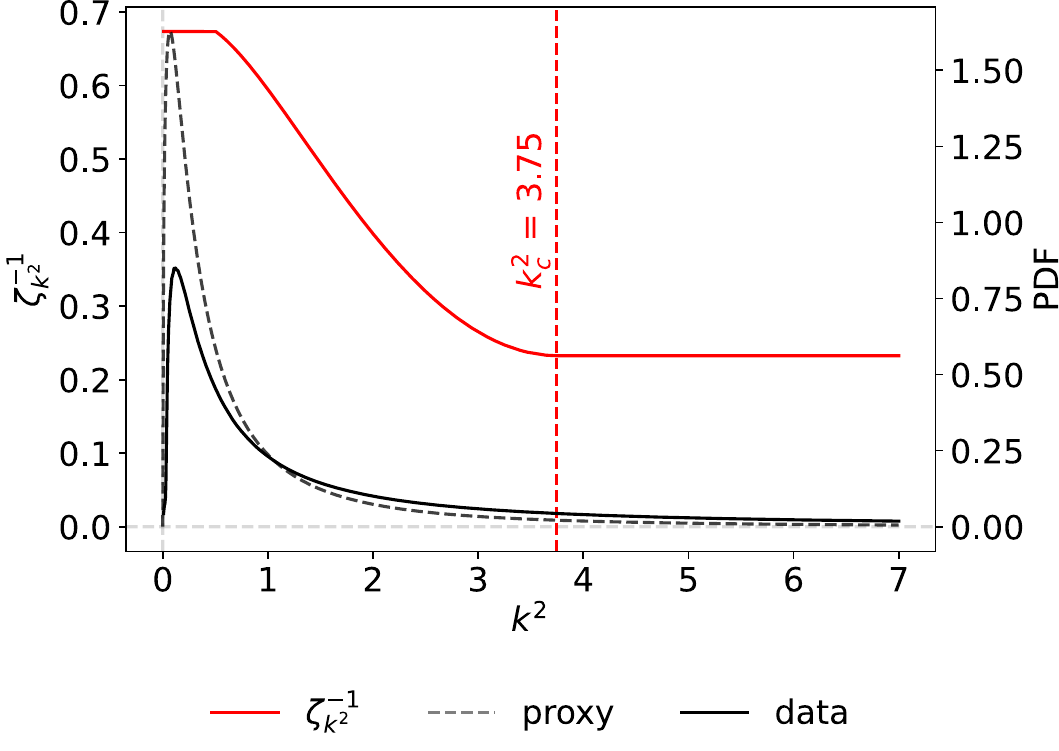}
    \end{minipage}

    \caption{%
        Local inverse adherence values (left axis, red curve) compared with spectral densities (right axis) across varying \snr values ($\beta$).
    }\label{fig:local_adherence}
\end{figure}

The local definitions above align with the anticipated behaviour of the statistical ensemble.
Numerical results show that fluctuations oscillate around the proxy distribution, a phenomenon visible at sufficiently high \uv scales (corresponding to small $\lambda$).
For $\lambda$ in this regime, $\zeta_{<\lambda}(\mu)$ remains essentially zero.
As the global trend drives the empirical spectrum away from the proxy, the probability of fluctuations crossing the proxy curve vanishes at a critical value $\lambda_c$, and $\zeta_{<\lambda}(\mu)$ becomes non-zero, provided the signal scale exceeds the typical fluctuation scale.
\Cref{fig:illustration} qualitatively illustrates this behaviour, showing that $\zeta_{<\lambda}(\mu)$ deviates significantly from zero for $\lambda > \lambda_c$.
Consequently, $\lambda_c$ serves as an approximation of the decoupling cut-off $\Lambda$, at which the empirical distribution begins to diverge from the proxy.
A more pragmatic approach, discussed in~\cite{RG3,RG4}, defines $\lambda_c^\prime$ at the point where $\dim_{\tau}(u_4) = 0$.
While these two values generally differ, both signify a substantial deviation from the MP ensemble.
We therefore adopt the conservative choice
\begin{equation}
    \Lambda = \min (\lambda_c,\lambda_c^\prime).
    \label{eq:Lambda}
\end{equation}

\subsection{Inverse Gaussian Distance and Inverse Adherence}

The inverse Gaussian distance operates in momentum space, the variable conjugate to the spectrum under the inference relation \eqref{inferenceZero}.
The transformation from the empirical spectral distribution $\mu$ to the momentum distribution $\rho$ erases the finite support $[\lambda_-, \lambda_+]$: in the continuous limit and under the shift convention of \Cref{LocalFieldTheory}, this interval is mapped onto $(0, +\infty)$.
Working on an unbounded support simplifies the comparison of distributions across different empirical realisations, since the divergence at the spectral edge no longer needs to be regularised by hand.

\begin{definition}{Inverse Gaussian Distance and Inverse Adherence}{inverseGaussianDistance}
    Let $\rho_1(p^2)$ and $\rho_2(p^2)$ denote two momentum distributions, each related to its spectral counterpart through \eqref{eq:empirical_momentum_density}.
    The local inverse Gaussian distance at scale $k^2$, denoted $g_{k^2}(\rho_1, \rho_2)$, and the local inverse adherence $\zeta_{k^2}^{-1}(\rho_1)$ are defined by
    \begin{equation}
        g_{k^2}(\rho_1,\rho_2)
        \defeq
        \max_{p^2 \in (0,k^2)} \,
        \left|  \eval{\dim_{\tau}\qty(u_4)}_{\rho_1(p^2)} - \eval{\dim_{\tau}\qty(u_4)}_{\rho_2(p^2)} \right|,
    \end{equation}
    \begin{equation}
        \zeta_{k^2}^{-1}(\rho_1)
        \defeq
        \min_{p^2 \in (0,k^2)} \,
        \left( \eval{\dim_{\tau}\qty(u_4)}_{\rho_*(p^2)} - \eval{\dim_{\tau}\qty(u_4)}_{\rho_1(p^2)} \right),
    \end{equation}
    where $\rho_*$ is the momentum distribution associated with the proxy of $\rho_1$, i.e.\ the inverse of the MP proxy of \Cref{def:MPdistributionProxy} transported to the momentum variable.
\end{definition}

\Cref{fig:local_adherence} illustrates the evolution of the critical momentum $k^2_c$, the momentum-space analogue of $\lambda_c$, across varying \snr $\beta$, using the realistic distribution of \Cref{fig:gianduja}.
The proxy $\rho_*$ is determined by matching the empirical momentum distribution over the numerically explored range $k^2 \in [0, 7]$; the upper limit $k^2 = 7$ is arbitrary but lies sufficiently far in the \uv to capture the relevant fluctuations.
For $\beta$ values associated with the strongest signals (see \Cref{fig:figplotgianduja}), the empirical distribution deviates from the \mpdistr proxy at high \uv scales, as shown by the pronounced local inverse adherence $\zeta^{-1}_{k^2}$.
For weaker \snr, the decoupling between the empirical distribution and the proxy occurs deeper in the \uv regime, and $\zeta^{-1}_{k^2}$ remains small, possibly indistinguishable from background statistical fluctuations.
The non-vanishing asymptotic behaviour of $\zeta_{k^2}^{-1}$ in the \uv is a consequence of the mapping from $[\lambda_-, \lambda_+]$ to $(0, +\infty)$: the support extends to arbitrarily large $k^2$, so the distance between the empirical distribution and its proxy accumulates over this domain.

In summary, the dimension of the quartic coupling $u_4$ provides a robust order parameter in the \ir.
The canonical dimensions of higher-order couplings, such as $u_6$, are more sensitive to local fluctuations.
Since the canonical dimensions depend on the logarithmic derivative of the spectral density ($\dd \ln \rho / \dd t$, see \Cref{canonicalu4}), they can become ill-defined or singular in regions with spectral gaps or abrupt density variations.
Higher-order couplings therefore function as more sensitive probes of \uv microstructures.
Real-world applications of the dimensional criterion will be discussed in \Cref{part4}, where different use cases will be extensively explored.

\section{Evidence of a Dimensional Phase Transition}\label{sec:dimensional_phase_transition}
\begin{figure}[t]
    \centering
    \includegraphics[width=0.7\textwidth]{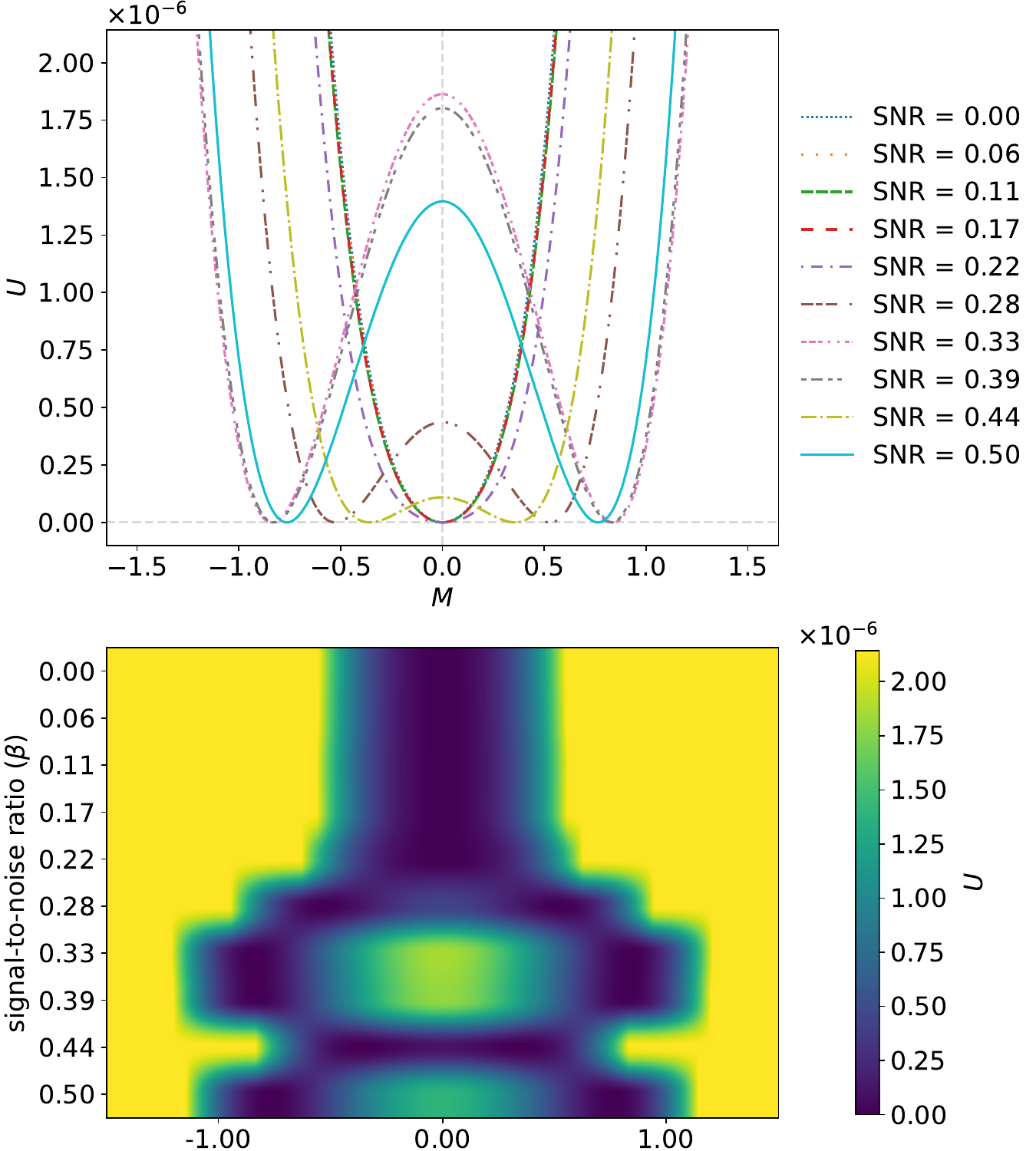}
    \caption{%
    Symmetry-breaking scenario at higher \snr ($\beta$) for the data set of \Cref{fig:gianduja}.
    The figure shows the \ir evolution of the effective potential, evaluated at the scale $k^2_{\text{\ir}}$, starting from the initial conditions at the mesoscopic scale $\Lambda$: $\bar{u}_2(\Lambda) = -8.24 \times 10^{-6}$, $\bar{u}_4(\Lambda) = 2.70 \times 10^{-6}$, and $\bar{u}_6(\Lambda) = 1.73 \times 10^{-6}$.
    }\label{fig:plotpot1}
\end{figure}

Throughout the previous sections, our analysis has focused on the vicinity of the Gaussian fixed point, the unique point at which all interaction couplings vanish and the theory reduces to a free, non-interacting form.
Because the canonical dimensions depend on the spectral density, and therefore on the scale, the system possesses no other global fixed point: a true \rg fixed point would require the couplings to become scale-invariant.
In this short concluding section, we probe the neighbourhood of the Gaussian fixed point in greater detail, where effects not captured by the perturbative analysis become accessible to the numerical flow without compromising the validity of the \lpa and the truncation at $\bar{P} = 3$ (see \Cref{flow1,flow2,flow3}).
This provides a numerical illustration of the dimensional phase transition that underlies \gsa as a dimensional mechanism for signal and anomaly detection in \dft.

We focus on \rg flow trajectories initialised near the Gaussian point, but sufficiently far from it that non-Gaussian interactions contribute to the dynamics.
Initial conditions are sampled at the mesoscopic scale $\Lambda$ using a Latin Hypercube scheme, generating $2.5 \times 10^3$ points in the domain $(\bar{u}_2, \bar{u}_4, \bar{u}_6) \in [-10^{-5}, 10^{-5}]^3$.
For each sampled initial condition, the flow equations are integrated numerically from $k^2 = \Lambda$ down to the infrared reference scale $k^2_{\text{\ir}}$.
We pay particular attention to the global morphology of the effective potential, which depends on the trajectory in a manner more robust to the specific choice of approximation than the individual coupling values.

\Cref{fig:plotpot1} shows the effective potential at $k^2_{\text{\ir}}$ as a function of the \snr $\beta$, with initial conditions chosen in the symmetry-restoration regime (the orange region for $D < 4$ of \Cref{fig_size_RS}).
As $\beta$ increases, the potential evolves from a single-well (symmetric) shape to a double-well shape, signalling the spontaneous breaking of the $\mathds{Z}_2$ symmetry, namely the sign-inversion symmetry $\varphi \to -\varphi$ of the order parameter.
\Cref{fig:plotpot2} tracks the symmetric phase volume, defined as the fraction of sampled initial conditions flowing to a $\mathds{Z}_2$-symmetric minimum, as a function of $\beta$.
This fraction exhibits a sharp drop, with maximal variations aligning with those of the canonical dimension $\dim_\tau(u_4)$ identified in \Cref{sec:gsa}.
The transition shown in \Cref{fig:plotpot1} is a direct consequence of the signal-induced deformation of the empirical eigenvalue distribution in the \ir.
Unlike standard spontaneous symmetry breaking in $\varphi^4$ models, typically driven by temperature, this transition is driven by a shift in the effective dimensionality of the system.
As the signal modifies the spectrum, the canonical dimensions of the couplings evolve, and the system is driven from a noise-dominated symmetric phase ($D \approx 3$) to a signal-dominated broken phase ($D > 4$).
We refer to this phenomenon as \emph{dimensional symmetry breaking}, which, to our knowledge, has not been previously identified in the literature.

\begin{figure}[t]
    \centering
    \includegraphics[width=0.7\textwidth]{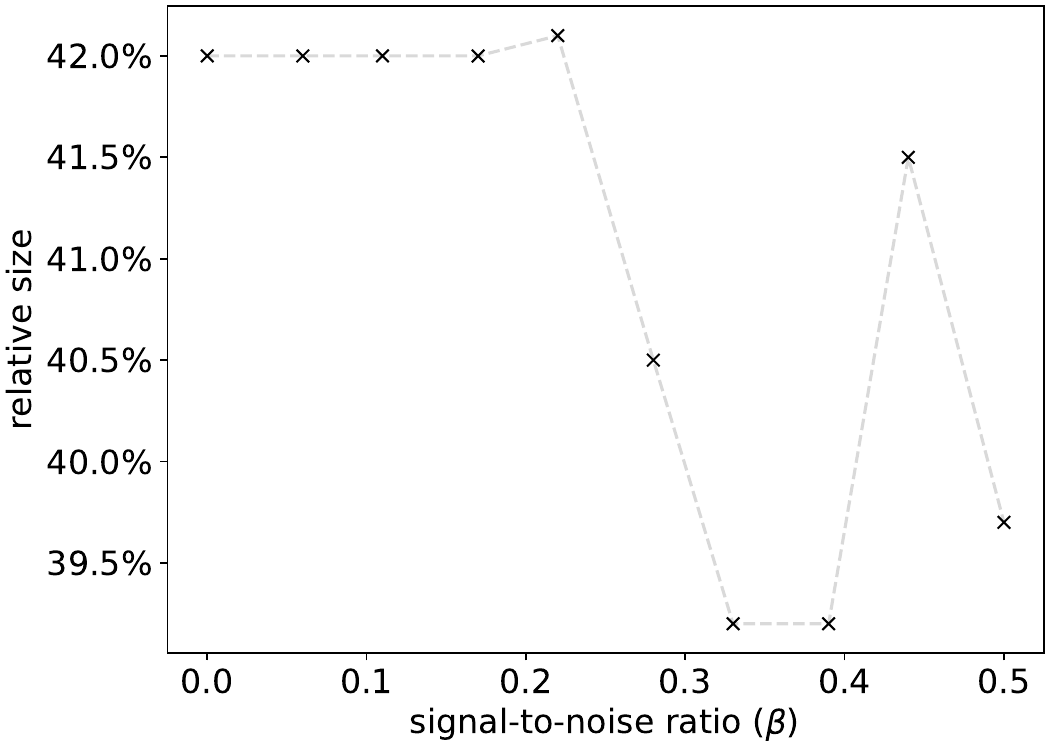}
    \caption{%
        Symmetric phase volume, relative to the total number of sampled initial conditions, as a function of the \snr ($\beta$).
    }\label{fig:plotpot2}
\end{figure}

\part{Real-World Applications and Advanced Techniques}\label{part4}

In the previous parts, we established the theoretical foundations of the \rg approach to signal detection and demonstrated its validity through the \gsa framework.
\Cref{part3} introduced the dimensional phase transition as a universal detection criterion, quantified the three detection thresholds $\beta_t$, $\beta_c$, and $\beta_O$, and validated the method against the exact Onsager solution of the two-dimensional Ising model.
Everything was verified under controlled conditions, primarily within the additive Gaussian noise model \eqref{eq:additive_model}.
Some natural questions arise: how does the framework behave when the Gaussian-noise idealisation is relaxed? What additional insight can it extract from data whose noise structure is itself rich and structured?

This last part addresses these questions along three complementary axes.
First, we extend the noise model beyond a single Gaussian source.
\Cref{Sec4-1} shows that multiple independent noise components generate a multimodal structure in the canonical dimension flow, providing an estimator for the number of distinct confounding sources.
The following section then tests the method against non-Gaussian and structured noises by blurring natural images through convolutions and injecting periodic interference patterns that mimic realistic degradation processes.
Second, we move from static to time-dependent data: using the coarsening dynamics of the two-dimensional Ising model, we measure the dynamic critical exponent $z$ through the spectral scaling of the cutoff index $\lambda_c(t)$, demonstrating that the \gsa naturally extends to temporal correlations.
Third, we confront the method with real-world sensor data from a hyperspectral image of the Martian surface acquired by the \emph{Compact Reconnaissance Imaging Spectrometer for Mars} (\textsc{crism}) instrument, where residual spectral correlations invisible to \pca-based benchmarks are detected in the post-\pca regime and shown to distinguish mineral classes with markedly different spectral signatures.

Each of these extensions validates the framework from a distinct perspective.
For the physicist, they confirm that the dimensional phase transition is not an artefact of the Gaussian-noise idealisation but survives under correlated, non-Gaussian, and temporally varying backgrounds.
For the data scientist, they provide evidence that the method operates reliably under the imperfect conditions of real-world data, where noise is rarely white, often structured, and accompanied by multiple confounding sources.
Together, the following sections demonstrate that the \rg approach to signal detection is not merely an elegant theoretical construction, but a practical tool whose domain of applicability extends well beyond the controlled setting in which it was originally formulated.

\section{Independent Noise Components}\label{Sec4-1}
Following the analysis of the cyclic phenomenon in \Cref{detectionTh}, we develop the corresponding interpretation in detail.
\Cref{fig:can_dim_mnist} shows the canonical dimensions evaluated at the \ir scale $k^2_{\text{\ir}}$ for two representative data sets: a realistic image and handwritten digits.
Two distinct regimes appear, one for each data set.
\begin{enumerate}
    \item For the handwritten digits, our results are consistent with the findings of \Cref{detectionTh}: the dimension of $u_4$ does not vanish but instead oscillates around the analytical value supplied by the \mpdistr proxy.
          By the detection criterion of \Cref{sec:gsa}, this is the signature of an absence of structured signal inside the bulk of the spectrum, a conclusion supported independently by the statistics of the eigenvectors discussed in \Cref{sec:gsa}.
          The isolated spikes, however, carry most of the information, while a faint residual signal remains visible in the bulk distribution.
    \item For the realistic image, by contrast, the canonical dimensions exhibit a richer dynamics: once $\beta$ exceeds the initial threshold $\beta_O$, the dimensions undergo a sequence of irregular cycles, climbing to successive local maxima before falling back, with the cycles demarcated by the set of thresholds $\{\beta_O^{(1)}, \beta_O^{(2)}, \dots, \beta_O^{(M_0)}\}$.
          The oscillations persist until a critical \snr $\beta_L$ is reached, beyond which the canonical dimensions resume a regular oscillation around the analytical values provided by the \mpdistr distribution proxy.
          At this stage the bulk of the matrix is well described by the proxy alone, and the original Gaussian matrix $Z$ has effectively decoupled from the signal.
\end{enumerate}

\begin{figure}[t]
    \centering
    \includegraphics[width=0.49\textwidth]{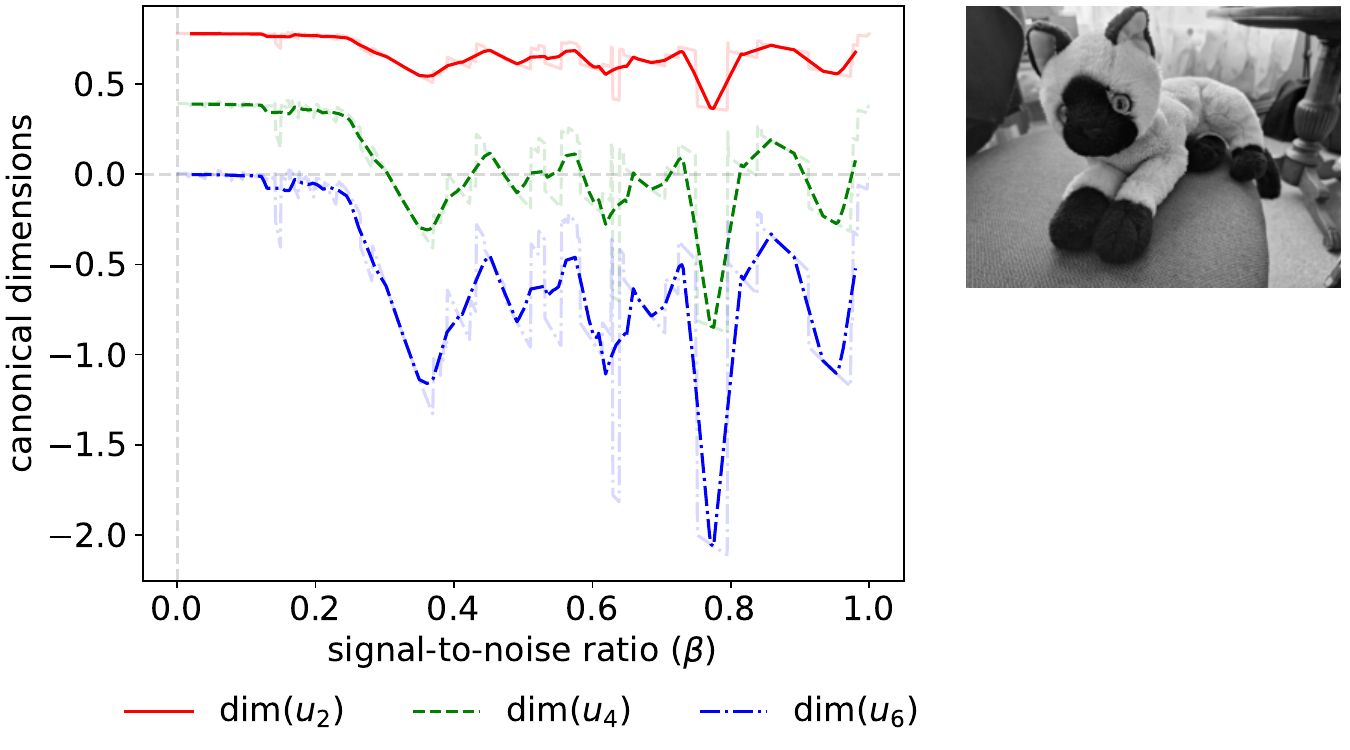}
    \hfil
    \includegraphics[width=0.49\textwidth]{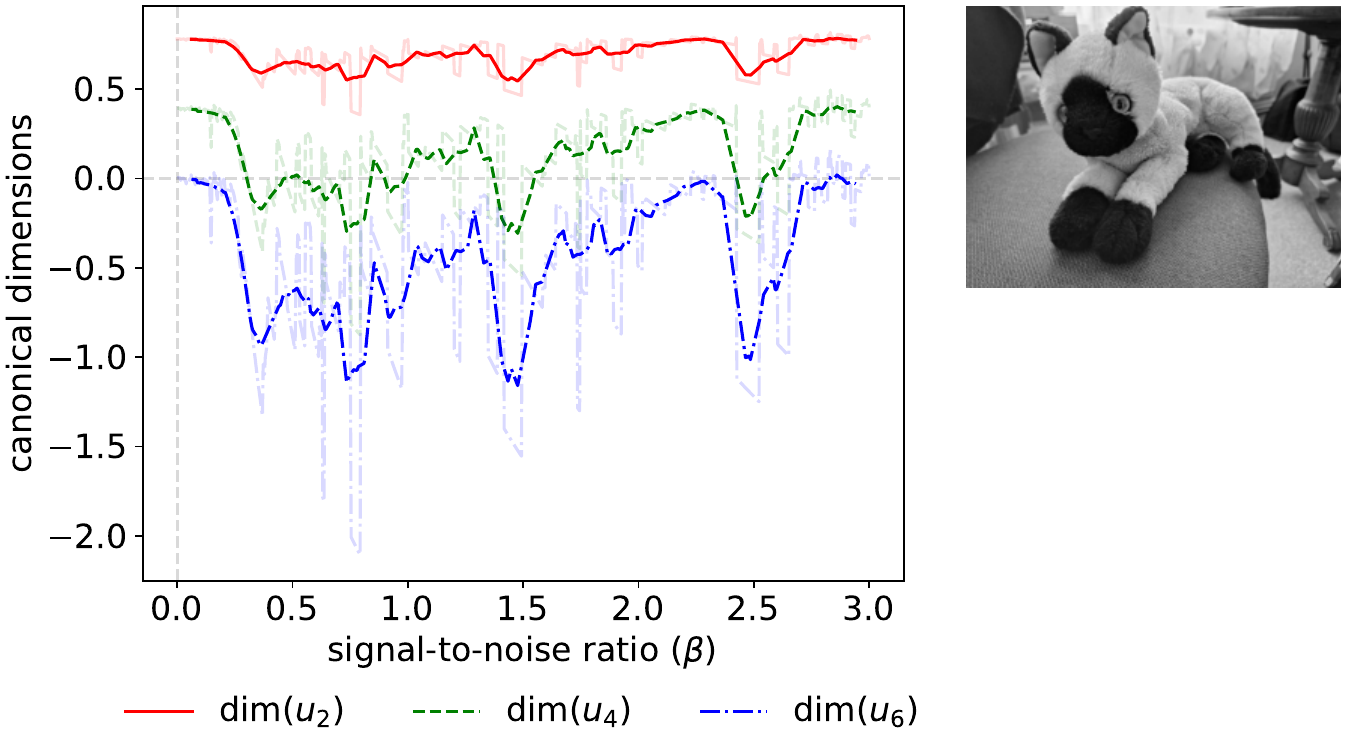}\\[0.25em]
    \includegraphics[width=0.49\textwidth]{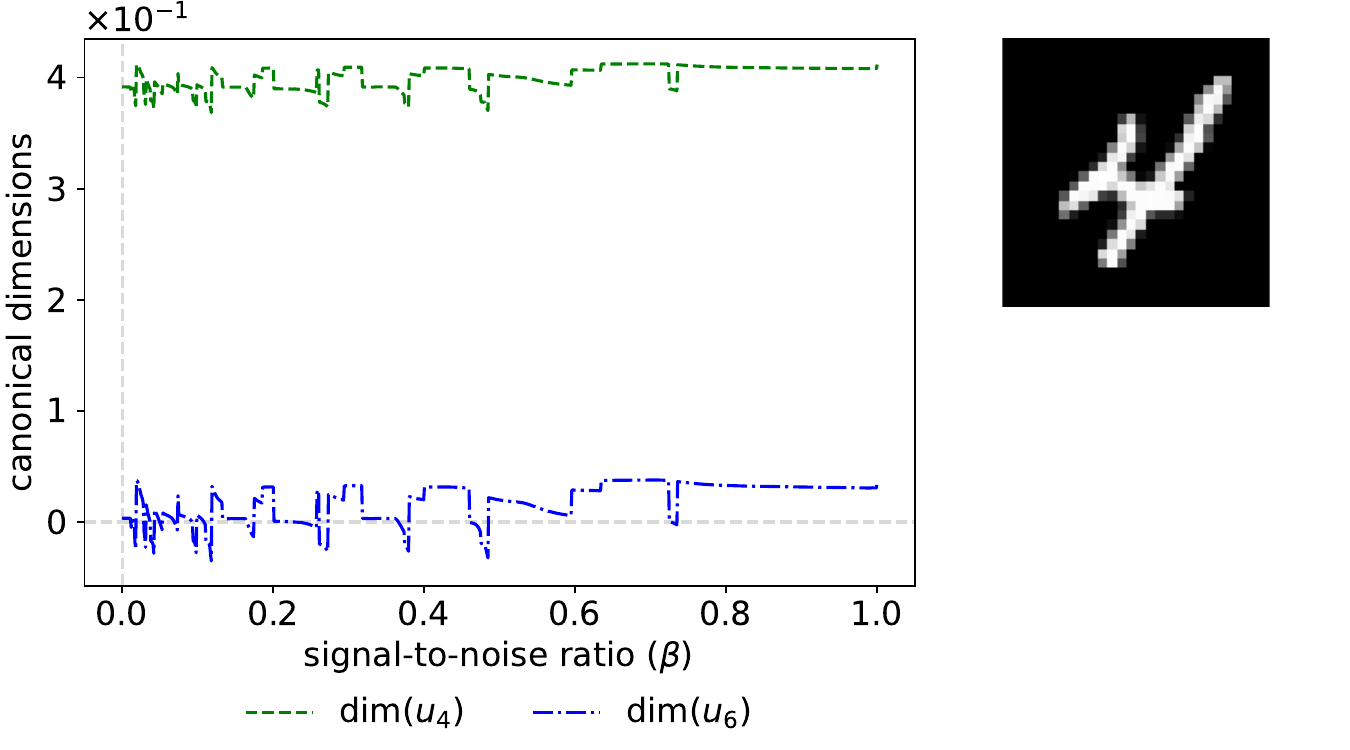}
    \hfil
    \includegraphics[width=0.49\textwidth]{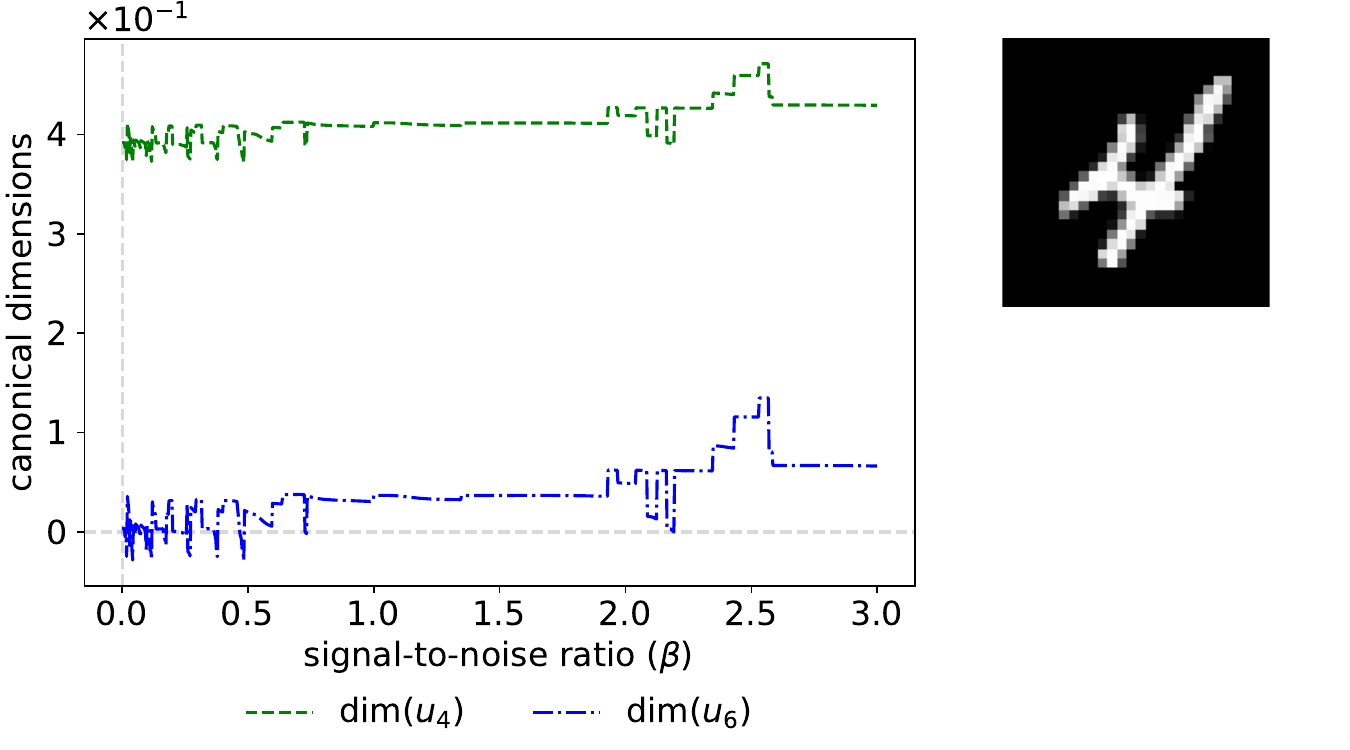}
    \caption{%
    Canonical dimensions at the \ir scale $k^2_{\text{\ir}}$ for a realistic image (top row) and a handwritten digit (bottom row).
    The left column covers $\beta \in [0, 1]$ and the right column $\beta \in [0, 3]$, illustrating the gradual decoupling of successive noise components with increasing \snr for the realistic image, and the absence of a comparable effect for the digit.
    }\label{fig:can_dim_mnist}
\end{figure}

These two regimes admit a common interpretation in terms of the noise distribution intrinsic to the data.
Within the additive model~\eqref{eq:additive_model}, we decompose the signal $S$ as
\begin{equation}
    S = S_0 + \sum\limits_{i = 1}^M \tilde{S}_i\qty(\omega_i).
    \label{eq:signal_decomp}
\end{equation}
The decomposition organises the constituents of the signal hierarchically: $S_0$ denotes the primary signal, which may combine discrete spectral spikes and a continuous part of the eigenvalue distribution, while the terms $\{\tilde{S}_i\}_{i=1, 2, \dots, M}$ describe systematic effects such as the sensor response, instrumental uncertainties, and background gradients.
The parameters $\omega_i$ act as confounders, quantifying the influence of these extraneous sources on the signal.
Although the components $\{ \tilde{S}_i \}_{i=1, 2, \dots, M}$ are usually classified as \enquote{noise}, they may follow arbitrary distributions and exhibit non-trivial interactions that are not captured by the simple additive form~\eqref{eq:additive_model}.
They are detectable by the \rg framework, since they represent a measurable departure from the reference \mpdistr background, and we identify them as sources of a \enquote{different type of signal}.
Throughout this section, we use the term \enquote{noise} loosely to include both the background distribution and the confounding sources.

Substituting the decomposition~\eqref{eq:signal_decomp} into the additive model~\eqref{eq:additive_model} yields a refined extensive-rank model
\begin{equation}
    Y
    =
    \beta S_0 + \qty(Z + \beta \sum\limits_{i = 1}^M \tilde{S}_i\qty(\omega_i))
    =
    \beta S_0 + \tilde{Z}_M\qty(\beta),
\end{equation}
in which the effective background distribution depends on both $\beta$ and the number of confounding components that have not yet decoupled.
As $\beta$ increases, the small eigenvalue distributions associated with the components $\{ \tilde{S}_i \}_{i=1, 2, \dots, M}$ begin to separate from the bulk, signalling the progressive decoupling of groups of \dof.
In the \rg language of \Cref{sec:dimensional_phase_transition}, each decoupling corresponds to a partial restoration of the \mpdistr universality class within a finite \snr window, and the sequence of such restorations reproduces the multimodal phase structure identified in~\cite{Landau2023}.
We therefore propose to estimate the number of independent noise components, $M$, by counting the distinct cycles, $M_0$, observed in the canonical dimension flow as a function of $\beta$.
Each cycle marks a subset of eigenvalues that exits the bulk, and hence the sequential decoupling of one \enquote{confounder} layer from the noise background.

\begin{figure}[t]
    \centering
    \begin{tikzpicture}
    \begin{axis}[
            width=10cm,
            height=7cm,
            axis lines=middle,
            axis line style={-Latex},
            xlabel={$\lambda$},
            ylabel={$\mu$},
            xlabel style={right, xshift=2pt},
            ylabel style={above, yshift=2pt},
            xmin=-0.25, xmax=3.6,
            ymin=-0.25, ymax=1.0,
            ticks=none,
            legend pos=north east,
            legend cell align=left,
            legend style={
                    at={(1.0,0.95)},
                    fill=white,
                    fill opacity=0.85,
                    draw=gray!50,
                    text opacity=1,
                    font=\footnotesize
                },
            smooth,
            thick,
            no marks,
            every axis plot/.append style={line width=1.2pt},
        ]
        \path[name path=axis] (axis cs:0,0) -- (axis cs:5,0);

        \addplot[
            name path=MP,
            domain=0.09:2.92,
            samples=200,
            color=black
        ]
        {sqrt((2.92-x)*(x - 0.09))/(2.0*pi*x*0.5)};
        \addlegendentry{main bulk};

        \addplot[
            name path=subMP1,
            domain=0.3:1.5,
            samples=200,
            color=red
        ]
        {sqrt((1.5-x)*(x - 0.3))/(2.0*pi*x*0.5)};
        \addlegendentry{$\beta \ll \beta_c$};

        \addplot[
            name path=subMP2,
            domain=0.3:2.88,
            samples=200,
            color=black!85!green,
            dashed
        ]
        {1.25*sqrt((2.88-x)*(x - 0.3))/(2.0*pi*x*0.75)};
        \addlegendentry{$\beta \simeq \beta_c$};

        \addplot[
            name path=subMP3,
            domain=0.3:3.3,
            samples=200,
            color=blue,
            dotted
        ]
        {1.50*sqrt((3.3-x)*(x - 0.3))/(2.0*pi*x*0.87)};
        \addlegendentry{$\beta > \beta_c$};

        \addplot[
            thick,
            color=black,
            fill=black!5
        ]
        fill between [
                of=MP and axis
            ];
    \end{axis}
\end{tikzpicture}
    \caption{%
        Schematic evolution of a sub-bulk eigenvalue distribution as a function of the \snr $\beta$.
        The three curves correspond to increasing values of $\beta$ relative to the critical threshold $\beta_c$ of the \mpdistr law: the sub-distribution widens (its variance grows) as $\beta$ increases, while remaining nested within the main bulk $Z$.
        The figure is meant as a visual guide to illustrate the process: sub-distributions of eigenvectors are, however, delocalised and diluted within the bulk.
    }\label{fig:mp_stretching}
\end{figure}
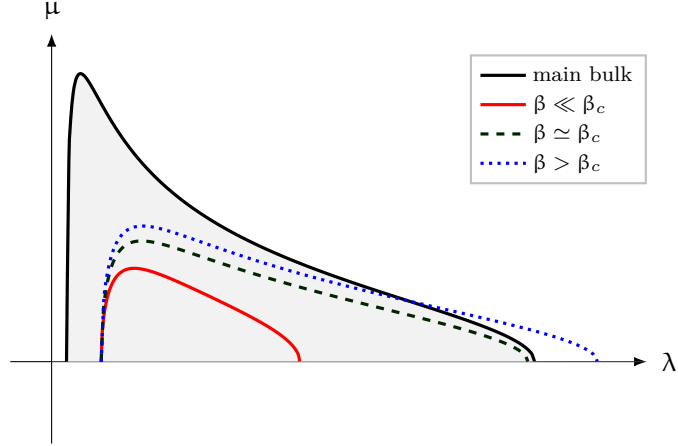

\Cref{fig:can_dim_mnist} already displays the two contributions identified in the decomposition: the primary signal $S_0$ and the confounding sources $\tilde{S}_i$.
When the analysis is restricted to the bulk eigenvalue distribution and the discrete spikes are excluded, the appearance of plateaus beyond a critical value of $\beta$ (most visible in the \textsc{mnist} image) marks the contribution of these confounding variables.
The dynamics of the canonical dimension corroborate the picture: as $\beta$ increases, the spikes of $S_0$ progressively exit the bulk, simplifying the additive model~\eqref{eq:additive_model} by removing the influence of the confounding sources one by one, in order of decreasing magnitude.
After the $n$-th decoupling, the residual eigenvalue distribution of the bulk is well approximated by an effective \mpdistr law with variance
\begin{equation}
    \sigma_{n+1:M}^2 \simeq \sum_{i=n+1}^M \sigma_i^2,
    \label{eq:residual_variance}
\end{equation}
where we assume, for simplicity, a weak correlation between the different components.\footnote{%
    Real-world examples notably display higher correlations and more pronounced cycles in the canonical dimensions (see, for instance, the realistic image panels in \Cref{fig:can_dim_mnist}).
    For the sake of simplicity, the discussion assumes weakly correlated sources, but is by no means limited by them.
}
The total variance of the original data matrix $X = \beta S + Z$ becomes
\begin{equation}
    \operatorname{Var}(X) = 1 + \beta^2 \sigma_{n+1:M}^2 > 1.
    \label{eq:varX}
\end{equation}
The unit contribution arises from the pure-noise matrix $Z$ of \eqref{eq:additive_model} (with $\sigma^2 = 1$ in our numerical experiments), while the term $\beta^2 \sigma_{n+1:M}^2$ collects the residual contributions of the surviving components $\{ \tilde{S}_i \}_{i=n+1, \dots, M}$.
The increasing $\beta$ also stretches the sub-distributions of the residual components inside the bulk $\tilde{Z}_M\qty(\beta)$, which is therefore no longer independent of the \snr.
The net effect is an increase of the internal variance of these sub-distributions, as illustrated schematically in \Cref{fig:mp_stretching}: the overall variance $\sigma_{n+1:M}^2$ therefore increases with $\beta$, even as the number of surviving components decreases.
As shown in \Cref{fig:can_dim_var}, this stretching produces a retarded descent of the canonical dimensions and accounts for the apparent increase observed when spikes depart the bulk.
The effect is most pronounced for the handwritten digit, which carries only weak signal remnants, whereas the realistic image tends to converge toward intermediate values because the correlations between its various sources are non-trivial.

\begin{figure}[t]
    \centering
    \includegraphics[width=0.45\textwidth]{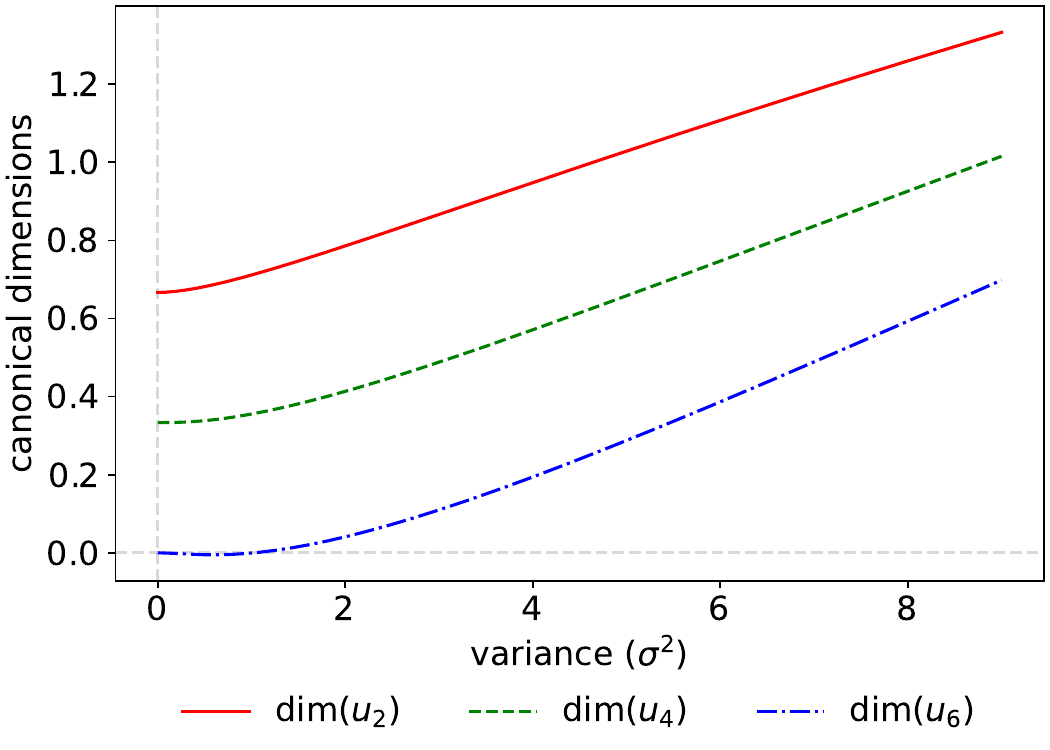}
    \quad
    \includegraphics[width=0.45\textwidth]{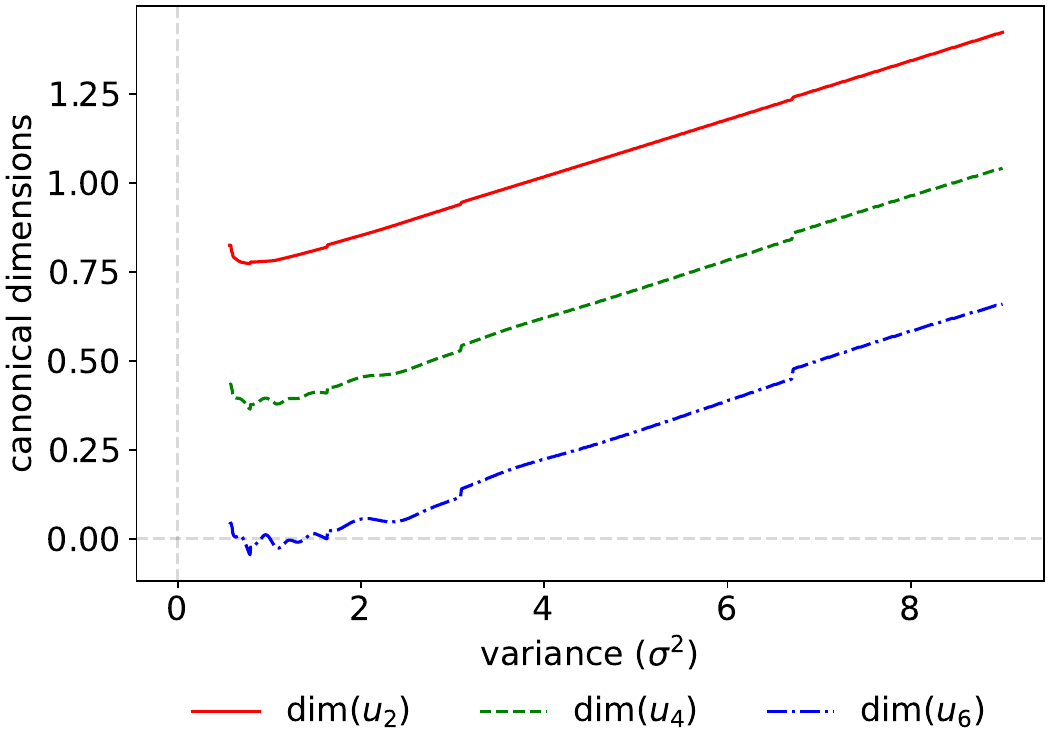}
    \caption{%
    Canonical dimensions at the scale $k^2_{\text{\ir}}$ as a function of the noise standard deviation $\sigma$ for the analytical \mpdistr law (left) and an empirical \mpdistr sample (right).
    The retarded descent of the empirical curves relative to the analytical baseline is a direct consequence of the stretching of the bulk sub-distributions described in the text.
    }\label{fig:can_dim_var}
\end{figure}

The mechanism, parametrised by the \snr $\beta$, can be summarised as follows.
\begin{enumerate}
    \item At low $\beta$, only the primary spikes of $S_0$ emerge from the empirical bulk, which remains closely approximated by the \mpdistr distribution; the signal is therefore only detectable if the bulk distribution itself is perturbed.
    \item At sufficiently high $\beta$, all spikes of $S_0$ exit the bulk, leaving only the eigenvalue distribution associated with the components $\tilde{S}_{i \ge 1}$.
    \item As $\beta$ continues to increase, the leading spikes of the $\tilde{S}_i$ components (typically associated with low-rank structures and hence limited in number) become detectable by \pca, eventually leaving only a weak residual intensity within the distribution of the original matrix $Z$.
    \item Over broad ranges of $\beta$ beyond the threshold $\beta_L$, these low-intensity distributions decouple from the bulk and remain undetectable.
          As a consequence, they influence only the global characteristics of the bulk distribution of $Z$, in particular its variance.
    \item Abrupt transitions in the behaviour of the canonical dimensions correspond to groups of spikes exiting the bulk distribution, as described in points~3 and~4.
\end{enumerate}
The mechanism is therefore linked to the presence of multiple independent noise sources within the data, each defined by a distinct exit threshold from the bulk distribution.
It generalises the notion of the \emph{bimodal connected phase} of~\cite{Landau2023} to a multimodal phase, in which each mode corresponds to an independent noise component.
In this scenario, the integer count $M_0 \le M$ provides a lower bound on the true number $M$ of intrinsic noise sources: $M_0$ is a conservative (non-overestimating) estimator because it counts only those cycles that are resolved by the canonical dimension flow, and any component whose decoupling falls below the numerical resolution is not counted.
Determining $M$ exactly remains challenging, as it would in principle require scanning all possible \snr values, which, at this stage, are not subject to an upper bound.

It is also important to note that the decoupling thresholds of the individual cycles are not a monotonically increasing function of $\beta$: the cycles in \Cref{fig:can_dim_mnist} have irregular widths, so that the position of each successive threshold along the $\beta$ axis oscillates rather than progressing uniformly.
\begin{remark}{Validity of Multi-Component Detection}{multidetection}
    The validity of this multi-component detection rests on a stability criterion.
    For each identified component $i$, the detection is meaningful only if the system remains within the basin of attraction of the reference universality class of \Cref{def:definitionMPclass} throughout the entire dimensional crossover interval $[\beta_t^{(i)}, \beta_O^{(i)}]$.
    Outside this basin, the canonical dimensions can no longer be interpreted as a deformation of the \mpdistr baseline, and the corresponding component must be discarded from the estimate of $M_0$.
\end{remark}

\section{Non-Gaussian Noise Sources}\label{sec:non_gaussian_noise}
In this section, we describe a \enquote{toy model} as a reference signal to illustrate the signal's manifestation.
This signal is subsequently corrupted by various types of noise, modelled after real-world scenarios, which allows for precise control over the \snr and the establishment of rigorous experimental conditions.
The image dimensions are $5000 \times 4000$ pixels.
For simplicity, the image has been converted to greyscale.
As an initial experiment, we consider a perturbation of the original image by convolving it with a Gaussian kernel.
By varying the kernel radius $R$, we induce a progressive blurring effect.
Specifically, we define the Gaussian kernel over a grid of size $(2R+1) \times (2R+1)$, where the value at pixel $(i, j)$ is given by:
\begin{equation}
    G_\sigma[i,j]
    \propto
    \exp\left(-\frac{i^2+j^2} {2\sigma^2}\right) \quad i,j \in \mathds{Z},
    \quad
    |i|,|j| \leq R
\end{equation}
where $\sigma \defeq R / 5$.
The resulting convolution is denoted as $(I * G)$, with the convolution operator defined as follows:
\begin{equation}
    (I * G)[i,j]
    \defeq
    \sum_{u=-R}^{R}\sum_{v=-R}^{R}
    I[i-u,\,j-v]\;
    G[u,v].
    \label{eq:Img_convolution}
\end{equation}
Note that to ensure this formula is well-defined, boundary conditions must be specified for the image.
In this study, we adopt periodic boundary conditions:
\begin{equation}
    I[i,j] \defeq I[i \bmod N_x,\; j \bmod N_y],
\end{equation}
Here, $N_x$ and $N_y$ are the number of sites in the $x$ and $y$ directions, respectively.
Once the boundary conditions are established, the convolution operator $I \mapsto I * G$ represents a linear, local, and translation-invariant mapping.
We now describe the data construction and preprocessing procedures.
The process begins with the transformation of data. First, according to \eqref{eq:Img_convolution}, we convolve the image with Gaussian kernels of size $R=20$, $R=200$, and $R=2000$ respectively.
Then we move on to sample generation, where we introduce additive white Gaussian noise to the transformed data as in \eqref{eq:additive_model}, characterised by a \snr $\beta \ge 0$.
Here, $Z \sim \mathcal{N}(0, 1)$ denotes an $N \times P$ matrix with \iid normally distributed entries, and $S \in \mathds{R}^{N \times P}$ represents the centred signal matrix (with $N=\num{5e3}$ and $P=\num{4e3}$).
The corresponding covariance matrix is subsequently constructed according to~\eqref{cov}.

\begin{figure}[t]
    \centering

    \begin{subfigure}{0.48\textwidth}
        \centering
        \includegraphics[width=\textwidth]{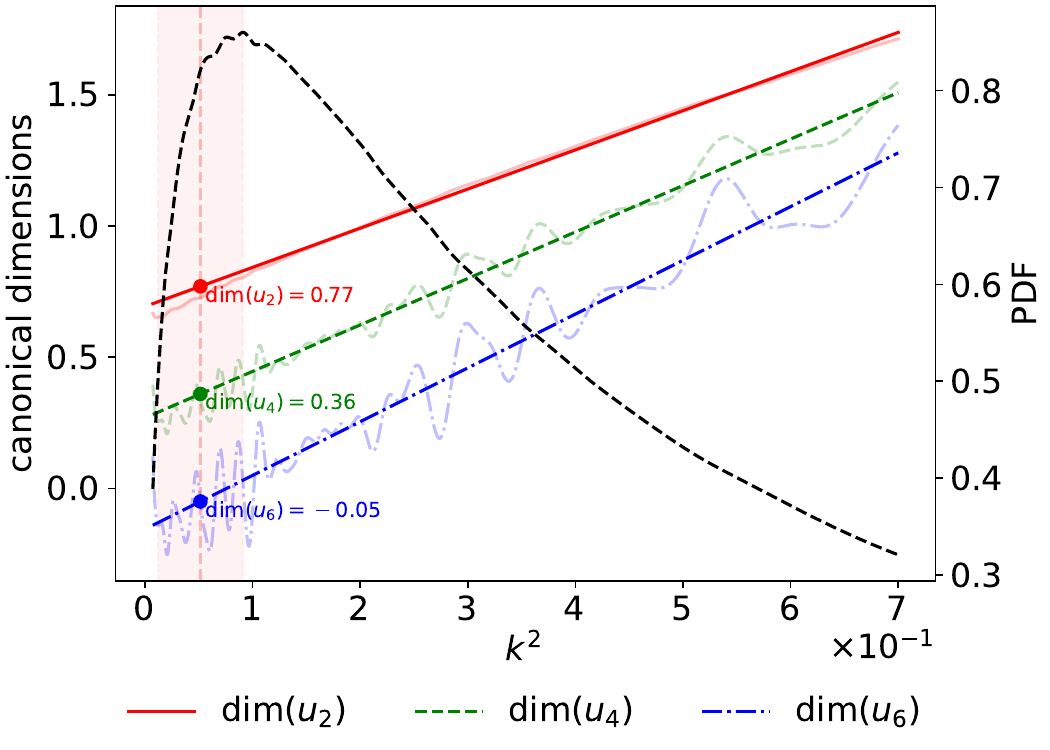}
        \caption{\snr=0.25, kernel size=20}
    \end{subfigure}
    \hfill
    \begin{subfigure}{0.48\textwidth}
        \centering
        \includegraphics[width=\textwidth]{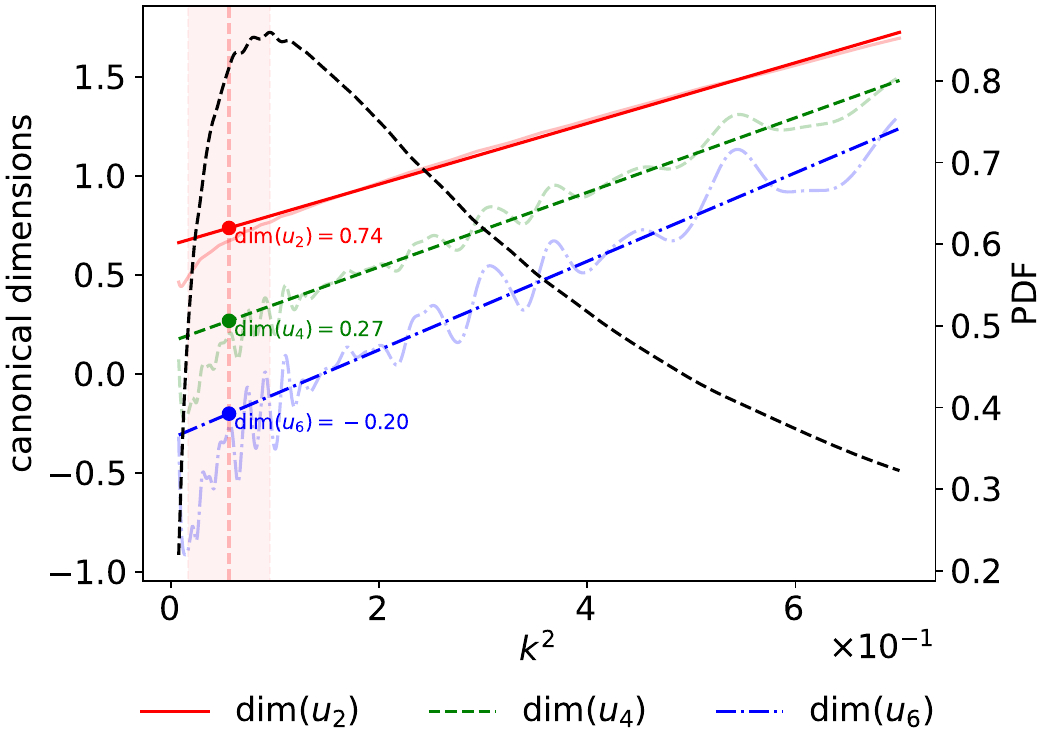}
        \caption{\snr=0.59, kernel size=20}
    \end{subfigure}

    \vspace{0.6cm}

    \begin{subfigure}{0.48\textwidth}
        \centering
        \includegraphics[width=\textwidth]{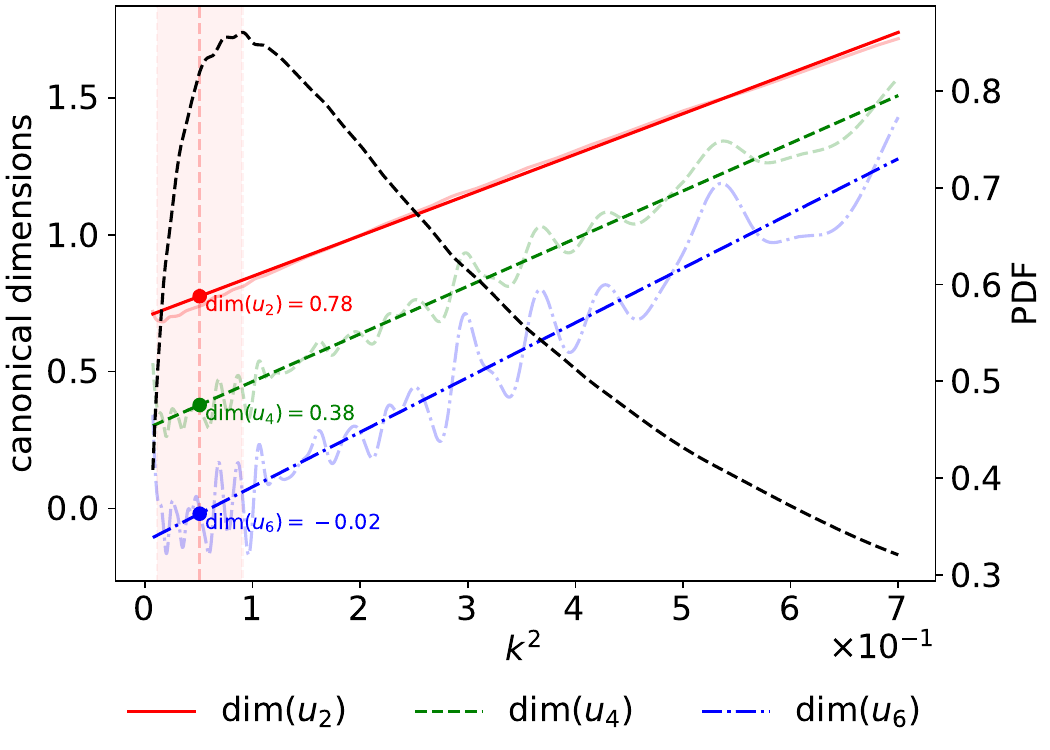}
        \caption{\snr=0.21, kernel size=2000}
    \end{subfigure}
    \hfill
    \begin{subfigure}{0.48\textwidth}
        \centering
        \includegraphics[width=\textwidth]{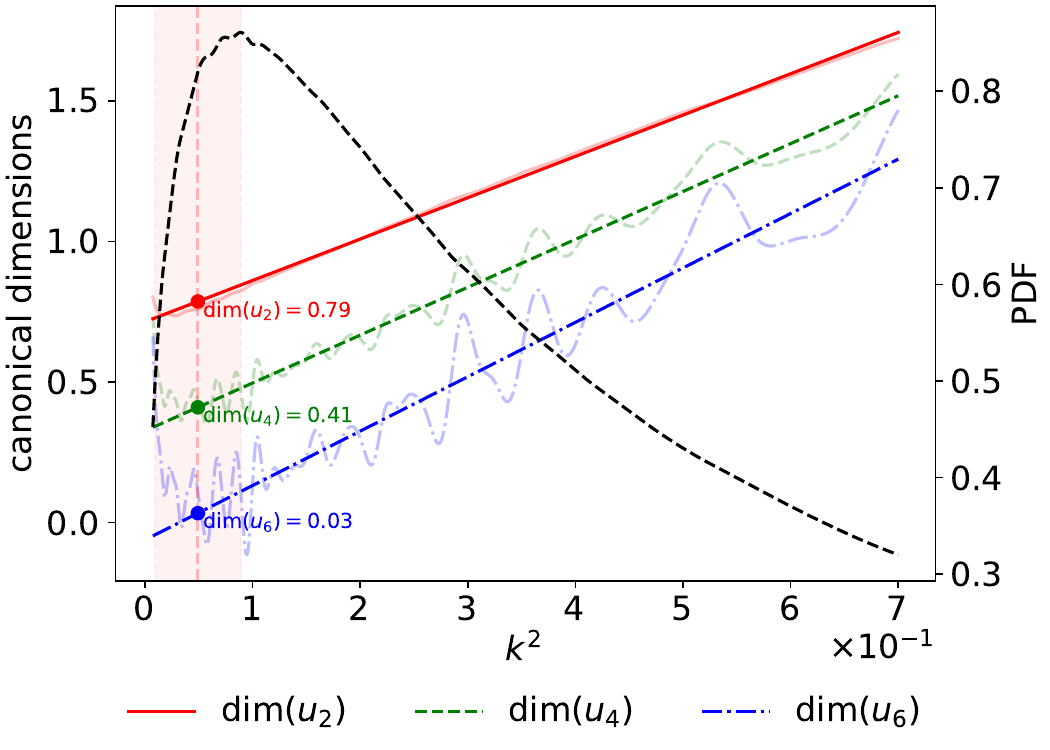}
        \caption{\snr=0.63, kernel size=2000}
    \end{subfigure}

    \caption{
        Canonical dimensions extracted from numerical data after kernel convolution for increasing kernel sizes for fixed \snr.
    }\label{fig:CEA_convolution_0}
\end{figure}

\begin{figure}[t]
    \centering

    \begin{subfigure}{0.48\textwidth}
        \centering
        \includegraphics[width=\textwidth]{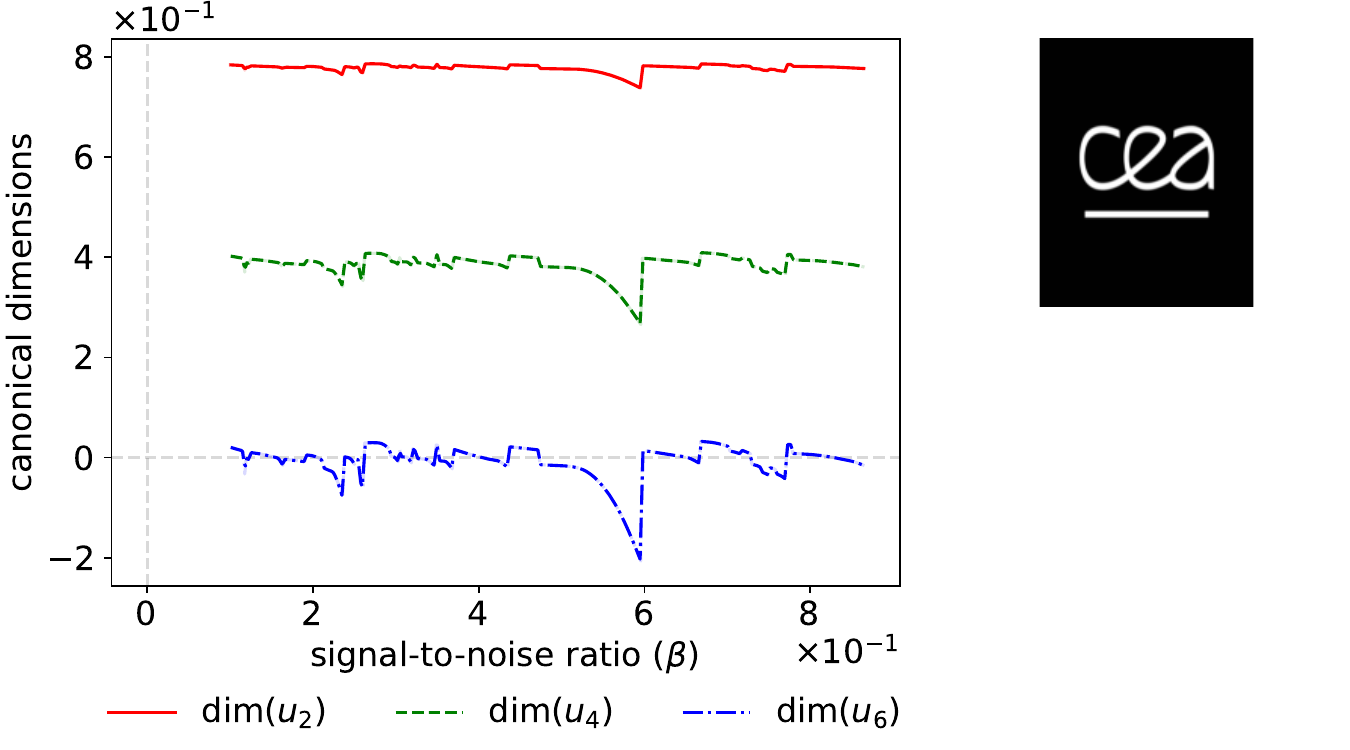}
        \caption{Kernel size $=20$}
    \end{subfigure}
    \hfill
    \begin{subfigure}{0.48\textwidth}
        \centering
        \includegraphics[width=\textwidth]{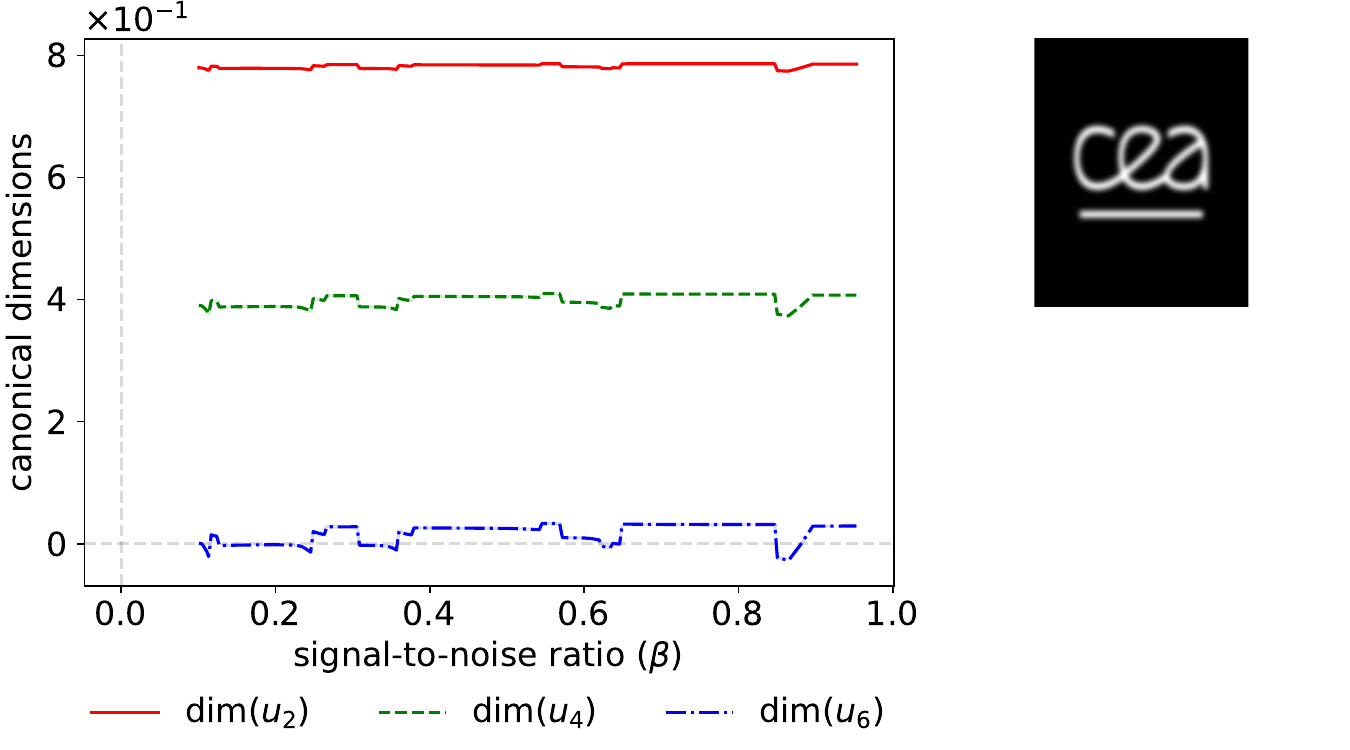}
        \caption{Kernel size $=200$}
    \end{subfigure}

    \vspace{0.6cm}

    \begin{subfigure}{0.6\textwidth}
        \centering
        \includegraphics[width=\textwidth]{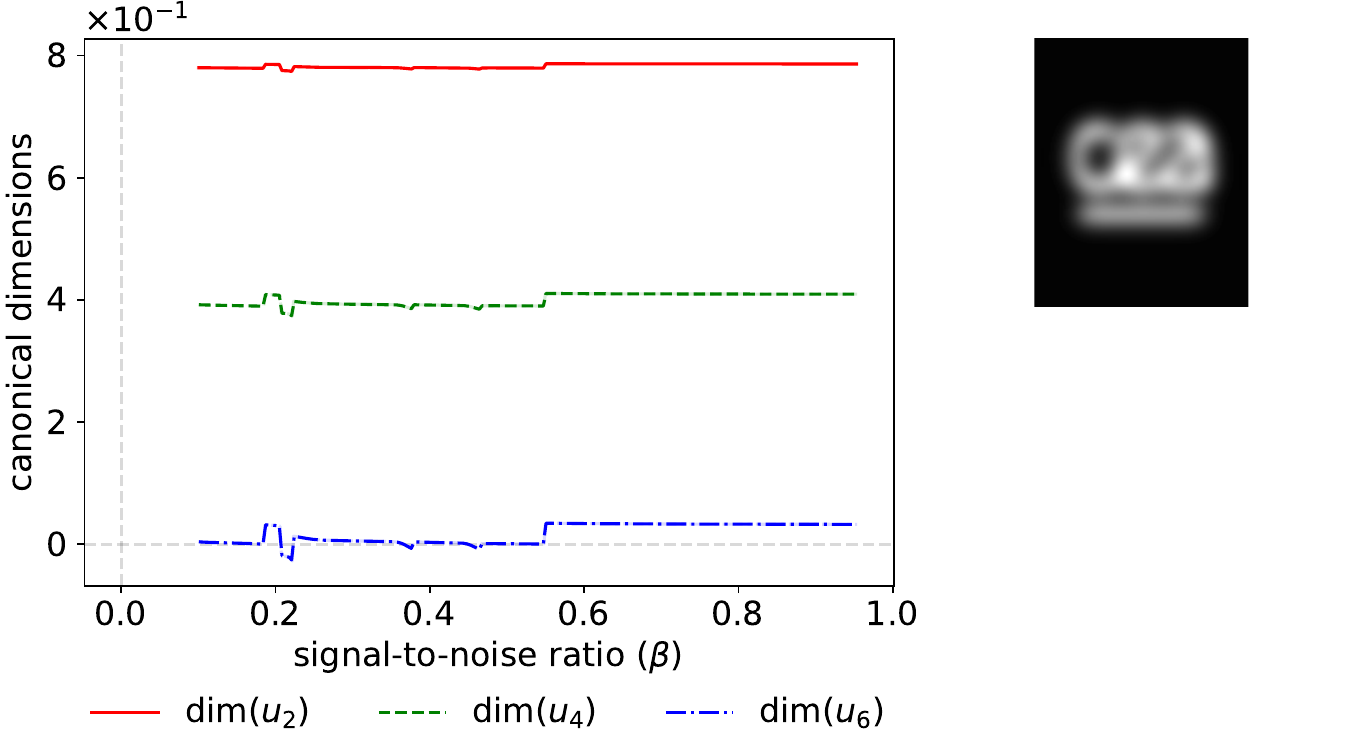}
        \caption{Kernel size $=2000$}
    \end{subfigure}

    \caption{
        Canonical dimensions at a fixed \ir scale (the red dashed line in \Cref{fig:CEA_convolution_0}) extracted from numerical data after kernel convolution, for increasing kernel sizes.
        The \snr is shown on the horizontal axis.
    } \label{fig:CEA_convolution}
\end{figure}

We focus here on the Gaussian behaviour of the \rg, and, specifically, of the canonical dimensions, the properties of which are summarised in \Cref{fig:CEA_convolution_0,fig:CEA_convolution} for varying parameter values.
As these figures demonstrate, a high \snr renders $\dim_{\tau}(u_6)$ negative, indicating that $u_6$ is irrelevant.
Recall that, according to \Cref{tab:thresholds}, the critical detection threshold, $\beta_c$, is defined as the scale at which the asymptotic canonical dimension of $u_4$ transitions to a negative value.

The rationale for this definition is supported by empirical evidence (detailed in the referenced work) that the quartic dimension exhibits lower sensitivity to the intrinsic fluctuations of the data.
Furthermore, subsequent sections of this article provide additional evidence to support this definition as a conservative estimate of significant deviations from the reference noise model induced by the signal.
The results presented in the figures suggest that no finite $\beta_c$ exists within the relevant range of \snr.
Although the sign change for $\dim_{\tau}(u_6)$ appears to surpass the typical scale of fluctuations (see~\cite{RG7}), establishing a precise scale for signal emergence in this context remains challenging.
These findings are consistent with those obtained for the \textsc{mnist} dataset in our previous study.

\begin{figure}[t]
    \centering

    \begin{subfigure}{0.48\textwidth}
        \centering
        \includegraphics[width=\textwidth]{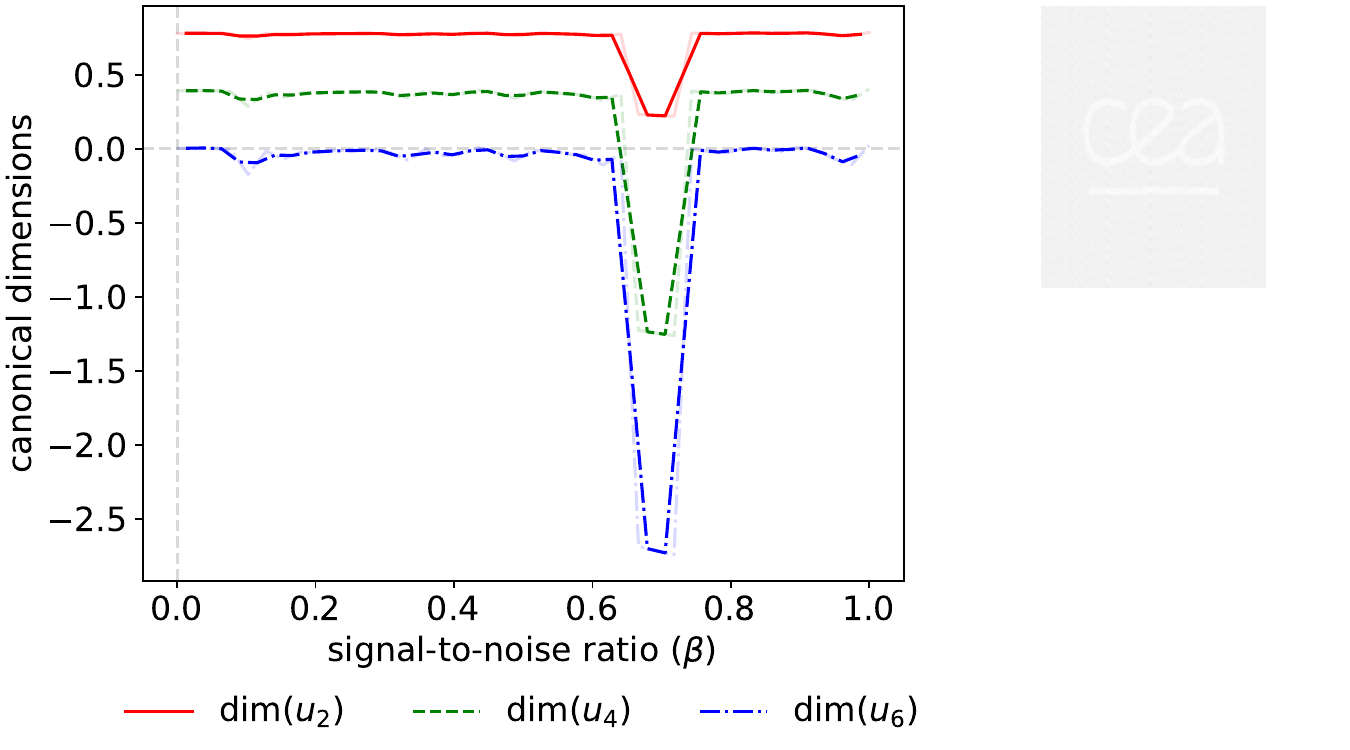}
        \caption{Periodic + Gaussian (log)}
    \end{subfigure}
    \hfill
    \begin{subfigure}{0.48\textwidth}
        \centering
        \includegraphics[width=\textwidth]{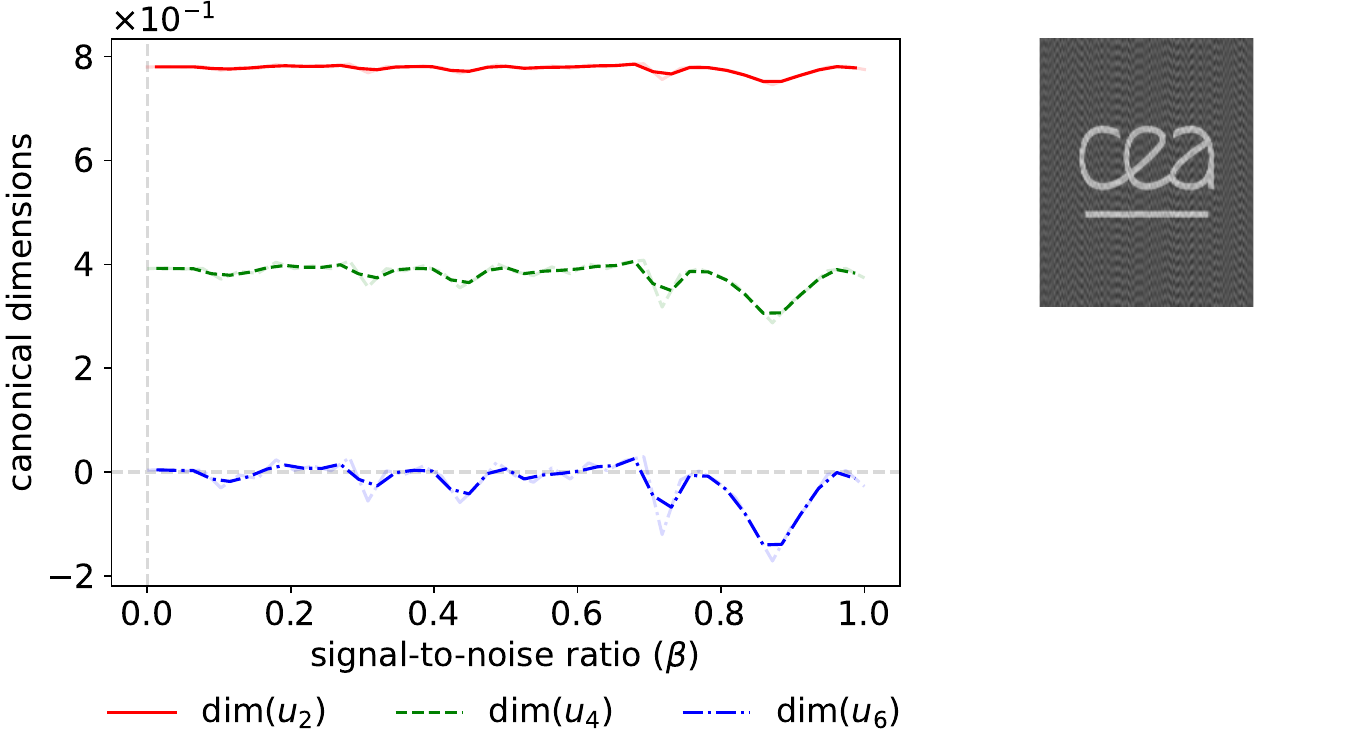}
        \caption{Periodic + Gaussian (additive)}
    \end{subfigure}

    \vspace{0.6cm}

    \begin{subfigure}{0.5\textwidth}
        \centering
        \includegraphics[width=\textwidth]{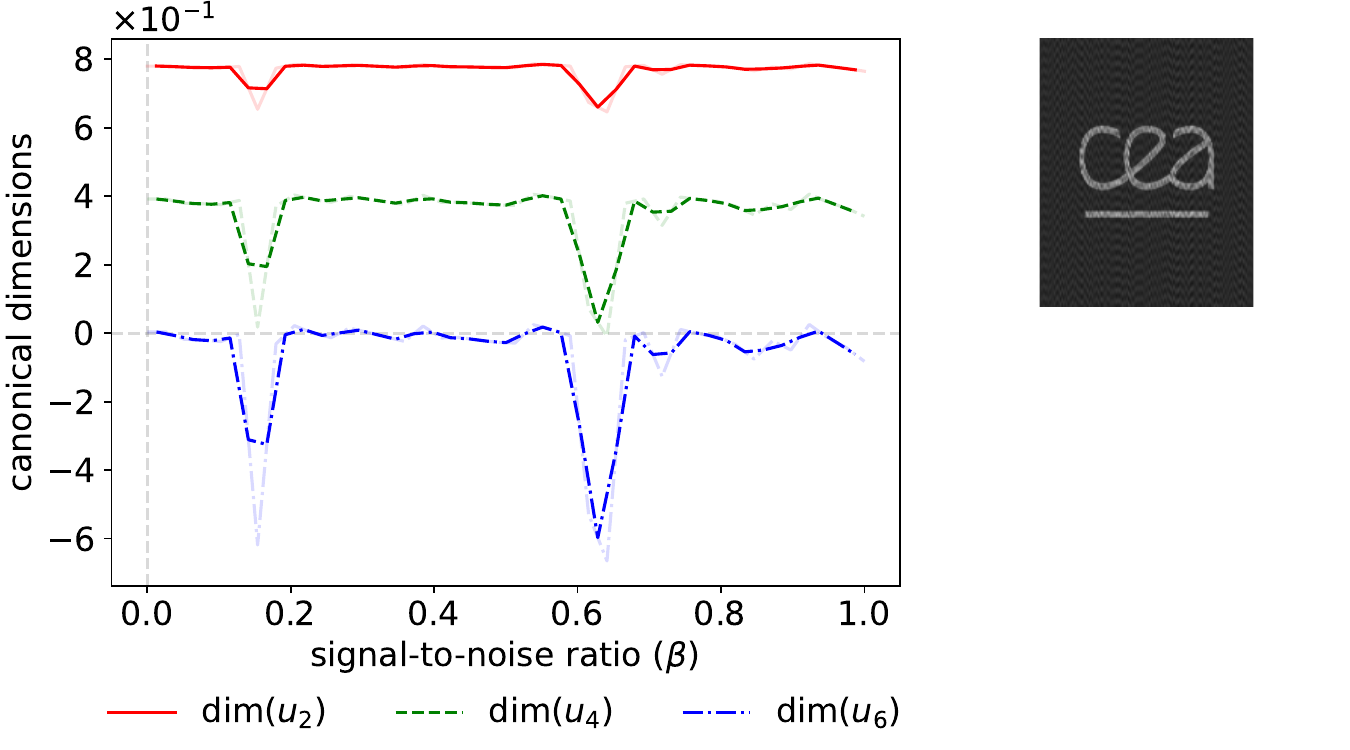}
        \caption{Periodic + Gaussian (multiplicative)}
    \end{subfigure}

    \caption{
        Canonical dimensions extracted after injecting a periodic component on top of Gaussian noise, compared across three combination rules (log, additive, multiplicative) as a function of the \snr.
    }
\end{figure}

As a second example, we consider a structured periodic perturbation in addition to the additive Gaussian background.
The perturbation oscillates along the horizontal direction, so that within each row $i$ its value varies periodically with the column index $j$.
The corresponding greyscale intensity matrix $P$ is thus defined by
\begin{equation}
    P_{ij}
    =
    \frac{255}{2}
    +
    \frac{A_i}{2}
    \sin\left(\frac{j}{5}+\phi_i\right).
\end{equation}
Here, each row has an independently sampled amplitude $A_i$ and phase $\phi_i$, distributed according to
\begin{equation}
    A_i\sim
    \left. \mathcal{N}\left(\frac{255}{2},20^2\right)\right|_{[0,255]},
    \qquad
    \phi_i\sim\mathcal{N}(0,1),
\end{equation}
respectively.
Where $\left. \mathcal{N}\left(\frac{255}{2},20^2\right)\right|_{[0,255]}$
denotes a truncated Gaussian distribution in the interval $[0,255]$.
The offset $255/2$ centres the oscillations at the midpoint of the greyscale range, while the factor $1/2$ ensures that $P_{ij}$ remains in $[0,255]$.
Both these perturbations are described in more detail in \Cref{algo:perturbation}.

\begin{algorithm}[!ht]
    \caption{Construction of samples with Gaussian-kernel or periodic perturbations}
    \label{algo:perturbation}

    \SetKwInOut{Input}{Input}
    \SetKwInOut{Output}{Output}
    \SetKwInOut{Let}{Let}

    \Input{$N>0$ (sample size), $q\in(0,1]$ (aspect ratio)}
    \Let{$P\gets\lfloor qN\rfloor$}

    \Input{$\beta\geq 0$ (signal-to-noise ratio)}
    \Input{$I\in[0,255]^{H\times W}$ (input image)}
    \Input{
        $Z\in\mathds{R}^{N\times P}$ (Gaussian noise),
        s.t.\ $Z_{ij}\sim\mathcal{N}(0,1)$
    }
    \Input{
        $\mathtt{version}\in\{\mathrm{A},\mathrm{B}\}$
    }

    $S\gets\operatorname{resize}(I,N,P)
        \in\mathds{R}^{N\times P}$
    \tcp*{Interpolate the image}

    \If{\emph{\texttt{version}}$~= A$: Gaussian-kernel transformation}{

    \Input{
        $R\in\mathds{N}$ (kernel radius), $\sigma=R/5$ (kernel width)
    }

    $\widehat{G}_{\sigma,R}[u,w]
        \gets
        \exp\!\left(
        -(u^2+w^2)/(2\sigma^2)
        \right)$,
    \quad $-R\leq u,w\leq R$

    $C_{\sigma,R}
        \gets
        \displaystyle
        \sum_{u=-R}^{R}
        \sum_{w=-R}^{R}
        \widehat{G}_{\sigma,R}[u,w]$

    $G_{\sigma,R}[u,w]
        \gets
        \widehat{G}_{\sigma,R}[u,w]/C_{\sigma,R}$
    \tcp*{Normalize the kernel}

    \For{$i\gets0$ \KwTo $N-1$}{
        \For{$j\gets0$ \KwTo $P-1$}{
            $S^{(\mathrm{A})}[i,j]\gets
                \displaystyle
                \sum_{u=-R}^{R}
                \sum_{w=-R}^{R}
                G_{\sigma,R}[u,w]\,
                S\!\left[
                    (i-u)\bmod N,\,
                    (j-w)\bmod P
                    \right]$
        }
    }

    $S^{(\mathrm{A})}\gets
        \left(S^{(\mathrm{A})} - \langle S^{(\mathrm{A})}\rangle \right)
        \left(\operatorname{Var}(S^{(\mathrm{A})}\right)^{-1/2}
    $
    \tcp*{Standardize after convolution}

    $X\gets\beta S^{(\mathrm{A})}+Z$
    \tcp*{Additive Gaussian background}
    }

    \If{\emph{\texttt{version}}$~= B$: periodic structured perturbation}{

        \Input{
            $\omega>0$ (angular frequency),
            $(\sigma_A)$ (amplitude standard deviation),
            $\sigma_\phi$ (phase width),
            $c\in\{\mathrm{add},\mathrm{mult},\mathrm{log}\}$ (combination rule)
        }

        \For{$i\gets0$ \KwTo $N-1$}{
            Draw
            $A_i\sim\mathcal{N}(127.5,\sigma_A^2)$ restricted to $[0,255]$
            and
            $\phi_i\sim\mathcal{N}(0,\sigma_\phi^2)$\;

            \For{$j\gets0$ \KwTo $P-1$}{
                $\Pi[i,j]\gets 127.5 +
                    (A_i / 2) \sin(\omega j+\phi_i)$
            }
        }

        \uIf{$c=\mathrm{add}$}{
            $X_0\gets S+\Pi$
        }
        \uElseIf{$c=\mathrm{mult}$}{
            $X_0\gets S\odot\Pi$
        }
        \ElseIf{$c=\mathrm{log}$}{
            $X_0\gets
                \log\left( S
                \odot
                \Pi \right)$,
        }
        $X_0\gets
            \left(X_0-\langle X_0\rangle\right)
            \left(\operatorname{Var}(X_0)\right)^{-1/2}$
        \tcp*{Standardization (global, over all pixels)}
        $X\gets\beta X_0+Z$
        \tcp*{Image with Gaussian background}
    }
\end{algorithm}

\section{Time Series Analysis: Measuring Critical Exponents}\label{sec:time_series}
As recalled in \Cref{critical}, in the vicinity of the critical point, the relaxation time $\tau$ diverges with the correlation length $\xi$ as $\tau=\xi^z$, where $z$ is a critical exponent.
Note that, unlike most critical exponents, there is no known exact analytical solution in two dimensions for the exponent $z$, and the best estimates, obtained via Monte Carlo simulations, yield $z \approx 2.17$~\cite{nightingale1996dynamic}.
Below the critical temperature, domain growth over time $t$ follows a $\sqrt{t}$ law.
This scaling arises because the characteristic size $L$ of a domain is essentially governed by its local curvature, which is proportional to $1/L$, yielding $\dd L/\dd t \propto 1/L$ (see \Cref{AppOJK}).
The curious reader can also see \Cref{sec:app1,critical} for further details.
However, exactly at the critical point ($T=T_c$), the influence of the anomalous dimension and of fluctuations in general becomes dominant, causing the mean-field approximation to break down.
Consequently, domain growth no longer strictly follows the principle of interfacial minimisation through curvature optimisation~\cite{livi2017nonequilibrium}.
At the critical point, we have:
\begin{equation}
    L(t) \sim t^{1/z}.
\end{equation}
In this section, we explain how the \gsa enables the measurement of $z$ in the vicinity of the critical temperature.
Once again, this provides a twofold validation: it not only validates the \gsa itself, but also establishes it as a reliable method in practice.

The principle of our measurement essentially consists in tracking the time evolution of the cutoff index $\lambda_c(t)$, obtained by minimising the \kl divergence.
This cutoff separates the coherent signal subspace (the spikes) from the noise subspace dominated by thermal fluctuations (the \mpdistr bulk).
It is remarkable that one can directly relate the spectral behaviour of $\lambda_c(t)$ to the spatial dynamics $L(t)$.
As a first approach, each domain (or \enquote{bubble}) within which the field is coherent acts as a macroscopic \dof, generating an eigenvalue that emerges from the noise.
The rank of the signal, measured by $\lambda_c(t)$, thus corresponds to the number of these independent bubbles.
At the time of quenching ($t=0$), the system is completely uncorrelated, the number of independent \dof is maximal, and the spectrum is indistinguishable from the noise.
Over time, the coarsening process causes the characteristic size of the domains to grow.
The physical space becomes structured, which translates spectrally into a decrease in the number of independent components.
Since the number of bubbles in the system scales as $N^2/L(t)^2 \propto t^{-2/z}$, the effective dimension of the signal subspace, $\lambda_c(t)$, decreases according to this same scaling law.\footnote{%
    By convention, the eigenvalues of the covariance matrix are always sorted in descending order (from highest to lowest energy): $\lambda_1 \ge \lambda_2 \ge \dots \ge \lambda_P$.
    Thus, $\lambda_c$ measures the size (or rank) of the signal.
}
Eventually ($t \gg 1$), when a single ordered phase occupies the entire grid, the signal retains only a single macroscopic \dof ($\lambda_c \to 1$), concentrating the entire signal intensity.

However, this approach is too naive.
Notably, it ignores the spectral gap between the outliers (associated with domain emergence during coarsening) and the \mpdistr bulk.
Furthermore, spatio-temporal correlations significantly reduce the effective number of \dof $q$ appearing in the definition of \Cref{thm:thMP}.

To be concrete, let us consider a discretised version of the Model~A on an $N \times N$ grid.
At high temperatures, where the correlation length is essentially limited to the lattice spacing, the correlation matrix spectrum follows a \mpdistr law with $q = N^2/T_w$, where $T_w$ denotes the number of time samples in the observation window (distinct from the thermodynamic temperature $T$).
Below the critical temperature, domain growth reduces the effective number of independent sites to $N_{\text{eff}}^2 \sim N^2 / L(t)^2 \sim t^{-2/z}$.
Simultaneously, temporal correlations are critical: the system \enquote{ages}, with correlations decaying via a power law (see \eqref{timecorr}) rather than exponentially.

To track these correlations, we employ the sliding window method, a standard technique for ageing systems that locally restores time-translation invariance.
We divide the total interval $T_{\text{int}}$ into $n$ overlapping segments $D_i$ of size $\tau$.
Let $t_i$ be the centre of window $i$.
Two regimes arise for the effective parameter $q_{\text{eff}}$:
\begin{itemize}
    \item \textbf{Short-time regime ($t_i \lesssim \tau$):} The window captures the global dynamics, $q_{\text{eff}} \sim N_{\text{eff}}^2/t_i \sim t_i^{-(2/z + 1)}$.
    \item \textbf{Long-time regime ($t_i \gg \tau$):} The window is constrained by its finite size $\tau$, $q_{\text{eff}} \sim N_{\text{eff}}^2/\tau \sim t_i^{-2/z}$.
\end{itemize}
The crossover between these regimes is governed by the correlation time $\tau$.
For $t_i \lesssim \tau$, the window captures the full temporal evolution and all $t_i$ samples contribute independently to the eigenvalue spectrum.
For $t_i \gg \tau$, the domain configuration evolves slowly compared to the window size.
Thus, configurations separated by more than $\tau$ are statistically equivalent because the system has not decorrelated.
The effective number of independent temporal samples therefore saturates at $\tau$, and the denominator in $q_{\text{eff}}$ transitions from $t_i$ to $\tau$.

Consequently, the spectral edge $\lambda_c$ follows a scaling function of the form:
\begin{equation} \label{eq:window_fit}
    \lambda _c = \sigma_0^2(1 + A t^{-\gamma /2})^2
\end{equation}
where $\gamma = 2/z+1$ at short times and $\gamma = 2/z$ at long times, with $\sigma_0$ and $A$ being numerically determined constants.

In order to compute the exponent $z$ we simulated the Model~A (see \Cref{critical}) quenching from high temperature to the critical temperature on a $100 \times 100$ lattice.
To record sufficiently dense snapshots of the evolving configurations, we used a time step $100$ times smaller than that employed in the critical-temperature analysis of \Cref{IsingCritical}.
In Method~I, we divided the resulting time series into overlapping windows of $10000$ time steps, as illustrated in \Cref{fig:sampling_windows}.
Each window therefore yields a tensor of dimensions $100 \times 100 \times 10000$ that is then reshaped into a $10000 \times 10000$ matrix, by aggregating the first two dimensions.

\begin{figure}[t]
    \centering
    \begin{tikzpicture}[
        >=Latex,
        interval/.style={thick, |<->|},
        boundaries/.style={dashed, thin}
    ]
    \draw[->, line width=0.9pt] (0,0) -- (6,0) node[right] {$t$};
    \def\numberofintervals{5}
    \foreach \i in {1,...,\numberofintervals} {
            \pgfmathsetmacro{\xmin}{0.35*(\i-1)}
            \pgfmathsetmacro{\xmax}{1.1 + 0.65*\i}
            \pgfmathsetmacro{\y}{0.05 + 0.5*\i}
            \draw[interval]
            (\xmin,\y) node[left] {$S_{\i}$} -- (\xmax,\y);
            \draw[boundaries] (\xmin, \y) -- (\xmin, 0);
            \draw[boundaries] (\xmax, \y) -- (\xmax, 0);
        }
    \node at (4,3.1) {$\cdots$};
\end{tikzpicture}
    \caption{%
        Illustration of overlapping sampling windows $S_1,\, S_2,\, \ldots$ along the time axis.
    }
    \label{fig:sampling_windows}
\end{figure}
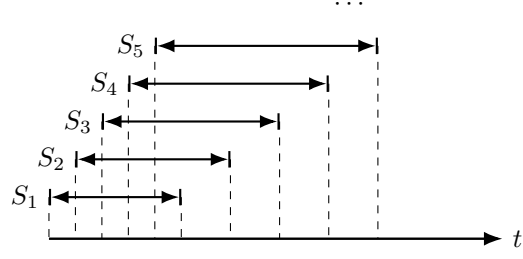

Method~II closely resembles Method~I, except that the window size increases over time, from $\num{4e3}$ time steps at the beginning of the evolution to $\num{1.5e4}$ at the end.
This adaptive windowing scheme better captures the rapid evolution immediately after the quench while providing improved statistics during the slower late-time relaxation.
However, the shorter initial windows yield matrices as small as $\num{4e3} \times \num{1e4}$ corresponding to fewer samples and consequently larger statistical fluctuations.
For both methods, we compute $\lambda_c$ separately within each window.
The resulting values are shown in \Cref{fig:lambda_windows}.

\begin{figure}[t]
    \centering
    \begin{subfigure}{0.49\textwidth}
        \centering
        \includegraphics[width=\linewidth]{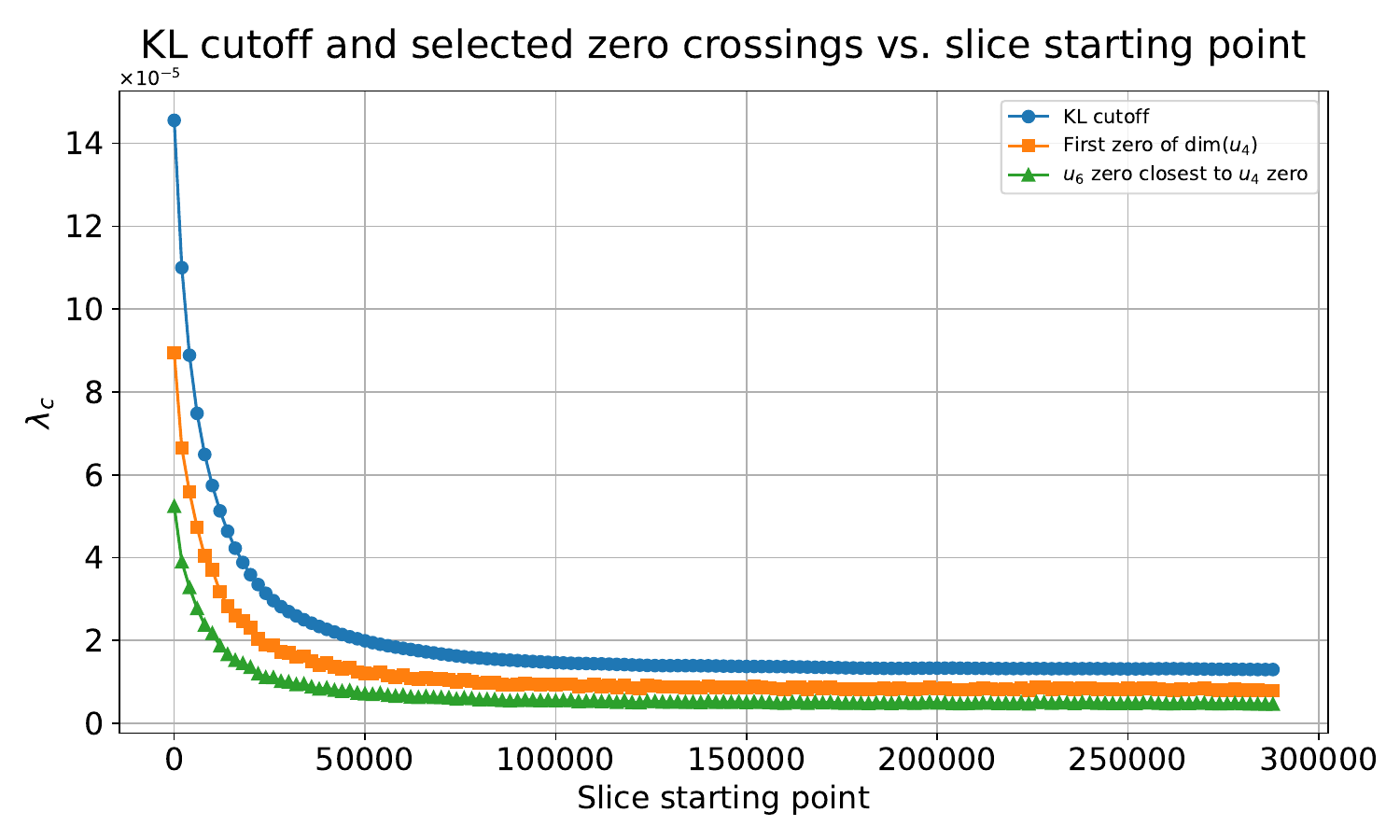}
        \caption{Method I}
        \label{fig:constant_window}
    \end{subfigure}
    \hfill
    \begin{subfigure}{0.49\textwidth}
        \centering
        \includegraphics[width=\linewidth]{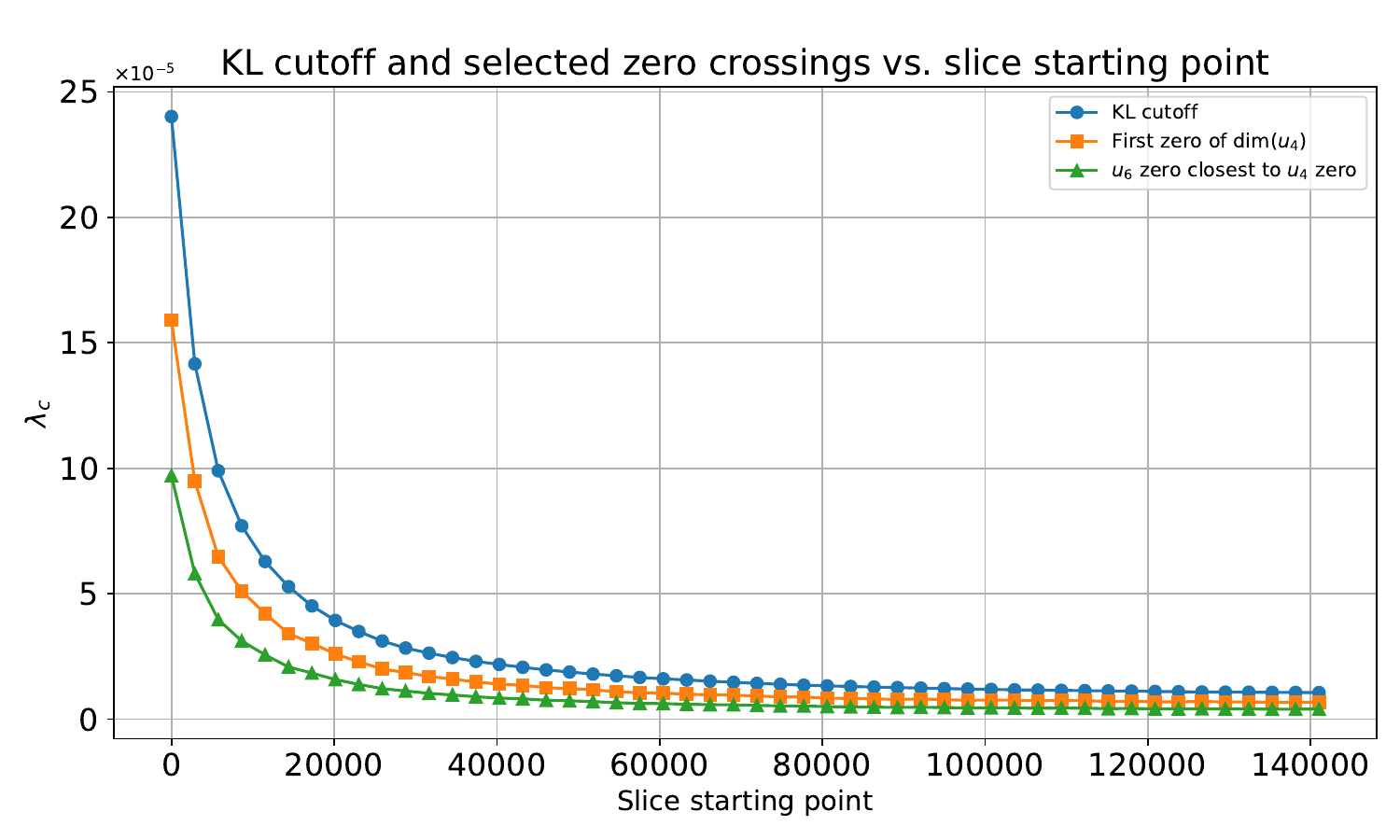}
        \caption{Method II}
        \label{fig:dynamic_window}
    \end{subfigure}
    \caption{$\lambda_c$ calculated for different slicing methods.}
    \label{fig:lambda_windows}
\end{figure}

Next, we fit \Cref{eq:window_fit} to the data shown in \Cref{fig:lambda_windows}.
For both methods, however, only the long-time scaling regime can be fitted, as the early-time evolution remains too rapid to be resolved by our windowing schemes.
For Method~II, the fitting procedure is relatively straightforward because the adaptive window size accounts for the different relaxation rates at early and late times.
We therefore exclude the first few data points, which lie outside the long-time scaling regime, and fit \Cref{eq:window_fit} to the remaining data, as shown in \Cref{fig:dynamic_window_fit}.
This yields $z=2.47 \pm 0.25$, $z=2.63 \pm 0.56$ and $z=2.63 \pm 0.59$ for the \kl-, $u_4$- and $u_6$-based methods respectively.

For Method~I, the fitting procedure requires greater care.
We seek an intermediate time interval that excludes both the early-time regime, which is not adequately resolved by our windowing scheme, and the late-time plateau associated with the approach to equilibrium.
To identify a suitable fitting range, we varied the lower and upper bounds of the interval and computed the mean squared error (\mse) for each fit.
A selection of the resulting \mse and estimates of $z$ is reported in \Cref{tab:z_values}.
The fit over the best interval is shown in \Cref{fig:constant_window_fit}.
Overall, the six estimates obtained using the two methods are consistent with the reference value $z \approx 2.17$~\cite{nightingale1996dynamic}.
Method~II yields systematically higher $z$ values, with a weighted average of $z = 2.51 \pm 0.21$ ($1.6\sigma$ above the reference).
This bias likely arises from the adaptive windowing scheme: shorter initial windows ($\num{4e3}$ time steps) produce matrices with fewer samples, inflating the spectral edge estimate at early times where the exponent is largest.
Method~I, which uses a fixed window size and restricts to the optimal fitting interval, gives $z = 2.30 \pm 0.11$ ($1.1\sigma$ from reference), in better agreement.

\begin{figure}[t]
    \centering
    \begin{subfigure}{0.49\textwidth}
        \centering
        \includegraphics[width=\linewidth]{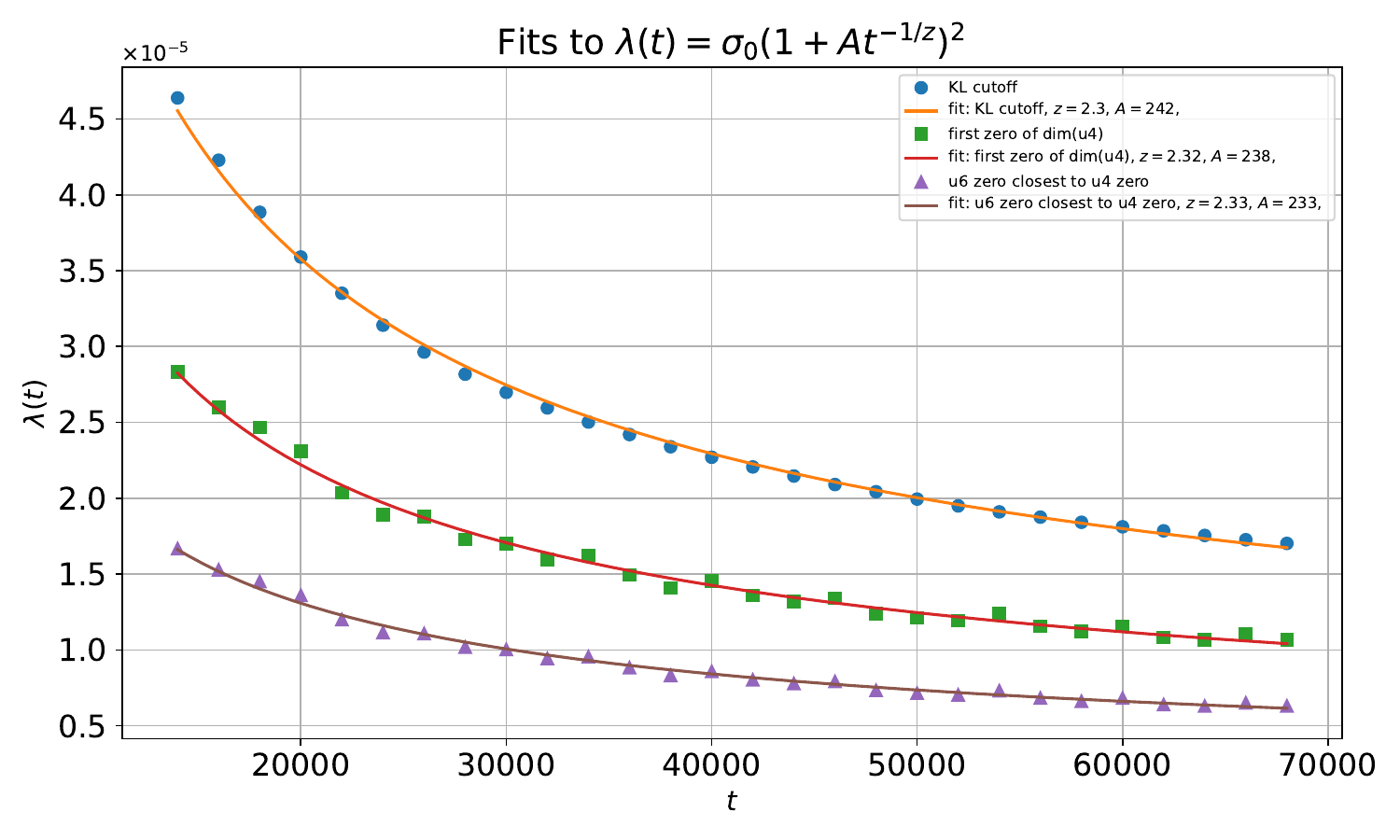}
        \caption{\Cref{eq:window_fit} fitted to the best chosen interval of Method I.}
        \label{fig:constant_window_fit}
    \end{subfigure}
    \hfill
    \begin{subfigure}{0.49\textwidth}
        \centering
        \includegraphics[width=\linewidth]{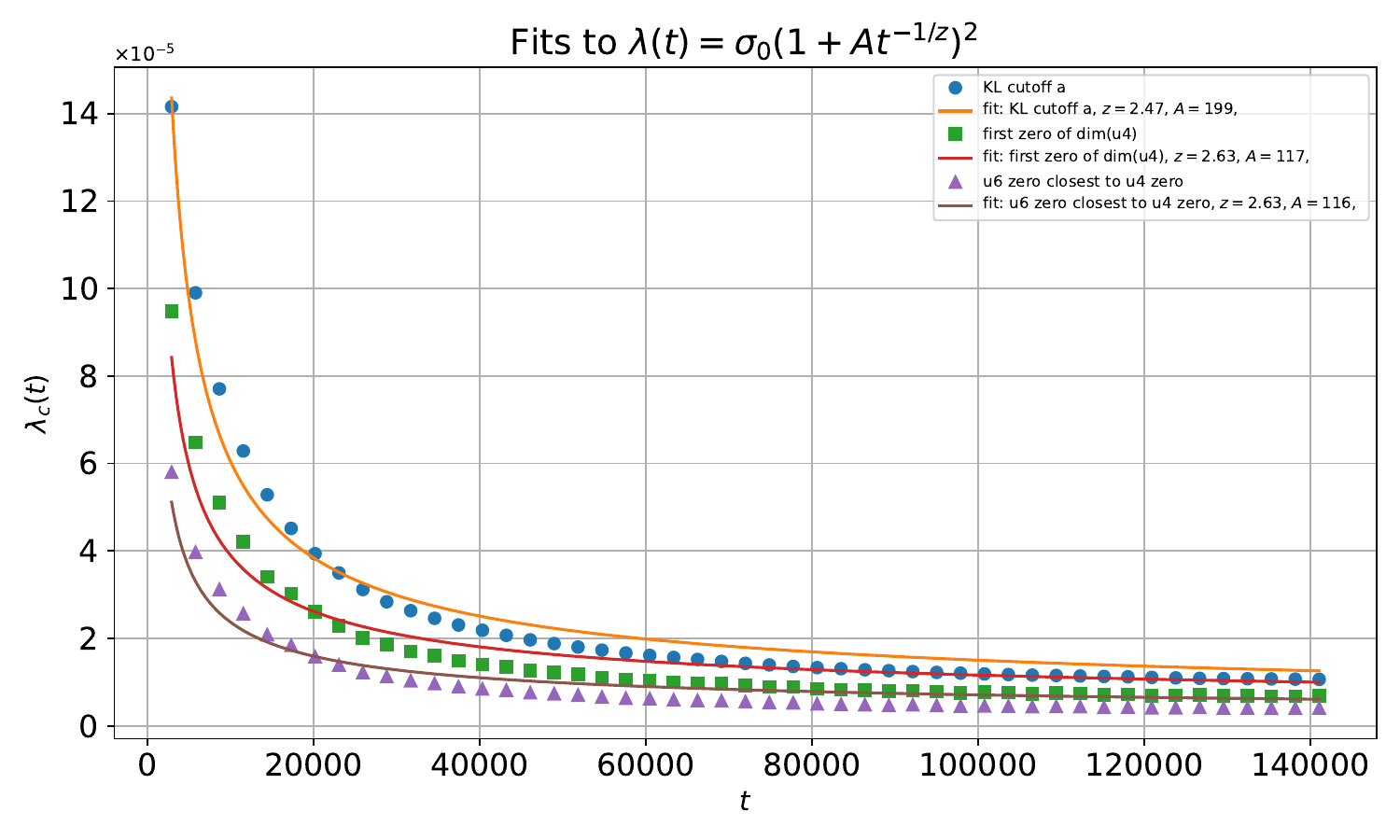}
        \caption{\Cref{eq:window_fit} fitted to the data obtained from  Method II.}
        \label{fig:dynamic_window_fit}
    \end{subfigure}
    \caption{\Cref{eq:window_fit} fitted to data from method I and II.}
    \label{fig:window_fit}
\end{figure}

\begin{table}[t]
    \centering
    \begin{tabular}{|c|c|c|c|}
        \hline
        \textbf{Starting Point} &
        \textbf{End Point}      &
        \textbf{$z$-Values}     &
        \textbf{MSE}                               \\
        \hline

        7000                    & 38000          &
        \begin{tabular}[c]{@{}c@{}}
            $2.56 \pm 0.638$ \\
            $3.09$           \\
            $5.10$
        \end{tabular}
                                &
        \begin{tabular}[c]{@{}c@{}}
            $1.92\times10^{-12}$ \\
            $6.23\times10^{-13}$ \\
            $2.12\times10^{-13}$
        \end{tabular}
        \\
        \hline

        7000                    & 58000          &
        \begin{tabular}[c]{@{}c@{}}
            $3.13$          \\
            $2.51 \pm 0.54$ \\
            $3.10$
        \end{tabular}
                                &
        \begin{tabular}[c]{@{}c@{}}
            $6.62\times10^{-13}$ \\
            $8.29\times10^{-13}$ \\
            $1.45\times10^{-13}$
        \end{tabular}
        \\
        \hline

        \textbf{10000}          & \textbf{58000} &
        \begin{tabular}[c]{@{}c@{}}
            $2.30 \pm 0.135$ \\
            $2.28 \pm 0.303$ \\
            $2.29 \pm 0.31$
        \end{tabular}
                                &
        \begin{tabular}[c]{@{}c@{}}
            $9.72\times10^{-14}$ \\
            $1.99\times10^{-13}$ \\
            $6.93\times10^{-14}$
        \end{tabular}
        \\
        \hline

        14000                   & 58000          &
        \begin{tabular}[c]{@{}c@{}}
            $2.79 \pm 2.55$  \\
            $2.34 \pm 0.596$ \\
            $2.35 \pm 0.557$
        \end{tabular}
                                &
        \begin{tabular}[c]{@{}c@{}}
            $3.13\times10^{-12}$ \\
            $1.86\times10^{-13}$ \\
            $6.16\times10^{-14}$
        \end{tabular}
        \\
        \hline

        14000                   & 78000          &
        \begin{tabular}[c]{@{}c@{}}
            $2.35 \pm 0.212$ \\
            $2.39 \pm 0.416$ \\
            $2.40 \pm 0.418$
        \end{tabular}
                                &
        \begin{tabular}[c]{@{}c@{}}
            $1.53\times10^{-13}$ \\
            $2.05\times10^{-13}$ \\
            $6.97\times10^{-14}$
        \end{tabular}
        \\
        \hline
    \end{tabular}

    \caption{%
        Fitted values of $z$ and corresponding \mse for different fitting intervals.
        The best candidate with the smallest overall \mse is indicated in bold.
        The missing uncertainty values could not be reliably estimated from the covariance matrix and are therefore not reported.}
    \label{tab:z_values}
\end{table}

\section{Hyperspectral Images}\label{sec:hyperspectral}
\begin{figure}[t]
    \centering
    \input{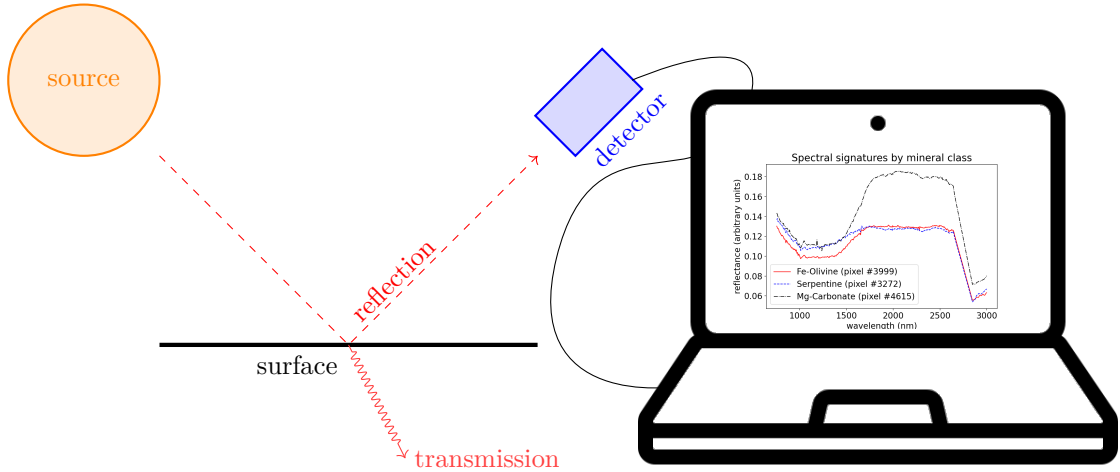}
    \caption{%
        Schematic representation of light and matter interaction and the resulting spectral measurement.
        In the planetary remote-sensing application considered here, sunlight is reflected from the surface and captured by an imaging spectrometer.
        The transmitted component is not available
        to the detector.
    }
    \label{fig:lightmatterinteraction}
\end{figure}

Computer vision deals with the analysis of visual information, often presented as images.
Recent advances in \ai have contributed to powerful new approaches for processing visual data, with large neural networks now capable of analysing complex scenes captured as multi-channel images.
Standard \rgb images condense colourimetric information into three channels that roughly correspond to primary colour bands, and they represent the vast majority of the data processed by modern \ai techniques.

From a scientific perspective, however, reducing spectral information to three arbitrary channels discards a significant portion of the information contained in the original scene.
Matter interacts with light across a broad electromagnetic spectrum, and capturing this interaction in detail requires far more than three channels.
The natural representation of this interaction is a spectrum (see \Cref{fig:lightmatterinteraction}): an illumination source irradiates the surface, and a detector measures the reflected intensity as a function of wavelength.
The result is a \emph{spectral} image $\mathrm{X} \in \mathds{R}^{H \times W \times C}$, where $H$ and $W$ denote the spatial dimensions and $C$ is the number of spectral bands.
Each pixel $\vec{x}_{ij} \in \mathds{R}^{C}$ records the spectrum of the light captured at the corresponding spatial location $(i, j)$, with $i = 1, 2, \dots, H$ and $j = 1, 2, \dots, W$.
Because the interaction between radiation and matter depends on molecular composition, each spectrum serves as a fingerprint of the material present on the surface.
These images are commonly referred to as multispectral or hyperspectral images (\hsi), depending on the number of bands: multispectral systems typically carry $\mathcal{O}(10)$ bands, while hyperspectral systems carry $\mathcal{O}(10^2)$ to $\mathcal{O}(10^3)$ bands.
Thanks to the emergence of high-performance optical sensors, \hsi has become an important tool for scientists, enabling the acquisition of detailed spectral information across a wide range of wavelengths.
\hsi is now used across diverse sectors, including industry, climate science, medicine, and agriculture~\cite{khan2018modern}.

\begin{figure}[t]
    \centering
    \begin{subfigure}{\textwidth}
        \centering
        \includegraphics[width=\textwidth]{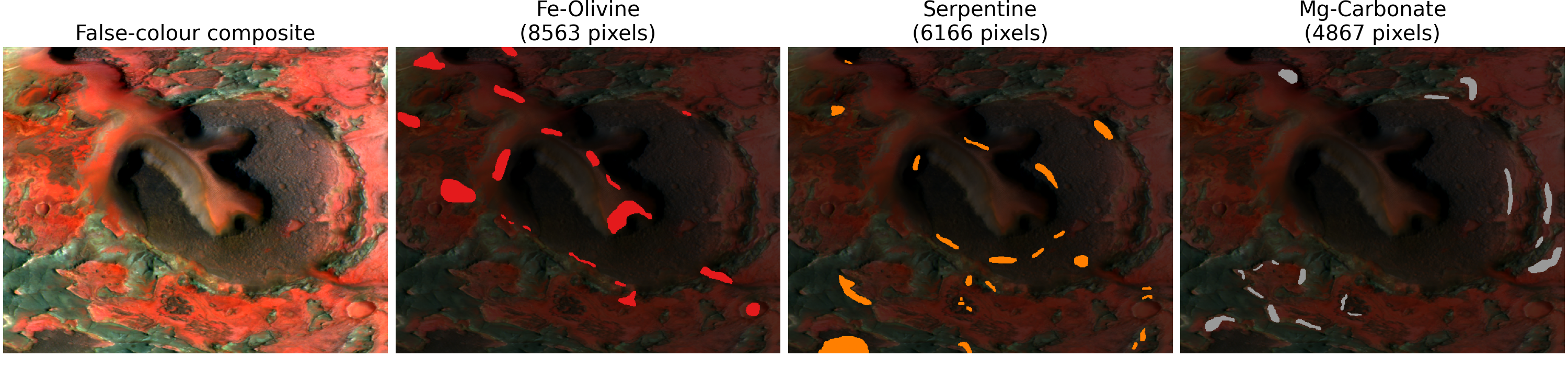}
        \caption{Visualisation of the three most represented mineral classes.}
    \end{subfigure}
    \begin{subfigure}{\textwidth}
        \centering
        \includegraphics[width=0.65\textwidth]{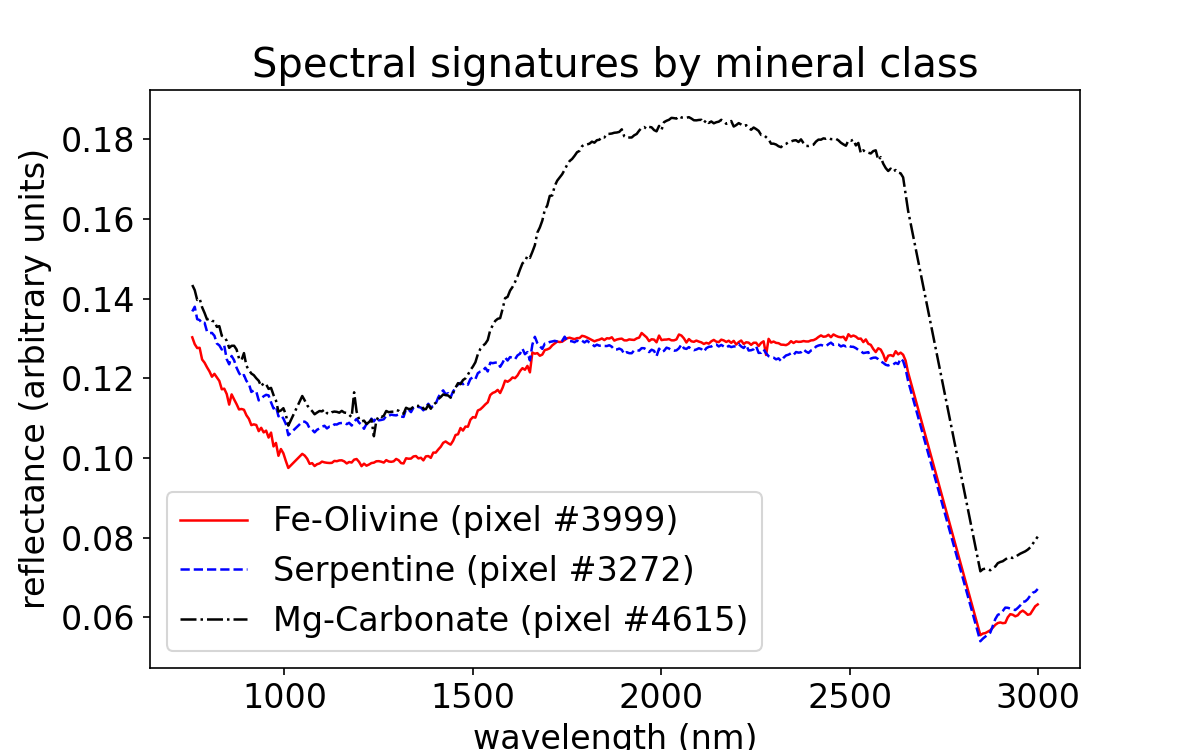}
        \caption{%
            Representative spectrum extracted from the dataset
            (\emph{Fe-Olivine class}).
        }
    \end{subfigure}
    \caption{%
        Visualisation of the hyperspectral cube for the \emph{Nili Fossae}
        formation on the Martian surface.
    }
    \label{fig:nfviz}
\end{figure}

From a data perspective, a hyperspectral image is structured as a high-dimensional data cube, commonly called the \emph{hyperspectral cube}.
Like any experimental acquisition, an \hsi is inherently corrupted by noise from multiple sources: electronic interference (often manifesting as quasi-white noise), photon (shot) noise, which follows Poisson statistics and is signal-dependent, unlike additive Gaussian noise, impulse (\textit{salt-and-pepper}) noise arising from sensor defects or transmission errors, receiver calibration errors, and localised pixel failures~\cite{ilesanmi2021methods}.
A robust denoising procedure is therefore a critical prerequisite for hyperspectral data analysis.
This task is made significantly more challenging by two properties of \hsi data:
\begin{enumerate}
    \item The high dimensionality of the data amplifies noise redundancy, particularly because of strong correlations between adjacent spectral bands.
    \item The multiplicity of noise sources, which can interact in complex ways and further complicates the denoising process.
\end{enumerate}

Denoising methods must be adapted to these specific constraints.
One of the most widely used approaches is the \mnf transform~\cite{green1988MNF}, a linear dimensionality reduction technique extending \pca.
The \mnf transform proceeds in two steps~\cite{green1988MNF}: first, the noise covariance matrix is estimated (typically from spatial differences of neighbouring pixels) and used to decorrelate and whiten the noise.
Second, a standard \pca rotation is applied to the noise-whitened data, ordering the resulting components by decreasing \snr rather than by variance alone.
The data are thereby decomposed into components with progressively lower \snr.
While applying \pca, or subsequent classification algorithms, to the high-\snr components generally yields effective results, standard methods may fail on the low-\snr components, where residual signal is too weak to be separated by variance-based techniques yet may still carry diagnostically valuable spectral information.

Before turning to the \gsa analysis, it is useful to spell out the correspondence between the elements of \hsi and the field-theoretic ingredients developed in the previous parts of this review.
In the language of \Cref{sec:gsa}, the data matrix $X \in \mathds{R}^{N \times P}$ has $N$ rows corresponding to pixels (samples) and $P$ columns corresponding to spectral bands (variables).
The \ecm is therefore a $P \times P$ object, and the spectral density of its eigenvalues is what enters the power counting and the construction of the effective action.
The continuum removal and normalisation procedures described below are designed to keep this spectral density as close as possible to the \mpdistr universality class in the absence of signal, so that any departure detected by \gsa can be attributed to genuine mineral-specific correlations rather than to instrumental artefacts.
The standard classification task (assigning a mineral label to each pixel) is a separate problem.
The original analysis of this dataset relies on a dedicated neural-network pipeline~\cite{xi2025mctgcl} and is outside the scope of the present review.
Here we take the labels as given and focus on what \gsa can say about the correlation structure of each mineral class independently.

This low-\snr subspace is precisely the regime in which the \gsa framework is designed to operate.
Unlike the images studied in \Cref{sec:gsa}, where the signal manifests through geometric patterns in the spatial domain, \hsi data embed signal in the \emph{colourimetric} (spectral) domain: different minerals produce distinct spectral signatures whose correlation structure is probed by \gsa directly on the \ecm, without requiring prior knowledge of absorption band positions or reference spectral libraries.
To illustrate how \gsa performs in this context, we conduct a controlled experiment on mineral spectra extracted from the \emph{} benchmark dataset~\cite{xi2025mctgcl,ximars2025}.

For this study, we use one annotated image from the \emph{Nili Fossae} region of the Martian surface, acquired by the \emph{Compact Reconnaissance Imaging Spectrometer for Mars} (\textsc{crism}) instrument aboard the Mars Reconnaissance Orbiter~\cite{murchie2007crisminstrument}.
As shown in \Cref{fig:nfviz}, the image contains several distinct mineral formations.
For our analysis, we select the three most represented classes: two silicates (\emph{Fe-Olivine} and \emph{Serpentine}) and one carbonate (\emph{Mg-Carbonate}).
Pixels belonging to each class are extracted into separate spectral matrices.
In all three cases the number of spectral bands is $C = P = 343$, so that the correlation matrix is of size $343 \times 343$.
For the three classes considered, the number of pixels and the resulting sample-to-band ratio are: $N_{\text{Fe-Olivine}} = 8563$ and $q = P/N \simeq 0.04$ for Fe-Olivine, $N_{\text{Serpentine}} = 6166$ and $q \simeq 0.06$ for Serpentine, $N_{\text{Mg-Carbonate}} = 4867$ and $q \simeq 0.07$ for Mg-Carbonate.
We note that the dataset is provided with standard atmospheric and photometric corrections already applied~\cite{xi2025mctgcl}.
These corrections have a limited impact on the \gsa analysis, which depends primarily on the correlation structure of the spectra rather than on absolute reflectance values.
Spectra are pre-processed following standard reflectance spectroscopy protocols~\cite{clark1984convexhull}.
First, we apply continuum removal via the upper convex hull method: for each spectrum, the convex envelope spanning the local reflectance maxima is computed, and the spectrum is divided by this envelope.
This normalisation isolates narrow absorption features from the broadband continuum, emphasising the mineral-specific spectral signatures while suppressing variations due to albedo and illumination geometry.
Second, each continuum-removed spectrum is normalised to unit integral, ensuring that the subsequent \ecm construction reflects correlation structure rather than absolute reflectance magnitude.
We then follow the \gsa procedure described in \Cref{sec:gsa}: isolated spikes are removed using exactly the same protocol as in that section, which restricts the analysis to the continuous bulk of the eigenvalue spectrum.
The key point is that this spike removal is applied \emph{by construction}, so that any structure detected by \gsa must originate from extensive-rank correlations within the spectral bulk.
Gaussian noise is then injected at varying \snr $\beta$ to probe the detectability of each mineral class.
The empirical canonical dimensions are evaluated at the same IR reference scale $k^2_{\text{IR}}$ defined in \Cref{sec:gsa}.
By this construction, the spectra studied here contain no isolated eigenvalue spikes.
The analysis is, by design, a \enquote{post-\pca} analysis.
Standard \pca-based methods excel at extracting the isolated spikes that carry the dominant low-rank signal component, and in a real application those spikes would already have been identified and subtracted.
\gsa operates on the residual continuous spectrum that \pca leaves behind, the regime where spike-based benchmarks are no longer applicable.
Consequently, a direct quantitative comparison with \pca-based benchmarks is not meaningful here.
The two approaches address complementary portions of the eigenvalue spectrum.
The relevant question is whether \gsa can detect residual spectral structure \emph{beyond} what \pca already captures, and the results below show that it can.
\begin{figure}[t]
    \centering
    \begin{subfigure}{0.45\textwidth}
        \centering
        \includegraphics[width=\textwidth]{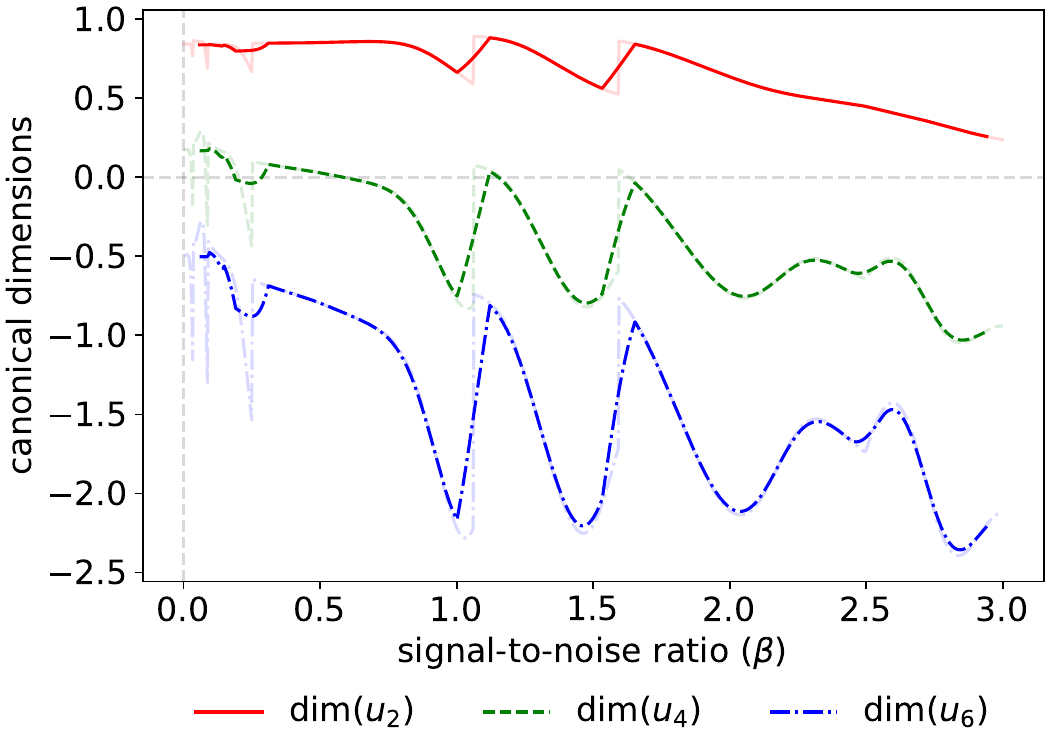}
        \caption{Fe-Olivine}
    \end{subfigure}
    \begin{subfigure}{0.45\textwidth}
        \centering
        \includegraphics[width=\textwidth]{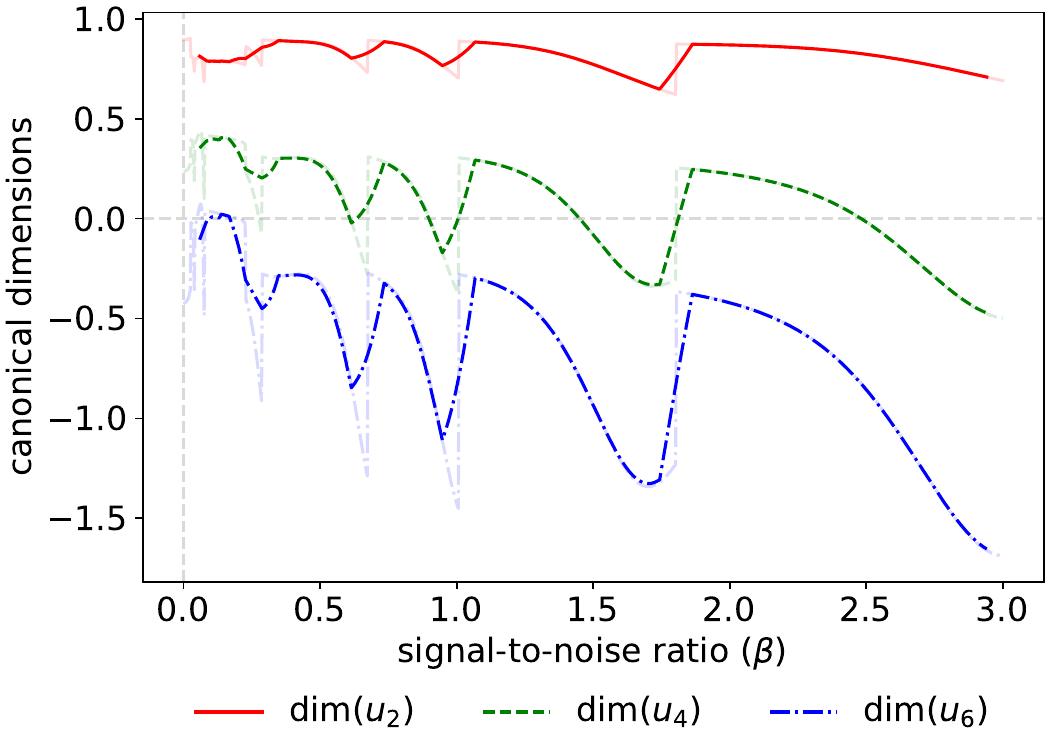}
        \caption{Serpentine}
    \end{subfigure}
    \begin{subfigure}{0.45\textwidth}
        \centering
        \includegraphics[width=\textwidth]{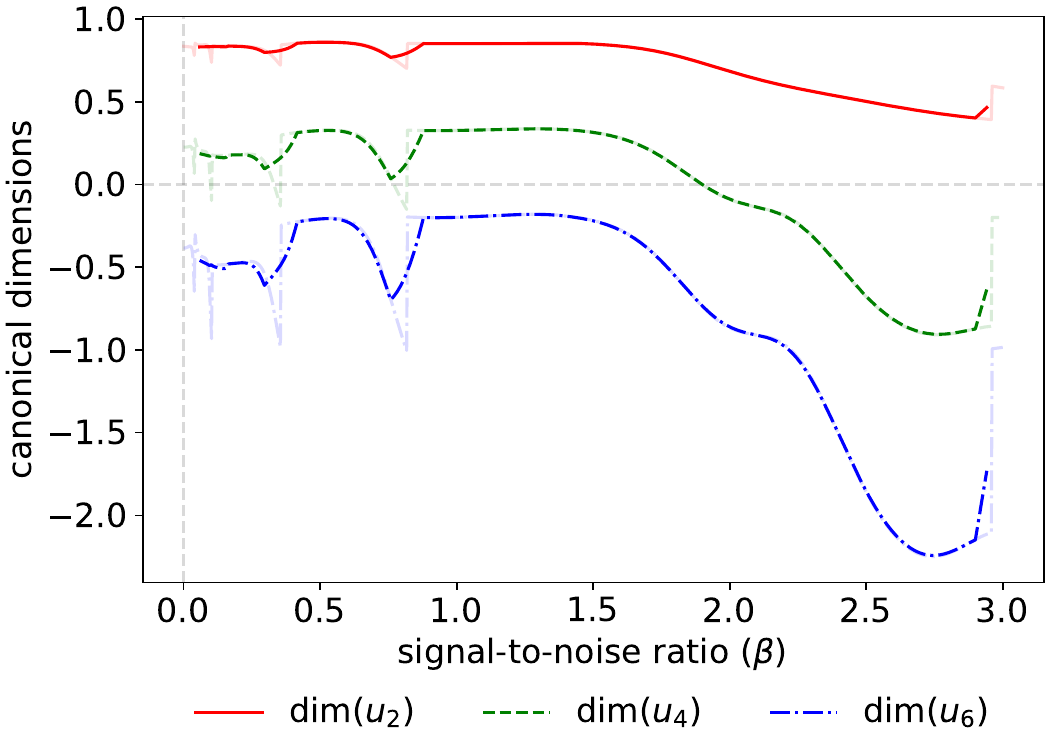}
        \caption{Mg-Carbonate}
    \end{subfigure}
    \caption{%
        Canonical dimensions of the three mineral classes as functions of the
        injected \snr $\beta \in [0, 3]$.
    }
    \label{fig:can_dim_mars}
\end{figure}

\Cref{fig:can_dim_mars} shows the canonical dimensions for each mineral class as a function of $\beta$, and \Cref{tab:hsi_thresholds} summarises the detection thresholds defined in \Cref{detectionTh}.
The two silicate classes exhibit qualitatively similar behaviour, with multiple visible peaks and distinct threshold values.
Fe-Olivine displays a rigidity threshold already at very low \snr ($\beta_t \simeq 0.05$), followed by a critical threshold at $\beta_c \simeq 0.15$ where $\dim_{\tau}(u_4)$ first crosses zero.
Its optimal threshold $\beta_O$, the first pronounced local minimum of $\dim_{\tau}(u_4)$, is located near $\beta_O \simeq 1.0$.
Serpentine shows a somewhat higher rigidity threshold ($\beta_t \simeq 0.1$) and a critical threshold at $\beta_c \simeq 0.6$, with its optimal threshold also in the vicinity of $\beta_O \simeq 1.0$.
Mg-Carbonate, which is notoriously difficult to detect in standard spectral analyses, exhibits a markedly different profile: despite an early rigidity loss at $\beta_t \simeq 0.1$, its critical threshold occurs much later at $\beta_c \simeq 1.9$, and the optimal threshold is reached only at $\beta_O \simeq 2.8$.

\begin{table}[t]
    \caption{Detection thresholds for the three mineral classes studied.}\label{tab:hsi_thresholds}
    \centering
    \begin{tabular}{@{}lccc@{}}
        \toprule
        \textbf{Mineral class} & $\beta_t$ (rigidity) & $\beta_c$ ($\dim_{\tau}(u_4)=0$) & $\beta_O$ (optimal) \\
        \midrule
        Fe-Olivine             & $\simeq 0.05$        & $\simeq 0.15$                    & $\simeq 1.0$        \\
        Serpentine             & $\simeq 0.10$        & $\simeq 0.60$                    & $\simeq 1.0$        \\
        Mg-Carbonate           & $\simeq 0.10$        & $\simeq 1.9$                     & $\simeq 2.8$        \\
        \bottomrule
    \end{tabular}
\end{table}

The differences between the three classes can be understood, at least qualitatively, from the spectral characteristics of each mineral.
Fe-Olivine exhibits deep, narrow absorption bands near \SI{1}{\micro\meter}, which produce sharp spectral features that remain distinguishable even at low \snr.
The corresponding eigenvalue spectrum develops correlated structure already at small $\beta$.
Serpentine displays broader hydration features near \SI{1.4}{\micro\meter} and \SI{2.3}{\micro\meter}, requiring higher \snr for reliable discrimination from the noise background.
Mg-Carbonate is characterised by particularly broad absorption bands near \SI{2.3}{\micro\meter} and \SI{2.5}{\micro\meter}: the larger spectral support of these features distributes the signal across a wider eigenvalue range, so that a greater number of eigenvectors is needed to capture the correlated structure.
In the \rg framework, this would translate to a prolonged dimensional phase transition~\cite{Landau2023}: the effective dimension remains above 4 over a wider interval of $\beta$, producing the long-tailed decay of $\dim_{\tau}(u_4)$ observed in \Cref{fig:can_dim_mars}.
This interpretation is consistent with our numerical results, although a fully quantitative mapping between absorption band width and dimensional phase transition length has yet to be established.

The presence of multiple peaks in the canonical dimension profiles, observed for both silicate classes, is consistent with the eigenvalue phase transition picture developed in~\cite{Landau2023} and in our previous work.
In the absence of isolated spikes, which are systematically removed in our pre-processing, these peaks suggest successive delocalisation transitions of eigenvector subsets as $\beta$ increases.
One natural interpretation is that each peak corresponds to a distinct spectral sub-feature of the mineral under study, such as a different crystalline phase or a different illumination geometry across the region.
The observation that the three classes display distinct peak structures is, on its own, a non-trivial result: the framework can distinguish between minerals on the basis of their \ecm structure alone, without invoking reference spectra or band positions.

\begin{remark}{Multiple peaks and successive noise components}{hsi_rem_cycles}
    The same peak structure is reminiscent of the cyclic behaviour of the canonical dimension analysed in \Cref{Sec4-1} for realistic images.
    In that context, successive peaks in $\dim_{\tau}(u_4)$ were associated with the progressive decoupling of independent noise components as the \snr increases.
    The fact that the \hsi data, which are of entirely different nature, exhibit a qualitatively similar pattern provides additional support for the interpretation that the peaks reflect the sequential emergence of structured sub-populations of eigenvalues.
    In the present case, the most natural candidates for these sub-populations are the distinct absorption features that characterise each mineral.
    The number and position of the peaks, rather than their exact height, may therefore carry the diagnostic information.
\end{remark}

These transitions occur well below the \bbp threshold~\cite{math3} (see \Cref{AppD}), which for the three sample-to-band ratios considered here would lie in the range $\beta_{\text{BBP}} = q^{1/4} \in [0.49, 0.56]$.
The corresponding \gsa thresholds ($\beta_t \simeq 0.05$--$0.10$) are an order of magnitude smaller, confirming that \gsa detects extensive-rank spectral structure while the eigenvalues in question are still fully embedded within the \mpdistr bulk.

The \frg framework, and \gsa in particular, thus reveals diagnostically useful spectral features from a perspective that is complementary to standard classification and spectral-matching techniques.
Whereas on standard images \gsa detects signal through geometric patterns that emerge as correlated pixel domains, here it probes the colourimetric correlation structure directly, distinguishing mineral classes by the spectral organisation of their absorption features rather than by spatial morphology.
The use case presented here focuses on planetary science, but the methodology extends naturally to other fields where spectral analysis is relevant, including remote sensing of the Earth and astrophysical observations~\cite{BioucasDias2012HSIunmixing}.


\section*{Acknowledgments}

The authors acknowledge support from the COMETA COST Action \href{https://www.cost.eu/actions/CA22130/}{CA22130}.
V.L., for his part, extends his warm thanks to the swan of destiny.

The authors also raise the alarm over the fate of our colleague and friend Dine Ousmane Samary, victim of arbitrary detention in their home country of Benin.
Interested readers can, if they so choose, sign the following petition for their release: \url{https://c.org/JBH2cJ6qhd}.
\clearpage

\printbibliography[heading=bibintoc]


\clearpage
\appendix

\section{Maximum Entropy Estimator}\label{App0}
The maximum entropy estimator~\cite{amari2016information,jaynes1982rationale} is the least informative statistical inference compatible with a prescribed set of expectation-value constraints.
Entropy maximisation is a standard reference prescription in statistical inference: among all densities consistent with the prescribed expectation values, it selects the one carrying the smallest amount of additional information beyond the constraints, and is therefore the natural default in the absence of further prior knowledge about the system.

\subsection{Generalities}

In brief, the principle amounts to considering a system described by a continuous variable $q \in \mathcal{X}$, where $\mathcal{X}$ is a given sample space (e.g.\ $\mathcal{X} \equiv \mathds{R}^N$), and governed by an unknown probability density $p(q)$.
Suppose further that we are given $K$ expectation-value constraints: for a collection of real-valued test functions $c_1(q), c_2(q), \dots, c_K(q)$, the corresponding expectations under $p$ are specified as
\begin{equation}
	\int_{\mathcal{X}}\, \dd q\, p(q) c_l(q)
	=
	\mathds{E}_{p}[c_l(q)]
	=
	a_l,
	\qquad 1 \leq l \leq K.
	\label{constMEE}
\end{equation}
The maximum-entropy estimate $P(q)$ is the probability density that:
\begin{enumerate}
	\item maximises the Shannon (differential) entropy
	      \begin{equation}
		      S[p] \defeq -\int \dd q\, p(q) \ln p(q),
		      \label{entropy}
	      \end{equation}
	\item reproduces the constraints~\eqref{constMEE} together with the normalisation $\int_{\mathcal{X}} \dd q\, p(q) = 1$.
\end{enumerate}
The two requirements are combined by extremising the augmented functional
\begin{equation}
	I[p] \defeq S[p] - \sum_{l=0}^{K} h_l \left(\int_{\mathcal{X}} \dd q\, p(q) c_l(q) - a_l\right),
	\label{funIMEE}
\end{equation}
where the sum is extended to $l=0$ with the convention $c_0(q) \equiv 1$ and $a_0 \equiv 1$, so that the corresponding Lagrange multiplier $h_0$ enforces the normalisation of the probability distribution.
At the extremum the functional derivative of $I$ with respect to $p$ vanishes pointwise:
\begin{equation}
	\eval{\frac{\delta I}{\delta p}}_{p=P} = 0.
	\label{extcond}
\end{equation}
The integrand of $I[p]$ is the pointwise Lagrangian
\begin{equation}
	\mathcal{L}[p](q) = -p(q) \ln p(q) - \sum_{l=0}^{K} h_l\, p(q) c_l(q) + \sum_{l=0}^{K} h_l a_l,
	\label{LagMEE}
\end{equation}
in which the test functions $c_l(q)$ and the multipliers $h_l$ are independent of $p$.
The functional derivative reduces to the partial derivative of $\mathcal{L}$ with respect to the field $p(q)$,
\begin{equation}
	\frac{\partial \mathcal{L}}{\partial p}
	= -\ln p(q) - 1 - \sum_{l=0}^{K} h_l\, c_l(q).
	\label{dLdpMEE}
\end{equation}
Setting~\eqref{dLdpMEE} to zero and solving for $p(q)$ yields
\begin{equation}
	p(q) = \exp\!\left(-1 - \sum_{l=0}^{K} h_l c_l(q)\right) = e^{-1-h_0}\, \exp\!\left(-\sum_{l=1}^{K} h_l c_l(q)\right).
	\label{Pexpr}
\end{equation}
The overall prefactor $e^{-1-h_0}$ is fixed by the normalisation $\int_{\mathcal{X}} \dd q\, P(q) = 1$ and defines the partition function
\begin{equation}
	Z \defeq \int_{\mathcal{X}} \dd q\, \exp\!\left(-\sum_{l=1}^{K} h_l c_l(q)\right),
	\label{ZdefMEE}
\end{equation}
so that the maximum-entropy density takes the canonical exponential form
\begin{equation}
	P(q) = \frac{1}{Z} \exp\!\left(-\sum_{l=1}^{K} h_l c_l(q)\right).
	\label{MEEfinal}
\end{equation}
The second variation of $I[p]$ at the critical point is strictly negative,
\begin{equation}
	\eval{\frac{\delta^{2} I}{\delta p^{2}}}_{p = P} = -\frac{1}{P(q)} < 0,
\end{equation}
which confirms that the stationary point is a strict maximum of the entropy under the stated linear constraints.

\subsection{The Boltzmann Distribution}

As a concrete illustration of the maximum-entropy principle, we derive the Boltzmann--Gibbs distribution of equilibrium statistical mechanics.
Consider a large physical system, such as a gas, with many generalised coordinates $\textbf{x} = (x_1,x_2,\dots, x_N) \in \mathds{R}^N$, governed by an unknown probability density $p(\textbf{x})$ with a prescribed mean energy
\begin{equation}
	\int \dd \textbf{x}\, H(\textbf{x})\, p(\textbf{x})
	=
	E.
	\label{constraintE}
\end{equation}
The function $H(\textbf{x})$ is the Hamiltonian of the system.
We look for the density $P_E(\textbf{x})$ that satisfies~\eqref{constraintE} and the normalisation $\int \dd \textbf{x}\, P_E(\textbf{x}) = 1$ while maximising the entropy~\eqref{entropy}.
The augmented functional to be extremised is
\begin{equation}
	I[p]
	=
	S[p]
	-
	a \left(\int \dd \textbf{x}\, H(\textbf{x})\, p(\textbf{x})-E\right)
	-
	b \left(\int \dd \textbf{x}\, p(\textbf{x})-1\right),
	\label{IBoltz}
\end{equation}
where $a$ is the Lagrange multiplier enforcing the mean-energy constraint~\eqref{constraintE} and $b$ enforces the normalisation.
The pointwise Lagrangian is
\begin{equation}
	\mathcal{L}_E[p](\textbf{x}) = -p(\textbf{x}) \ln p(\textbf{x}) - a H(\textbf{x})\, p(\textbf{x}) + a E - b\, p(\textbf{x}) + b,
\end{equation}
and its partial derivative with respect to $p(\textbf{x})$ is
\begin{equation}
	\frac{\partial \mathcal{L}_E}{\partial p}
	=
	-\ln p(\textbf{x}) - 1 - a H(\textbf{x}) - b.
	\label{dLdpBoltz}
\end{equation}
Setting this to zero and solving for $p(\textbf{x})$ gives
\begin{equation}
	p(\textbf{x})
	=
	\exp\!\left(-1 - a H(\textbf{x}) - b\right)
	=
	e^{-b-1}\, e^{-a H(\textbf{x})}.
	\label{presolBoltz}
\end{equation}
The normalisation $\int \dd \textbf{x}\, p(\textbf{x}) = 1$ fixes $e^{-b-1}$ as the inverse of the partition function
\begin{equation}
	Z_E
	\defeq
	\int \dd \textbf{x}\, e^{-a H(\textbf{x})},
	\label{ZdefE}
\end{equation}
so that the maximum-entropy density takes the canonical Boltzmann--Gibbs form
\begin{equation}
	P_E(\textbf{x})
	=
	\frac{e^{-a H(\textbf{x})}}{Z_E}.
	\label{BoltzmannFinal}
\end{equation}
The Lagrange multiplier $a$ is identified with the inverse temperature $\beta \equiv 1 / (k_B T)$, where $T$ is the absolute temperature and $k_B$ is Boltzmann's constant.

\section{Local Theory and Universality Arguments}\label{App1}
This section summarises the arguments from~\cite{RG1,RG2,RG3,RG4,RG5}, used to justify the local field theory as a first-order approximation at the tail of the spectrum.
The core rationale relies on the concept of universality of matrix models.
For instance, the \mpdistr distribution consistently emerges across a wide array of data analysis contexts, ranging from finance and biology to medicine and machine learning.

\subsection{From Data to Effective Local Field Theory}

\begin{figure}[t]
    \centering
    \begin{tikzpicture}[
        >=Latex,
        scale=0.8,
        every node/.style={scale=0.8}
    ]
    \def\s{1.0}

    \definecolor{spinblue}{RGB}{172,215,230}
    \definecolor{darkblue}{RGB}{0,0,139}
    \definecolor{darkred}{RGB}{139,0,0}

    \tikzset{
    box/.style={draw=spinblue, fill=spinblue!40, thick, rounded corners=2pt},
    boxtext/.style={align=center, text width=5cm},
    arrow/.style={->, thin, black!80},
    darrow/.style={->, thin, black!80, dashed, dash pattern=on 3pt off 3pt},
    rarow/.style={-{Latex}, red, very thick},
    axis/.style={->, thin, black},
    }


    \draw[box]
    ({-6.029*\s}, {2.966*\s}) rectangle ({-1.940*\s}, {4.911*\s})
    node[above left] {\Large $B$}
    node[midway, boxtext] {{\large \textbf{Binary data}} \\[0.35em] {Empirical first and second order cumulants}};

    \draw[box]
    ({1.398*\s}, {2.893*\s}) rectangle ({5.863*\s}, {5.063*\s})
    node[above left] {\Large $C$}
    node[midway, boxtext] {{\large \textbf{Statistical inference}} \\[0.35em] {$p(\sigma) \propto e^{-H(\sigma)}$} \\ {$H(\sigma) = -\frac{1}{2} s_i K_{ij} s_j + h_i s_i$}};

    \draw[box]
    ({4.274*\s}, {-1.758*\s}) rectangle ({7.462*\s}, {1.641*\s})
    node[above left] {\Large $D$}
    node[midway, boxtext, yshift=3em] {{\large \textbf{Denoising}}};

    \draw[box]
    ({-7.417*\s}, {-1.610*\s}) rectangle ({-4.168*\s}, {1.395*\s})
    node[above left] {\Large $A$}
    node[midway, boxtext, yshift=3em] {{\large \textbf{ECM spectrum}}};

    \draw[box]
    ({-5.991*\s}, {-5.058*\s}) rectangle ({-1.902*\s}, {-3.114*\s})
    node[above left] {\Large $B^{\prime}$}
    node[midway, boxtext] {{\large \textbf{Other dataset}} \\[0.35em] {(same universality law)}};

    \draw[box]
    ({-2.3*\s}, {-1.096*\s}) rectangle ({1.995*\s}, {1.075*\s})
    node[above left] {\Large $C^{\prime\prime}$}
    node[midway, boxtext] {{\large \textbf{Effective field theory}} \\[0.45em] {\huge \textbf{?}}};

    \draw[box]
    ({1.399*\s}, {-5.124*\s}) rectangle ({5.864*\s}, {-3.091*\s})
    node[above left] {\Large $C^{\prime}$}
    node[midway, boxtext] {{\large \textbf{Statistical inference}}};


    \draw[axis] ({-6.713*\s}, {-1.294*\s}) -- ({-6.713*\s}, {0.301*\s}) node[right] {$\mu(\lambda)$};
    \draw[axis] ({-6.713*\s}, {-1.294*\s}) -- ({-4.513*\s}, {-1.294*\s}) node[above] {$\lambda$};
    \draw[thick, darkblue, smooth]
    (-6.713*\s, -1.294*\s) .. controls (-6.50*\s, 0.42*\s) and (-6.32*\s, 0.10*\s) .. (-4.725*\s, -1.294*\s);

    \draw[axis] ({4.779*\s}, {-1.294*\s}) -- ({4.779*\s}, {0.384*\s}) node[right] {$\mu(\lambda)$};
    \draw[axis] ({4.779*\s}, {-1.294*\s}) -- ({6.979*\s}, {-1.294*\s}) node[above] {$\lambda$};
    \draw[thick, darkblue, smooth]
    (4.779*\s, -1.294*\s) node[below] {\footnotesize \uv} .. controls (4.99*\s, 0.17*\s) and (5.17*\s, 0.18*\s) .. (6.751*\s, -1.294*\s) node[below] {\footnotesize \ir};

    \draw[thin, red, dashed, dash pattern=on 4pt off 2pt]
    ({6.5*\s}, {-1.321*\s}) -- ({6.5*\s}, {-0.804*\s})
    node[above] {\small $\Lambda$};


    \draw[arrow] ({-5.857*\s}, {1.726*\s}) -- ({-4.105*\s}, {2.777*\s});
    \draw[arrow] ({-1.762*\s}, {3.749*\s}) -- ({1.061*\s}, {3.751*\s});
    \draw[arrow] ({4.160*\s}, {2.782*\s}) -- ({5.125*\s}, {1.892*\s});
    \draw[arrow] ({-5.782*\s}, {-1.867*\s}) -- ({-4.057*\s}, {-3.003*\s});
    \draw[arrow] ({-1.725*\s}, {-4.251*\s}) -- ({1.185*\s}, {-4.293*\s});
    \draw[arrow] ({3.396*\s}, {-2.945*\s}) -- ({5.794*\s}, {-1.930*\s});
    \draw[darrow] ({3.278*\s}, {-2.931*\s}) -- ({-0.114*\s}, {-1.227*\s});
    \draw[darrow] ({3.238*\s}, {2.773*\s}) -- ({-0.246*\s}, {1.240*\s});


    \draw[rarow] ({2.051*\s}, {-0.017*\s}) -- ({4.180*\s}, {-0.015*\s});
    \draw[rarow] ({-3.978*\s}, {0.013*\s}) -- ({-2.448*\s}, {0.015*\s});
\end{tikzpicture}
    \caption{%
        Illustration of the general strategy.
        In the vicinity of the universal \mpdistr class, there must exist a universal field theory capable of determining the cutoff $\Lambda$.
    }\label{figplan}
\end{figure}
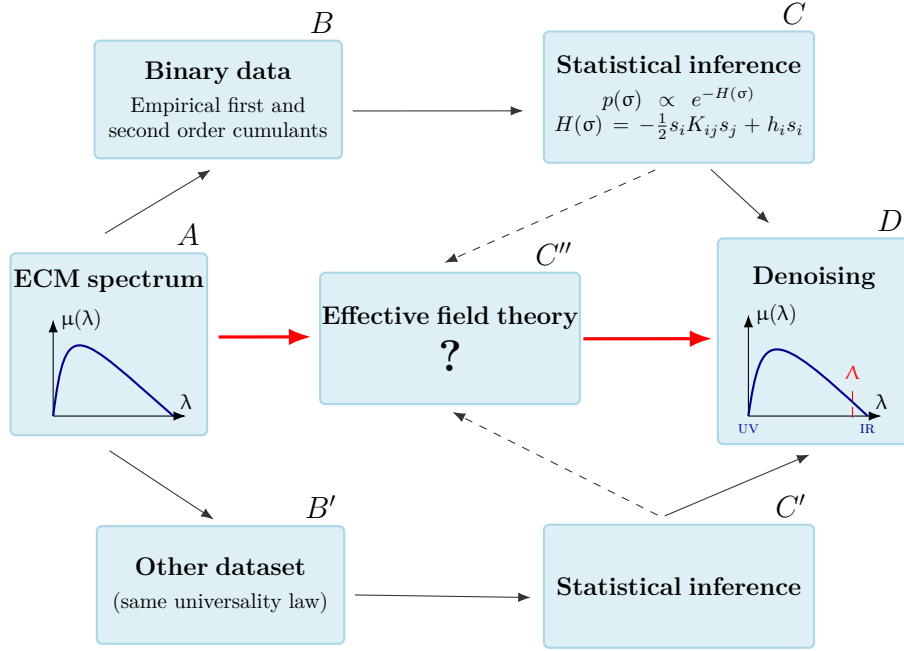

Near a universality class, the eigenvalue distribution of the correlation matrix becomes entirely agnostic of the actual nature of the microscopic \dof.
This remains true whether we deal with complex dynamics of financial markets or with biological cellular correlations.
Therefore, if an effective theory can be constructed for a particular case, successfully isolating the signal from the noise in the tail of the spectrum, it must be able to perform this exact separation for any distribution belonging to the same universality class.
\Cref{figplan} summarises our overall strategy, taking binary data as a starting point toward a universal field theory.

The set of spectra in the vicinity of the \mpdistr represents a broad family of data types.
In the case of binary data (discussed below), one can infer a maximum entropy effective distribution and aim to identify a cutoff $\Lambda$ separating signal from noise (steps $B$-$C$-$D$).
Although this pathway can be mapped independently for all possible data (steps $B^\prime$-$C^\prime$-$D$), the underlying universality suggests the existence of a more general theory ($C^{\prime\prime}$).
This framework depends exclusively on the spectrum, yielding a common approximation for all specific inference cases in the neighbourhood of the \mpdistr universality class.

\begin{figure}[t]
    \centering
    \begin{tikzpicture}
    \def\s{1.0}
    \definecolor{spinblue}{RGB}{172,215,230}

    \draw[thin, black!60]
    ({-6.32*\s}, {-1.78*\s})
    -- ({-3.10*\s}, {-1.30*\s});
    \draw[thin, black!60]
    ({-3.10*\s}, {-1.30*\s})
    -- ({-3.40*\s}, {0.44*\s});
    \draw[thin, black!60]
    ({-3.10*\s}, {-1.30*\s})
    -- ({0.16*\s}, {-2.22*\s});
    \draw[thin, black!60]
    ({-3.10*\s}, {-1.30*\s})
    -- ({3.40*\s}, {-0.44*\s});
    \draw[thin, black!60]
    ({-1.08*\s}, {-1.04*\s})
    -- ({1.22*\s}, {1.02*\s});
    \draw[thin, black!60]
    ({1.22*\s}, {1.02*\s})
    -- ({6.32*\s}, {1.78*\s});
    \draw[thin, black!60]
    ({1.22*\s}, {1.02*\s})
    -- ({-0.16*\s}, {2.22*\s});
    \draw[thin, black!60]
    ({-6.30*\s}, {-1.79*\s})
    -- ({0.14*\s}, {-2.22*\s});

    \foreach \x/\y/\d in {
    {-6.30}/{-1.79}/1,    
    {-3.42}/{0.43}/1,     
    {1.20}/{1.07}/1,      
    {6.39}/{1.82}/1,      
    {-1.08}/{-0.94}/1,    
    {-3.10}/{-1.26}/-1,   
    {0.14}/{-2.22}/-1,    
    {-0.17}/{2.18}/-1,    
    {3.39}/{-0.43}/-1     
    } {
    \pgfmathsetmacro{\px}{\x*\s}
    \pgfmathsetmacro{\py}{\y*\s}

    \draw[fill=spinblue!50, draw=spinblue] (\px,\py) circle (0.4);

    \ifnum\d=1
        \draw[-{Latex}, red, very thick] (\px,\py-0.3) -- (\px,\py+0.3);
    \else
        \draw[-{Latex}, red, very thick] (\px,\py+0.3) -- (\px,\py-0.3);
    \fi
    }
\end{tikzpicture}
    \caption{%
        A typical neural network can be modelled as a spin system, where the activity of each neuron corresponds to the \enquote{up} or \enquote{down} orientation of a spin (indicated by the red arrows), conventionally assigned the values $S=\pm 1$.}\label{figreseau}
\end{figure}
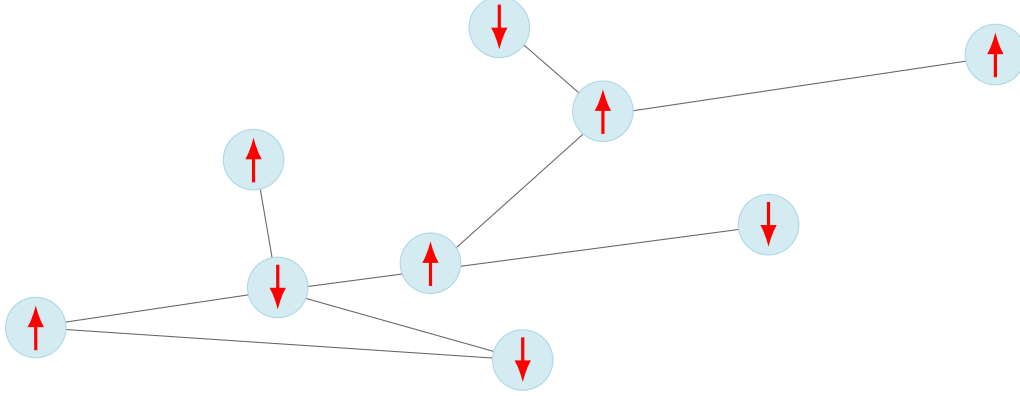

A rudimentary model of biological or artificial neural activity (\Cref{figreseau}) provides a first example of a system where one can infer an effective distribution capable of isolating a signal within the vicinity of the \mpdistr.
Furthermore, this class of problems frequently maps onto a lattice Ising-like model, whose macroscopic behaviour is closely related to the dynamics detailed in \Cref{fromisingTofield}.

Typically, the empirical information regarding such systems is encoded in the \ecm~\eqref{corr}.
The stochastic nature of the data, such as neural activity, is fundamentally governed by an unknown underlying probability distribution, which can be approximated via statistical inference.
The approach requiring the fewest a priori assumptions is the maximum entropy principle (see \Cref{App0}), which yields the least biased distribution consistent with the available information.
For a given spin configuration $\mathcal{S} \defeq \{S_1,S_2,\cdots,S_N\}$, this parameterized probability distribution $p(\mathcal{S}, \theta)$ must satisfy the following empirical constraints (assuming a zero mean):
\begin{equation}
    \langle S_i S_j \rangle \defeq \sum_{\mathcal{S}} \, S_i S_j p(\mathcal{S},\theta) = C_{ij}.
\end{equation}
The maximum entropy principle allows us to explicitly write the probability distribution $p(\mathcal{S},\theta)$:
\begin{equation}
    p(\mathcal{S},\theta)= \frac{1}{Z} \exp \left(\sum_{i,j}\, S_i K_{ij} S_j\right),
\end{equation}
where the partition function, which normalises the probability distribution, is defined as the sum over all possible spin configurations:
\begin{equation}
    Z
    \defeq
    \sum_{\mathcal{S}}\, \exp \left(\sum_{i,j}\, S_i K_{ij} S_j\right).
\end{equation}

The maximum entropy model has been considered in a series of papers \cite{Bial1,Bial2,Bial3} to describe the electrical activity of neurons in the brain.
Experimentally, the pairwise correlation function is estimated by constructing discrete-time samples $\delta x_i(t)$ for each neuron $i$ (with zero mean).
To extract the relevant macroscopic features of the distribution, the authors employ a specific coarse-graining procedure, analogous to a block-spin partial integration.
More precisely, they construct a coarse-grained distribution $\bar{p}(\tilde{\sigma})$ by replacing the original microscopic variables $S_i$ with:
\begin{equation}
    S_i \to \tilde{S}_i
    \propto
    \sum_{j=1}^N\left(\sum_{\mu=1}^\Lambda u_{i}^{(\mu)} u_j^{(\mu)} \right)S_j,
\end{equation}
where $u_{i}^{(\mu)}$ is the $i$-th component of the normalised eigenvector associated with the eigenvalue $\lambda_\mu$ of $K$.
The effective distribution $\bar{p}(\tilde{\sigma})$ is then obtained by averaging over the spin configurations $\mathcal{S}$ while holding $\tilde{\sigma}=\{\tilde{S}_i \}$ fixed.

However, for our purposes, this probability distribution remains overly sensitive to the specific problem at hand.
We can achieve greater generality by introducing a continuous variable and applying the standard heuristic argument used to derive Landau theory from the Ising model in the vicinity of a phase transition.
This procedure is standard and extensively documented in textbooks, such as \cite{de2006random}.
It relies on the following identity:
\begin{equation}
    \int \dd x \,e^{k x} e^{-\frac{1}{2 \sigma^2} x^2} \propto e^{\frac{1}{2} \sigma^2 k^2}.
\end{equation}
Hence, introducing $N$ real variables $\phi_i \in \mathds{R}$, we have:
\begin{equation}
    \exp \left(\sum_{i,j}\, S_i K_{ij} S_j\right) \propto \int_{-\infty}^{+\infty} \left[\prod_{i=1}^N \dd \phi_i\right] \, e^{-\frac{1}{2}\sum_{i,j}^N\phi_i K_{ij}^{-1} \phi_j} \prod_{j=1}^N e^{\phi_j S_j}.
\end{equation}
Summing over the remaining spins, we get:
\begin{equation}
    Z=\int_{-\infty}^{+\infty} \left[\prod_{i=1}^N \dd \phi_i\right] \, \exp\left(-\frac{1}{2}\sum_{i,j}^N\phi_i K_{ij}^{-1} \phi_j+\sum_{i=1}^N \ln (2\cosh(\phi_i))\right),
    \label{partitooninf1}
\end{equation}
and in terms of these new variables $\phi$, the correlation functions between spins read as:
\begin{equation}
    C_{ij}=\langle \tanh (\phi_i) \tanh(\phi_j) \rangle.
\end{equation}

\begin{remark}{Inverse Matrix Justification}{inverseK}
    It should be noted that this formal construction \cite{de2006random,itzykson1991statistical1,itzykson1991statistical2}, which justifies the transition to a field-theoretical description from the discrete Ising model, rests upon the existence of the inverse matrix $K^{-1}$.
    For the Ising model, this procedure is rigorously justified only on large scales and in the vicinity of the phase transition, where the Fourier transform
    \begin{equation}
        K(\vec{p}) \propto 2 \sum_{\ell=1}^{D} \cos p_\ell,
    \end{equation}
    admits the expansion $K(\vec{p}) \propto 2D - (\vec{p})^2 + \mathcal{O}(p_\ell^{4})$ for $(\vec{p})^2 \to 0$.
    This quadratic form is clearly invertible (see \Cref{part1}).
    We assume that this manipulation remains valid at least in the \ir tail of the spectrum.
    While this derivation does not require strict mathematical rigour, it is motivated by the principle of \rg universality: the essential features of the effective description are governed primarily by the form and structure of the interactions.
\end{remark}

If we suppose that fluctuations around the mean field are sufficiently small, $\cosh(\phi_i)$ can be expanded as a power series in $\phi_i$ within the symmetric phase, provided the \ir mass is large enough.
Indeed, close to the Gaussian fixed point, the expectation values of field products can be evaluated using Wick's theorem (see \Cref{thm:theoremWick}).
However, the resulting expression establishes a convoluted relationship between the empirical correlation matrix $C_{ij}$ and the effective two-point function of the model, $G_{ij}$, defined as:
\begin{equation}
    G_{ij}
    \defeq
    \frac{1}{Z}\int_{-\infty}^{+\infty} \left[\prod_{k=1}^N \dd\phi_k\right] \phi_i\phi_j \exp\left(-\frac{1}{2}\sum_{l,m}^N\phi_l K_{lm}^{-1} \phi_m+\sum_{l=1}^N \ln (2\cosh(\phi_l))\right).
\end{equation}
While assuming a sufficiently large mass simplifies this relation to $G_{ij} \approx C_{ij}$, dealing with this formulation remains analytically cumbersome.
A more elegant and direct alternative is to perform the change of variable
\begin{equation}
    \Phi_i \defeq \tanh \phi_i.
\end{equation}
The Jacobian of this transformation yields:
\begin{equation}
    \dd \Phi_i = (1-\Phi_i^2)\,\dd \phi_i \quad \implies \quad \dd \phi_i = \frac{\dd \Phi_i}{1-\Phi_i^2} = \dd \Phi_i \exp\left(-\ln (1-\Phi_i^2)\right).
\end{equation}
Using the identity $\cosh \phi_i= \cosh(\arctanh\,\Phi_i)=(1-\Phi_i^2)^{-1/2}$, the partition function \eqref{partitooninf1} can be rewritten as:
\begin{equation}
    Z=\int_{-1}^{+1} \left[\prod_{k=1}^N \dd\Phi_k\right] \exp\left(-\frac{1}{2}\sum_{l,m}^N \arctanh(\Phi_l) K_{lm}^{-1} \arctanh(\Phi_m)-\frac{3}{2}\sum_{l=1}^N \ln (1-\Phi_l^2)\right).
    \label{partitooninf2}
\end{equation}
Once again, the logarithm can be expanded in terms of local interactions $\sum_i \Phi_i^{2n}$.
Expanding the $\arctanh$ contributions, we recover the standard Gaussian kinetic term $\frac{1}{2}\sum_{i,j}^N \Phi_i K_{ij}^{-1} \Phi_j$, supplemented by interactions involving higher even powers of the fields that couple directly to the kernel $K^{-1}$.
In standard field theory, these non-local couplings correspond to derivative interactions (proportional to powers of the Laplacian $\Delta$ or generalised momenta $p^2$) that become irrelevant in the deep \ir limit.
In this macroscopic limit, $K^{-1}$ is dominated by the mass term, meaning these momentum-dependent couplings simply shift the bare couplings generated by the expansion of the logarithm.
In the symmetric phase, the fields $\Phi_i$ fluctuate closely around zero.

\begin{remark}{Positivity of the Quartic Coupling}{positiveCondition}
    It might be surprising that in the integral \eqref{partitooninf2}, the expansion of $\ln (1-\Phi_l^2)$ (i.e.\ $\ln (1-\Phi_l^2) = -\Phi_l^2 - \Phi_l^4/2 - \mathcal{O}(\Phi_l^6)$) appears to have the wrong sign with respect to stability.
    From this perspective, the corrections to the quartic (and higher-order) couplings arising from the mass term are essential.
    For the sake of simplicity, let us assume that $K^{-1}_{ij} = M_0 \delta_{ij}$ (neglecting derivative terms, which thus yields the effective local potential).
    Expanding the $\arctanh(\Phi_i)$ terms gives the following local potential $V(\Phi)$ (see \Cref{equationAnsatz}):
    \begin{equation*}
        V(\Phi)= \left( \frac{1}{2} M_0 - \frac{3}{2} \right) \Phi^2 + \left( \frac{1}{3} M_0 - \frac{3}{4} \right) \Phi^4 + \mathcal{O}(\Phi^6).
    \end{equation*}
    The effective mass is therefore $m^2 \defeq \frac{1}{2} M_0 - \frac{3}{2}$, which vanishes at $M_0 = 3$ (the critical temperature).
    As long as $M_0 > 3$ (above the critical point), the mass is positive.
    The effective quartic coupling $u_4(M_0) \defeq \frac{1}{3} M_0 - \frac{3}{4}$ evaluates to $u_4(3) = 1/4$ at the critical point, proving it is, indeed, positive.
\end{remark}

\subsection{The Renormalisation Group and Landau Theory}

When we construct the Landau theory, we assume that a sufficient number of \rg steps have been performed, and that the coupling constants have evolved from their microscopic definitions.
The universal properties of the theory are dictated not by the specific coupling values, but by the form of the interactions and the spatial dimension.
This follows the locality prescription \cite{RG5}.
This motivates us to define the theory space for the most general field theory at the spectral tail as follows:

\begin{definition}{Theory Space at the Spectral Tail}{theorySpaceSpectralTail}
    We assume the existence of a real variable $\phi_i$ that mimics the effective behaviour of the true microscopic \dof, the statistics of which are described by the following partition function:
    \begin{equation}
        Z(J)
        \defeq
        \int_{-\infty}^{+\infty} \left[\prod_{i=1}^N \dd \phi_i\right] \, \exp\left(-\frac{1}{2}\sum_{i,j}^N\phi_i K_{ij}^{-1} \phi_j-V(\Phi) \right),
        \label{equationAnsatz}
    \end{equation}
    where $\Phi \defeq \{\phi_i\}$ denotes the field configuration, and the relevant part of the potential $V(\Phi)$ is assumed to be a sum of powers of the field evaluated at the same point (i.e.\ a local potential), namely:
    \begin{equation}
        V(\Phi)\defeq \sum_{i=1}^N \left(\sum_{k=1}^\infty \, \frac{u_{2k}}{(2k)!} \, \phi_i^{2k}\, \right),
    \end{equation}
    for a set of coupling constants $\{u_{2k}\}$.
    We finally impose the requirement that the two-point correlation function matches the \ecm:
    \begin{equation}
        \frac{\delta^2}{\delta J_k \delta J_l} \ln Z= C_{ij}.
    \end{equation}
\end{definition}

The theory space defined above remains viable only if the dominant interactions remain stable during the flow, at least within the physical region of interest.
More specifically, we will see that even if other interactions are generated, the local theory forms a stable and relevant subspace in the \rg sense near the Gaussian fixed point.

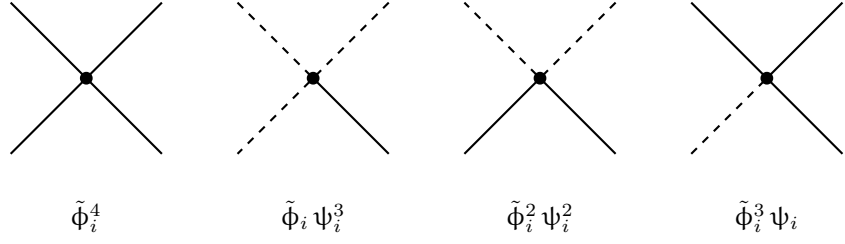
\begin{figure}[t]
    \centering
    \begin{tikzpicture}[thick]
    \draw
    (1.0, 1.0) -- (-1.0, -1.0);
    \draw
    (-1.0, 1.0) -- (1.0, -1.0);
    \fill[fill=black, draw=black] (0.0, 0.0) node[below, yshift=-1.5cm] {$\tilde{\phi}_i^4$} circle (2pt);

    \begin{scope}[xshift=3cm]
        \draw[dashed]
        (1.0, 1.0) -- (-1.0, -1.0);
        \draw[dashed]
        (-1.0, 1.0) -- (0.0, 0.0);
        \draw
        (0.0, 0.0) -- (1.0, -1.0);
        \fill[fill=black, draw=black] (0.0, 0.0) node[below, yshift=-1.5cm] {$\tilde{\phi}_i\, \psi_i^3$} circle (2pt);
    \end{scope}

    \begin{scope}[xshift=6cm]
        \draw[dashed]
        (1.0, 1.0) -- (0.0, 0.0) -- (-1.0, 1.0);
        \draw
        (-1.0, -1.0) -- (0.0, 0.0) -- (1.0, -1.0);
        \fill[fill=black, draw=black] (0.0, 0.0) node[below, yshift=-1.5cm] {$\tilde{\phi}_i^2\, \psi_i^2$} circle (2pt);
    \end{scope}

    \begin{scope}[xshift=9cm]
        \draw
        (1.0, 1.0) -- (0.0, 0.0);
        \draw
        (-1.0, 1.0) -- (1.0, -1.0);
        \draw[dashed]
        (-1.0, -1.0) -- (0.0, 0.0);
        \fill[fill=black, draw=black] (0.0, 0.0) node[below, yshift=-1.5cm] {$\tilde{\phi}_i^3\, \psi_i$} circle (2pt);
    \end{scope}
\end{tikzpicture}
    \caption{
        Interactions generated by the decomposition into fast and slow modes.
        Dashed and solid lines represent the fields $\psi$ and $\tilde{\phi}$, respectively.
    }\label{figintstep}
\end{figure}

Let us now investigate how the interactions evolve by considering a single step of the \rg flow, focusing on the spectral tail beyond a certain scale $\Lambda_0$.
Assuming that we are in the vicinity of the Gaussian fixed point,\footnote{%
    Or in a regime analogous to the approximation discussed in \Cref{LocalFieldTheory} (see \Cref{eqrho}).
}
we restrict our attention to the quartic regime, where
\begin{equation}
    V(\Phi) \defeq \frac{u_4}{4!}\sum_i \phi_i^4.
\end{equation}
To a first approximation, the field can be decomposed in the eigenbasis of the \ecm as
\begin{equation}
    \phi_i
    \defeq
    \sum_{\lambda} u_i^{(\lambda)} \varphi_\lambda.
\end{equation}
We then integrate out the modes contained within the interval $I_\Lambda \defeq [\Lambda_0,\Lambda]$, decomposing the field as $\phi_i = \tilde{\phi}_i + \psi_i$, where $\psi_i$ represents the \dof within this window $I_\Lambda$.
The expansion of $(\tilde{\phi}_i+\psi_i)^4$ generates terms that can be classified into three distinct categories, shown in\Cref{figintstep}.

We can compute the effective Hamiltonian $\tilde{H}$ for the fields $\{\tilde{\phi}_i\}$ by integrating out the \uv modes within the window $I_\Lambda$.
We obtain:
\begin{equation}
    H[\Phi] \defeq H[\tilde{\Phi}]+H[\Psi]+U[\tilde{\Phi},\Psi],
\end{equation}
where the last term corresponds to the interaction between the fields $\Psi$ and $\tilde{\Phi}$, as illustrated in \Cref{figintstep}.
Notice that $U[\tilde{\Phi},\Psi]$ contains no kinetic contributions, since:
\begin{equation}
    \sum_{i,j} \psi_i C_{ij} \tilde{\phi}_j= \sum_\lambda \lambda \,\psi_\lambda \tilde{\phi}_\lambda =0.
\end{equation}
Hence, we get:
\begin{equation}
    \tilde{H} \defeq - \ln \left( \int \dd \psi \, e^{-H[\tilde{\Phi}]-H[\Psi]-U[\tilde{\Phi},\Psi]}\right)= H[\tilde{\Phi}] + (\text{fluctuations}),
\end{equation}
where fluctuations can be computed using perturbation theory, with the leading-order contributions depicted in \Cref{figinteff}.
These correspond to Feynman diagrams, where the internal lines (edges involved in loops) represent the propagators associated with the field $\Psi$, namely $C^{(\psi)}_{ij} \defeq \sum_{\lambda\in I_\Lambda} \lambda u_i^{(\lambda)} u_j^{(\lambda)}$.

\begin{figure}[t]
    \centering
    \begin{tikzpicture}[thick]
    \draw (-1.0, 1.0) -- (0.0, 0.0) -- (-1.0, -1.0);
    \draw (4.0, 1.0) -- (3.0, 0.0) -- (4.0, -1.0);
    \draw[dashed] (0.0, 0.0) .. controls (1.0, 1.0) and (2.0, 1.0) .. (3.0, 0.0);
    \draw[dashed] (0.0, 0.0) .. controls (1.0, -1.0) and (2.0, -1.0) .. (3.0, 0.0);
    \fill (0.0, 0.0) circle (2pt);
    \fill (3.0, 0.0) circle (2pt);

    \begin{scope}[xshift=7cm]
        \draw (-1.0, -1.0) -- (0.0, 0.0);
        \draw[dashed] (0.0, 0.5) circle (0.5);
        \draw[dashed] (0.0, 0.0) -- (2.5, 0.0);
        \draw (2.5, 0.0) -- (3.5, 1.0);
        \draw (2.5, 0.0) -- (3.75, 0.0);
        \draw (2.5, 0.0) -- (3.5, -1.0);
        \fill (0.0, 0.0) circle (2pt);
        \fill (2.5, 0.0) circle (2pt);
    \end{scope}
\end{tikzpicture}
    \caption{The two typical connected configurations up to order $u_4^2$.}\label{figinteff}
\end{figure}
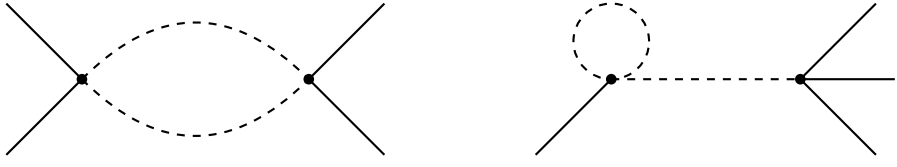

Explicitly, the two diagrams in \Cref{figinteff} can be evaluated, yielding contributions which, up to irrelevant numerical constants, read as follows:
\begin{align}
    V^{(1)} & \propto \sum_{i,j} \, \tilde{\phi}_i^2 \tilde{\phi}_j^2 (C^{(\psi)})^2_{ij}
    \\
    V^{(2)} & \propto \sum_{i,j} \, \tilde{\phi}_i \tilde{\phi}_j^3 (C^{(\psi)})_{jj}(C^{(\psi)})_{ij} + i \leftrightarrow j.
\end{align}

Assuming that the fields vary smoothly with the site indices \cite{RG0}, which is consistent with the standard Landau framework for ferromagnetism, we can expand $\tilde{\phi}_j = \tilde{\phi}_i + \mathcal{O}(\vert i-j \vert)$.
This yields:
\begin{align}
    V^{(1)} & \approx \sum_{i,j} \, \tilde{\phi}_i^4 (C^{(\psi)})^2_{ij}
    \\
    V^{(2)} & \approx \sum_{i,j} \, \tilde{\phi}_i^4  (C^{(\psi)})_{ii}(C^{(\psi)})_{ji} + i \leftrightarrow j.
\end{align}
Decomposing $C^{(\psi)}$ in its eigenbasis
\begin{equation}
    (C^{(\psi)})_{ji} \defeq \sum_{\lambda\in I_\Lambda} \lambda u_i^{(\lambda)} u_j^{(\lambda)},
\end{equation}
and using the fact that
\begin{equation}
    \sum_i u_i^{(\lambda)} u_i^{(\lambda^{\prime})}=\delta_{\lambda\lambda^{\prime}},
\end{equation}
we find for $V^{(1)}$:
\begin{equation}
    V^{(1)} \approx \sum_{i} \sum_{\lambda\in I_\Lambda} \, \tilde{\phi}_i^4 \, \lambda^2 u_i^{(\lambda)} u_i^{(\lambda)},
\end{equation}
Let us assume that the spectrum in this regime is dominated by noise.
Consequently, due to the delocalised nature of the eigenvectors, the expectation value of $(u_i^{(\lambda)})^2$ must scale as
\begin{equation}
    \mathds{E}[(u_i^{(\lambda)})^2] = \frac{1}{N}
\end{equation}
(see \eqref{PTdist}) for sufficiently large $N$.
This yields:
\begin{equation}
    V^{(1)} \approx \sum_{i} \, \tilde{\phi}_i^4  \frac{1}{N}\sum_{\lambda\in I_\Lambda} \, \lambda^2 \to \sum_{i} \, \tilde{\phi}_i^4 \int_{I_\Lambda} \, \dd \lambda\, \mu(\lambda)\lambda^2.
\end{equation}

\begin{figure}[t]
    \centering
    \begin{tikzpicture}[thick]
    \draw (-3., -1.25) -- (3., -1.25);
    \draw[dashed] (0.00, 0.00) circle (1.25);
\end{tikzpicture}
    \caption{One-loop contribution to the kinetic action correction.}\label{masstadpole}
\end{figure}
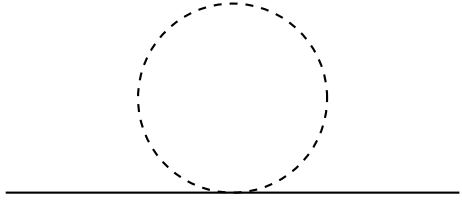

Similarly, the second interaction term $V^{(2)}$ involves the sum $\sum_j u_j^{(\lambda)} = N \langle u^{(\lambda)} \rangle \to 0$, so that $V^{(2)}$ vanishes.
It is worth noting that, regardless of this delocalisation argument, this term is not \onepi and therefore does not contribute to the effective \onepi vertex.
Consequently, $V^{(1)}$ provides the leading (one-loop) correction to the effective local quartic coupling constant in $\tilde{H}$.
Note, however, that in the localisation limit of the eigenvectors, we recover a $\tilde{\phi}_i^4$-type interaction, albeit confined to a more restricted index domain.
This shows that the class of local interactions is stable under the flow.
The correction to the kinetic term can be evaluated in an analogous manner.
This one-loop correction (at order $u_4$) is represented by the tadpole diagram in \Cref{masstadpole} and reads explicitly as follows:
\begin{equation}
    \delta K \propto \sum_i \, \tilde{\phi}_i^2 \, (C^{(\psi)})_{ii}.
\end{equation}
Using the same arguments as before, we get:
\begin{equation}
    \delta K \propto \sum_i \, \tilde{\phi}_i^2 \, \int_{I_\Lambda}\, \dd \lambda\,\mu(\lambda) \lambda,
\end{equation}
which is diagonal.
As a consequence, only the mass is shifted at leading order in perturbation theory, mirroring the behaviour of standard $\phi^4$ field theory.

The transition from the Ising model to the Landau theory, fully justified by the \rg framework \cite{Wilson1}, demonstrates that, within a quasi-local approximation of the propagator, an interaction of the form \eqref{ultralocalint} or \eqref{ultralocalint2}, which doubles the number of states owing to definition \eqref{GaussianModelBeyond}, generates precisely the same flow as a local $\phi_i^4$ interaction, and we recover the exact same one-loop corrections as those computed herein.
Since the non-perturbative formalism maps onto an effective one-loop calculation, we expect that, within the \lpa framework where the anomalous dimension is neglected, we should recover identical flow equations.
In standard Landau theory, the equivalence between the real-space interaction $\phi^4(\vec{x})$ and its momentum-space counterpart $\delta(\vec{p}_1+\vec{p}_2+\vec{p}_3+\vec{p}_4) \prod_{i=1}^4 \varphi(\vec{p}_i)$ is mediated by the ordinary Fourier transform.
Although such a transformation is lacking in the present context, the equivalence of the two models is guaranteed by the identity of their respective flows: the two forms generate the same Feynman diagrams, and then the same large $N$ flow equations by construction.
The choice of model \eqref{ultralocalint} is therefore justified by its simplicity and tractability, serving as a minimal representation of a more comprehensive universal theory.

\begin{remark}{Neglecting Derivative Interactions}{derivative}
    This approximation systematically neglects derivative interactions.
    While these couplings can play a role, their physical relevance must ultimately be assessed by the renormalisation group flow itself.
    As a first-order approximation, we shall therefore restrict our analysis to the local theory, where fields interact at the same point.
    To estimate the signal-to-noise boundary and validate this framework, we assume this local truncation to be sufficient.
    A more thorough investigation incorporating derivative terms lies beyond the scope of the \lpa considered in this review and is deferred to future work.
\end{remark}

\begin{remark}{Quartic Interaction and Eigenvector Delocalization}{quartic}
    Due to the delocalised nature of the eigenvectors in the bulk, we expect the overlap $\sum_i u_i^{(\lambda_1)}u_i^{(\lambda_2)}u_i^{(\lambda_3)}u_i^{(\lambda_4)}$ to be sharply peaked around configurations where the eigenvalues $\lambda_a$ are equal in pairs.
    In the vicinity of the Gaussian fixed point, the quartic term behaves as $\sum_i\phi_i^4 \sim \left(\sum_\lambda\varphi_\lambda^2\right)^2$, which structurally resembles an $O(N)$-invariant interaction.
    The local approximation naturally includes this limiting case, which directly mirrors, for instance, the pairing configurations $p_1=-p_2$ and $p_3=-p_4$ in the quartic sector of~\eqref{ultralocalint}.
\end{remark}

A heuristic justification for the local theory~\eqref{ultralocalint} can also be provided, based once again on a universality argument.
Consider, for instance, a binary spin system near the 3D Ising critical point.
In the equilibrium regime, by the \clt, the Gaussian theory provides an excellent approximation for the coarse-grained variables.
Consequently, at sufficiently large scales relative to the correlation length, the Landau framework must hold.
This corresponds to the Hamiltonian:
\begin{equation}
    H_{\text{L}}
    \defeq
    \frac{1}{2} \int \, \dd^D\vec{p}\, M(-\vec{p}\,)(\vec{p}\,^2+m^2)M(\vec{p})+ \frac{g}{4!} \int\, \prod_i \dd^D\vec{p}_i  \,\delta(\vec{p}_1+\vec{p}_2+\vec{p}_3+\vec{p}_4) \prod_{i=1}^4 M(\vec{p}_i).
\end{equation}
The momentum distribution $\rho_{\mathds{R}^D}((\vec{p})^{2}) \propto ((\vec{p})^{2})^{\frac{D-2}{2}}$, together with the locality principle adopted for the interactions, fixes the scaling dimension of the coupling constant $g$ to $[g]=4-D$, thereby determining its relevance or irrelevance.
In the present framework, we assume $D=3$, which implies that $g$ is a relevant coupling (for which the dimension is positive).

Since the distribution is essentially Gaussian for the \clt argument, the spectrum of the correlation matrix is expected to fall into the \mpdistr universality class.
This allows us to construct a theory of correlations based on the maximum entropy principle, as established above.
In the spectral tail, the \mpdistr distribution scales as $\rho(p^2) \sim (p^2)^{1/2}$, which matches precisely the spatial distribution $\rho_{\mathds{R}^3}((\vec{p})^{2})$ in Landau theory.
The primary objective of Landau theory is to capture these long-range correlations.
Notably, the physical scales are inverted within the spectrum: the spectral tail, which represents the \ir regime from the perspective of information relevance, corresponds to the \uv regime of the network, meaning that the most Gaussian region dictates the large-scale behaviour.
Although the two theories differ in general, they encode the same information.
In this specific asymptotic tail, the identity of the two momentum distributions implies that the interactions must share the same functional form.
At this scale, their \rg flows are identical.
Consequently, the interaction \eqref{ultralocalint} reproduces the exact same correlations as Landau theory in the three-dimensional spectral tail.

\section{Large-N Solution and Hartree Approximation}\label{sec:app1}
This appendix is directed to the interested data scientist who wants to better understand the physical phenomena behind the computation of the critical temperature in thermal statistical systems.
Here, we derive the logarithmic shift of the critical temperature in two dimensions within the Hartree approximation~\cite{mulet2007langevin,livi2017nonequilibrium}.
While this Gaussian reference captures the role of fluctuations in suppressing long-range order, it remains insufficient to reach the exact thermodynamic limit where the system size $L$ goes to infinity, for the reasons detailed in \Cref{part1}.

\subsection{Hartree Approximation and Linearisation}

Let $x \in \mathds{R}^d$ and consider a family of $N$ fields $\phi_i(x,t)$ whose evolution is governed by a Langevin-type equation:
\begin{equation}
    \frac{\dd}{\dd t} \phi_i(x,t)
    =
    \Delta \phi_i(x,t)
    -
    2 N \phi_i(x,t) V^\prime(\phi^2(x,t))
    +
    \eta_i(x,t),
\end{equation}
where $\Delta$ denotes the standard Laplacian on $\mathds{R}^d$, and $\eta_i(x,t)$ is a Gaussian white noise with zero mean and a correlation function given by
\begin{equation}
    \langle \eta_i(x,t) \eta_j(x^\prime,t^\prime) \rangle = 2 T \delta_{ij} \delta(x-x^\prime) \delta (t-t^\prime),
\end{equation}
satisfying the fluctuation-dissipation theorem with the convention that temperature enters through the noise amplitude.\footnote{%
    This convention differs from the simulation setup of \Cref{IsingCritical}, where temperature appears in the drift term, with unit-variance noise. Both choices are physically equivalent and yield the same equilibrium distribution $P(\phi) \propto e^{-H/T}$.%
}

Using $\phi^2 \defeq N^{-1} \sum_i \phi_i^2(x,t)$, we shall adopt a quartic potential of the form:
\begin{equation}
    V(\phi^2) \defeq \frac{1}{2}a \phi^2 + \frac{1}{4} b (\phi^2)^2,
\end{equation}
with $a \defeq (T-T_0)$ and $b \geq 0$.
When $b = 0$, the system becomes unstable at $T = T_0$, leading to an exponential divergence of the field for $T < T_0$.
The introduction of a non-zero interaction ($b > 0$) stabilises the dynamics, ensuring convergence toward the equilibrium states defined by $V^\prime = 0$.
However, this interaction also competes with the development of long-range order, where the system features a diverging spatial correlation length.
The physical critical temperature is thus suppressed to $T_c < T_0$.
Furthermore, when initialised at infinite temperature and suddenly quenched to $T < T_c$, the system exhibits an equilibration time that scales with $N$.

In the large-$N$ limit, the quantity $\phi^2$ is assumed to be sharply peaked around its expectation value $R^2$.
Moreover, for sufficiently large $t$, $R$ can be treated as time-independent.
The dynamic equation thus linearises, yielding in Fourier space:
\begin{equation}
    \frac{\dd}{\dd t} \phi_i(k,t) = -k^2\phi_i(k,t) - (a + b R^2) \phi_i(k,t) + \eta_i(k,t).
\end{equation}
The critical temperature $T_c$ is now $T_c = T_0 - b R^2 \le T_0$.
The solution to the previous equation becomes:
\begin{equation}
    \phi_i(k,t)
    =
    \phi_i(k,0)e^{-(k^2+a_{\text{eff}}) t}
    +
    \int_0^t \dd t^\prime \eta_i(k,t^\prime) e^{-(k^2+a_{\text{eff}}) (t-t^\prime)}.
\end{equation}

\subsection{Correlation Functions and Critical Temperature}

If we assume purely random initial conditions ($\phi_i(k,0) \sim \mathcal{N}(0,1)$), which represent a high-temperature phase ($T \gg 1$), and let the system go through an instantaneous quench at $t=0$ to a temperature $T < T_c$, the correlation function of the Fourier modes can be evaluated by performing a statistical average over the noise fluctuations:
\begin{equation}
    \langle \phi(k, t) \phi(k', t) \rangle
    =
    (2\pi)^d \delta(k+k')\, C(k, t),
\end{equation}
where
\begin{equation}
    C(k, t)
    =
    e^{-2(k^2+a_{\text{eff}})t} + T\, \frac{1 - e^{-2(k^2+a_{\text{eff}})t}}{k^2+a_{\text{eff}}}.
    \label{eq:correlationskk}
\end{equation}
For sufficiently large $t$, $C(k, t) \approx \frac{T}{k^2+a_{\text{eff}}}$, leading to the self-consistent equation:\footnote{%
    The memory of the initial conditions decays exponentially.
}

\begin{equation}
    a_{\text{eff}}
    =
    a + b T \int \frac{\dd^d k}{(2\pi)^d} \frac{1}{k^2+a_{\text{eff}}}.
\end{equation}
The critical temperature is precisely defined by $a_{\text{eff}} = 0$:
\begin{equation}
    0
    =
    a + b T_c \int \frac{\dd^d k}{(2\pi)^d} \frac{1}{k^2}
    =
    a + \frac{b T_c}{2\pi} \ln \left(\frac{M}{\mu}\right),
\end{equation}
where $M$ and $\mu$ are the \uv and \ir cut-offs, respectively, restricting the momentum scale to $\mu \leq \vert k \vert \leq M$.

The Hartree approximation, which extends standard mean-field theory, consists of applying this self-consistent relation to finite $N$, and specifically to $N = 1$.
For a discrete grid of size $L \times L$ ($L \in \mathds{N}$) with lattice spacing $\ell$, the first Brillouin zone is given by $[-\pi/\ell, \pi/\ell]^2$.
Consequently, the \uv cut-off is set to $M = \pi/\ell$, while the \ir cut-off is $\mu = 2\pi/(L \ell)$, yielding:
\begin{equation}
    T_c = \frac{T_0}{1 + \frac{b}{2\pi} \ln \left(\frac{L}{2}\right)}.
\end{equation}
Within the Hartree approximation, the critical temperature vanishes in the thermodynamic limit, reflecting the breakdown of this mean-field description for $d \leq 2$ (the exact two-dimensional Ising model, with $N = 1$, retains a finite critical temperature, as shown by Onsager).
However, the underlying physical mechanism remains valid: fluctuations hinder the development of long-range order, suppressing $T_c$ below its bare value $T_0$.
Taking the Fourier transform of the transient contribution to~\eqref{eq:correlationskk}, which dominates the early-time coarsening regime, we find that the real-space correlation function exhibits a spatial decay governed by $\exp \left(-\frac{\vert x-x^\prime \vert^2}{8 t} \right)$, with the correlation length growing as $t^{1/2}$.
Consequently, the system reaches equilibrium only on timescales proportional to the total number of grid sites, while the characteristic size of the spin domains grows as $t^{1/2}$.

To gain further insight, it is instructive to investigate the two-time correlation functions.
A straightforward calculation yields:
\begin{equation}
    \langle \phi(0,0) \phi(0,t) \rangle \sim e^{-a_{\text{eff}}t}\,.
\end{equation}
The relaxation time
\begin{equation}
    \tau \defeq \frac{1}{a_{\text{eff}}} \sim \frac{1}{T-T_c}
\end{equation}
diverges at the transition point, a phenomenon known as \emph{critical slowing down}.

\section{The Ohta-Jasnow-Kawasaki Approximation}\label{AppOJK}
An elementary approach to temporal correlations (the ageing effect) in the low-temperature phase reached by a quench below the critical temperature relies on the method proposed by Ohta, Jasnow, and Kawasaki (\ojk) \cite{Bray_1994}.
The idea rests on the observation that, in the coarsening regime at low temperature, the field $\phi(x,t)$ (with $x\in \mathds{R}^{D}$) takes essentially only two values, $\phi(x,t) = \pm M_0$, away from the domain walls.
Here, $M_0$ denotes the spontaneous magnetisation at the minimum of the potential well.
The \ojk approximation introduces a smooth auxiliary field $m(x,t)$ whose zero set marks the domain walls, such that\footnote{%
    At long times, in one dimension, the local profile of a wall is that of a soliton (a kink).
}
\begin{equation}
    \phi(x,t) = M_0\,\frac{m(x,t)}{|m(x,t)|} = M_0\,\operatorname{sign}\!\left(m(x,t)\right).
\end{equation}

An evolution equation for $m(x,t)$ follows from the Allen-Cahn equation satisfied by $\phi$.
In the sharp-interface limit, the Allen-Cahn dynamics reduces to mean-curvature flow, that is, the normal velocity $v_n$ of a domain wall is proportional to the local mean curvature $\kappa$.
Since domain walls are the zero set of $m$ (i.e. $\{ x \mid m(x,t) = 0 \}$), the unit outward normal is $\hat{n} = \nabla m / \|\nabla m\|$, and the mean curvature is $\kappa = \nabla \cdot \hat{n}$.
Using the level-set identity
\begin{equation}
    \kappa\,|\nabla m| = |\nabla m|\,\nabla\cdot\!\left(\frac{\nabla m}{|\nabla m|}\right) = \Delta m - \frac{\nabla m \cdot \nabla|\nabla m|}{|\nabla m|},
\end{equation}
the level-set evolution $\partial_t m = \kappa\,|\nabla m|$ separates into a linear Laplacian term and a non-linear term.
A second key ingredient of the \ojk approximation is to average the non-linear term over a region large compared with the interface thickness but small on the scale of the curvature radius, assuming that the local orientations of the domain walls are statistically independent.
Averaging over the orientations of the normal vectors reduces the non-linear term to a second Laplacian contribution, with a dimension-dependent coefficient $a(D)$.

A rescaling of the time variable by the factor $1 + a(D)$ absorbs this coefficient, leaving
\begin{equation}
    \frac{\partial m(x,t)}{\partial t} = \Delta m(x,t),
\end{equation}
where $\Delta$ denotes the standard Laplacian on $\mathds{R}^{D}$.
The linearity of such diffusion equation preserves the Gaussianity of $m$ at all later times.
The joint distribution of $m(x,t)$ and $m(x,t_w)$ is thus bivariate Gaussian and its heat-kernel form gives a closed expression for the two-point function.
Together these properties allow us to compute the temporal correlation that exhibits the ageing effect:
\begin{equation}
    \langle m(x,t)\,m(x,t_w)\rangle = \int \frac{\dd^{D}k}{(2\pi)^{D}}\,\langle m(k,0)\,m(-k,0)\rangle\,e^{-k^{2}(t+t_w)} = K \int \frac{\dd^{D}k}{(2\pi)^{D}}\,e^{-k^{2}(t+t_w)},
\end{equation}
where the first equality uses the Fourier representation of the diffusion propagator $m(k,t) = m(k,0)\,e^{-k^{2}t}$ and the second invokes a white-noise initial condition, namely that $m(x,0)$ is a centred Gaussian random field with covariance
\begin{equation}
    \langle m(x,0)\,m(x^{\prime},0)\rangle = K\,\delta(x-x^{\prime}).
\end{equation}
Performing the Gaussian integral in $D$ dimensions gives the position-space form
\begin{equation}
    \langle m(x,t)\,m(x,t_w)\rangle \propto (t+t_w)^{-D/2}.
\end{equation}

The normalised correlation coefficient $\gamma(t,t_w)$ between the Gaussian variables $m_1 = m(x,t)$ and $m_2 = m(x,t_w)$ is therefore
\begin{equation}
    \gamma(t,t_w) \equiv \frac{\langle m(t)\,m(t_w)\rangle}{\sqrt{\langle m^{2}(t)\rangle\,\langle m^{2}(t_w)\rangle}} = \left(\frac{4\,t\,t_w}{(t+t_w)^{2}}\right)^{D/4}.
\end{equation}

The temporal correlation of the original order parameter $\phi$ follows directly from the construction $\phi = M_0\,\operatorname{sign}(m)$:
\begin{equation}
    C(t,t_w) = \langle \phi(x,t)\,\phi(x,t_w)\rangle = M_0^{2}\,\bigl\langle \operatorname{sign}\!\left(m(x,t)\right)\,\operatorname{sign}\!\left(m(x,t_w)\right)\bigr\rangle.
\end{equation}
Let $m_1$ and $m_2$ be two centred Gaussian random variables with unit variance and correlation coefficient $\gamma$.
Their joint density is
\begin{equation}
    P(m_1,m_2;\gamma) = \frac{1}{2\pi\sqrt{1-\gamma^{2}}}\,\exp\!\left(-\frac{m_1^{2} + m_2^{2} - 2\gamma\,m_1 m_2}{2(1-\gamma^{2})}\right),
\end{equation}
and the expectation of interest is
\begin{equation}
    I(\gamma) = \int_{-\infty}^{+\infty} \dd m_1 \dd m_2\, \int_{-\infty}^{+\infty} \operatorname{sign}(m_1)\,\operatorname{sign}(m_2)\,P(m_1,m_2;\gamma).
\end{equation}
Applying Price's theorem to the bivariate Gaussian gives
\begin{equation}
    \langle \operatorname{sign}(m_1)\,\operatorname{sign}(m_2)\rangle = \frac{2}{\pi}\,\arcsin(\gamma),
\end{equation}
which yields
\begin{equation}
    C(t,t_w) = \frac{2M_0^{2}}{\pi}\,\arcsin\!\left[\left(\frac{4\,t\,t_w}{(t+t_w)^{2}}\right)^{D/4}\right],
\end{equation}
and, in the ageing limit $t \gg t_w$,
\begin{equation}
    C(t,t_w) \approx \frac{2M_0^{2}}{\pi}\left(\frac{4\,t_w}{t}\right)^{D/4}.
\end{equation}
The system retains a non-trivial memory of the waiting time $t_w$ (its age), with a power-law decay replacing the exponential relaxation characteristic of an equilibrium correlation function.

\section{Recovering the Ising Model from Model~A}\label{sec:app2}
This appendix establishes the correspondence between the continuous Langevin dynamics of Model~A and the equilibrium partition function of the two-dimensional Ising model.
Starting from the discrete gradient formulation of the kinetic action, we obtain the Boltzmann distribution of the discretised theory and show that the large-$\varepsilon$ limit of the on-site potential reduces it to the standard Ising model.

\subsection{Discretisation and Partition Function}

To establish the correspondence with the Ising model, we use the discrete Laplacian~\eqref{DiscreteLaplacian} and introduce the discrete gradient $(\nabla_{\text{dis}} \varphi)_{ij}$ at the site $(i,j)$, with components
\begin{equation}
    (\nabla_{\text{dis}} \varphi)_{ij}
    =
    \begin{pmatrix}
        \varphi_{i+1,j}-\varphi_{ij} \\
        \varphi_{i,j+1}-\varphi_{ij}
    \end{pmatrix}.
\end{equation}
Define the quadratic form
\begin{equation}
    I[\varphi]
    \defeq
    \frac{1}{2}\, \sum_{i,j} [(\varphi_{i+1,j}-\varphi_{ij})^2+(\varphi_{i,j+1}-\varphi_{ij})^2],
    \label{eq:Iphi}
\end{equation}
in which the lattice spacing has been set to $\ell = 1$ (see~\eqref{DiscreteLaplacian}).
A direct computation gives
\begin{equation}
    \frac{\partial I}{\partial \varphi_{ij}}=-(\Delta_{\text{dis}}\varphi)_{ij},
    \label{eq:dIdphi}
\end{equation}
so that $I[\varphi]$ is the discrete action whose Euler--Lagrange equation reproduces the linear part of the Model~A dynamics.

We take the explicit Euler--Maruyama discretisation of the kinetic equation
\begin{equation}
    \varphi_{i,j}^{n+1}
    =
    \varphi_{i,j}^{n}
    +
    \Delta t\left[\frac{J}{T} (\Delta_{\text{dis}}\varphi)^{n}_{i,j} - V^{\prime}(\varphi_{i,j}^{n})\right]+\sqrt{2\Delta t}\,\xi_{i,j}^{n},
    \label{eq:kinetics}
\end{equation}
where $J>0$ is the ferromagnetic coupling and $\xi_{i,j}^{n} \sim \mathcal{N}(0,1)$ is a standard Gaussian noise, independent across sites and time steps.
For the on-site potential $V$ of~\eqref{VIsing}, namely
\begin{equation}
    V(\varphi) = -\frac{1}{2}\,\varepsilon\,\varphi^{2} + \frac{1}{4}\,\varepsilon\,\varphi^{4},
\end{equation}
the drift coefficient in~\eqref{eq:kinetics} is minus the gradient of the Hamiltonian
\begin{equation}
    H[\varphi] = J\,I[\varphi] + T\sum_{i,j} V(\varphi_{ij}),
\end{equation}
so the discrete Langevin equation admits the Boltzmann equilibrium distribution $P_{\text{eq}}(\varphi) \propto e^{-H[\varphi]/T}$.

In the continuous-time limit, the corresponding equilibrium partition function reduces to
\begin{equation}
    Z^{(\text{eq})}_{\Delta t\to 0}
    =
    \int \prod_{i,j}\dd \varphi_{ij}\,
    e^{-\frac{J}{T}I[\varphi]}\, e^{-V[\varphi]},
    \label{eq:Zeq}
\end{equation}
where $V[\varphi] \equiv \sum_{i,j} V(\varphi_{ij})$.

\subsection{Limits and Reduction to the Ising Model}

The cubic $V^{\prime}(\varphi) = \varepsilon\varphi(\varphi^{2}-1)$ has three critical points, $\varphi = 0, \pm 1$, of which $\varphi = 0$ is a local maximum and $\varphi = \pm 1$ are the two symmetric minima.
Expanding $V$ around either minimum gives
\begin{equation}
    V[\varphi] = V[\varphi_{0}] + \frac{1}{2}\,V^{\prime\prime}(\varphi_{0})\,(\varphi - \varphi_{0})^{2} + \mathcal{O}\!\left((\varphi - \varphi_{0})^{3}\right),
    \label{eq:Vexpand}
\end{equation}
with
\begin{equation}
    V(\pm 1) = -\frac{\varepsilon}{4},
    \qquad
    V^{\prime\prime}(\pm 1) = 2\,\varepsilon.
    \label{eq:Vminima}
\end{equation}

In the limit $\varepsilon \to \infty$, the on-site factor $e^{-V(\varphi_{ij})}$ concentrates at the two minima $\varphi_{ij} = \pm 1$ on every site (behaving as a Dirac delta function), so the integral in~\eqref{eq:Zeq} is well approximated by the sum
\begin{equation}
    Z^{(\text{eq})}_{\Delta t\to 0}
    \;\approx\;
    \sum_{\{\varphi_{ij}=\pm 1\}}\, e^{-\frac{J}{T}I[\varphi]}.
    \label{eq:Zsum}
\end{equation}
For any configuration in $\{\varphi_{ij}=\pm 1\}$ we have the algebraic identity $(\varphi_{I} - \varphi_{J})^{2} = 2(1 - \varphi_{I}\varphi_{J})$, and therefore
\begin{equation}
    I[\varphi]
    =
    2 N^{2} - \sum_{\langle I,J \rangle} \varphi_{I}\varphi_{J},
    \label{eq:ItoIsing}
\end{equation}
where the capital letters $I, J$ denote lattice sites with coordinates $I = (i,j)$ on the $N \times N$ grid, and $\langle I,J \rangle$ indicates that the sum runs over nearest-neighbour pairs, each counted once ($2 N^{2}$ pairs in total under periodic boundary conditions).
Substituting~\eqref{eq:ItoIsing} into~\eqref{eq:Zsum} yields
\begin{equation}
    Z^{(\text{eq})}_{\Delta t\to 0}
    =
    e^{-\frac{2 N^{2} J}{T}}\;
    \sum_{\{\varphi_{I}=\pm 1\}}\, e^{\frac{J}{T}\sum_{\langle I,J \rangle}\, \varphi_{I}\varphi_{J}}.
    \label{eq:ZIsing}
\end{equation}
The prefactor $e^{-2 N^{2} J / T}$ depends only on $N$ and $T$ and is independent of the spin configuration, so~\eqref{eq:ZIsing} coincides with the standard 2D Ising partition function up to an irrelevant multiplicative factor.

\section{Short Review of the Spiked Matrix Models}\label{AppD}
In this appendix we briefly review the paradigm of the spiked matrix model.
In particular, we provide only the statements and definitions needed for the understanding of the subject, and the proof of the \bbp phase transition is only sketched.
Readers interested in this topic may consult the given references.

The spiked matrix model~\cite{aubin2019spiked,spike1,Bouchaud3} is used to study the eigenvalue distribution of a random matrix perturbed by a low-rank signal.
The latter is usually defined by isolated eigenvectors planted into a random matrix noise.
Despite its apparent simplicity, this framework provides a highly valuable statistical model for \pca.\footnote{%
    Though often presented as a simple dimensionality-reduction algorithm, \pca is actually an \emph{eigenvalue problem} in mathematics.
    The objective of \pca is, first of all, finding the largest eigenvalue in a random matrix background.
    As a by-product of the technique, it also enables the reduction of dimensionality by ranking the principal components by their variance.
}
In their seminal paper~\cite{BBP}, Baik, Ben Arous, and P\'ech\'e showed that the spiked sample-covariance model (a special case of the Wishart ensemble) exhibits a sharp phase transition when the signal strength, represented by the spike, reaches a critical threshold.
The same phenomenon, with the same universal mechanism, arises for spiked Wigner matrices and was established by F\'eral and Kargin, and by P\'ech\'e~\cite{Feral2007,Peche2006}.
Detection and recovery of the planted signal are only possible beyond this point.
Below this critical value, the largest eigenvalue remains embedded within the bulk of delocalised eigenvectors.

\subsection{Spiked Wigner Model and Phase Transition}

This section starts with a brief review of single-spike matrix models and their underlying phase transition, focusing on the simplified case of a \emph{Gaussian Wigner model} perturbed by a single deterministic vector.
More concretely, we model the data under investigation as an $N \times N$ real random matrix $Q$, whose entries split into a sum of two contributions:
\begin{equation}
    Q_{ij}= \beta u_i u_j + \frac{1}{\sqrt{N}} M_{ij},\label{spikePCA}
\end{equation}
where $u=(u_1,\dots, u_N) \in \mathds{R}^N$ with $\lvert u \rvert = 1$ is a normalised vector, and $M$ is an $N \times N$ real symmetric matrix whose entries are randomly distributed according to the \emph{Gaussian Orthogonal Ensemble} (\goe).
Specifically, the off-diagonal entries follow $\mathcal{N}(0,1)$, while the diagonal entries follow $\mathcal{N}(0,2)$.
The parameter $\beta\in \mathds{R}^+$ quantifies the signal strength.

\begin{figure}[t]
    \centering
    \includegraphics[width=0.7\textwidth]{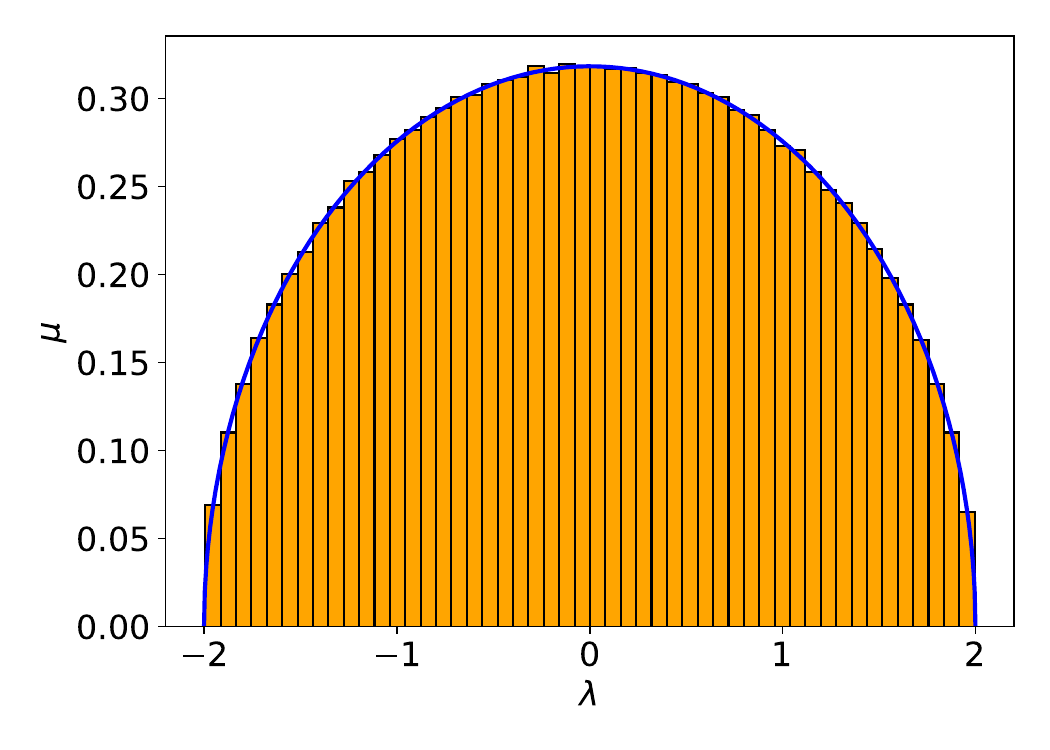}
    \caption{
        Illustration of the convergence toward Wigner's law.
        The histogram shows the eigenvalue distribution for a Wigner matrix of size $10^4$.
        The blue line materialises the limiting Wigner semicircle ($\mu_W$) law.
    }\label{WignrLaw}
\end{figure}

Let $\{\lambda_\mu\}_{\mu=1}^{N}$ denote the set of eigenvalues of $Q$.
The case $\beta=0$ corresponds to pure white noise, and in the limit $N\to \infty$, the empirical eigenvalue distribution $\mu_E$ (with $\lambda\in \mathds{R}$) is given by:
\begin{equation}
    \mu_E(\lambda)=\frac{1}{N} \sum_{\mu=1}^{N} \delta(\lambda-\lambda_\mu).
\end{equation}
The distribution converges weakly toward the Wigner semicircle law $\mu_W(\lambda)$:
\begin{equation}
    \mu_W(\lambda)=\frac{1}{2\pi}\sqrt{4-\lambda^2} \, \mathbf{1}_{[-2,2]}(\lambda),
\end{equation}
where the indicator function is
\begin{equation}
    \mathbf{1}_{[a,b]}(\lambda)
    =
    \begin{cases}
        1 & \text{if} \;\; \lambda \in [a, b], \\
        0 & \text{otherwise}.
    \end{cases}
\end{equation}
\Cref{WignrLaw} shows an example of the convergence of the empirical eigenvalue distribution toward the Wigner semicircle law.

In general, for $\beta \neq 0$, we have the following theorem for the spiked Wigner model~\cite{Feral2007,Peche2006,Peche2006Erratum,spike1}:
\begin{theorem}{Largest Eigenvalue of a Spiked Wigner Matrix}{Th1}
    Let $Q=\beta uu^{\mathsf{T}}+M/\sqrt{N}$ be a spiked Wigner matrix with $u^{\mathsf{T}}u=1$ and $M$ a \goe{} matrix.
    We have:
    \begin{itemize}
        \item For $\beta \leq 1$, the largest eigenvalue of $Q$ converges almost surely to $2$ as $N\to \infty$, and the fluctuations are of order $N^{-2/3}$ with the Tracy--Widom limiting distribution.
        \item For $\beta>1$, the largest eigenvalue converges almost surely to $\beta+1/\beta >2$ as $N\to \infty$, and the fluctuations are of order $N^{-1/2}$ with a Gaussian limiting distribution.
    \end{itemize}
\end{theorem}
\begin{proof}
    The following is a sketch of the proof, which nonetheless highlights the main ingredients.
    Consider a Wigner matrix $M$ perturbed by a single spike, resulting in a matrix $Q$ defined in equation~\eqref{spikePCA}, which we rewrite as
    \begin{equation}
        Q_{ij}= \beta u_i u_j + X_{ij},
    \end{equation}
    where $X_{ij} \defeq \frac{1}{\sqrt{N}} M_{ij}$.
    We assume $N$ to be sufficiently large.
    Let $\lambda$ be an eigenvalue of $Q$ lying outside the support of the spectrum of $X$, which in the limit $N \to \infty$ is the interval $[-2, 2]$.
    Let $u^{(\lambda)}$ be the associated eigenvector.
    The eigenvalue equation reads
    \begin{equation}
        \lambda u^{(\lambda)}_i
        = \beta u_i \, (u \cdot u^{(\lambda)}) + \sum_{j} X_{ij}\, u^{(\lambda)}_j,
    \end{equation}
    which can be rewritten as
    \begin{equation}
        \sum_{j} \left(\lambda \delta_{ij}-X_{ij}\right) u^{(\lambda)}_j
        = \beta u_i \, (u \cdot u^{(\lambda)}).
    \end{equation}
    Assuming that $\lambda$ lies outside the spectrum of $X$, the matrix with entries $\lambda \delta_{ij}-X_{ij}$ is invertible, so we may introduce the resolvent
    \begin{equation}
        G_X(\lambda) \defeq \left(\lambda\mathds{1} - X\right)^{-1},
    \end{equation}
    and obtain
    \begin{equation}
        \frac{u_i^{(\lambda)}}{\beta}= \sum_{j}(G_X(\lambda))_{ij}\, u_j \, (u \cdot u^{(\lambda)}).
    \end{equation}
    Taking the inner product on both sides with $u$ and simplifying, it follows that
    \begin{equation}
        \frac{1}{\beta}=\sum_{i,j} u_i\, (G_X(\lambda))_{ij}\, u_j\,.
    \end{equation}
    Let $\{v^{(\lambda_\mu)}\}$ denote the eigenvectors of $G_X(\lambda)$ with corresponding eigenvalues $(\lambda-\lambda_\mu)^{-1}$, where $\lambda_\mu$ runs over the eigenvalues of $X$.
    Expanding the resolvent in the eigenbasis of $X$ yields
    \begin{equation}
        \frac{1}{\beta}=\sum_{\mu} \left(\sum_{i} u_i v^{(\lambda_\mu)}_i\right)^2 (\lambda-\lambda_\mu)^{-1}.
    \end{equation}
    In the large-$N$ limit, the resolvent is self-averaging: the random matrix $G_X(\lambda)$ concentrates around its expectation.
    Taking the expectation, and using the fact that the eigenvectors of $X$ are delocalised, with components following the Porter--Thomas law of \eqref{PTdist}, we find
    \begin{equation}
        \mathds{E}\left[\left(\sum_{i} u_i v^{(\lambda_\mu)}_i\right)^2\right]
        = \sum_i u_i^2 \, \mathds{E}\left[\left(v^{(\lambda_\mu)}_i\right)^2\right]
        \simeq \frac{1}{N}\sum_i u_i^2 = \frac{1}{N}\,.
    \end{equation}
    Combined with the convergence of the empirical spectral measure of $X$ toward the Wigner semicircle, this yields the master equation
    \begin{equation}
        \frac{1}{\beta}=\mathfrak{g}(\lambda).\label{masterEq}
    \end{equation}
    The function $\mathfrak{g}$ is the limiting Stieltjes transform of $X$, defined for a complex variable $z$ outside the support of the spectrum of $X$ as
    \begin{equation}
        \mathfrak{g}(z) \defeq \lim_{N\to\infty} \frac{1}{N}\Tr \frac{1}{z\mathds{1}-X}
        = \frac{z-\sqrt{z^2-4}}{2}.
    \end{equation}
    The explicit form follows from the self-consistency relation $g = (z-g)^{-1}$ satisfied by the resolvent of the Gaussian matrix model, namely $g^2 - zg + 1 = 0$, whose branch vanishing as $1/z$ for $\lvert z\rvert\to\infty$ gives the formula above.
    Substituting into \eqref{masterEq} and isolating the square root, we find
    \begin{equation}
        \sqrt{\lambda^2-4} = \lambda - \frac{2}{\beta}.
    \end{equation}
    Squaring both sides (a step valid for $\lambda>2$, outside the support of the spectrum of $X$, and $\beta>0$) and simplifying gives the linear equation
    \begin{equation}
        \lambda^2 - \frac{4\lambda}{\beta} + \frac{4}{\beta^2} = \lambda^2 - 4,
    \end{equation}
    whose unique solution is
    \begin{equation}
        \lambda=\beta+\frac{1}{\beta}.
    \end{equation}
    The condition on $\beta$ can be read directly from the properties of the Stieltjes transform in \eqref{masterEq}.
    On the real axis and for $z > 2$, $\mathfrak{g}$ is strictly decreasing, with $\mathfrak{g}(2)=1$ and $\lim_{z\to\infty}\mathfrak{g}(z)=0$.
    A solution of \eqref{masterEq} therefore exists if and only if $1/\beta \leq 1$, i.e.\ $\beta \geq 1$.
    The threshold is reached at $\beta^* = 1$, where $\lambda = \lambda_+ = 2$, located exactly at the upper edge of the Wigner bulk.
    For $\beta < 1$, the equation has no solution outside the bulk, and the spike is entirely absorbed by the spectrum.
\end{proof}
From the point of view of signal detection, this result means that:
\begin{enumerate}
    \item As soon as $\beta>1$, the signal can be detected and recovered, and the leading eigenvector provides an asymptotically efficient estimator.
          \pca is in this sense optimal in the low-rank regime.
    \item For $\beta<1$, signal detection and recovery are impossible in the large-$N$ limit unless additional structural assumptions on the prior are imposed.
          In this case \pca ceases to be optimal and Bayesian estimators dominate.
    \item For the critical value $\beta=1$, consistency tests exist to distinguish the spike from the bulk, see \cite{mei2018landscape}.
\end{enumerate}

While we focused on the Gaussian ensemble, the previous statement holds for a wide class of random symmetric matrices as the matrix size approaches infinity, provided that the entries are \iid random variables with zero mean and finite variance.
This defines the real symmetric Wigner ensemble.
More precisely, following \cite{spike1,Bouchaud3}:
\begin{definition}{Wigner Ensemble}{Wig}
    The real symmetric Wigner ensemble is a fundamental class of statistical models for real symmetric $N\times N$ random matrices.
    A matrix $M$ belongs to this ensemble if and only if:
    \begin{itemize}
        \item The upper-triangular entries $M_{ij}$ with $1\leq i \leq j \leq N$ are independent and identically distributed with zero mean ($\mathds{E}(M_{ij})=0$ for $i<j$, $\mathds{E}(M_{ii})=0$).
        \item The off-diagonal entries have unit variance ($\mathds{E}(M_{12}^2)=1$).
        \item All moments of the distribution of a single entry are finite ($\mathds{E}(\lvert M_{ij}\rvert^k)<\infty$ for all $k\in\mathds{N}$),
    \end{itemize}
    and $M$ is symmetric, $M^{\mathsf{T}} = M$.
    The lower-triangular entries are then determined by $M_{ji} = M_{ij}$.
\end{definition}
Wigner's law emerges from the universality of high-dimensional random matrices, analogous to the \clt for scalar random variables.
We state the full result below~\cite{spike1,Wigner,Krajewski1}:
\begin{theorem}{Wigner's Law}{Wigner}
    Let $Y=M/\sqrt{N}$ with $M$ a real symmetric $N\times N$ matrix belonging to the Wigner ensemble.
    Then, as $N\to \infty$, the empirical eigenvalue distribution $\mu_E$ of $Y$ converges weakly to the Wigner semicircle law $\mu_W$.
\end{theorem}
The Wigner semicircle law represents a universal property of matrix ensembles, as it is virtually independent of the specific probability measure of the entries.

\subsection{Spiked Wishart Ensemble and Universality}

Another important and universal result for random matrices concerns the \mpdistr distribution.
This law describes the limiting spectral density of the white \emph{Wishart} ensemble, defined for an $N\times P$ rectangular matrix $X$ with \iid real entries.
Following the standard nomenclature of \rmt, the corresponding $P\times P$ sample-covariance matrix $Y = X^{\mathsf{T}}X/N$ belongs to the Laguerre Orthogonal Ensemble, i.e.\ the real symmetric white Wishart ensemble:
\begin{definition}{White Real Wishart Ensemble}{Wis}
    The white real Wishart ensemble is the set of statistical models for positive-definite $P\times P$ matrices of the form
    \begin{equation}
        Y=\frac{1}{N} X^{\mathsf{T}} X,
    \end{equation}
    where $X$ is an $N\times P$ matrix with independent and identically distributed real entries of zero mean and finite variance (with all higher moments assumed finite as well).

    Non-white Wishart matrices can be defined more generally as in~\cite{Bouchaud3}:
    \begin{equation}
        \mathds{E}(X_{ij}X_{kl})=C_{ik}\delta_{jl},
    \end{equation}
    for some covariance matrix $C$ with zero mean and finite variance $\sigma^2$, where the Kronecker delta enforces column-wise decorrelation.
\end{definition}
Wishart matrices also possess the property of converging toward a deterministic limiting distribution, namely the \mpdistr law (see \Cref{thm:thMP}).

\Cref{thm:thMP} generalises to other universality classes, and for $N, P \to \infty$ the position of the \emph{outlier} in the density spectrum is given by~\eqref{masterEq}.
For the spiked Wishart ensemble, inverting the closed form of the \mpdistr total Stieltjes transform
\begin{equation}
    \mathfrak{g}_{\mathrm{MP}}(z) = \frac{z - \sigma^2(1-q) - \sqrt{(z - \lambda_+)(z - \lambda_-)}}{2 \sigma^2 q z}
\end{equation}
gives, after isolating the square root and squaring (a step valid for $\lambda > \lambda_+$ where the square root is real and positive),
\begin{equation}
    \lambda = \frac{\beta\left(\beta + \sigma^2(1-q)\right)}{\beta - q\sigma^2}.
\end{equation}
The detection threshold is obtained by setting $\lambda = \lambda_+$ in the master equation~\eqref{masterEq}, giving the threshold $\beta^* = \sigma^2(\sqrt{q} + q)$.
For $\beta \leq \beta^*$ the spike is absorbed by the bulk.\footnote{%
    In the convention of the appendix (sample spike, $C = X^{\mathsf{T}} X/N + \beta v v^{\mathsf{T}}$), this formula differs from the original \bbp (population spike) result~\cite{BBP,Peche2006,Feral2007,Peche2006Erratum} by the substitution $q \mapsto 1/q$ in the last term, reflecting the difference between the population spike model $\Sigma = \mathds{1} + \theta v v^\mathsf{T}$ sampled via $S = Y \Sigma Y^\mathsf{T}/n$ and the sample spike model used here.
}


The existence of universality theorems explains why the mathematical tools of random matrix theory are so powerful in practice for \pca.
Universality implies that the statistical features of a dataset do not depend on the specific details of its underlying distribution, but rather on the finiteness of its moments.
Universal distributions fit empirical noisy spectra closely in high-dimensional settings.
Real-world data are inherently noisy, and the components of eigenvectors associated with universal distributions are \emph{delocalised} and possess maximum entropy density.\footnote{%
    For Gaussian ensembles, completely delocalised eigenvectors satisfy the concentration bound: with high probability, none of their entries significantly exceeds the typical scale of $\sim \sqrt{\log N/N}$.
    Specifically, for any eigenvector $u_k$ in the eigenbasis $\mathcal{B}=\{u_1,\dots,u_N\}$, the infinity norm satisfies the tail estimate
    \begin{equation*}
        \mathds{P}\left(\lVert u_k \rVert_{\infty}\geq \frac{c\sqrt{\log N}}{\sqrt{N}} \right) \leq N^{-D(c)},
    \end{equation*}
    with $D(c)=c^2/2-1$ for any $c>\sqrt{2}$ (see \cite{ErdosSchleinYau2009AnnProb,ErdosSchleinYau2009CMP}).
}

Consequently, they carry no structural information, providing the ideal mathematical representation of pure noise.
In contrast, as is the case for the single-spike matrix model discussed above, structural information is encoded within \emph{localised eigenvectors}, whose weight concentrates on a small subset of components.
While universal bulk distributions model purely noisy data, these localised states capture the underlying signal.
Historically, this universal behaviour was first highlighted in nuclear physics through the seminal works of Wigner, Dyson, Gaudin, Mehta, and others~\cite{Meh2004}.
These pioneering studies demonstrated the universality of the \emph{Wigner surmise} for systems containing many interacting particles governed by partially unknown laws.
Since then, universality has been observed across almost all domains of science, ranging from nuclear physics to biology, chemistry, and economics.

\end{document}